\documentclass[11pt,a4paper,twoside]{report}
\usepackage[a4paper, twoside, bindingoffset=1.5cm, inner=2.5cm, outer=2.5cm,
	top=3.25cm, bottom=3.25cm, headsep=1cm]{geometry}
\usepackage{tikz,pgfplots,fp}
\usetikzlibrary{positioning, decorations.markings, patterns, fit}
\usepackage{CJKutf8}
\usepackage{indentfirst}
\usepackage{rotating}
\usepackage{xspace}
\usepackage{makeidx}
\usepackage[titletoc]{appendix}
\usepackage[style=ieee-alphabetic, backend=bibtex]{biblatex}
\usepackage{amsmath,amssymb,amsthm,thmtools,amsfonts,mathtools,stmaryrd,nccmath}
\usepackage{accents,bm,bbm}
\usepackage{algorithm, algpseudocode}

	\algnewcommand\algorithmicforeach{\textbf{for each}}
	\algdef{S}[FOR]{ForEach}[1]{\algorithmicforeach\ #1\ \algorithmicdo}
	\algdef{SE}[DOWHILE]{Do}{DoWhile}{\algorithmicdo}[1]{\algorithmicwhile\ #1}
	\makeatletter 
	 
	\@addtoreset{algorithm}{chapter} 
	\makeatother
\usepackage{enumitem,subcaption}
\usepackage{tabularx, xltabular, multirow}
	\keepXColumns
	\newcolumntype{R}{X}
	\newcolumntype{L}{>{\hsize=.3\hsize}X}
	\newcolumntype{Y}{>{\centering\arraybackslash}X}
	\newcolumntype{Z}{>{\hsize=.5\hsize}X}
\makeindex
	\renewcommand\listfigurename{List of Figures and Tables}
\newlength{\somelength}
\makeatletter
\def\degree#1{\gdef\@degree{#1}}
\def\programme#1{\gdef\@programme{#1}}
\def\supervisor#1{\gdef\@supervisor{#1}}
\def\institute#1{\gdef\@institute{#1}}
\def\authorFirstName#1{\gdef\@authorFirstName{#1}}
\def\authorLastName#1{\gdef\@authorLastName{#1}}
\def\maketitle{
	\thispagestyle{empty}
	\vspace*{1in}
	\begin{center}
	\Huge\textbf{\@title}
	\vskip 1in
	\author{\@authorFirstName\ \@authorLastName}
	\Large\textbf{\@authorLastName, \@authorFirstName}
	\vskip 1in
	\large
	A Thesis Submitted in Partial Fulfillment \\
	of the Requirements for the Degree of \\
	\@degree \\
	in \\
	\@programme \\
	\vskip 1in
	\@institute\\
	\@date
	\end{center}
	\vfill
	\newpage\thispagestyle{empty}
	\noindent\begin{xltabular}{\textwidth}{XZ}
	This work is licensed under a Creative Commons ``Attribution 4.0 International'' license.
	& \centering\raisebox{-.5\height}{\includegraphics[width=.3333\textwidth]{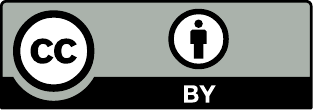}}
	\end{xltabular}
    \newpage}
\makeatother
\renewcommand{\mathbf}[1]{{\bm{#1}}}
\makeatletter
\def\munderbar#1{\underline{\sbox\tw@{$#1$}\dp\tw@\z@\box\tw@}}
\makeatother
\declaretheorem[name=Theorem, numberwithin=chapter]{theorem}
\declaretheorem[name=Lemma, sibling=theorem]{lemma}
\declaretheorem[name=Proposition, sibling=theorem]{proposition}
\declaretheorem[name=Corollary, sibling=theorem]{corollary}

\declaretheorem[name=Definition, sibling=theorem, style=definition, qed=\qedsymbol]{definition}
\declaretheorem[name=Example, sibling=theorem, style=definition, qed=\qedsymbol]{example}
\declaretheorem[name=Postulate, style=definition]{postulate}
\declaretheorem[name=Remark, sibling=theorem, style=remark]{remark}
\DeclareMathSymbol{\Alpha}{\mathalpha}{operators}{"41}
\DeclareMathSymbol{\Beta}{\mathalpha}{operators}{"42}
\DeclareMathSymbol{\Epsilon}{\mathalpha}{operators}{"45}
\DeclareMathSymbol{\Zeta}{\mathalpha}{operators}{"5A}
\DeclareMathSymbol{\Eta}{\mathalpha}{operators}{"48}
\DeclareMathSymbol{\Iota}{\mathalpha}{operators}{"49}
\DeclareMathSymbol{\Kappa}{\mathalpha}{operators}{"4B}
\DeclareMathSymbol{\Mu}{\mathalpha}{operators}{"4D}
\DeclareMathSymbol{\Nu}{\mathalpha}{operators}{"4E}
\DeclareMathSymbol{\Omicron}{\mathalpha}{operators}{"4F}
\DeclareMathSymbol{\Rho}{\mathalpha}{operators}{"50}
\DeclareMathSymbol{\Tau}{\mathalpha}{operators}{"54}
\DeclareMathSymbol{\Chi}{\mathalpha}{operators}{"58}
\DeclareMathSymbol{\omicron}{\mathord}{letters}{"6F}

\newcommand{\Integers}{\mathbb{Z}}
\newcommand{\Rationals}{\mathbb{Q}}
\newcommand{\Reals}{\mathbb{R}}
\newcommand{\Complex}{\mathbb{C}}
\newcommand{\prob}{\mathrm{P}}
\DeclareFontFamily{U} {MnSymbolA}{}
\DeclareFontShape{U}{MnSymbolA}{m}{n}{
  <-6> MnSymbolA5
  <6-7> MnSymbolA6
  <7-8> MnSymbolA7
  <8-9> MnSymbolA8
  <9-10> MnSymbolA9
  <10-12> MnSymbolA10
  <12-> MnSymbolA12}{}
\DeclareFontShape{U}{MnSymbolA}{b}{n}{
  <-6> MnSymbolA-Bold5
  <6-7> MnSymbolA-Bold6
  <7-8> MnSymbolA-Bold7
  <8-9> MnSymbolA-Bold8
  <9-10> MnSymbolA-Bold9
  <10-12> MnSymbolA-Bold10
  <12-> MnSymbolA-Bold12}{}
\DeclareSymbolFont{MnSyA} {U} {MnSymbolA}{m}{n}
\DeclareMathSymbol{\medslash}{\mathrel}{MnSyA}{210}
\newcommand{\boldtheta}{{\boldsymbol\theta}}
\newcommand{\boldeta}{{\boldsymbol\eta}}
\newcommand{\boldtau}{{\boldsymbol{T}}}
\DeclareMathOperator{\tr}{tr}
\makeatletter
\DeclareRobustCommand{\bigotimes}{
	\mathop{\vphantom{\sum}\mathpalette\bigotimes@\relax}\slimits@}
\renewcommand{\bigotimes@}[2]{\vcenter{\sbox\z@{$#1\sum$}
	\hbox{\resizebox{.9\dimexpr\ht\z@+\dp\z@}{!}{$\m@th\otimes$}}}}
\makeatother
\DeclareMathOperator{\Tensor}{\bigotimes}
\DeclareMathOperator{\st}{s.t.}
\DeclareMathOperator{\wt}{wt} % weight
\DeclareMathOperator{\var}{Var} % Variance
\DeclareMathOperator{\support}{Supp}
\DeclareMathOperator{\diag}{diag}
\DeclareMathOperator{\rank}{rank}
\DeclareMathOperator{\perm}{perm}
\makeatletter
\DeclareRobustCommand{\bigstar}{
	\mathop{\vphantom{\sum}\mathpalette\bigstar@\relax}\slimits@}
\newcommand{\bigstar@}[2]{\vcenter{\sbox\z@{$#1\sum$}
	\hbox{\resizebox{.8\dimexpr\ht\z@+\dp\z@}{!}{$\m@th\star$}}}}
\makeatother
\newcommand{\defeq}{\triangleq}%or \triangleq &or \mathrel{\mathop:}=

\newcommand{\defpropto}{\propto}
\renewcommand{\leq}{\leqslant}
\renewcommand{\geq}{\geqslant}
\renewcommand{\subset}{\subseteq}
\newcommand{\transp}{\mathsf{T}}
\newcommand{\Herm}{\dag}
\newcommand{\bra}[1]{\left\lvert#1\right\rangle}
\newcommand{\ket}[1]{\left\langle#1\right\rvert}
\newcommand{\braket}[1]{\bra{#1}\!\ket{#1}}
\newcommand{\conj}[1]{\overline{#1}}
\newcommand{\expectation}[1]{\left\langle{#1}\right\rangle}
\newcommand{\expectationwrt}[2]{\left\langle{#1}\right\rangle_{\!\!#2}}
\newcommand{\abs}[1]{\left\lvert#1\right\rvert}
    \newcommand{\bigabs}[1]{\bigl\lvert#1\bigr\rvert}
    \newcommand{\biggabs}[1]{\biggl\lvert#1\biggr\rvert}

\newcommand{\norm}[1]{\left\lVert#1\right\rVert}
\newcommand{\size}[1]{\left\lvert#1\right\rvert}
\newcommand{\D}{\mathrm{d}} % derivative
\newcommand{\lvec}[1]{\accentset{\leftharpoonup}{#1}} % Require Pkg "accents"
\newcommand{\rvec}[1]{\accentset{\rightharpoonup}{#1}} % Require Pkg "accents"
\newcommand{\grad}{\nabla}
\newcommand{\id}{\mathcal{I}}
\newcommand{\nb}[1]{{\partial #1}} % neighbor
\DeclareMathOperator{\VEC}{vec}
\newcommand{\ceil}[1]{\left\lceil#1\right\rceil}

\newcommand{\tensor}{\otimes}
\newcommand{\xk}{\setminus}
\newcommand{\hadamard}{\odot}
\newcommand{\set}[1]{\mathcal{#1}}
\newcommand{\operator}[1]{\mathcal{#1}}
\newcommand{\system}[1]{\mathsf{#1}} 
\newcommand{\rv}[1]{\mathsf{#1}}

\newcommand{\vect}[1]{\mathbf{#1}}
\newcommand{\hilbert}{\mathcal{H}}
\newcommand{\entropy}{\boldsymbol{\mathrm{H}}}
\newcommand{\qEntropy}{\entropy}
\newcommand{\mutualInfo}{\boldsymbol{\mathrm{I}}}
\newcommand{\qmutualInfo}{\mutualInfo}
\newcommand{\infoRate}{\mathrm{I}}
\newcommand{\entropicRate}{\mathrm{H}}
\newcommand{\capacity}{\mathrm{C}}
\newcommand{\DensOp}{\mathfrak{D}}
\newcommand{\LinearOp}{\mathfrak{L}}
\newcommand{\HermitianOp}{\mathfrak{L}_\Herm}
\newcommand{\PositiveOp}{\mathfrak{L}_{+}}
\newcommand{\StPositiveOp}{\mathfrak{L}_{++}}
\newcommand{\ProbSp}{\mathfrak{P}}
\DeclarePairedDelimiterX{\infdivx}[2]{(}{)}{#1\delimsize\|#2}
\DeclarePairedDelimiterX{\inner}[2]{\langle}{\rangle}{#1\delimsize\vert#2}
\DeclarePairedDelimiterX{\setdef}[2]{\{}{\}}{#1\delimsize\vert#2}
\newcommand{\infdiv}[2]{\mathrm{D}\infdivx*{#1}{#2}}
\newcommand{\IRUB}{\bar{\infoRate}}
\newcommand{\IRLB}{\munderbar{\infoRate}}
\newcommand{\helmholtz}{F_{\mathsf{H}}}
\newcommand{\gibbs}{F_{\mathsf{G}}}
\newcommand{\bethe}{F_{\mathsf{B}}}
\newcommand{\What}{\hat{W}}
\newcommand{\QW}{(Q\mkern-1.5mu W)}
\newcommand{\QWaux}{(Q\mkern-1.5mu\hat{W})}

\newcommand{\vx}{\vect{x}}
\newcommand{\tx}{\tilde{x}}
\newcommand{\tvx}{\tilde{\vect{x}}}
\newcommand{\cx}{\check{x}}
\newcommand{\cvx}{\check{\vect{x}}}

\newcommand{\tX}{\tilde{X}}

\newcommand{\vy}{\vect{y}}
\newcommand{\ty}{\tilde{y}}

\newcommand{\cy}{{\check{y}}}
\newcommand{\cvy}{\check{\vect{y}}}

\newcommand{\tY}{\tilde{Y}}

\newcommand{\vs}{\vect{s}}
\newcommand{\ts}{\tilde{s}}
\newcommand{\tvs}{\tilde{\vect{s}}}
\newcommand{\cs}{\check{s}}

\newcommand{\tS}{\tilde{S}}

\newcommand{\vt}{\vect{t}}
\newcommand{\tit}{\tilde{t}}
\newcommand{\tvt}{\tilde{\vect{t}}}

\newcommand{\muY}{\mu^{\rv{Y}}}
\newcommand{\muXY}{\mu^{\rv{XY}}}

\newcommand{\bmuY}{\bar \mu^{\rv{Y}}}
\newcommand{\bmuXY}{\bar \mu^{\rv{XY}}}

\newcommand{\sigmaY}{\sigma^{\rv{Y}}}
\newcommand{\sigmaXY}{\sigma^{\rv{XY}}}

\newcommand{\bsigmaY}{\bar \sigma^{\rv{Y}}}
\newcommand{\bsigmaXY}{\bar \sigma^{\rv{XY}}}

\newcommand{\lambdaY}{\lambda^{\rv{Y}}}
\newcommand{\lambdaXY}{\lambda^{\rv{XY}}}

\newcommand{\pgood}{p_{\mathrm{g}}}
\newcommand{\pbad}{p_{\mathrm{b}}}
\pgfdeclarelayer{bg}
\pgfdeclarelayer{fg}
\pgfsetlayers{bg,main,fg}
\tikzset{-*/.style={
    shorten >=#1,
    decoration={markings,
                mark={at position 1 with {\draw[fill] circle [radius=#1];}}},
    postaction=decorate},
    -*/.default=1.5pt}
\tikzset{*-/.style={
    shorten <=#1,
    decoration={markings,
                mark={at position 0 with {\draw[fill] circle [radius=#1];}}},
    postaction=decorate},
    *-/.default=1.5pt}
\tikzset{*-*/.style={
    shorten <=#1,
    shorten >=#1,
    decoration={markings,
                mark={at position 0 with {\draw[fill] circle [radius=#1];}},
                mark={at position 1 with {\draw[fill] circle [radius=#1];}}},
    postaction=decorate},
    *-*/.default=1.5pt}
\tikzset{double distance=1.5pt}%Set default double distance to 1.5pt
\tikzset{multiple/.style={%Edge with more than 3 paths
    double,
    decoration={markings,
                mark={at position .5 with {\node[draw=none] {$\medslash$};}}},
    postaction=decorate}}
\newcommand{\aka}{a.k.a.\ } % always followed by a space
\newcommand{\eg}{\emph{e.g.}} % always end with ","
\newcommand{\ie}{\emph{i.e.}} % always end with ","
\newcommand{\wrt}{w.r.t.\ } % always followed by a space
\newcommand{\pmf}{PMF\xspace}
\newcommand{\pmfs}{PMFs\xspace}
\newcommand{\pdf}{PDF\xspace}
\newcommand{\pdfs}{PDFs\xspace}
\newcommand{\LHS}{LHS\xspace} 
\newcommand{\RHS}{RHS\xspace} 
\makeatletter
\ifnum\inputlineno=\m@ne
\let\showlineno\@empty
\else
\def\showlineno{\the\inputlineno}
\fi
\makeatother
\let\oldmarginpar\marginpar
\renewcommand{\marginpar}[1]{%
	\leavevmode%
	\oldmarginpar[\raggedleft#1]{\raggedright#1}% Do not delete this '%'.
	\ignorespacesafterend\ignorespaces%
}
\begin{document}
% Title page********************************************************************
\title{Factor~Graphs for Quantum~Information~Processing}
\authorFirstName{Michael Xuan}
\authorLastName{CAO}
\degree{Doctor of Philosophy}
\programme{Information Engineering}
\supervisor{Prof. VONTOBEL, Pascal Olivier}
\institute{The Chinese University of Hong Kong}
\date{April 2021}
\maketitle
\pagenumbering{roman}
%*******************************************************************************
% Abstract**********************************************************************
%\chapter*{Abstract}
\makeatletter
\noindent\begin{xltabular}{\textwidth}{lX}
	\multicolumn{2}{l}{\textbf{Abstract of thesis entitled:}}\\
	\multicolumn{2}{l}{\@title}\\
	\hline
	\textbf{Submitted by} & \@authorLastName, \@authorFirstName\\
	\cline{2-2}
	\textbf{For the degree of} & \@degree\\
	\cline{2-2}
	\multicolumn{2}{l}{\textbf{At the Chinese University of Hong Kong in \@date}}
\end{xltabular}\vspace{-2pt}
\makeatother
Statistical graphical models are frameworks that use graphs to represent dependencies among a large number of random variables.
Such frameworks have been proven useful for developing analytical/algorithmic methods in characterizing the behavior of systems involving a large number of variables.
Among these frameworks, the factor graphs are a popular variant.
A factor graph represents a factorization of a multivariate function, and in many practical examples, a factorization of some probability function.
It is a recurring problem in physics and engineering to compute the marginals (or partition sums) of a given multivariate function.
With the factor graph representation, such marginals (or partition sums) can be computed/estimated via a class of algorithms known as belief-propagation algorithms.
\par
In recent years, quantum computing and quantum communications have received increasing attention.
In these applications, information is represented and processed by suitable quantum systems.
Quantum theory, which is a generalization of probability theory, describes the uncertainty of quantum systems.
Classical graphical models, which are based on probability theory, can no longer efficiently describe the dependencies among such systems.
\par
In this thesis, we are interested in generalizing factor graphs and the relevant methods toward describing quantum systems.
Two generalizations of classical graphical models are investigated, namely double-edge factor graphs (DeFGs) and quantum factor graphs (QFGs).
Conventionally, a factor in a factor graph represents a nonnegative real-valued local functions.
Two different approaches to generalize factors in classical factor graphs yield DeFGs and QFGs, respectively.
We proposed/re-proposed and analyzed generalized versions of belief-propagation algorithms for DeFGs/QFGs.
As a particular application of the DeFGs, we investigate the information rate and their upper/lower bounds of classical communications over quantum channels with memory.
In this study, we also propose a data-driven method for optimizing the upper/lower bounds on information rate.
\begin{CJK*}{UTF8}{bsmi}
%\chapter*{摘要}
\newpage
\makeatletter
\noindent\begin{xltabular}{\textwidth}{lX}
	\multicolumn{2}{l}{\textbf{標題如下之論文之摘要}}\\
	\multicolumn{2}{l}{用於量子信息處理的因子圖方法的探討}\\
	\hline
	\textbf{提交者} & 曹\ 軒\\
	\cline{2-2}
	\textbf{所申請之學位} & 哲學博士\\
	\cline{2-2}
	\multicolumn{2}{l}{\textbf{二零二零年十一月\ 於香港中文大學}}
\end{xltabular}
\makeatother
統計圖形模型是使用圖形表示大量隨機變量之間的依存關係的框架。
在研究涉及大量變量的系統的行為時，此類框架有助開發相關的分析手段和計算方法。
在這類框架中，因子圖（factor graph）是一種流行的方法。
因子圖用圖表示因式分解。
而在許多實例中，它往往用於表示某些概率函數的分解。
計算給定因式分解的邊際（marginal）或配分和（partition sum）是物理學和工程學中常見的問題。
利用因子圖，這樣的邊際或配分和可用置信傳播算法（belief-propagation algorithm）（估）算出。

近年來，量子計算和量子通信領域所受關注愈增。
在這些應用中，訊息由量子系統（的狀態）表示。
量子論（quantum theory）是描述量子系統不確定性的理論，它是概率論的延伸。
由於前述圖模型上的方法均基於概率論，它們難以高效地描述此類系統之間的依賴性。

在本文中，我們關注因子圖和相關方法的拓展，以期用於描述量子系統。
本研究涉及了兩個拓展模型：雙邊因子圖（DeFGs）和量子因子圖（QFGs）。
傳統上，因子圖中的因子表示的是非負實值局部函數。
而DeFG和QFG分別對應了兩種對經典因子圖中的因子的拓展。
本文分別提出/重新提出並分析了針對DeFG/QFG而拓展的置信傳播算法。
我們還研究了量子記憶信道（quantum channel with memory）的經典通信率及其上下界。
該研究使用了DeFG來表示量子系統。
在本研究中，我們提出了一種基於數據來優化信息速率的上下界的方法。
\end{CJK*}
\cleardoublepage
%*******************************************************************************
% Introduction******************************************************************
\chapter*{Preface\markboth{PREFACE}{}}
\addcontentsline{toc}{chapter}{Preface}
%In the past few years, I have been working as a Ph.D. student under the supervision of Prof. Pascal VONTOBEL.
%The majority of my work has been devoted to generalizing the methods of factor graphs to quantum settings.
%The generalization received partial success: Not all of the quantum counterparts of the results from the methods of factor graphs hold as we wished, and interestingly, some results that hold precisely in classical cases only hold approximately in quantum setups.
%Overall, the study seemed to have generated more questions than it solved.
%This thesis is a summary of these explorations.
%\section*{Backgrounds\markboth{}{BACKGROUNDS}}
% Factor graphs
Factor graphs~\cite[Section~2.1]{kschischang2001factor, loeliger2004introduction, wainwright2008graphical} are a type of graphical models that uses bipartite graphs (or some equivalence) to describe \emph{factorizations} of multivariate functions, where a factorization of a function is a decomposition of itself into the \emph{product of a several sub-functions} (see~\eqref{eq:cfg:global:function:1}).
Though the method itself does not impose any constraints on the functions to be factorized, it is often used to represent the factorizations of probability functions and thus can be used to illustrate the dependencies among the involved random variables.
Many interesting computer science and physics problems can be modeled as computing the \emph{marginals} of certain functions with some knowledge of its factorization structure.
Examples include problems in coding theory (\eg, decoding of low-density parity-check (LDPC) codes~\cite{gallager1962low}) and problems in statistical mechanics (\eg, the Ising models (see, \eg,~\cite{molkaraie2013partition})).
Besides visualizing such factorizations, factor graphs also provide an efficient approach to compute or estimate the aforementioned marginals using a class of algorithms known as \emph{belief-propagation} (BA) algorithms.
\par
% Generalization to quantum systems
In recent years, quantum computing and quantum communications have recieved increasing attention.
Though technically speaking, such systems \emph{can} be described by probability theory, it is hugely inefficient to do so since a discrete quantum system's state space is continuous.
For example, in a quantum system of $n$-qubits, the states of the system can be $\sum_{\vx_1^n\in\mathbb{F}_2^n} \alpha_{\vx_1^n} \bra{\vx_1^n}$ for any complex numbers $\{\alpha_{\vx_1^n}\}_{\vx_1^n\in\mathbb{F}_2^n}$ such that $\sum_{\vx_1^n\in\mathbb{F}_2^n} \abs{\alpha_{\vx_1^n}}^2 = 1$.
In this case, describing the system using probability theory would require a $(2^{n+1}-1)$-dimensional (real) continuous random variable.
On the other hand, due to the physical nature of observations (as we understand them so far), it is equivalent to describe the same system using a $2^n$-by-$2^n$ density matrix (operator).
The latter description is the key to quantum theory, a generalization of probability theory (see Section~\ref{sec:basic_quantum_theory}).
Thus, solving the interference problem on a system involving quantum components using the traditional method of factor graphs and probability theory would be extremely inefficient.
This has motivated us to consider generalizations of the model and the methods of factor graphs for quantum theory.
%*******************************************************************************
\section*{Problems and Contributions\markboth{}{PROBLEMS AND CONTRIBUTIONS}}
As introduced below, we consider three particular topics/problems in this thesis.
The first two topics are two different generalizations of the method of factor graphs with quantum theory in mind.
The last topic solves a specific problem in quantum information theory and demonstrates how factor graphs could be helpful in these problems.
%*******************************************************************************
\subsection*{Double-Edge Factor Graphs}
In Chapter~\ref{chapter:DeFGs}, we focus on a generalized class of factor graphs called double-edge factor graphs (DeFGs).
A DeFG represents a factorization of the form $g(\vs,\tvs;\vx) = \prod_{a\in\set{F}} f_a(\vs_\nb{a},\tvs_\nb{a};\vx_{\delta{a}})$, where $f_a$ is complex-valued for each $a$ and where the matrix associated with $f_a$, \ie, $[f_a(\vx_{\delta{a}})]_{\vs_\nb{a},\tvs_\nb{a}} \defeq f_a(\vs_\nb{a},\tvs_\nb{a};\vx_{\delta{a}})$, is PSD for each $\vx_{\delta{a}}$.
This class of graphical models is closely related to the factor graphs for quantum probabilities introduced in~\cite{loeliger2012factor, loeliger2017factor} and is  helpful in representing quantum systems.
In particular, the results of quantum dynamics can be expressed as marginals of some global functions represented by suitable DeFGs.

In our work, we are interested in the problem of computing the marginals and partition sums of DeFGs.
We generalize the ``closing-the-box'' operations to DeFGs, and show that the marginals/partition sum of an acyclic DeFG can be computed using a sequence of ``closing-the-box'' operations.
We also generalize the belief-propagation (BP) algorithm to DeFGs.
In particular, we define the holographic transformation of DeFGs, and derive a loop calculus expansion for DeFGs.
We study multiple numerical examples in which BP algorithms for DeFGs show promising results.

Parts of the contents of Chapter~\ref{chapter:DeFGs} have been published in~\cite{cao2017double}.
%*******************************************************************************
\subsection*{Quantum Factor Graphs}
In Chapter~\ref{chapter:QFGs}, we focus on a graphical model called \emph{quantum factor graphs} (QFGs) \cite{Leifer08}.
The graphical model is a direct generalization of the \emph{bifactor networks} proposed in the same paper, in which the authors considered a ``factorization'' in the form of $\rho\defeq\bigstar_{a\in\set{F}}\rho_a$ where $\star$ is a commutative and associative binary operator of linear operators (matrices).
The target is to compute the (partial) trace of $\rho$.
The major results in~\cite{Leifer08} are based on additional commutativity assumptions.

In our work, we consider a different setup without commutativity assumptions, but with the local operators being ``close'' to identity operators.
We establish an approximated distributivity of $\star$ over (partial) trace functions under such a setup and generalize the ``closing-the-box'' operations to QFGs.
We also show that the trace of an acyclic QFG can always be approximated via a sequence of ``closing-the-box'' operations.
The belief-propagation algorithm for QFGs, which coincides with the quantum belief-propagation algorithm in~\cite{Leifer08}, is then a natural and heuristic generalization of the ``closing-the-box'' based method to QFGs with cycles.
We generalize the Bethe approximation to QFGs and define Helmholtz, Gibbs, and the Bethe free energies for this new setup.
We show that, provided that all local operators are close to $I$, the fixed points of the belief-propagation algorithm for QFGs approximately correspond to the stationary points of the Bethe free energy.
A numerical demonstration is also included at the end of that chapter.

Parts of the contents of Chapter~\ref{chapter:QFGs} have been published in~\cite{cao2016quantum}.
%*******************************************************************************
\subsection*{Bounding and Estimating the Classical Information Rate of Quantum Channels with Memory}
In Chapter~\ref{chapter:QCwM}, we focus on the information rate of classical communications over a finite-dimensional quantum channel with memory.
A quantum channel with memory is a CPTP map (see Definition~\ref{def:quantum:channel}) from $\hilbert_\system{A}\tensor\hilbert_\system{S}$ to $\hilbert_\system{B}\tensor\hilbert_{\system{S}'}$ where $\system{A}$ and $\system{B}$ are, respectively, the input and output systems and where $\system{S}$ and $\system{S}'$ are, respectively, the memory systems before and after the channel use (see Definition~\ref{def:quantum:channel:memory}).
In particular, we are interested in computing and bounding the information rate of classical communications over a quantum channel with memory using only separable-state ensemble and local measurements.
This restriction is equivalent to the scenario where no quantum computing device is present at the sending or receiving end or where our manipulation of the channel is limited to a single-channel use.
The difficulty of the problem lies with the presence of quantum memory.
In the most straightforward situation, where the memory system exhibits classical properties under certain ensembles and measurements, the classical communication setup is equivalent to a finite-state-machine channel (FSMC)~\cite{gallager1968information}.
Though the evaluation of the information rate of an FSMC is nontrivial in general, efficient stochastic methods for estimating and bounding this quantity have been developed~\cite{arnold2006simulation, sadeghi2009optimization}.

Our work is inspired by~\cite{arnold2006simulation}, where the authors proposed an efficient stochastic method in estimating the information rate of indecomposable FSMCs, and~\cite{sadeghi2009optimization}, where the authors proposed upper and lower bounds of the information rate of FSMCs based on auxiliary FSMCs and efficient methods for optimizing the bounds.
We propose algorithms for estimating the information rate of such communication setups, along with algorithms for bounding the information rate based on so-called auxiliary channels.
Some of the algorithms are generalized versions of the methods in~\cite{arnold2006simulation, sadeghi2009optimization}.
We also discuss suitable graphical models for doing the relevant computations.
We emphasize that the auxiliary channels are learned in a data-driven approach, \ie, only input/output sequences of the actual channel are needed, but not the channel model of the actual channel.

Chapter~\ref{chapter:QCwM} is adapted from~\cite{cao2020bounding}, with the majority of its contents published in~\cite{cao2017estimating, cao2019optimizing, cao2020bounding}.
\par
%*******************************************************************************
\section*{Structure of this Thesis\markboth{}{STRUCTURE OF THIS THESIS}}
The rest of this thesis consists of the four main chapters, followed by a summary, appendices, bibliography, and an index of terms.
Chapter~1 provides an introduction to some preliminary knowledge regarding factor graphs and basic quantum information theory.
Readers may skip this chapter if they are already familiar with these topics.
Chapters~2 and~3 discuss double-edge factor graphs and quantum factor graphs, respectively.
Chapter~4 studies classical information rate of quantum channels with memory.
These three chapters are relatively independent of each other, and readers may read them in any desired order.
%*******************************************************************************
% Acknowledgment***************************************************************
\chapter*{Acknowledgment\markboth{ACKNOWLEDGEMENT}{}}
\addcontentsline{toc}{chapter}{Acknowledgment}
I want to thank my supervisor Prof. Pascal VONTOBEL for his patient guidance in the past few years, and my co-supervisor, Prof. Chandra NAIR, for his suggestions and help.
I also want to thank Dr. Joseph RENES for hosting my internship during the winter of 2019 and enlightening discussions.

I also want to thank the thesis committee members, Prof. Chandra NAIR, Prof. Angela ZHANG, Prof. Changhong ZHAO, Prof. Cheuk-Ting LI, and Prof. Henry PFISTER, for their time in reading this thesis and for their helpful comments.

Finally, I want to thank my wife, July, for providing continuous joy and support to my life.

The work in this thesis was supported in part by the Research Grants Council of the Hong Kong Special Administrative Region, China, under Project CUHK 14209317 and Project CUHK 14207518.
%*******************************************************************************
% Notations*********************************************************************
\chapter*{Notations\markboth{NOTATIONS}{}}
%\section*{Sets}
\noindent\begin{xltabular}{\textwidth}{L|R}
	\multicolumn{2}{l}{\textbf{\Large Sets}}\\
	\hline
	$\set{X}$ & a set\\
	$\set{A}\sqcup\set{B}$ & union of disjoint sets $\set{A}$ and $\set{B}$\\
	$\set{B}\xk\set{A}$ & the relative complement of $\set{A}$ with respect to $\set{B}$\\
	$\set{B}\xk a$ & equivalent to $\set{B}\xk\{a\}$\\
	$\Integers,\Integers_{>0},\Integers_{<0},\Integers_{\geqslant 0}$ & the sets of integers, positive integers, negative integers, non-negative integers, respectively\\
	$\Rationals,\Rationals_{>0},\Rationals_{<0},\Rationals_{\geqslant 0}$ & the sets of rationals, positive rationals, negative rationals, non-negative rationals, respectively\\
	$\Reals,\Reals_{>0},\Reals_{<0},\Reals_{\geqslant 0}$ & the sets of real numbers, positive real numbers, negative real numbers, non-negative real numbers, respectively\\
	$\Complex$ & the sets of complex numbers\\
	\hline
\end{xltabular}
%\section*{Probability and Quantum Theory}
\noindent\begin{xltabular}{\textwidth}{L|R}
	\multicolumn{2}{l}{\textbf{\Large Probability and Quantum Theory}}\\
	\hline
	$\rv{X}$ & a random variable\\
	$\prob_\rv{X}$ & probability mass function (\pmf) or probability density function (\pdf) \wrt $\rv{X}$\\
	$\expectation{\rv{X}}$ & expectation of the random variable $\rv{X}$ \\
	$\expectationwrt{f(\rv{X})}{\ b}$ & expectation of $f(\rv{X})$, where the \pmf of $\rv{X}$ is $b$\\
	$\rv{X}_1^n$ & a list of random variables, \ie, $(\rv{X}_1,\ldots,\rv{X}_n)$\\
	$\ProbSp(\set{X})$ & the set of all \pmfs or \pdfs over the set $\set{X}$\\
	$\system{S}$ & a quantum system\\
	$\hilbert_\system{S}$ & state space associated with system $\system{S}$, which is a Hilbert space\footnote{All Hilbert spaces involved in this thesis are finite-dimensional, unless stated otherwise.}\\
	$\rho_\system{S}$ & a density operator associated with system $\system{S}$\\
	$\system{S}_0^n$ & a list of quantum systems, \ie, $(\system{S}_1,\ldots,\system{S}_n)$\\
	$\LinearOp(\hilbert)$ & the set of all linear operators (\ie, linear transformations) on $\hilbert$\\
	$\PositiveOp(\hilbert)$ & the set of all positive operators on $\hilbert$\\
	$\StPositiveOp(\hilbert)$ & the set of all strictly positive operators on $\hilbert$\\
	$\DensOp(\hilbert)$ & the set of all density operators on $\hilbert$\\
	\hline
\end{xltabular}
%\section*{Linear Algebra and Matrices}
\noindent\begin{xltabular}{\textwidth}{L|R}
	\multicolumn{2}{l}{\textbf{\Large Linear Algebra and Matrices}}\\
	\hline
	$\vx$ & a vector\\
	$\vx_1^n$ & a vector of length $n$, \ie, $\vx_1^n=(x_1,x_2,\ldots,x_n)$\\
	$\vx_i$ & $i$-th entry of the vector $\vx$\\
	$\norm{\vect{v}}_p$ & $\ell_p$-norm of the real vector $\vx$, \ie, $\norm{\vect{v}}_p \defeq \left(\sum_{i=1}^n \abs{\vect{v}_i}^p\right)^{\frac{1}{p}}$\\
	$A$ & a matrix\\
	$A\geqslant 0$ & a PD matrix, \ie, $A\in\Complex^{n\times n}$ for some positive integer $n$ and $\mathbf{v}^\Herm A\mathbf{v}\geqslant 0$ for all $\mathbf{v}\in\Complex^n$\\
	$A> 0$ & a PSD matrix, \ie, $A\in\Complex^{n\times n}$ for some positive integer $n$ and $\mathbf{v}^\Herm A\mathbf{v}> 0$ for all $\mathbf{v}\in\Complex^n$\\
	$[T]$ & a matrix representation of the linear transformation $T$\\
	$A_{i,j}$ & $(i,j)$-th entry of the matrix $A$\\
	$\VEC(A)$ & vectorization of an $n\times m$ matrix $A$, \ie, $A_{i,j}=\VEC(A)_{i+n\cdot (j-1)}$ \\
	$\tr(A),\!\tr(T),\!\tr(\rho)$ & trace of the matrix $A$, of the transformation $T$, of the operator $\rho$, respectively\\
	$\bra{\phi}$ & a column vector\\
	$\ket{\phi}$ & a row vector adjoint to $\bra{\phi}$, equivalently $\ket{\phi}$ is the linear transformation $\bra{\psi}\mapsto\inner{\phi}{\psi}$\\
	$\inner{\phi}{\psi}$ & inner product between the vectors $\bra{\phi}$ and $\bra{\psi}$\\	
	$\inner{\phi_{\system{A}}}{\psi_{\system{AB}}}$ & equivalent to $(\ket{\phi_\system{A}}\tensor I_\system{B})(\bra{\psi_{\system{AB}}})$\\
	$\exp{A}$ & matrix exponential of the matrix $A$\\
	$\log{A}$ & matrix logarithm of the matrix $A$\\
	$\tensor$ & tensor product between two vector spaces or two operators\\
	$\hadamard$ & Hadamard product between two matrices or vectors\\
	$\perp$ & perpendicular ($\mathbf{v}\perp\mathbf{u} \iff \inner{\mathbf{v}}{\mathbf{u}}=0$)\\
	\hline
\end{xltabular}
%\section*{Classical and Quantum Information Theory}
\noindent\begin{xltabular}{\textwidth}{L|R}
	\multicolumn{2}{l}{\textbf{\Large Classical and Quantum Information Theory}}\\
	\hline
	$\entropy(\rv{X}), \entropy(p)$ & Shannon entropy of the random variable $\rv{X}$ and the \pmf $p$, respectively\\
	$\qEntropy(\system{A}), \qEntropy(\rho)$ & von Neumann entropy of the quantum system $\system{A}$ and the density operator $\rho$, respectively\\
	$\mutualInfo(\rv{X}:\rv{Y})$ & mutual information between the random variables $\rv{X}$ and $\rv{Y}$\\
	$\qmutualInfo(\system{A}:\system{B})$ & quantum mutual information between the systems $\system{A}$ and $\system{B}$\\
	$\infdiv{p}{q}$ & relative entropy between the \pmfs$p$ and $q$\\
	$\infdiv{\rho}{\sigma}$ & quantum relative entropy between the density operators $\rho$ and $\sigma$\\
	\hline
\end{xltabular}
%\section*{Graph Theory}
\noindent\begin{xltabular}{\textwidth}{L|R}
	\multicolumn{2}{l}{\textbf{\Large Graph Theory}}\\
	\hline
	$\set{G}=(\set{V},\set{E},R)$ & a graph: $\set{V}$ is called the vertex set, $\set{E}$ is called the edge set, and $R:\set{E}\to\set{V}\times\set{V}$ is called the relationship function\footnote{All graphs involved in this thesis are simple, finite, and undirectional.} of $\set{G}$\\
	$\set{G}=(\set{V}_1,\set{V}_2,\set{E})$ & a bipartite graph, \ie, $\set{G}=(\set{V}_1\sqcup\set{V}_2,\set{E}\subset\set{V}_1\times\set{V}_2,R:e\mapsto e)$\\
	$\deg(v)$ & the degree of the vertex $v$\\
	$\nb{v}$ & the set of neighbors of $v$\\
	$(v_1-\cdots-v_n)$ & a backtrackless walk in a graph: requiring $v_i\neq v_{i+1}$ and $v_i$ connected to $v_{i+1}$\\
	\hline
\end{xltabular}
%\section*{Others}
\noindent\begin{xltabular}{\textwidth}{L|R}
	\multicolumn{2}{l}{\textbf{\Large Others}}\\
	\hline
	$\star$ & star product between two Hermitian operators (see Eqs.~\eqref{eq:def:star} and~\eqref{eq:def:star:2})\\
	\hline
\end{xltabular}
\newpage
%\section*{Acronyms}
\begin{xltabular}{\textwidth}{L|R}
	\multicolumn{2}{l}{\textbf{\Large Acronyms}}\\
	\hline
	CC-QSC & classical-input classical-output quantum-state channel\\
	CPTP & completely positive trace-preserving\\
	CtB & closing-the-box\\
	DeFG & double-edge factor graph\\
	FG & (classical) factor graph\\
	FSMC & finite-state-machine channel\\
	i.i.d. & independent and identically distributed\\
	NFG & normal factor graph\\
	PD & (strictly) positive definite\\
	\pdf & probability density function\\
	\pmf & probability mass function\\
	PSD & positive semi-definite\\
	QFG & quantum factor graph\\
	\hline
\end{xltabular}

\noindent Note that, unless stated otherwise (only in a few numerical examples), all appearances of logarithm in this thesis should be treated as natural logarithm.
%*******************************************************************************
% Lists*************************************************************************
\listoffigures\addcontentsline{toc}{chapter}{\listfigurename}
\listofalgorithms\addcontentsline{toc}{chapter}{\listalgorithmname}
\chapter*{List of Publications}\addcontentsline{toc}{chapter}{List of Publications}
\noindent
\begin{itemize}
	\item M. X. Cao and P. O. Vontobel, ``Quantum factor graphs: closing-the-box  operation and variational approaches,'' in \textit{Proceedings International Symposium on Information Theory and Its Applications (ISITA)}, Monterey, CA, USA, 2016, pp. 651--655.
	\item M. X. Cao and P. O. Vontobel, ``Estimating the information rate of a channel with classical input and output and a quantum state,'' in \textit{Proceedings IEEE International Symposium on Information Theory (ISIT)}, Aachen, Germany, 2017, pp. 3205--3209.
	\item M. X. Cao and P. O. Vontobel, ``Double-edge factor graphs: Definition, properties, and examples,'' in \textit{Proceedings IEEE Information Theory Workshop (ITW)}, Kaohsiung, Taiwan, 2017, pp. 136--140.
	\item M. X. Cao and P. O. Vontobel, ``Optimizing bounds on the classical information rate of quantum channels with memory,'' in \textit{Proceedings IEEE International Symposium on Information Theory (ISIT)}, Paris, France, 2019, pp. 265--269.
	\item M. X. Cao and P. O. Vontobel, ``Bounding and estimating the classical information rate of quantum channels with memory,'' \textit{IEEE Transactions on Information Theory}, vol. 66, no. 9, pp. 5601--5619, 2020.
\end{itemize}
%*******************************************************************************
% TOC***************************************************************************
\cleardoublepage
\tableofcontents
%*******************************************************************************
% Main Chapters*****************************************************************
\cleardoublepage \setcounter{page}{0} \pagenumbering{arabic} \pagestyle{headings}
%*******************************************************************************
% Chapter 1 Preliminaries*******************************************************
\chapter{Preliminaries}\label{chapter:preliminaries}
The preliminaries of this thesis, as introduced in this chapter, include basic results regarding \emph{classical} factor graphs (FGs) and basic quantum information theory, each comprising one section of the chapter.
Readers may skip the correcponding section(s) provided familiarity with the topic(s).
We emphasize the reviewing nature of this chapter; namely, the results are either known contributions, or derived rather straightforwardly.
%*******************************************************************************
\section{Factor Graphs and Belief-Propagation Algorithms}\label{sec:CFGs}
This section reviews factor graphs, belief-propagation (BP) algorithms, and several interpretations of BP algorithms' outputs.
%*******************************************************************************
\subsection{Factor Graphs and Normal Factor Graphs}
Given a real-valued function $g$ of multiple variables, a factorization\index{factorization} of $g$ is an expression in the form
\begin{equation}\label{eq:cfg:global:function:1}
g\big((x_i)_{i\in\set{V}}\big) = \prod_{a \in \set{F}} f_a(\mathbf{x}_{\partial a}),
\end{equation}
where $\set{V}$, $\set{F}$ are some finite sets and where, for each $a\in\set{F}$, $\partial a\subset\set{V}$ denotes the indices of the arguments of $f_a$.
A factorization as in~\eqref{eq:cfg:global:function:1} can be represented by a \emph{bipartite} graph between $\set{V}$ and $\set{F}$, with $i\in\set{V}$ being a neighbor of $a\in\set{F}$ if and only if $i\in\nb{a}$ (see Figure~\ref{fig:cfg:generic}); conversely, such a graph is called a \emph{factor graph} (FG) \emph{representing~\eqref{eq:cfg:global:function:1}}.
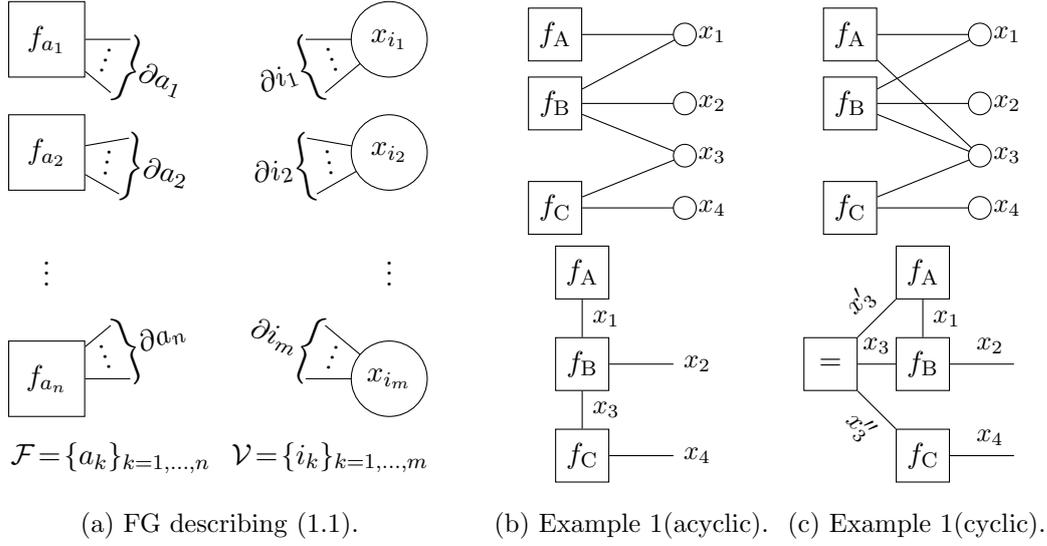
\begin{figure}\centering
\begin{subfigure}[b]{0.48\columnwidth}\centering
	\begin{tikzpicture}[node/.style={draw=none},
    	factor/.style={rectangle, minimum size=1cm, draw},
    	varNode/.style={circle,minimum size=1cm,draw,inner sep=0pt,outer sep=0pt},
    	node distance = 1.5cm]
		\node[factor] (A) {$f_{a_1}$};
		\node[factor] (B)[below of=A] {$f_{a_2}$};
		\node (E)[below of=B] {};
		\node[factor] (C)[below of=E] {$f_{a_{n}}$};
		\node[varNode] (x1)[right=3.5cm of A] {$x_{i_1}$};
		\node[varNode] (x2)[right=3.5cm of B] {$x_{i_2}$};
		\node[varNode] (xn)[right=3.5cm of C] {$x_{i_{m}}$};
		\path (B) edge[draw=none] node {$\vdots$} (C);
		\path (x2) edge[draw=none] node {$\vdots$} (xn);
		\node[below=0.2cm of C.south west, anchor = north west, xshift=-1mm]
		    {$\set{F}\!=\!\{a_k\}_{k=1,\ldots,n}$};
		\node[below=0.35cm of xn.south east, anchor = north east, xshift=3mm]
		    {$\set{V}\!=\!\{i_k\}_{k=1,\ldots,m}$};
		
		\pgfmathsetmacro\rad{1.1};
		\pgfmathsetmacro\y{\rad *cos(90)};
		\pgfmathsetmacro\x{\rad *sin(90)};
		\path (A) edge[draw] ([xshift=\x cm, yshift=-\y cm]A);
		\pgfmathsetmacro\y{\rad *cos(70)};
		\pgfmathsetmacro\x{\rad *sin(70)};
		\path (A) edge[draw=none] node[sloped,above=-8.2pt] {$\vdots$} node[sloped,above=0pt,pos=0.6,anchor=west] {$\bigg\}\nb{a_1}$} ([xshift=\x cm, yshift=-\y cm]A);
		\pgfmathsetmacro\y{\rad *cos(50)};
		\pgfmathsetmacro\x{\rad *sin(50)};
		\path (A) edge[draw] ([xshift=\x cm, yshift=-\y cm]A);
		\pgfmathsetmacro\y{\rad *cos(100)};
		\pgfmathsetmacro\x{\rad *sin(100)};
		\path (B) edge[draw] ([xshift=\x cm, yshift=-\y cm]B);
		\pgfmathsetmacro\y{\rad *cos(80)};
		\pgfmathsetmacro\x{\rad *sin(80)};
		\path (B) edge[draw=none] node[sloped,above=-8.2pt] {$\vdots$} node[sloped,above=0pt,pos=0.6,anchor=west] {$\bigg\}\nb{a_2}$} ([xshift=\x cm, yshift=-\y cm]B);
		\pgfmathsetmacro\y{\rad *cos(60)};
		\pgfmathsetmacro\x{\rad *sin(60)};
		\path (B) edge[draw] ([xshift=\x cm, yshift=-\y cm]B);
		\pgfmathsetmacro\y{\rad *cos(130)};
		\pgfmathsetmacro\x{\rad *sin(130)};
		\path (C) edge[draw] ([xshift=\x cm, yshift=-\y cm]C);
		\pgfmathsetmacro\y{\rad *cos(110)};
		\pgfmathsetmacro\x{\rad *sin(110)};
		\path (C) edge[draw=none] node[sloped,above=-8.2pt] {$\vdots$} node[sloped,above=0pt,pos=0.6,anchor=west] {$\bigg\}\nb{a_n}$} ([xshift=\x cm, yshift=-\y cm]C);
		\pgfmathsetmacro\y{\rad *cos(90)};
		\pgfmathsetmacro\x{\rad *sin(90)};
		\path (C) edge[draw] ([xshift=\x cm, yshift=-\y cm]C);
		\pgfmathsetmacro\y{\rad *cos(90)};
		\pgfmathsetmacro\x{\rad *sin(90)};
		\path (x1) edge[draw] ([xshift=-\x cm, yshift=-\y cm]x1);
		\pgfmathsetmacro\y{\rad *cos(70)};
		\pgfmathsetmacro\x{\rad *sin(70)};
		\path (x1) edge[draw=none] node[sloped,above=-8.2pt] {$\vdots$} node[sloped,above=0pt,pos=0.6,anchor=east] {$\nb{i_1}\bigg\{$}  ([xshift=-\x cm, yshift=-\y cm]x1);
		\pgfmathsetmacro\y{\rad *cos(50)};
		\pgfmathsetmacro\x{\rad *sin(50)};
		\path (x1) edge[draw] ([xshift=-\x cm, yshift=-\y cm]x1);
		\pgfmathsetmacro\y{\rad *cos(100)};
		\pgfmathsetmacro\x{\rad *sin(100)};
		\path (x2) edge[draw] ([xshift=-\x cm, yshift=-\y cm]x2);
		\pgfmathsetmacro\y{\rad *cos(80)};
		\pgfmathsetmacro\x{\rad *sin(80)};
		\path (x2) edge[draw=none] node[sloped,above=-8.2pt] {$\vdots$} node[sloped,above=0pt,pos=0.6,anchor=east] {$\nb{i_2}\bigg\{$}  ([xshift=-\x cm, yshift=-\y cm]x2);
		\pgfmathsetmacro\y{\rad *cos(60)};
		\pgfmathsetmacro\x{\rad *sin(60)};
		\path (x2) edge[draw] ([xshift=-\x cm, yshift=-\y cm]x2);
		\pgfmathsetmacro\y{\rad *cos(130)};
		\pgfmathsetmacro\x{\rad *sin(130)};
		\path (xn) edge[draw] ([xshift=-\x cm, yshift=-\y cm]xn);
		\pgfmathsetmacro\y{\rad *cos(110)};
		\pgfmathsetmacro\x{\rad *sin(110)};
		\path (xn) edge[draw=none] node[sloped,above=-8.2pt] {$\vdots$} node[sloped,above=0pt,pos=0.6,anchor=east] {$\nb{i_m}\bigg\{$}  ([xshift=-\x cm, yshift=-\y cm]xn);
		\pgfmathsetmacro\y{\rad *cos(90)};
		\pgfmathsetmacro\x{\rad *sin(90)};
		\path (xn) edge[draw] ([xshift=-\x cm, yshift=-\y cm]xn);		
	\end{tikzpicture}
	\caption{FG describing \eqref{eq:cfg:global:function:1}.}
	\label{fig:cfg:generic}
\end{subfigure}
\begin{subfigure}[b]{0.25\columnwidth}\centering
	\begin{tikzpicture}[node/.style={draw=none},
    	factor/.style={rectangle, minimum size=.7cm, draw},
    	varNode/.style={circle,minimum size=.3cm,draw,inner sep=0pt,outer sep=0pt}]
		\node[factor] (A) {$f_{\mathrm{A}}$};
		\node[factor] (B)[below = 0.2cm of A] {$f_{\mathrm{B}}$};
		\node (E)[below = 0.2cm of B] {};
		\node[factor] (C)[below = 0.2cm of E] {$f_{\mathrm{C}}$};
		\node[varNode] (x1)[right=1.2cm of A] {};
	    	\node[font=\small, right=-3pt of x1] {$x_1$};
		\node[varNode] (x2)[right=1.2cm of B] {};
	    	\node[font=\small, right=-3pt of x2] {$x_2$};
		\node[varNode] at (E-|x2) (x3) {};
	    	\node[font=\small, right=-3pt of x3] {$x_3$};
		\node[varNode] (x4)[right=1.2cm of C] {};
	    	\node[font=\small, right=-3pt of x4] {$x_4$};
		\path (A) edge (x1);
		\path (B) edge (x1);
		\path (B) edge (x2);
		\path (B) edge (x3);
		\path (C) edge (x3);
		\path (C) edge (x4);
		
		\path (C) edge[draw=none] node[pos=0, factor, yshift=-.5cm, anchor=north] (A) {$f_{\mathrm{A}}$} (x4);
		%\node[factor, below=.5cm of C] (A) {$f_{\mathrm{A}}$};
		\node[factor] (B)[below=.5cm of A] {$f_{\mathrm{B}}$};
		\node[factor] (C)[below=.5cm of B] {$f_{\mathrm{C}}$};
		\path (A) edge node[right, font=\small] {$x_1$} (B);
		\path (B) edge node[right, font=\small] {$x_3$} (C);
		\path (B) edge node[pos=1, right, font=\small] {$x_2$} ([xshift=1.2cm]B);
		\path (C) edge node[pos=1, right, font=\small] {$x_4$} ([xshift=1.2cm]C);
	\end{tikzpicture}
	\caption{Example 1(acyclic).}
	\label{fig:cfg:1:a}
\end{subfigure}
\begin{subfigure}[b]{0.25\columnwidth}\centering
    \begin{tikzpicture}[node/.style={draw=none},
    	factor/.style={rectangle, minimum size=.7cm, draw},
    	varNode/.style={circle, minimum size=.3cm,draw,inner sep=0pt, outer sep=0pt}]
		\node[factor] (A) {$f_{\mathrm{A}}$};
		\node[factor] (B)[below = 0.2cm of A] {$f_{\mathrm{B}}$};
		\node (E)[below = 0.2cm of B] {};
		\node[factor] (C)[below = 0.2cm of E] {$f_{\mathrm{C}}$};
		\node[varNode] (x1)[right=1.2cm of A] {};
	    	\node[font=\small, right=-3pt of x1] {$x_1$};
		\node[varNode] (x2)[right=1.2cm of B] {};
	    	\node[font=\small, right=-3pt of x2] {$x_2$};
		\node[varNode] (x3) at (E-|x2) {};
	    	\node[font=\small, right=-3pt of x3] {$x_3$};
		\node[varNode] (x4)[right=1.2cm of C] {};
	    	\node[font=\small, right=-3pt of x4] {$x_4$};
		\path (A) edge (x1);
		\path (A) edge (x3);
		\path (B) edge (x1);
		\path (B) edge (x2);
		\path (B) edge (x3);
		\path (C) edge (x3);
		\path (C) edge (x4);
		
		\path (C) edge[draw=none] node[factor, yshift=-.5cm, anchor=north] (A) {$f_{\mathrm{A}}$} (x4);
		\node[factor] (B)[below=.5cm of A] {$f_{\mathrm{B}}$};
		\node[factor] (C)[below=.5cm of B] {$f_{\mathrm{C}}$};
		\path (A) edge node[right, font=\small] {$x_1$} (B);
		\path (B) edge node[above, font=\small, pos=1, anchor=south east] {$x_2$} ([xshift=1.2cm]B);
		\path (C) edge node[above, font=\small, pos=1, anchor=south east] {$x_4$} ([xshift=1.2cm]C);
		\node[factor] (E)[left=.5cm of B] {$=$};
		\path (E) edge node[above, font=\small, sloped] {$x'_3$} (A);
		\path (E) edge node[above, font=\small, sloped] {$x_3$} (B);
		\path (E) edge node[below, font=\small, sloped] {$x''_3$} (C);
	\end{tikzpicture}
	\caption{Example 1(cyclic).}
	\label{fig:cfg:1:b}
\end{subfigure}
\caption{Some factor graphs.}
\end{figure}
%*******************************************************************************
\begin{example}[Two simple examples]\label{ex:cfg:1}
The factor graph in Figure~\ref{fig:cfg:1:a} depicts the factorization
\begin{equation}\label{eq:classical:factorization:example:1}
    g(x_1, x_2, x_3,x_4) = f_{\mathrm{A}}(x_1) \cdot f_{\mathrm{B}}(x_1,x_2,x_3) \cdot f_{\mathrm{C}}(x_3,x_4).
\end{equation}
Note that in this factor graph, $\deg(i)\leq 2$ for each $i\in\set{V}$.
In this case, one can simplify the graph by drawing all the vertices in $\set{V}$ as edges, as shown in the bottom part of the figure.
Note that vertices in $\set{V}$ with degree 1 result in so-called half-edges.
The bottom graph is often known as a \emph{normal} factor graph (NFG)\index{normal factor graph}~\cite{forney2001codes, al2011normal}.
\par
%***************************************************************************
The factor graph in Figure~\ref{fig:cfg:1:b} depicts the factorization
\begin{equation}\label{eq:classical:factorization:example:2}
g(x_1, x_2, x_3,x_4) = f_{\mathrm{A}}(x_1,x_3) \cdot f_{\mathrm{B}}(x_1,x_2,x_3) \cdot f_{\mathrm{C}}(x_3,x_4).
\end{equation}
In comparison with the previous example, this factor graph is cyclic.
Despite the fact that $\deg(x_3)>2$, we can still redraw this graph as an NFG by introducing an \emph{equality} node, as depicted in the bottom part of Figure~\ref{fig:cfg:1:b}.
\end{example}
%*******************************************************************************
We formalize the above discussion as follows.
\begin{definition}[Factor Graph]\index{factor graph}
A \emph{factor graph} (FG) is a bipartite graph $\set{G}=(\set{V},\set{F},\set{E}\subset \set{V}\times\set{F})$ associated with a variable set $\mathfrak{V}$ and a factor set $\mathfrak{F}$, where
\begin{itemize}
	\item $\mathfrak{V}=\{\set{X}_i\}_{i\in\set{V}}$ is indexed by $\set{V}$, and each element of $\mathfrak{V}$ is a set (\aka alphabets);
	\item $\mathfrak{F}=\{f_a\}_{a\in\set{F}}$ is indexed by $\set{F}$, and $f_a:\bigtimes_{i\in\nb{a}}\set{X}_i\to \Reals$ for each $a\in\set{F}$.
\end{itemize}
The function $g(\vx)\defeq\prod_{a\in\set{F}} f_a(\vx_\nb{a})$ is called the \emph{global function} of $\set{G}$, and in this case, $\set{G}$ is said to be representing the factorization $g(\vx)=\prod_{a\in\set{F}} f_a(\vx_\nb{a})$.\footnote{Here, $\vx\defeq(x_i)_{i\in\set{V}}$. This also applies for later appearances of ``$\vx$'' as an argument of global functions.}
A factor graph is \emph{normal} if the degree of any vertex in $\set{V}$ is at most 2.
\end{definition}
\begin{remark}
We often redraw the vertices in $\set{V}$ in an NFG as edges, as in Figure~\ref{fig:cfg:1:a}.
\end{remark}
\begin{remark}
Any factor graph can be converted into an NFG by properly introducing equality node(s), as in the bottom part of~Figure~\ref{fig:cfg:1:b}.
\end{remark}
%*******************************************************************************
\subsection{Marginals and Partition Sums}
The \emph{marginals} (or \emph{partition sum}) of a multivariable function are the results of the summation over some (or all) of its arguments.
Computing marginals and partition sums of factorizations is a recurring problem in physics and computer science, as demonstrated in the following two examples.
%*******************************************************************************
\begin{example}[A 2-D Ising Model~\cite{molkaraie2013partition}]
Figure~\ref{fig:cfg:2:a} depicts the factorization of the configuration probability of a 2-D $n\times n$ Ising model:
\begin{equation}
p^{(\beta)}(\vx_{1,1}^{n,n}) \propto g^{(\beta)}(\vx_{1,1}^{n,n}) \defeq \prod_{i,j=1,\ldots,n} h_{i,j}^{(\beta)}(x_{i,j})\cdot \prod_{i,j,i'j'\in\{1,\ldots,n\}\atop \abs{i-i'}+\abs{j-j'}=1} f^{(\beta)}(x_{i,j},x_{i',j'}),
\end{equation}
where 
\begin{itemize}
\item $x_{i,j}\in\{-1,1\}$ describes the spin configuration of the particle at location $(i,j)$;
\item $\beta\in(0,+\infty)$ is the inverse temperature;
\item $h_{i,j}^{(\beta)}(x_{i,j})\defeq\exp\left(-\beta\cdot\tilde{h}_{i,j}\cdot x_{i,j}\right)$, where $(\tilde{h}_{i,j}\!\cdot\! x_{i,j})$ is the energy contributed by the external magnetic field on the particle at $(i,j)$;
\item $f^{(\beta)}(x_{i,j},x_{i',j'})\defeq \exp\left(-\beta\cdot x_{i,j}\, x_{i',j'}\right)$, where $(x_{i,j}\!\cdot\! x_{i',j'})$ is the energy contributed by the ferromagnetic interaction between the particles $(i,j)$ and $(i',j')$.
\end{itemize}
In this case, the \emph{partition function}\index{partition function} (as a function of $\beta$) of the system is given by 
\begin{equation}
Z(\beta) = \sum_{\vx_{1,1}^{n,n}}\exp\Bigg( -\beta \cdot \Bigg( \sum_{i,j=1,\ldots,n} \tilde{h}_{i,j} \cdot x_{i,j} + \hspace{-16pt} \sum_{i,j,i'j'\in\{1,\ldots,n\} \atop \abs{i-i'}+\abs{j-j'}=1} \hspace{-16pt} x_{i,j} x_{i',j'} \Bigg) \Bigg) = \sum_{\vx_{1,1}^{n,n}} g^{(\beta)}(\vx_{1,1}^{n,n}),
\end{equation}
and the \emph{Helmholtz free energy}\index{free energy!Helmholtz} of the system is
\begin{equation}
F = -\beta^{-1} \log{Z(\beta)},
\end{equation}
which is the maximum useful work obtainable from the system.\footnote{Note that the term ``partition function'' and ``Helmholtz free energy'' are established terms in thermodynamics, and describe a physical system's statistical properties. These two terms were later borrowed, and defined differently in the context of statistical models and graphical methods.}
\end{example}
%*******************************************************************************
We borrow the terms \emph{partition function} and \emph{Helmholtz free energy} from statistical mechanics.
\begin{definition}[Partition Sum, Free energy]\label{def:parition:free:energy}
Given a factor graph $\set{G}$ describing the factorization $g(\vx)=\prod_{a\in\set{F}}f_a(\vx_\nb{a})$, the \emph{partition sum}\index{partition sum} of $\set{G}$ is defined as
\begin{equation}
Z(\set{G})\defeq\sum_x g(x) = \sum_x \prod_{a\in\set{F}} f_{a}(x_{\partial a}).
\end{equation}
Also, the \emph{(Helmholtz) free energy}\index{free energy!Helmholtz} of $\set{G}$ is defined as $\helmholtz(\set{G})\defeq -\log{Z(\set{G})}$.
\end{definition}
%*******************************************************************************
\begin{example}[Bit-wise Decoding of an LDPC Code~\cite{gallager1962low}]
Figure~\ref{fig:cfg:2:b} depicts the factorization of the conditional distribution of the inputs $(\rv{X}_1,\ldots,\rv{X}_n)$ given the detected outputs $(y_1,\ldots,y_n)$, using a regular LDPC code and {i.i.d.}~copies of a memoryless channel.
Namely,
    \[
    \prob_{\rv{X}_1^n|\rv{Y}_1^n}(\vx_1^n|\vy_1^n) \propto
    g(\vx_1^n,\vy_1^n) \defeq
    \prod_{\ell=1}^n \prob_{\rv{Y}|\rv{X}}(y_\ell|x_\ell)\cdot\prod_{a=1}^d f_+(\vx_\nb{a}),
    \]
where 
    \begin{itemize}
    \item the input/output alphabet are some finite sets;
    \item $\prob_{\rv{Y}|\rv{X}}$ is the conditional \pmf characterizing the memoryless channel in use;
    \item $d$ is a positive integer smaller than $n$ and where, for each $a\in\{1,\ldots,d\}$, $\nb{a}\subset\{1,\ldots,n\}$ and $\size{\nb{a}}=k$ for some positive integer $k\ll n$;
    \item $f_+(\vect{v})\defeq \begin{cases}1 & \sum_{i=1}^k v_i = 0_\mathbb{F}\\ 0 & \text{otherwise}\end{cases},\ \forall \vect{v}\in\mathbb{F}^{k}$.
    \end{itemize}
In this case, bit-wise decoding is equivalent to computing the marginals
\[
\prob_{\rv{X}_i|\rv{Y}_1^n}(x_i|\vy_1^n) \propto
\sum_{\vx_{\{1,\ldots,n\}\xk{i}}} g(\vx_1^n,\vy_1^n) =
\sum_{\vx_{\{1,\ldots,n\}\xk{i}}} \prod_{\ell=1}^n \prob_{\rv{Y}|\rv{X}}(y_\ell|x_\ell) \cdot \prod_{a=1}^d f_+(\vx_\nb{a})
\]
for each $i=1,\ldots,n$.
\end{example}
%*******************************************************************************
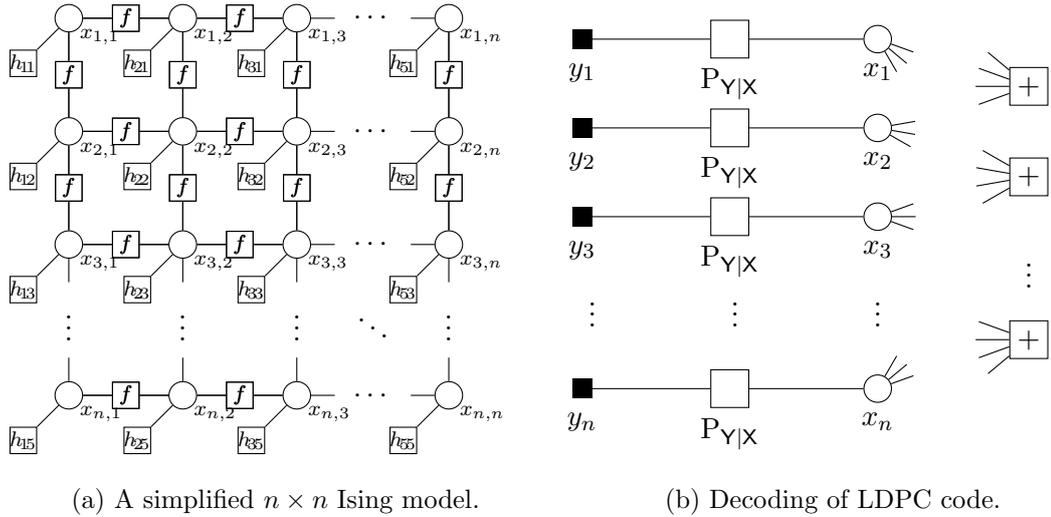
\begin{figure}
%*******************************************************************************
\begin{subfigure}[b]{.5\textwidth}
\begin{tikzpicture}[node/.style={draw=none},
    factor/.style={rectangle, minimum size=.35cm, draw},
    varNode/.style={circle, minimum size=.3cm, font=\small,draw},
    node distance=1.5cm]
	\pgfmathsetmacro\dis{1.5};
	\foreach \x in {1,2,3,5}{\foreach\y in {1,2,3,5}{
		\ifthenelse{\x=5}{
			\def\X{n};
			\pgfmathsetmacro\posX{\dis*(\x-2)+2};
		}{
			\def\X{\x};
			\pgfmathsetmacro\posX{\dis*\x};
		};
		\ifthenelse{\y=5}{
			\def\Y{n};
			\pgfmathsetmacro\posY{-\dis*(\y-2)-2};
		}{
			\def\Y{\y};
			\pgfmathsetmacro\posY{-\dis*\y};
		};
		\node[varNode] (x\x\y) at (\posX cm,\posY cm){};
		\node [below right = -.17cm of x\x\y, font=\scriptsize] {$x_{\Y,\X}$};
	}}
	\foreach \x in {1,2,3,5}{
		\path (x\x3) -- (x\x5) coordinate[midway] (m\x4);
		%\pgfmathsetmacro\posX{\dis*\x};
		%\pgfmathsetmacro\posY{-\dis*4};
		%\node at (\posX cm,\posY cm){$\vdots$};
		\node at (m\x4){$\vdots$};
	}
	\foreach \y in {1,2,3,5}{
		%\pgfmathsetmacro\posX{\dis*4};
		%\pgfmathsetmacro\posY{-\dis*\y};
		%\node at (\posX cm,\posY cm){$\cdots$};
		\path (x3\y) -- (x5\y) coordinate[midway] (m4\y);
		\node at (m4\y){$\cdots$};
	}
	\path (x33) -- (x55) coordinate[midway] (m44);
	\node at (m44){$\ddots$};
	%\pgfmathsetmacro\posX{\dis*4};
	%\pgfmathsetmacro\posY{-\dis*4};
	%\node at (\posX cm,\posY cm){$\ddots$};
	\foreach \x in {1,2,3,5}{\foreach\y in {1,2,3,5}{
		\ifthenelse{\x>1}{
			\ifthenelse{\x=5}{
				\path (x\x\y) edge ([xshift=-.5cm]x\x\y);
			}{
				\pgfmathtruncatemacro\X{\x-1};
				\path (x\x\y) -- (x\X\y) coordinate[midway] (m\x\y\X\y);
				\node[factor] (f\x\y\X\y) at (m\x\y\X\y) {};
				\node[font=\scriptsize] at (m\x\y\X\y) {$f$};
				\path (x\x\y) edge (f\x\y\X\y);
				\path (f\x\y\X\y) edge (x\X\y);
			}
		}{};
		\ifthenelse{\y>1}{
			\ifthenelse{\y=5}{
				\path (x\x\y) edge ([yshift=.5cm]x\x\y);
			}{
				\pgfmathtruncatemacro\Y{\y-1};
				\path (x\x\y) -- (x\x\Y) coordinate[midway] (m\x\y\x\Y);
				\node[factor] (f\x\y\x\Y) at (m\x\y\x\Y) {};
				\node[font=\scriptsize] at (m\x\y\x\Y) {$f$};
				\path (x\x\y) edge (f\x\y\x\Y);
				\path (f\x\y\x\Y) edge (x\x\Y);
			}
		}{};
		\ifthenelse{\x<5}{
			\ifthenelse{\x=3}{
				\path (x\x\y) edge ([xshift=.5cm]x\x\y);
			}{
				\pgfmathtruncatemacro\X{\x+1};
				\path (x\x\y) -- (x\X\y) coordinate[midway] (m\x\y\X\y);
				\node[factor] (f\x\y\X\y) at (m\x\y\X\y) {};
				\node[font=\scriptsize] at (m\x\y\X\y) {$f$};
				\path (x\x\y) edge (f\x\y\X\y);
				\path (f\x\y\X\y) edge (x\X\y);
			}
		}{};
		\ifthenelse{\y<5}{
			\ifthenelse{\y=3}{
				\path (x\x\y) edge ([yshift=-.5cm]x\x\y);
			}{
				\pgfmathtruncatemacro\Y{\y+1};
				\path (x\x\y) -- (x\x\Y) coordinate[midway] (m\x\y\x\Y);
				\node[factor] (f\x\y\x\Y) at (m\x\y\x\Y) {};
				\node[font=\scriptsize] at (m\x\y\x\Y) {$f$};
				\path (x\x\y) edge (f\x\y\x\Y);
				\path (f\x\y\x\Y) edge (x\x\Y);
			}
		}{};
	}}
	\foreach \x in {1,2,3,5}{\foreach\y in {1,2,3,5}{
		\node[factor] (h\x\y) [below left = .4cm of x\x\y] {};
		\node [font=\scriptsize] at (h\x\y) {$h_{\!\x\!\y}$};
		\path (x\x\y) edge (h\x\y);
	}}
\end{tikzpicture}
\caption{A simplified $n\times n$ Ising model.}\label{fig:cfg:2:a}
\end{subfigure}
%*******************************************************************************
\begin{subfigure}[b]{.5\textwidth}
\begin{tikzpicture}[node/.style={draw=none},
    factor/.style={rectangle, minimum size=.5cm, draw},
    varNode/.style={circle, minimum size=.3cm, font=\small, draw},
    fixedNode/.style={rectangle, minimum size=.2cm, fill},
    node distance=1.2cm]
	\node[factor] (f1) {}; \node at (f1) {$+$};
	\node[factor] (f2)[below of=f1] {}; \node at (f2) {$+$};
	\node (f)[below of=f2] {$\vdots$}; %\node at (f) {};
	\node[factor] (fn)[below=.3cm of f] {}; \node at (fn) {$+$};

	\path (f2) -- (f) coordinate[midway] (M);
	\node[varNode] (y3) [left=1.8cm of M] {};\node [below=0pt of y3] {$x_3$};
	\node[varNode] (y2) [above=.8cm of y3] {}; \node [below=0pt of y2] {$x_2$};
	\node[varNode] (y1) [above=.8cm of y2] {}; \node [below=0pt of y1] {$x_1$};
	\node (y) [below=.6cm of y3] {$\vdots$};
	\node[varNode] (yn) [below=.5cm of y] {}; \node [below=0pt of yn] {$x_n$};
	
	\pgfmathsetmacro\rad{.5};
	\pgfmathsetmacro\y{\rad *cos(70)};
	\pgfmathsetmacro\x{\rad *sin(70)};
	\path (y1) edge[draw] ([xshift=\x cm, yshift=-\y cm]y1);
	\pgfmathsetmacro\y{\rad *cos(50)};
	\pgfmathsetmacro\x{\rad *sin(50)};
	\path (y1) edge[draw] ([xshift=\x cm, yshift=-\y cm]y1);
	\pgfmathsetmacro\y{\rad *cos(30)};
	\pgfmathsetmacro\x{\rad *sin(30)};
	\path (y1) edge[draw] ([xshift=\x cm, yshift=-\y cm]y1);
	
	\pgfmathsetmacro\y{\rad *cos(100)};
	\pgfmathsetmacro\x{\rad *sin(100)};
	\path (y2) edge[draw] ([xshift=\x cm, yshift=-\y cm]y2);
	\pgfmathsetmacro\y{\rad *cos(80)};
	\pgfmathsetmacro\x{\rad *sin(80)};
	\path (y2) edge[draw] ([xshift=\x cm, yshift=-\y cm]y2);
	\pgfmathsetmacro\y{\rad *cos(60)};
	\pgfmathsetmacro\x{\rad *sin(60)};
	\path (y2) edge[draw] ([xshift=\x cm, yshift=-\y cm]y2);
	
	\pgfmathsetmacro\y{\rad *cos(110)};
	\pgfmathsetmacro\x{\rad *sin(110)};
	\path (y3) edge[draw] ([xshift=\x cm, yshift=-\y cm]y3);
	\pgfmathsetmacro\y{\rad *cos(90)};
	\pgfmathsetmacro\x{\rad *sin(90)};
	\path (y3) edge[draw] ([xshift=\x cm, yshift=-\y cm]y3);
	\pgfmathsetmacro\y{\rad *cos(70)};
	\pgfmathsetmacro\x{\rad *sin(70)};
	\path (y3) edge[draw] ([xshift=\x cm, yshift=-\y cm]y3);
	
	\pgfmathsetmacro\y{\rad *cos(110)};
	\pgfmathsetmacro\x{\rad *sin(110)};
	\path (yn) edge[draw] ([xshift=\x cm, yshift=-\y cm]yn);
	\pgfmathsetmacro\y{\rad *cos(130)};
	\pgfmathsetmacro\x{\rad *sin(130)};
	\path (yn) edge[draw] ([xshift=\x cm, yshift=-\y cm]yn);
	\pgfmathsetmacro\y{\rad *cos(150)};
	\pgfmathsetmacro\x{\rad *sin(150)};
	\path (yn) edge[draw] ([xshift=\x cm, yshift=-\y cm]yn);
	
	\pgfmathsetmacro\rad{.7};
	\pgfmathsetmacro\y{\rad *cos(110)};
	\pgfmathsetmacro\x{\rad *sin(110)};
	\path (f1) edge[draw] ([xshift=-\x cm, yshift=-\y cm]f1);
	\pgfmathsetmacro\y{\rad *cos(130)};
	\pgfmathsetmacro\x{\rad *sin(130)};
	\path (f1) edge[draw] ([xshift=-\x cm, yshift=-\y cm]f1);
	\pgfmathsetmacro\y{\rad *cos(90)};
	\pgfmathsetmacro\x{\rad *sin(90)};
	\path (f1) edge[draw] ([xshift=-\x cm, yshift=-\y cm]f1);
	\pgfmathsetmacro\y{\rad *cos(70)};
	\pgfmathsetmacro\x{\rad *sin(70)};
	\path (f1) edge[draw] ([xshift=-\x cm, yshift=-\y cm]f1);
	
	\pgfmathsetmacro\y{\rad *cos(100)};
	\pgfmathsetmacro\x{\rad *sin(100)};
	\path (f2) edge[draw] ([xshift=-\x cm, yshift=-\y cm]f2);
	\pgfmathsetmacro\y{\rad *cos(120)};
	\pgfmathsetmacro\x{\rad *sin(120)};
	\path (f2) edge[draw] ([xshift=-\x cm, yshift=-\y cm]f2);
	\pgfmathsetmacro\y{\rad *cos(80)};
	\pgfmathsetmacro\x{\rad *sin(80)};
	\path (f2) edge[draw] ([xshift=-\x cm, yshift=-\y cm]f2);
	\pgfmathsetmacro\y{\rad *cos(60)};
	\pgfmathsetmacro\x{\rad *sin(60)};
	\path (f2) edge[draw] ([xshift=-\x cm, yshift=-\y cm]f2);
	
	\pgfmathsetmacro\y{\rad *cos(70)};
	\pgfmathsetmacro\x{\rad *sin(70)};
	\path (fn) edge[draw] ([xshift=-\x cm, yshift=-\y cm]fn);
	\pgfmathsetmacro\y{\rad *cos(50)};
	\pgfmathsetmacro\x{\rad *sin(50)};
	\path (fn) edge[draw] ([xshift=-\x cm, yshift=-\y cm]fn);
	\pgfmathsetmacro\y{\rad *cos(90)};
	\pgfmathsetmacro\x{\rad *sin(90)};
	\path (fn) edge[draw] ([xshift=-\x cm, yshift=-\y cm]fn);
	\pgfmathsetmacro\y{\rad *cos(110)};
	\pgfmathsetmacro\x{\rad *sin(110)};
	\path (fn) edge[draw] ([xshift=-\x cm, yshift=-\y cm]fn);
	
	\node[factor] (c1) [left=1.5cm of y1] {}; \node[below=0cm of c1] {$\prob_{\rv{Y}|\rv{X}}$};
	\node[factor] (c2) [left=1.5cm of y2] {}; \node[below=0cm of c2] {$\prob_{\rv{Y}|\rv{X}}$};
	\node[factor] (c3) [left=1.5cm of y3] {}; \node[below=0cm of c3] {$\prob_{\rv{Y}|\rv{X}}$};
	\node (c) [left=1.5cm of y] {$\vdots$};
	\node[factor] (cn) [left=1.5cm of yn] {}; \node[below=0cm of cn] {$\prob_{\rv{Y}|\rv{X}}$};
	
	\node[varNode, draw=none] (x1) [left=1.5cm of c1] {}; \node[below=0cm of x1] {$y_1$};
	\node[fixedNode](X1) at (x1) {};
	\node[varNode, draw=none] (x2) [left=1.5cm of c2] {}; \node[below=0cm of x2] {$y_2$};
	\node[fixedNode](X2) at (x2) {};
	\node[varNode, draw=none] (x3) [left=1.5cm of c3] {}; \node[below=0cm of x3] {$y_3$};
	\node (x) [left=1.5cm of c] {$\vdots$}; 
	\node[fixedNode](X3) at (x3) {};
	\node[varNode, draw=none] (xn) [left=1.5cm of cn] {}; \node[below=0cm of xn] {$y_n$};
	\node[fixedNode](Xn) at (xn) {};
	
	\draw (x1.center) -- (c1) -- (y1);
	\draw (x2.center) -- (c2) -- (y2);
	\draw (x3.center) -- (c3) -- (y3);
	\draw (xn.center) -- (cn) -- (yn);
\end{tikzpicture}
\caption{Decoding of LDPC code.}\label{fig:cfg:2:b}
\end{subfigure}
%*******************************************************************************
\caption{Examples of applications of factor graphs.}
\end{figure}
%*******************************************************************************
% Closing the box operations and message-passing
\subsection{Computing the Marginals/Partition Sums of Acyclic Factor Graphs}\label{subsec:marginal:acyclic:FGs}
In both examples in the last section, the quantities of interest are some marginals (or partition sum) of the factorization involved.
However, direct computation of the marginals (or partition sum) is not scalable, as the number of configurations grows exponentially \wrt the number of variables.
Utilizing the distributivity of $\cdot$ over $+$, one can compute the marginals more efficiently in some, but not all, cases (\eg, a complete bipartite factor graph).
This is particularly the case for acyclic factor graphs, where marginals can be computed by summing over exactly \emph{one} variable at each step.
For example, the marginal \wrt $x_0$ of the NFG in Figure~\ref{fig:CtB:Step_by_Step}(0) can be computed in $n+1$ steps as
\begin{align}\label{eq:markov:1}
p(x_0) &\defeq \sum_{x_1,\ldots,x_n} \mu(x_0) \cdot \prod_{\ell=1}^n p(x_\ell|x_{\ell-1})\\
&= \mu(x_0) \cdot \sum_{x_1} \!\left(\! p(x_1|x_0)\sum_{x_2} \!\left(\! p(x_2|x_1) \cdots \sum_{x_{n-1}} \!\left(\! p(x_{n-1}|x_{n-2}) \sum_{x_n} p(x_n|x_{n-1}) \!\right)\! \cdots \!\right)\! \!\right)\!. \nonumber
\end{align}
Such computations (by exploiting distributivity) can be illustrated as a sequence of ``closing-the-box'' (CtB) operations\index{closing-the-box operations}, where at each step we replace the subgraph \emph{in the box} with the result of the summation over the variables \emph{inside} the box (as shown in Figure~\ref{fig:CtB:Step_by_Step}(1)--($n\!+\!1$)).
\begin{figure}\centering
\begin{tikzpicture}[node/.style={draw=none},
	factor/.style={rectangle, minimum size=.7cm, draw},
	node distance=2cm]
	% Original NFG
	\node[factor] (X00) {$\mu$};
	\node[left = 1.5cm of X00.west, anchor = west] {0)};
	\node[factor, right of = X00] (X01) {$\prob$};
	\node[factor, right of = X01] (X02) {$\prob$};
	\node[right of = X02] (X0m) {$\cdots$};
	\node[factor,right of = X0m] (X0n-1) {$\prob$};
	\node[factor,right of = X0n-1] (X0n) {$\prob$};
	\path (X00) edge[draw] node[above=0] {$x_0$} (X01);
	\path (X01) edge[draw] node[above=0] {$x_1$} (X02);
	\path (X02) edge[draw] node[above=0] {$x_2$} (X0m);
	\path (X0m) edge[draw] node[above=0] {$x_{n\!-\!2}$} (X0n-1);
	\path (X0n-1) edge[draw] node[above=0] {$x_{n\!-\!1}$} (X0n);
	\path (X0n) edge[draw] node[above=0] {$x_n$} ([xshift=1.5cm]X0n);
	% Step 1
	\node[factor, below = 1cm of X00] (X10) {$\mu$};
	\node[left = 1.5cm of X10.west, anchor = west] {1)};
	\node[factor, right of = X10] (X11) {$\prob$};
	\node[factor, right of = X11] (X12) {$\prob$};
	\node[right of = X12] (X1m) {$\cdots$};
	\node[factor,right of = X1m] (X1n-1) {$\prob$};
	\node[factor,right of = X1n-1] (X1n) {$\prob$};
	\path (X10) edge[draw] node[above=0] {$x_0$} (X11);
	\path (X11) edge[draw] node[above=0] {$x_1$} (X12);
	\path (X12) edge[draw] node[above=0] {$x_2$} (X1m);
	\path (X1m) edge[draw] node[above=0] {$x_{n\!-\!2}$} (X1n-1);
	\path (X1n-1) edge[draw] node[above=0] {$x_{n\!-\!1}$} (X1n);
	\path (X1n) edge[draw] node[above=0] {$x_n$} ([xshift=1.5cm]X1n);
	\draw[dashed] ([xshift=-2mm,yshift=2mm]X1n.north west) rectangle ([xshift=17mm,yshift=-2mm]X1n.south east);
	% Step 2
	\node[factor, below = 1cm of X10] (X20) {$\mu$};
	\node[left = 1.5cm of X20.west, anchor = west] {2)};
	\node[factor, right of = X20] (X21) {$\prob$};
	\node[factor, right of = X21] (X22) {$\prob$};
	\node[right of = X22] (X2m) {$\cdots$};
	\node[factor,right of = X2m] (X2n-1) {$\prob$};
	\node[factor,right of = X2n-1, pattern=north west lines] (X2n) {};
	\path (X20) edge[draw] node[above=0] {$x_0$} (X21);
	\path (X21) edge[draw] node[above=0] {$x_1$} (X22);
	\path (X22) edge[draw] node[above=0] {$x_2$} (X2m);
	\path (X2m) edge[draw] node[above=0] {$x_{n\!-\!2}$} (X2n-1);
	\path (X2n-1) edge[draw] node[above=0] {$x_{n\!-\!1}$} (X2n);
	\draw[dashed] (X2n.north west) -- ([xshift=-2mm,yshift=-2mm]X1n.south west);
	\draw[dashed] (X2n.north east) -- ([xshift=17mm,yshift=-2mm]X1n.south east);
	\draw[dashed] ([xshift=-2mm,yshift=2mm]X2n-1.north west) rectangle ([xshift=2mm,yshift=-2mm]X2n.south east);
	% Step 3
	\node[factor, below = 1cm of X20] (X30) {$\mu$};
	\node[left = 1.5cm of X30.west, anchor = west] (l3) {3)};
	\node[factor, right of = X30] (X31) {$\prob$};
	\node[factor, right of = X31] (X32) {$\prob$};
	\node[right of = X32] (X3m) {$\cdots$};
	\node[factor,right of = X3m, pattern=north west lines] (X3n-1) {};
	\path (X30) edge[draw] node[above=0] {$x_0$} (X31);
	\path (X31) edge[draw] node[above=0] {$x_1$} (X32);
	\path (X32) edge[draw] node[above=0] {$x_2$} (X3m);
	\path (X3m) edge[draw] node[above=0] {$x_{n\!-\!2}$} (X3n-1);
	\draw[dashed] (X3n-1.north west) -- ([xshift=-2mm,yshift=-2mm]X2n-1.south west);
	\draw[dashed] (X3n-1.north east) -- ([xshift=2mm,yshift=-2mm]X2n.south east);
	\draw[dashed] ([xshift=-2mm,yshift=2mm]X32.north west) rectangle ([xshift=2mm,yshift=-2mm]X3n-1.south east);
	% Step n
	\node[factor, below = 1.5cm of X30] (X40) {$\mu$};
	\node[left = 1.5cm of X40.west, anchor = west] (ln) {$n$)};
	\path (l3) edge[draw = none] node {$\vdots$} (ln);
	\node[factor, right of = X40] (X41) {$\prob$};
	\node[factor, right of = X41, pattern=north west lines] (X42) {};
	\path (X40) edge[draw] node[above=0] {$x_0$} (X41);
	\path (X41) edge[draw] node[above=0] {$x_1$} (X42);
	\path (X42.north west) edge[draw,dashed] node (mark1) {} ([xshift=-2mm,yshift=-2mm]X32.south west);
	\path (X42.north east) edge[draw,dashed] node (mark2) {} ([xshift=2mm,yshift=-2mm]X3n-1.south east);
	\path (mark1) edge[draw=none] node {$\cdots$} (mark2);
	\draw[dashed] ([xshift=-2mm,yshift=2mm]X41.north west) rectangle ([xshift=2mm,yshift=-2mm]X42.south east);
	% Step n+1
	\node[factor, below = 1cm of X40] (X50) {$\mu$};
	\node[left = 1.5cm of X50.west, anchor = west] {$n\!+\!1$)};
	\node[factor, right of = X50, pattern=north west lines] (X51) {};
	\path (X50) edge[draw] node[above=0] {$x_0$} (X51);
	\draw[dashed] (X51.north west) -- ([xshift=-2mm,yshift=-2mm]X41.south west);
	\draw[dashed] (X51.north east) -- ([xshift=2mm,yshift=-2mm]X42.south east);
\end{tikzpicture}
\caption{A chain NFG, and the process to compute the marginal \wrt $x_0$ as described by~\eqref{eq:markov:1}.}
\label{fig:CtB:Step_by_Step}
\end{figure}
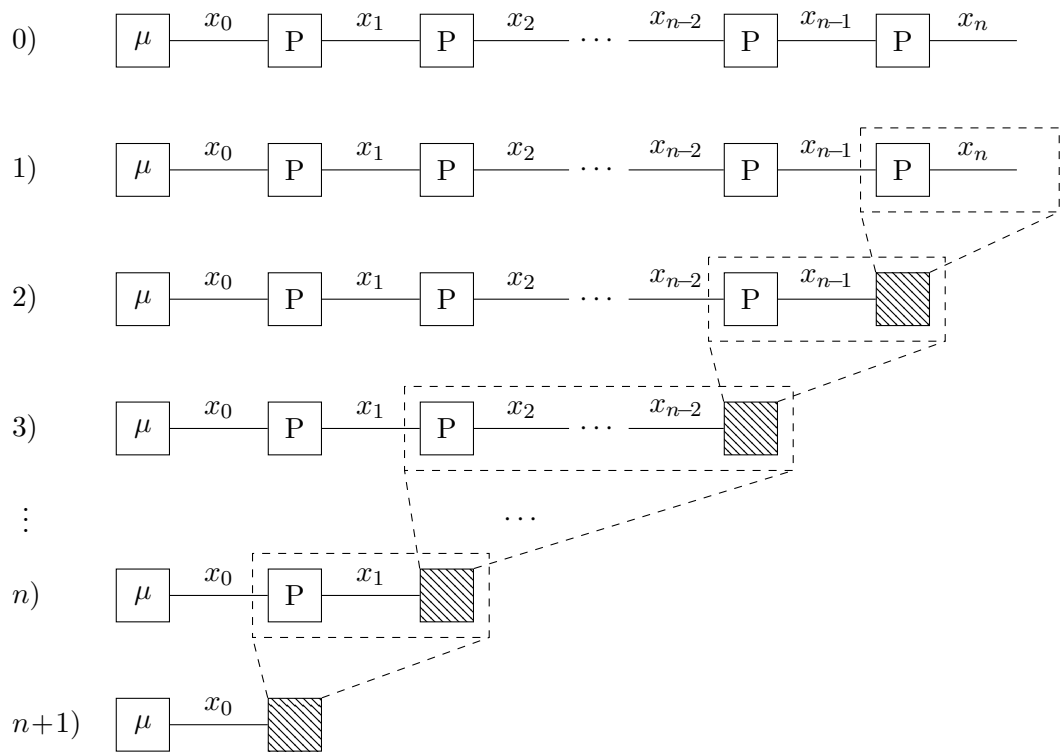
\par
%*******************************************************************************
% Message-passing interpretation
Equivalently, the above process can also be described as a \emph{message-passing algorithm}\index{message-passing algorithm}.
Namely, for an acyclic factor graph $\left((\set{V},\set{F},\set{E}),\{\set{X}_i\}_{i\in\set{V}},\{f_a\}_{a\in\set{F}}\right)$, there exists a unique set of functions (\aka \emph{messages}) $\{m_{i\to a},m_{a\to i}:\set{X}_i\to \Reals\}_{(i,a)\in\set{E}}$ such that
\begin{align}
    \label{eq:msg:acyclic:update:1}
    m_{i\to a}(x_i) &= \prod_{c\in\nb{i}\xk{a}} m_{c\to i}(x_i)\\
    \label{eq:msg:acyclic:update:2}
	m_{a\to i}(x_i) &= \sum_{\vx_{\nb{a}\xk{i}}} f_a(\vx_\nb{a})\cdot\prod_{j\in\nb{a}\xk{i}} m_{j\to a}(x_j)
\end{align}
for all $(i,a)\in\set{E}$, where we treat a vacuous product as constant 1.
Intuitively speaking, since the factor graph is acyclic, one can always construct a set of messages satisfying~\eqref{eq:msg:acyclic:update:1} and~\eqref{eq:msg:acyclic:update:2} from the leaves up to the root.
We state and prove this result as the following theorem.
\begin{theorem}\label{thm:BP:tree}
Consider an acyclic factor graph $((\set{V},\set{F},\set{E}),\{\set{X}_i\}_{i\in\set{V}},\{f_a\}_{a\in\set{F}})$.
There exists a unique set of messages $\{m_{i\to a},m_{a\to i}:\set{X}_i\to \Reals\}_{(i,a)\in\set{E}}$ satisfying~\eqref{eq:msg:acyclic:update:1} and~\eqref{eq:msg:acyclic:update:2}, and
\begin{align}
	\label{eq:msg:bi}
    \sum_{\vx_{\set{V}\xk{i}}} \prod_{a\in\set{F}} f_a(\vx_\nb{a}) &= \prod_{c\in\nb{i}} m_{c\to i}(x_i) && \forall x_i,\, \forall i\in\set{V},\\
    \label{eq:msg:ba}
    \sum_{\vx_{\set{V}\xk{\nb{a}}}} \prod_{c\in\set{F}} f_{c}(\vx_\nb{c}) &= f_a(\vx_\nb{a})\cdot\prod_{i\in\nb{a}} m_{i\to a}(x_i) && \forall\vx_\nb{a},\, \forall a\in\set{F}.
\end{align}
\end{theorem}
\begin{proof}
Without loss of generality, we assume the graph $(\set{V},\set{F},\set{E})$ to be connected. \par
\textbf{Existence:}
Since the bipartite graph $(\set{V},\set{F},\set{E})$ is acyclic, by removing any edge $(i,a)\in\set{E}$, the graph is split into two disjoint bipartite graphs $(\set{V}_i^a,\set{F}_i^a,\set{E}_i^a)$ and $(\set{V}_a^i,\set{F}_a^i,\set{E}_a^i)$ where $i\in\set{V}_i^a$, $a\in\set{F}_a^i$, and $(\set{V},\set{F},\set{E}) = (\set{V}_i^a\sqcup\set{V}_a^i,\set{F}_i^a\sqcup\set{F}_a^i,\set{E}_i^a\sqcup\{(i,a)\}\sqcup\set{E}_a^i)$.
For each $(i,a)\in\set{E}$, we define
	\begin{align*}
	m_{i\to a}(x_i) &\defeq \sum_{x_j:\:j\in\set{V}_i^a\xk{i}}\ \prod_{c\in\set{F}_i^a} f_c(\vx_\nb{c}),\\
	m_{a\to i}(x_i) &\defeq \sum_{x_j:\:j\in\set{V}_a^i}\ \prod_{c\in\set{F}_a^i} f_c(\vx_\nb{c}).
	\end{align*}
Because the bipartite graph $(\set{V},\set{F},\set{E})$ is acyclic, we observe:
\begin{enumerate}
\item For each $(i,a)\in\set{E}$, $(\set{V}_i^a,\set{F}_i^a,\set{E}_i^a)=(\{i\}\cup\bigsqcup_{c\in\nb{i}\xk{a}}\set{V}_c^i, \bigsqcup_{c\in\nb{i}\xk{a}}\set{F}_c^i,\bigsqcup_{c\in\nb{i}\xk{a}}\set{E}_c^i)$,
\item For each $(i,a)\in\set{E}$, $(\set{V}_a^i,\set{F}_a^i,\set{E}_a^i)=(\bigsqcup_{j\in\nb{a}\xk{i}}\set{V}_j^a,\{a\}\cup\bigsqcup_{j\in\nb{a}\xk{i}}\set{F}_j^a,\bigsqcup_{j\in\nb{a}\xk{i}}\set{E}_j^a)$,
\item For each $i\in\set{V}$, $(\set{V},\set{F},\set{E})=(\{i\}\cup\bigsqcup_{a\in\nb{i}}\set{V}_a^i,\bigsqcup_{a\in\nb{i}}\set{F}_a^i,\setdef{(i,a)}{a\in\nb{i}}\cup\bigsqcup_{a\in\nb{i}}\set{E}_a^i)$,
\item For each $a\in\set{F}$, $(\set{V},\set{F},\set{E})=(\bigsqcup_{i\in\nb{a}}\set{V}_i^a,\{a\}\cup\bigsqcup_{i\in\nb{a}}\set{F}_i^a,\setdef{(i,a)}{i\in\nb{a}}\cup\bigsqcup_{i\in\nb{a}}\set{E}_i^a)$,
\end{enumerate}
which implies~\eqref{eq:msg:acyclic:update:1},~\eqref{eq:msg:acyclic:update:2},~\eqref{eq:msg:bi}, and~\eqref{eq:msg:ba}, respectively.
\par
\textbf{Uniqueness:}
Assume there exist two different sets of messages $\{m_{i\to a},m_{a\to i}\}_{(i,a)\in\set{E}}$ and $\{m'_{i\to a},m'_{a\to i}\}_{(i,a)\in\set{E}}$, both satisfying~\eqref{eq:msg:acyclic:update:1},~\eqref{eq:msg:acyclic:update:2},~\eqref{eq:msg:bi}, and~\eqref{eq:msg:ba}.
Without loss of generality, let $m_{i_0\to a_0}\neq m'_{i_0\to a_0}$.
By invoking~\eqref{eq:msg:acyclic:update:1} and~\eqref{eq:msg:acyclic:update:2} alternatively, one can construct a path (as long as possible) $(a_0-i_0-a_i-i_1-\cdots)$ such that $m_{a_k\to i_{k\!-\!1}}\neq m'_{a_k\to i_{k\!-\!1}}$ and $m_{i_k\to a_k}\neq m'_{i_k\to a_k}$ for each $k$.
Since $(\set{V},\set{F},\set{E})$ is finite and acyclic, such a path must be of finite length, with the last vertex being a leaf of the graph.
If the last vertex is from $\set{V}$, say $i_\text{last}$, we have $m_{i_\text{last}\to a_\text{last}} = \mathbf{1} = m_{i_\text{last} \to a_\text{last}}$.
Otherwise, if it is from $\set{F}$, say $a_\text{last}$, we have $m_{a_\text{last} \to i_{\text{last}\!-\!1}} = f_{a_\text{last}} = m'_{a_\text{last}\to i_{\text{last}\!-\!1}}$.
In both cases, we have a contradiction.
Therefore, the set of messages satisfying~\eqref{eq:msg:acyclic:update:1},~\eqref{eq:msg:acyclic:update:2},~\eqref{eq:msg:bi}, and~\eqref{eq:msg:ba} must be unique.
\end{proof}
\par
%*******************************************************************************
% Belief-Propagation Algorithm
\subsection{Belief-Propagation Algorithms \& BP Fixed Points}
%*******************************************************************************
% Introducing BPA and BP fixed points
For generic factor graphs,  the belief-propagation (BP) algorithms\index{belief-propagation algorithm} are \emph{heuristic} generalizations of the message-passing algorithm described by~\eqref{eq:msg:acyclic:update:1} and~\eqref{eq:msg:acyclic:update:2}.
Such generalizations are made by initializing all the messages as constant functions and updating them according to~\eqref{eq:msg:acyclic:update:1} and~\eqref{eq:msg:acyclic:update:2}.
Namely,
\begin{align}
	\label{eq:msg:BP:update:1}
	m_{i\to a}^{(t)}(x_i) &\propto \prod_{c\in\nb{i}\xk{a}} m_{c\to i}^{(t)}(x_i) &&\forall (i,a)\in\set{E},\\
    \label{eq:msg:BP:update:2}
	m_{a\to i}^{(t)}(x_i) &\propto \sum_{\vx_{\nb{a}\xk{i}}} f_a(\vx_\nb{a})\cdot\prod_{j\in\nb{a}\xk{i}} m_{j\to a}^{(t-1)}(x_j) &&\forall (i,a)\in\set{E},
\end{align}
where $t\in\{1,2,\ldots\}$ and where the initial message $m_{i\to a}^{(0)}$ is some constant function for each $(i,a)\in\set{E}$.
Notice that, in the above updating rules, all new messages are computed using the old messages from the last ``batch''.
Such a \emph{sequence} of updates is known as the synchronous \emph{schedule}\index{schedule} \index{schedule!synchronous} (\aka flooding schedule\index{schedule!flooding}).
Though asynchronous schedules\index{schedule!asynchronous} do exist,\footnote{Thus, the belief-propagation algorithms are a class of algorithms instead of a single algorithm.} and can be helpful in many specific situations (see, \eg,~\cite{sharon2007efficient,fan2017scalable}), they are beyond our scope of discussion.
Algorithm~\ref{alg:bpa} lists the belief-propagation algorithm with the flooding schedule.
\begin{algorithm}
\caption{Belief-Propagation Algorithm (Flooding Schedule with Timeout)}
\label{alg:bpa}
\begin{algorithmic}[1]
	\Require{A factor graph $\Big(\set{G}=(\set{V},\set{F},\set{E}\subset \set{V}\times\set{F}),\mathfrak{V}=\{\set{X}\}_{i\in\set{V}},\mathfrak{F}=\{f_a\}_{a\in\set{F}}\Big)$, $\varepsilon \geqslant 0$;}
	\Ensure{Messages $\{m_{i\to a}, m_{a\to i}:\set{X}_i\to \mathbb{R}_{+}\}_{(i,a)\in\set{E}}$, $\mathsf{FLAG}\in\{\mathrm{completed},\mathrm{timeout}\}$.}
	\ForAll{$(i,a)\in\set{E}$}
	\State{$m_{i\to a}^{(0)}(x_i)\gets \size{\set{X}_i}^{-1}\ $ for each $x_i\in\set{X}_i$;}
	%\State{$m_{a\to i}^{(0)}(x_i)\gets \size{\set{X}_i}^{-1}\ $ for each $x_i\in\set{X}_i$;}
	\EndFor
	\State{$t\gets 0$;}
	\Do
	\State{$t\gets t+1$;}
	\ForAll{$(i,a)\in\set{E}$}
	%\State{$m_{i\to a}^{(t)}(x_i)\defpropto \prod_{c\in\nb{i}\xk{a}} m_{c\to i}^{(t-1)}(x_i)\ $ for each $x_i\in\set{X}_i$;}
	\State{$m_{a\to i}^{(t)}(x_i)\defpropto \sum_{\vx_{\nb{a}\xk{i}}}f(\vx_\nb{a})\cdot\prod_{j\in\nb{a}\xk{i}} m_{j\to a}^{(t-1)}(x_j)\ $ for each $x_i\in\set{X}_i$;}
	\EndFor
	\ForAll{$(i,a)\in\set{E}$}
	\State{$m_{i\to a}^{(t)}(x_i)\defpropto \prod_{c\in\nb{i}\xk{a}} m_{c\to i}^{(t)}(x_i)\ $ for each $x_i\in\set{X}_i$;}
	\EndFor
	\DoWhile{$\left(\neg\mathsf{timeout}\right) \land \left( \exists(i,a)\in\set{E} \text{ s.t. }\norm{m_{i\to a}^{(t)}-m_{i\to a}^{(t-1)}}_2>\varepsilon\text{ or }\norm{m_{a\to i}^{(t)}-m_{a\to i}^{(t-1)}}_2>\varepsilon\right)$}
	\Comment{$\mathsf{timeout}=\mathsf{false}$ unless the operating time exceeds a pre-selected waiting time.}
	\If{$\mathrm{timeout}$}
        \State{$\mathsf{FLAG}\gets\mathrm{timeout}$;}
	\Else
	    \State{$\mathsf{FLAG}\gets\mathrm{completed}$;}
	    \ForAll{$(i,a)\in\set{E}$}
	    \State{$m_{i\to a}\gets m_{i\to a}^{(t)}$;}
	    \State{$m_{a\to i}\gets m_{a\to i}^{(t)}$;}
	    \EndFor
	\EndIf
\end{algorithmic}
\end{algorithm}
\par
%*******************************************************************************
% BP Fixed Points
We consider an instance of BP algorithms to be ``completed'', if the messages $\{m_{i\to a}^{(t)}$, $m_{a\to i}^{(t)}\}_{(i,a)}$ ``converge'' under the updating rules~\eqref{eq:msg:BP:update:1} and~\eqref{eq:msg:BP:update:2}, namely when the changes of messages are \emph{negligible} after applying the updates.
In this sense, the result of a ``completed'' BP algorithm will always be (or at least be close to) a \emph{fixed point} of the update, namely a set of messages that stay unchanged \wrt \eqref{eq:msg:BP:update:1} and~\eqref{eq:msg:BP:update:2}.
We call such sets of messages BP fixed points (see definition below).
\begin{definition}[BP Fixed point]\label{def:BP:fixed:points}
\index{BP fixed points}
A set of messages $\{m_{i\to a}, m_{a\to i}\}_{i,a}$ is said to be a BP fixed point if
\begin{align}
\label{eq:msg:BP:fixed:1}
m_{i\to a}(x_i) &\propto \prod_{c\in\nb{i}\xk{a}} m_{c\to i}(x_i) &&\forall (i,a)\in\set{E},\\
\label{eq:msg:BP:fixed:2}
m_{a\to i}(x_i) &\propto \sum_{\vx_{\nb{a}\xk{i}}} f_a(\vx_\nb{a})\cdot\prod_{j\in\nb{a}\xk{i}} m_{j\to a}(x_j) &&\forall (i,a)\in\set{E}.
\end{align}
In this case, the set $\{m_{i\to a}, m_{a\to i}\}_{i,a}$ is also called a set of fixed-point messages.
\end{definition}
\begin{corollary}\label{cor:BP:tree}
Consider Algorithm~\ref{alg:bpa} with acyclic input factor graph and with no ``timeout'' constraint (\ie, the pre-selected waiting time infinity).
For any $\varepsilon\geqslant 0$, the algorithm will always end with ``$\ \mathsf{FLAG}=\mathrm{completed}$'' within finite time.
\end{corollary}
\begin{proof}
Let $\{m_{i\to a}, m_{a\to i}\}_{(i,a)\in\set{E}}$ be defined as in the proof of Theorem~\ref{thm:BP:tree}, and let $\{m^{(t)}_{i\to a}, m^{(t)}_{a\to i}\}_{(i,a)\in\set{E}}$ denote the messages in Algorithm~\ref{alg:bpa} at timestamp $t$.
For any $i,a,t$ such that $m^{(t)}_{a\to i}\not\propto m_{a\to i}$, using the same construction in the ``Uniqueness'' part of the proof of Theorem~\ref{thm:BP:tree}, one can construct a path $i-a-i_1-a_1-\cdots-a_{t-1}-i_t$ such that $m^{(0)}_{i_t\to a_{t-1}} \not\propto m_{i_t\to a_{t-1}}$. Since the factor graph is acyclic, the diameter of the factor graph must be at least $2t+1$.
Thus, for any $t\geqslant\ceil{\frac{D}{2}}$, $m^{(t)}_{a\to i}\propto m_{a\to i}$, where $D$ denotes the diameter of the graph.
Similarly, $m^{(t)}_{i\to a}\propto m_{i\to a}$ for all $t\geqslant\ceil{\frac{D}{2}}$.
In other words, for any $\varepsilon\geqslant 0$
\[
\norm{m_{i\to a}^{(t)}-m_{i\to a}^{(t-1)}}_2\leqslant\varepsilon\text{ and }\norm{m_{a\to i}^{(t)}-m_{a\to i}^{(t-1)}}_2\leqslant\varepsilon \quad \forall(i,a)\in\set{E} \quad \forall t\geqslant\ceil{\frac{D}{2}}+1,
\]
namely, the algorithm will ``complete'' within finite time.
\end{proof}
\par
As a direct result of Corollary~\ref{cor:BP:tree}, for acyclic factor graphs, taking $\varepsilon=0$, Algorithm~\ref{alg:bpa} (with no ``timeout'' constraint) will always produce the BP-fixed point within finite time.
For cyclic factor graphs in general, however, such results do not generalize.
On the one hand, there exist instances of BP algorithms that fail to converge.
(More interestingly, there are examples of Gaussian message passing algorithms\footnote{Gaussian message passing algorithms or GMPAs are a special class of BP algorithms on continuous-alphabet factor graphs with additional constraints on its factors.} (GMPAs) in which the convergence depends on the update schedule~\cite{fan2017scalable}.)
On the other hand, there also exist factor graphs with multiple BP fixed points.
Though the algorithm has been justified for some special classes of factor graphs, \eg, walk-summable Gaussian graphical models~\cite{malioutov2006walk}, more generic works have been focusing on various interpretations of BP fixed points~\cite{yedidia2005constructing, vontobel2013counting, chertkov2006loop}.
In the remainder of this section, we review two of these interpretations of BP fixed points for factor graphs with non-negative local functions.
%*******************************************************************************
\subsection{Interpretations of BP fixed points: The Variational Approach and Bethe's Approximation}
\label{subsec:FG:Bethe}
%*******************************************************************************
In this section, we review a method known as the variational approach together with Bethe's approximation (see, \eg,~\cite{yedidia2005constructing}).
The idea behind such methods has a deep root in statistical mechanics.
Note that computing the marginals is equivalent to computing the partition sums (of some slightly modified factor graphs) and thus is equivalent to computing the Helmholtz free energies.
The idea of this method consists of two parts:
The variational approach and the Bethe free energy.
%*******************************************************************************
\subsubsection{The Variational Approach}
The variational approach is a method adopted from variational mechanics, in which one considers the Helmholtz free energy as the minimal Gibbs free energy over all possible configurations.
The Gibbs free energy of a factor graph is defined as follows.
\begin{definition}[Gibbs free energy] \label{def:Gibbs:free:energy} \index{free energy!Gibbs}
Given a factor graph $\set{G}$ describing the factorization $g(\vx)=\prod_{a\in\set{F}}f_a(\vx_\nb{a})$ with $f_a\geqslant 0$ for each $a$, the \emph{Gibbs free energy} (\aka the variational free energy) of $\set{G}$ is defined as
\begin{equation}
\gibbs(b) \defeq -\sum_{a\in\set{F}}\sum_{\vx}b(\vx)\log{f_a(\vx_\nb{a})} + \sum_{\vx}b(\vx)\log{b(\vx)}
\end{equation}
for each \pmf $b(\vx)$ on $\bigtimes_{i\in\set{V}}\set{X}_i$.
Note that for the case where there exists some $\vx$ such that $b(\vx)>0$ and $f_a(\vx_a)=0$ for some $a$, we take the convention that $\gibbs(b) \defeq +\infty$.
\end{definition}
\begin{proposition}\label{prop:Gibbs:geq:Helmholtz}
    For any \pmf $b$ on $\bigtimes_{i\in\set{V}}\set{X}_i$, we have
	\begin{equation}
	\gibbs(b) \geqslant \helmholtz,
	\end{equation}
	with equality if and only if $b(\vx_\set{V}) \propto \prod_{a\in\set{F}}f_a(\vx_\nb{a})$.
\end{proposition}
\begin{proof}
The proof is done by showing $\gibbs(b)- \helmholtz=\infdiv{b}{p}$, where the \pmf $p(\vx)\defeq Z^{-1}\cdot \prod_{a\in\set{F}}f_a(\vx_\nb{a})$.
We omit the details.
\end{proof}
As a result of Proposition~\ref{prop:Gibbs:geq:Helmholtz}, the Helmholtz free energy $\helmholtz$ can be obtained by minimizing $\gibbs(b)$.
However, such a minimization problem is generally intractable since the problem's dimension grows exponentially \wrt the number of variables.
This has motivated physicists (since the 1900s) to develop ``good'' yet tractable approximations to $\gibbs$.
The Bethe free energy~\cite{bethe1935statistical} reviewed below is one of these approximations.
%*******************************************************************************
\subsubsection{The Bethe free energy}
\begin{definition}[Bethe free energy] \label{def:bethe:energy} \index{free energy!Bethe}
Given a factor graph $\set{G}$ describing the factorization $g(\vx)=\prod_{a\in\set{F}}f_a(\vx_\nb{a})$, the \emph{Bethe free energy} is the function
\begin{equation}\label{eq:def:bethe:free:energy}
\begin{aligned}
\bethe\left(\{b_a\}_{a\in\set{F}},\{b_i\}_{i\in\set{V}}\right) \defeq & -\sum_{a\in\set{F}} \sum_{\vx_\nb{a}} b_a(\vx)\log{f_a(\vx_\nb{a})} + \sum_{a\in\set{F}} b_a(\vx_\nb{a}) \log{b_a(\vx_\nb{a})}\\
&-\sum_{i\in\set{V}} (d_i-1)\cdot \sum_{x_i}b_i(x_i)\log(b_i(x_i)),
\end{aligned}
\end{equation}
where the domain of $\bethe$ is 
\[
\set{L}(\set{G})\defeq\setdef*{\big(\{b_a\}_{a\in\set{F}},\{b_i\}_{i\in\set{V}}\big)}
{\begin{array}{lll}
	b_a\in\ProbSp(\bigtimes_{i\in\nb{a}}\set{X}_i) & &\forall a\in\set{F}\\
	b_i\in\ProbSp(\set{X}_i) & &\forall i\in\set{V}\\
	\sum_{\vx_{\nb{a}\xk{i}}}b_a(\vx_\nb{a})= b_i(x_i) &\forall x_i &\forall (i,a)\in\set{E}
\end{array}}
\]
and where $d_i\defeq\deg(i)$ is the degree of the vertex $i$ for each $i\in\set{V}$.
\end{definition}
The set $\set{L}(\set{G})$ is often known as the \emph{local marginal polytope}\index{local marginal polytope}.
The marginals of a global \pmf always compose an element in $\set{L}(\set{G})$.
However, the converse is not necessarily true.
Namely, if we define the set $\set{M}(\set{G})$ (often known as the \emph{marginal polytope}) as
	\[
	\set{M}(\set{G})\defeq\setdef*{\big(\{b_a\}_{a\in\set{F}},\{b_i\}_{i\in\set{V}}\big)}
	{\begin{array}{lll}
	    \exists b\in\ProbSp(\bigtimes_{i\in\set{V}}\set{X}_i) \text{ s.t. } \\
	    b_a(\vx_\nb{a}) = \sum_{\vx_{\set{V}\xk\nb{a}}}b(\vx) &\forall\vx_\nb{a} &\forall a\in\set{F}\\
	    b_i(x_i) = \sum_{\vx_{\set{V}\xk i}}b(\vx) &\forall x_i &\forall i\in\set{V}
	\end{array}},
	\]
then $\set{M}(\set{G})\subset
\set{L}(\set{G})$, but there do exist $\set{G}$ such that $\set{M}(\set{G})\subsetneq\set{L}(\set{G})$.
For acyclic factor graphs, however, these two sets are the same.
\begin{lemma}[{see, \eg,~\cite[Proposition~4.1]{wainwright2008graphical}}] \label{prop:tree:Fb:equals:Fg}
Given an acyclic factor graph $\set{G}$, for any set of \pmfs $\left(\{b_a\}_{a\in\set{F}},\{b_i\}_{i\in\set{V}}\right)\in\set{L}(\set{G})$, the function
\begin{equation}\label{eq:belief:mock:up}
b(\vx) \defeq \frac{\prod_{a\in\set{F}}b_a(\vx_\nb{a})}{\prod_{i\in\set{V}}\bigl(b_i(x_i)\bigr)^{(d_i-1)}}
\end{equation}
is a \pmf over $\bigtimes_{i\in\set{V}}\set{X}_i$.
Notably, in this case, it holds that $\set{L}(\set{G})=\set{M}(\set{G})$.
\end{lemma}
\begin{proof}
Since it is obvious that $b(\vx)\geqslant 0$ for all $\vx$, it suffices to check that $\sum_{\vx}b(\vx)=1$.
This is proven by exploiting the acyclic structure of the factor graph.
We omit the details.
\end{proof}
For acyclic factor graphs, Lemma~\ref{prop:tree:Fb:equals:Fg} enables us to write 
\begin{equation}\label{eq:bethe:equals:gibbs}
\gibbs(b)=\bethe\left(\{b_a\}_{a\in\set{F}},\{b_i\}_{i\in\set{V}}\right)
\quad \forall \left(\{b_a\}_{a\in\set{F}},\{b_i\}_{i\in\set{V}}\right) \in \set{L}(\set{G}),
\end{equation}
where the \pmf $b$ is defined \emph{based on} $\left(\{b_a\}_{a\in\set{F}},\{b_i\}_{i\in\set{V}}\right)$ as in~\eqref{eq:belief:mock:up}.
By showing that the minimizer of $\gibbs$ can always be expressed in a form as in~\eqref{eq:belief:mock:up}, one can show that the Bethe free energy and the variational free energy share the same minimum.
Therefore, for acyclic factor graphs, the Helmholtz free energy $\helmholtz$ can be obtained by minimizing the Bethe free energy (see Theorem~\ref{thm:tree:bethe:exact}).
\begin{theorem}\label{thm:tree:bethe:exact}
For acyclic factor graphs, it holds that 
\begin{equation}
\min_{(\{b_a\}_{a\in\set{F}},\{b_i\}_{i\in\set{V}})\in\set{L}(\set{G})} \bethe(\{b_a\}_{a\in\set{F}},\{b_i\}_{i\in\set{V}})
= \min_{b\in\ProbSp(\bigtimes_{i\in\set{V}}\set{X}_i)} \gibbs(b)
= \helmholtz.
\end{equation}
\end{theorem}
\begin{proof}
As mentioned before, it suffices to prove that $b^\star(\vx)\defeq Z^{-1}\cdot \prod_{a\in\set{F}} f_a(\vx_\nb{a})$ can be expanded into $b^\star(\vx)=\prod_{a\in\set{F}}b_a(\vx_\nb{a})\cdot\prod_{i\in\set{V}}b_i^{(1-d_i)}(x_i)$.
However, since the factor graph is acyclic, by Theorem~\ref{thm:BP:tree}, we have
\begin{align*}
	\prod_{a\in\set{F}}b_a(\vx_\nb{a})\cdot\prod_{i\in\set{V}}b_i^{(1-d_i)}(x_i) &\stackrel{\text{(a)}}{\propto} \prod_{a\in\set{F}}\left(f_a(\vx_\nb{a})\cdot\prod_{i\in\nb{a}}m_{i\to a}(x_i)\right) \cdot \prod_{i\in\set{V}}\left(\prod_{c\in\nb{i}}m_{c\to i}(x_i)\right)^{(1-d_i)}\\
	&\stackrel{\text{(b)}}{=} \prod_{a\in\set{F}}f_a(\vx_\nb{a}) \cdot\prod_{(i,a)\in\set{E}} m_{i\to a}(x_i) \cdot\prod_{(i,a)\in\set{E}} m_{a\to i}^{(1-d_i)}(x_i)\\
	&\stackrel{\text{(c)}}{=} \prod_{a\in\set{F}}f_a(\vx_\nb{a}) \cdot\prod_{(i,a)\in\set{E}} \prod_{c\in\nb{i}\xk{a}} m_{c\to i}(x_i) \cdot\prod_{(i,a)\in\set{E}} m_{a\to i}^{(1-d_i)}(x_i)\\
	&\stackrel{\text{(d)}}{=} \prod_{a\in\set{F}}f_a(\vx_\nb{a}) \cdot\prod_{(i,a)\in\set{E}} m_{a\to i}^{(d_i-1)}(x_i) \cdot\prod_{(i,a)\in\set{E}} m_{a\to i}^{(1-d_i)}(x_i)\\
	&= \prod_{a\in\set{F}}f_a(\vx_\nb{a}) \propto b^\star(\vx),
\end{align*}
where $\{m_{i\to a},m_{a\to i}:\set{X}_i\to \Reals\}_{(i,a)\in\set{E}}$ is the unique BP fixed point for the acyclic factor graph.
Notice that we have used Theorem~\ref{thm:BP:tree} for (a) and~\eqref{eq:msg:acyclic:update:1} for (c), whereas (b) and (d) are results of counting.
Finally, by Lemma~\ref{prop:tree:Fb:equals:Fg}, one can prove that $b^\star(\vx)=\prod_{a\in\set{F}}b_a(\vx_\nb{a})\cdot\prod_{i\in\set{V}}b_i^{(1-d_i)}(x_i)$.
\end{proof}
%*******************************************************************************
\subsubsection{Bethe's approximation}
Motivated by the above discussion, Bethe's approximation is defined as follows.
\begin{definition}[Bethe's approximation]\label{def:bethe:approximation}
Given a factor graph $\set{G}$, Bethe's approximation (of the partition sum) is defined as
\begin{equation}
Z_{\mathsf{B}}(\set{G}) \defeq \exp\left(-\min_{(\{b_a\}_{a\in\set{F}},\{b_i\}_{i\in\set{V}})\in\set{L}(\set{G})} \bethe\left(\{b_a\}_{a\in\set{F}},\{b_i\}_{i\in\set{V}}\right)\right).
\end{equation}
Bethe's approximation \emph{at} $(\{b_a\}_{a\in\set{F}},\{b_i\}_{i\in\set{V}})\in\set{L}(\set{G})$ is defined as 
\begin{equation}
Z_{\mathsf{B}}\left(\{b_a\}_{a\in\set{F}},\{b_i\}_{i\in\set{V}}\right) \defeq \exp\left(-\bethe\left(\{b_a\}_{a\in\set{F}},\{b_i\}_{i\in\set{V}}\right)\right).\qedhere
\end{equation}
\end{definition}
%*******************************************************************************
By Theorem~\ref{thm:tree:bethe:exact}, it is clear that $Z_{\mathsf{B}}(\set{G})=Z(\set{G})$ for acyclic $\set{G}$'s.
For cyclic factor graphs, such ``exactness'' breaks down.
Though some connections between $Z_\mathsf{B}$ and $Z$ have been established~\cite{vontobel2013counting}, except for some families of factor graphs (\ie,~\cite{ruozzi2012bethe, vontobel2013bethe}), there is little we can guarantee about how accurate such approximations are\footnote{One can construct a sequence of factor graphs $\{\set{G}_k\}_k$ such that $Z(\set{G}_k)\to 0$ but $Z_\mathsf{B}(\set{G}_k)\to +\infty$.}.
On the other hand, however, the study of the Bethe free energy minimization problem reveals a deeper connection between the minimization problem of $\bethe$ and the fixed points of BP algorithms (as stated in the famous theorem below).
\begin{theorem}[{see~\cite[Theorem 2]{yedidia2005constructing}}] \label{thm:bethe:BPA}
Given a factor graph $((\set{V},\set{F},\set{E}),$$\{\set{X}_i\}_{i\in\set{V}},$ $\{f_a\}_{a\in\set{F}})$ with non-negative local functions, $\left(\{b_a\}_{a\in\set{F}},\{b_i\}_{i\in\set{V}}\right)\in\set{L}(\set{G})$ is an interior stationary point of $\bethe$ if there exists a positive BP fixed point $\{m_{i\to a}$, $ m_{a\to i}: \set{X}_i\to \Reals_{>0}\}_{(i,a)\in\set{E}}$ such that
\begin{align}
\label{eq:def:ba:msg}
b_a(\vx_\nb{a}) &\propto f(\vx_\nb{a})\cdot\prod_{i\in\nb{a}}m_{i\to a}(x_i),\\
\label{eq:def:bi:msg}
b_i(x_i) &\propto \prod_{a\in\nb{a}} m_{a\to i}(x_i).
\end{align}
Conversely, a set of positive messages $\{m_{i\to a}$, $ m_{a\to i}\}_{(i,a)\in\set{E}}$ is a BP fixed point if $\left(\{b_a\}_{a\in\set{F}},\{b_i\}_{i\in\set{V}}\right)$ defined according to~\eqref{eq:def:ba:msg} and~\eqref{eq:def:bi:msg} form an internal stationary point of $\bethe$.
\end{theorem}
\begin{proof}
The proof uses the Lagrange multiplier theory (see~\cite{yedidia2005constructing} for details).
\end{proof}
For factor graphs with positive local functions, the above theorem allows us to interpret BP algorithms as a process to find the stationary points of the Bethe free energy (see~\cite[Section~IV-C]{yedidia2005constructing}).
For BP fixed points with messages containing zero entries, it is unclear (but plausible) whether they still correspond to some stationary points of $\bethe$ (see~\cite[Conjecture~1]{yedidia2005constructing}).
%*******************************************************************************
% Loop Calculus
\subsection{Interpretations of BP fixed points: The Loop Calculus} \label{subsec:FG:LoopCalculus}
%*******************************************************************************
The loop calculus~\cite{chertkov2006loop, chernyak2007loop, mori2015loop} is another approach to characterize Bethe's approximation; and provides another interpretation for the BP fixed points.
The major result of the method can be summarized as the following theorem.
\begin{theorem}[Loop calculus]\label{thm:loop:calculus} \index{loop calculus}
Let $\set{G}=((\set{V}, \set{E}, \set{F}), \{\set{X}_i\}_{i\in\set{V}}, \{f_a\}_{a\in\set{F}})$ be a factor graph with non-negative local functions.
For any interior stationary point $(\{b_a\}_{a\in\set{F}},$ $\{b_i\}_{i\in\set{V}})$ of $\bethe$, it holds that
\begin{equation} \label{eq:lc:main}
Z(\set{G}) = Z_{\mathsf{B}}\left(\{b_a\}_{a\in\set{F}},\{b_i\}_{i\in\set{V}}\right) \cdot \sum_{E\in\mathfrak{L}(\set{E})}\mathcal{K}(E),
\end{equation}
where the set of \emph{generalized loops} $\mathfrak{L}(\set{E})$ is defined as 
\begin{equation}\label{eq:def:extended:loops}
\mathfrak{L}(\set{E}) \triangleq \left\{ E\subset\mathcal{E}:\: \size{\{(i,a)\}_{a\in\nb{i}}\cap E} \neq 1 \ \forall i\in\mathcal{V},\, \size{\{(i,a)\}_{i\in\nb{a}}\cap E} \neq 1 \ \forall a\in\mathcal{F} \right\}
\end{equation}
and where
\begin{equation}\label{eq:lc:binary}
\mathcal{K}(E) \defeq \prod_{a\in\set{F}}
\expectationwrt{ \prod_{i\in\nb{a},(i,a)\in E} \frac{\rv{X}_i-\expectationwrt{\rv{X}_i}{b_i}}{\var{(\rv{X}_i)}}}{b_a} \cdot \prod_{i\in\set{V}} \expectationwrt{ \left( \frac{\rv{X}_i-\expectationwrt{\rv{X}_i}{b_i}}{\var{(\rv{X}_i)}}\right)^{d_i(E)}}{b_i}
\end{equation}
for the case with all alphabets $\set{X}_i$ being binary~\cite{chertkov2006loop} and where
\begin{equation} \label{eq:lc:general}
\mathcal{K}(E) \defeq \sum_{\vy_E:\:y_{(i,a)}\neq 0 \atop \forall (i,a)\in E}\prod_{a\in\set{F}} \expectationwrt{ \prod_{i\in\nb{a},\atop(i,a)\in E} \frac{\partial\log{b_i(\rv{X}_i)}} {\partial\theta^{i,a}_{y_{i,a}}}}{\!b_a} \cdot \prod_{i\in\set{V}} \expectationwrt{ \prod_{a\in\nb{i},\atop(i,a)\in E} \frac{\partial\log{b_i(\rv{X}_i)}}{\partial\eta^{i,a}_{y_{i,a}}}}{\!b_i}
\end{equation}
for non-binary alphabets~\cite{mori2015loop}.
Here, $(\theta^{i,a}_y)_{y=1,\ldots,\size{\set{X}_i}-1}$ and $(\eta^{i,a}_y)_{y=1,\ldots,\size{\set{X}_i}-1}$ are some dual coordinate systems~\cite{amari2000methods} for $\log(\ProbSp(\set{X}_i))\defeq\{\log{p}:p\in\ProbSp(\set{X}_i)\}$.
\end{theorem}
\begin{remark}
Note that $\emptyset\in\mathfrak{L}(\set{E})$, and for acyclic factor graphs, $\mathfrak{L}(\set{E})=\{\emptyset\}$.
\end{remark}
\begin{remark}
For any factor graph, $\mathcal{K}(\emptyset) = 1$; thus for acyclic factor graphs, $Z(\set{G}) = Z_{\mathsf{B}}(\{b_a\}_{a\in\set{F}},\{b_i\}_{i\in\set{V}})$ at interior stationary points $(\{b_a\}_{a\in\set{F}},\{b_i\}_{i\in\set{V}})$.
\end{remark}
One method to prove the theorem is to express $Z(\set{G})$ as (see~\cite{wainwright2008graphical})
\begin{equation}
Z(\set{G}) = Z_{\mathsf{B}}\left(\{b_a\}_{a\in\set{F}},\{b_i\}_{i\in\set{V}}\right) \cdot
	\frac{\prod_{a\in\set{F}}b_a(\vx_\nb{a})}{\prod_{i\in\set{V}}b_i^{(d_i-1)}(x_i)}
\end{equation}
and expanding the terms on the right-hand side (see~\cite{sudderth2008loop}).
Another method~\cite{chertkov2006loop, mori2015loop}, reviewed below, involves the \emph{Holant theorem}~\cite{al2011holographic}.
\begin{theorem}[Holant theorem] \label{thm:holant} \index{Holant theorem}
Consider a factor graph $\set{G}$ representing the factorization $g(\vx)=\prod_{a\in\set{F}} f_a(\vx_\nb{a})$, where each alphabet $\set{X}_i$ is finite.
Given any set of invertible matrices $\{\phi_{i,a}\in\Reals^{\set{X}_i\times\set{Y}_{i,a}}\}_{(i,a)\in\set{E}}$ (where $\size{\set{X}_i}=\size{\set{Y}_{i,a}}$), one can express the partition sum $Z(\set{G})$ as
\begin{equation}\label{eq:holant}
Z(\set{G}) = \sum_{\vy} \prod_{a\in\set{F}} \hat{f}_a(\vy_{\nb{a},a}) \cdot \prod_{i\in\set{V}} \hat{h}_i(\vy_{i,\nb{i}}),
\end{equation}
where
\begin{align}
\label{eq:holographic:1}
\hat{f}(\vy_{\nb{a},a}) &\defeq \sum_{\vx_\nb{a}} f_a(\vx_\nb{a})\prod_{i\in\nb{a}} \hat{\phi}_{i,a}(y_{i,a},x_i),\\
\label{eq:holographic:2}
\hat{h}_i(\vy_{i,\nb{i}}) &\defeq \sum_{x_i} \prod_{a\in\nb{i}} \phi_{i,a}(x_i,y_{i,a}).
\end{align}
Here, $\hat{\phi}_{i,a}$ is the inverse matrix of $\phi_{i,a}$; and $\phi_{i,a}(x,y)$ and $\hat{\phi}_{i,a}(y,x)$ denote the $(x,y)$-th and the $(y,x)$-th entry of $\phi_{i,a}$ and $\hat{\phi}_{i,a}$, respectively.
\end{theorem}
\begin{proof}
This is a direct result of elementary linear algebra.
We omit the details.
\end{proof}
Using the elements from the theorem above, we consider the factorization 
\begin{equation}\label{eq:holographic:3}
	\tilde{g}(\vy_{\set{E}}) \defeq
	\prod_{a\in\set{F}} \hat{f}_a(\vy_{\nb{a},a})\cdot 
	\prod_{i\in\set{V}} \hat{h}_i(\vy_{i,\nb{i}}).
\end{equation}
The corresponding factor graph $\tilde{\set{G}}$ describing~\eqref{eq:holographic:3} is known as a \emph{holographic transform}\index{holographic transformation} of the original factor graph $\set{G}$ (\wrt $\{\phi_{i,a},\,\hat{\phi}_{i,a}\}_{(i,a)\in\set{E}}$).
Theorem~\ref{thm:holant} can then be rephrased as ``holographic transformations do not change the partition sum.''
For the sake of proving Theorem~\ref{thm:loop:calculus}, we consider a holographic transform such that 
\begin{align}
	\label{eq:lc:assumption:1}
	\hat{f}_a(\vy_{\nb{a},a}) &= 0,\text{ if } \wt(\vy_{\nb{a},a}) = 1 \quad \forall a\in\set{F},\\
	\label{eq:lc:assumption:2}
	\hat{h}_i(\vy_{i,\nb{i}}) &= 0,\text{ if } \wt(\vy_{i,\nb{i}}) = 1 \quad \forall i\in\set{V},
\end{align}
where we have assumed that each alphabet $\set{Y}_{i,a}$ contains an elements labeled as $0$.
This limits the support of $\tilde{g}$ to those $\vy_{\set{E}}$'s such that the indices of the non-zero entries of $\vy_\set{E}$ form a generalized loop, \ie, $\support(\tilde{g})\subset\setdef*{\vy_\set{E}}{\exists E\in\mathfrak{L}(\set{E})\,\st\ y_{i,a}\!=0\,\forall (i,a)\notin E}$.
Substituting~\eqref{eq:holographic:1} and~\eqref{eq:holographic:2} into~\eqref{eq:lc:assumption:1} and~\eqref{eq:lc:assumption:2}, respectively, one can rewrite the latter as
\begin{align}
\label{eq:lc:assumption:3}
\left[\sum_{\vx_{\nb{a}\xk{i}}}f_a(\vx_\nb{a})\cdot
\prod_{j\in\nb{a}\xk{i}}\hat{\phi}_{i,a}(0,x_j)\right]_{x_i}
&\perp \hat{\phi}_{i,a}(y,\cdot) &&\forall y\neq 0,\\
\label{eq:lc:assumption:4}
\left[\prod_{c\in\nb{i}\xk{a}}\phi_{i,c}(x_i,0)\right]_{x_i}
&\perp \phi_{i,a}(\cdot,y) &&\forall y\neq 0,
\end{align}
respectively.
Since $\phi_{i,a}$ and $\hat{\phi}_{i,a}$ are the inverse matrix of each other,~\eqref{eq:lc:assumption:3} and~\eqref{eq:lc:assumption:4} (and thus~\eqref{eq:lc:assumption:1} and~\eqref{eq:lc:assumption:2}) are equivalent to
\begin{align}
\label{eq:lc:assumption:5}
\phi_{i,a}(\cdot,0) &\propto
\left[\sum_{\vx_{\nb{a}\xk{i}}}f_a(\vx_\nb{a})\cdot
\prod_{j\in\nb{a}\xk{i}}\hat{\phi}_{i,a}(0,x_j)\right]_{x_i},\\
\label{eq:lc:assumption:6}
\hat{\phi}_{i,a}(0,\cdot) &\propto
\left[\prod_{c\in\nb{i}\xk{a}}\phi_{i,c}(x_i,0)\right]_{x_i},
\end{align}
respectively.
Compare the above with~\eqref{eq:msg:acyclic:update:1} and~\eqref{eq:msg:acyclic:update:2}.
Eqs.~\eqref{eq:lc:assumption:5} and~\eqref{eq:lc:assumption:6} (and thus~\eqref{eq:lc:assumption:1} and~\eqref{eq:lc:assumption:2}) are equivalent to the existence of some fixed-point messages $\{m_{i\to a}, m_{a\to i}\}_{(i,a)}$ such that
\begin{align}
	\label{eq:lc:assumption:7}
	\phi_{i,a}(x_i,0) &= c_{i,a}\cdot m_{a\to i}(x_i),\\
	\label{eq:lc:assumption:8}
	\hat{\phi}_{i,a}(0,x_i) &= \hat{c}_{i,a}\cdot m_{i\to a}(x_i),
\end{align}
where the constants $c_{i,a}$, $\hat{c}_{i,a}$ satisfy $(c_{i,a}\cdot \hat{c}_{i,a})^{-1}=Z_{i,a}(m_{i\to a},m_{a\to i})\defeq \sum_{x_i} m_{i\to a}(x_i)\cdot m_{a\to i}(x_i)$ for each $(i,a)\in\set{E}$.
Therefore, retracing the steps above, given any positive BP fixed point $\{m_{i\to a}, m_{a\to i}\}_{(i,a)\in\set{E}}$, we can construct a holographic transform satisfying~\eqref{eq:lc:assumption:7} and~\eqref{eq:lc:assumption:8}, and thus satisfying~\eqref{eq:lc:assumption:1} and~\eqref{eq:lc:assumption:2}.
\par
If all of the alphabets are binary (\ie, $\set{X}_i=\set{Y}_{i,a}=\mathbb{F}_2$ for all $i\in\set{V}$ and $(i,a)\in\set{E}$), since $\phi_{i,a}\cdot\hat{\phi}_{i,a}=\hat{\phi}_{i,a}\cdot\phi_{i,a}=\left[\begin{smallmatrix}1 & 0 \\ 0 & 1\end{smallmatrix}\right]$ for each $(i,a)\in\set{E}$, it must hold that
\begin{align}
\label{eq:lc:binary:1}
\phi_{i,a}(x_i,1) &= (-1)^{\conj{x_i}}c_{i,a}\cdot m_{i\to a}(\conj{x_i})	,\\
\label{eq:lc:binary:2}
\hat{\phi}_{i,a}(1,x_i) &= (-1)^{\conj{x_i}}\hat{c}_{i,a}\cdot m_{a\to i}(\conj{x_i}),
\end{align} where $\conj{0}\defeq 1$ and $\conj{1}\defeq 0$.
The proof of~\eqref{eq:lc:binary} in Theorem~\ref{thm:loop:calculus} involves substituting~\eqref{eq:lc:binary:1} and~\eqref{eq:lc:binary:2} into~\eqref{eq:holant} and applying the following corollary.
\begin{corollary}\label{cor:Zb:decompose}
Consider a factor graph $\set{G}$ representing the factorization $g(\vx)=\prod_{a\in\set{F}} f_a(\vx_\nb{a})$.
For a collection of positive fixed-point messages $\{m_{i\to a}, m_{a\to i}\}_{i,a}$, we have
\begin{equation}
Z_{\mathsf{B}}\left(\{b_a\}_{a\in\set{F}},\{b_i\}_{i\in\set{V}}\right) = \frac{\prod_{a\in\set{F}} Z_a(\{m_{i\to a}\}_{i\in\nb{a}}) \cdot \prod_{i\in\set{V}} Z_i(\{m_{a\to i}\}_{a\in\nb{i}})} {\prod_{(i,a)\in\set{E}} Z_{i,a}(m_{i\to a},m_{a\to i})},
\end{equation}
where the \pmfs $\{b_a\}_{a\in\set{F}}$ and $\{b_i\}_{i\in\set{V}}$ are defined according to~\eqref{eq:def:ba:msg} and~\eqref{eq:def:bi:msg} and where
\begin{align*}
Z_a(\{m_{i\to a}\}_{i\in\nb{a}}) &\defeq \sum_{\vx_\nb{a}}f(\vx_\nb{a}) \cdot \prod_{i\in\nb{a}} m_{i\to a}(x_i) && \forall a\in\set{F},\\
Z_i(\{m_{a\to i}\}_{a\in\nb{i}}) &\defeq \sum_{x_i}\prod_{a\in\nb{i}} m_{a\to i}(x_i) && \forall i\in\set{V},\\
Z_{i,a}(m_{i\to a},m_{a\to i}) &\defeq \sum_{x_i} m_{i\to a}(x_i)\cdot m_{a\to i}(x_i) && \forall (i,a)\in\set{E}.
\end{align*}
\end{corollary}
\begin{proof}
This is a direct result of Theorem~\ref{thm:bethe:BPA} by substituting~\eqref{eq:def:ba:msg} and~\eqref{eq:def:bi:msg} into~\eqref{eq:def:bethe:free:energy}.
We omit the details.
\end{proof}
\begin{proof}[Proof of~\eqref{eq:lc:binary}~{\cite[see][Proof of Lemma~5]{mori2015loop}}]
Using the aforementioned holographic transform, and by Theorem~\ref{thm:holant}, we have
\begin{fleqn}\begin{equation*}
	Z(\set{G}) = \sum_{\vy_\set{E}\in\support(\tilde{g})}
	    \prod_{a\in\set{F}} \hat{f}_a(\vy_{\nb{a},a})\cdot 
	    \prod_{i\in\set{V}} \hat{h}_i(\vy_{i,\nb{i}})
\end{equation*}\end{fleqn}% phantom
\begin{fleqn}\begin{equation*}
	\phantom{Z(\set{G})}
	=\sum_{\vy_\set{E}\in\support(\tilde{g})}\prod_{a\in\set{F}} \sum_{\vx_\nb{a}} f_a(\vx_\nb{a})\prod_{i\in\nb{a}} \hat{\phi}_{i,a}(y_{i,a},x_i) \cdot \prod_{i\in\set{V}} \sum_{x_i} \prod_{a\in\nb{i}} \phi_{i,a}(x_i,y_{i,a})
\end{equation*}\end{fleqn}
\begin{fleqn}\begin{equation*}
	\begin{aligned}
	\phantom{Z(\set{G})}
	=&\left(\prod_{a\in\set{F}} \sum_{\vx_\nb{a}} f_a(\vx_\nb{a})\prod_{i\in\nb{a}} \hat{c}_{i,a}\cdot m_{i\to a}(x_i) \cdot \prod_{i\in\set{V}} \sum_{x_i} \prod_{a\in\nb{i}} c_{i,a}\cdot m_{a\to i}(x_i)\right)\cdot \\
	&\sum_{\vy_\set{E}\in\support(\tilde{g})} \prod_{a\in\set{F}} \frac{\sum_{\vx_\nb{a}} f_a(\vx_\nb{a})\prod_{i\in\nb{a}} \hat{\phi}_{i,a}(y_{i,a},x_i)}{\sum_{\vx_\nb{a}} f_a(\vx_\nb{a})\prod_{i\in\nb{a}} \hat{\phi}_{i,a}(0,x_i)} \prod_{i\in\set{V}} \frac{\sum_{x_i} \prod_{a\in\nb{i}} \phi_{i,a}(x_i,y_{i,a})}{\sum_{x_i} \prod_{a\in\nb{i}} \phi_{i,a}(x_i,0)}
	\end{aligned}
\end{equation*}\end{fleqn}
\begin{fleqn}\begin{equation*}
	\phantom{Z(\set{G})}
	\overset{\text{(a)}}{=} Z_{\mathsf{B}}\left(\{b_a\}_{a},\{b_i\}_{i}\right) \cdot \sum_{\vy_\set{E}\in\support(\tilde{g})} \prod_{a\in\set{F}} \expectationwrt{\prod_{i\in\nb{a}}\frac{\hat{\phi}_{i,a}(y_{i,a},\rv{X}_i)}{\hat{\phi}_{i,a}(0,\rv{X}_i)}}{b_a} \prod_{i\in\set{V}}\expectationwrt{\prod_{a\in\nb{i}}\frac{\phi_{i,a}(\rv{X}_i,y_{i,a})}{\phi_{i,a}(\rv{X}_i,0)}}{b_i}
\end{equation*}\end{fleqn}
\begin{fleqn}\begin{equation*}
	\phantom{Z(\set{G})}
	\overset{\text{(b)}}{=} Z_{\mathsf{B}}\left(\{b_a\}_{a},\{b_i\}_{i}\right) \cdot \sum_{E\in\mathfrak{L}(\set{E})} \prod_{a\in\set{F}} \expectationwrt{\prod_{i\in\nb{a},\atop (i,a)\in E}\frac{\hat{\phi}_{i,a}(1,\rv{X}_i)}{\hat{\phi}_{i,a}(0,\rv{X}_i)}}{b_a} \prod_{i\in\set{V}}\expectationwrt{\prod_{a\in\nb{i},\atop (i,a)\in E}\frac{\phi_{i,a}(\rv{X}_i,1)}{\phi_{i,a}(\rv{X}_i,0)}}{b_i}
\end{equation*}\end{fleqn}
\begin{fleqn}\begin{equation*}
	\phantom{Z(\set{G})}
	\overset{\text{(c)}}{=} Z_{\mathsf{B}}\left(\{b_a\}_{a},\{b_i\}_{i}\right) \cdot \sum_{E\in\mathfrak{L}(\set{E})}
	\begin{aligned}[t]
	&\prod_{a\in\set{F}} \expectationwrt{\prod_{i\in\nb{a},\atop (i,a)\in E}\frac{(-1)^{\conj{\rv{X}_i}} m_{a\to i}(\conj{\rv{X}_i})\cdot m_{i\to a}(\conj{\rv{X}_i})}{m_{i\to a}(\rv{X}_i)\cdot m_{i\to a}(\conj{\rv{X}_i})}}{b_a}\\
	&\cdot\prod_{i\in\set{V}}\expectationwrt{\prod_{a\in\nb{i},\atop (i,a)\in E}\frac{(-1)^{\conj{\rv{X}_i}} m_{i\to a}(\conj{\rv{X}_i})\cdot m_{a\to i}(\conj{\rv{X}_i})}{m_{a\to i}(\rv{X}_i)\cdot m_{a\to i}(\conj{\rv{X}_i})}}{b_i}
	\end{aligned}
\end{equation*}\end{fleqn}
\begin{fleqn}\begin{equation*}
	\phantom{Z(\set{G})}
	\overset{\text{(d)}}{=} Z_{\mathsf{B}}\left(\{b_a\}_{a},\{b_i\}_{i}\right) \cdot \sum_{E\in\mathfrak{L}(\set{E})}
	\prod_{a\in\set{F}} \expectationwrt{\prod_{i\in\nb{a},\atop (i,a)\in E}\frac{(-1)^{\conj{\rv{X}_i}} b_i(\conj{\rv{X}_i})}{\sqrt{b_i(0)b_i(1)}}}{b_a}\!\!\!
	\prod_{i\in\set{V}}\expectationwrt{\prod_{a\in\nb{i},\atop (i,a)\in E}\frac{(-1)^{\conj{\rv{X}_i}} b_i(\conj{\rv{X}_i})}{\sqrt{b_i(0)b_i(1)}}}{b_i}.
\end{equation*}\end{fleqn}
Step (a) is based on Corollary~\ref{cor:Zb:decompose}, together with the observation that
\begin{align*}
\frac{\sum_{\vx_\nb{a}} f_a(\vx_\nb{a})\prod_{i\in\nb{a}} \hat{\phi}_{i,a}(y_{i,a},x_i)}{\sum_{\vx_\nb{a}} f_a(\vx_\nb{a})\prod_{i\in\nb{a}} \hat{\phi}_{i,a}(0,x_i)}
&=\sum_{\vx_\nb{a}} \frac{f_a(\vx_\nb{a})\prod_{i\in\nb{a}} \hat{\phi}_{i,a}(0,x_i)}{\sum_{\tvx_\nb{a}} f_a(\tvx_\nb{a})\prod_{i\in\nb{a}} \hat{\phi}_{i,a}(0,\tx_i)} \cdot \prod_{i\in\nb{a}}
  \frac{\hat{\phi}_{i,a}(y_{i,a},x_i)}{\hat{\phi}_{i,a}(0,x_i)},\\
\frac{\sum_{x_i} \prod_{a\in\nb{i}} \phi_{i,a}(x_i,y_{i,a})}{\sum_{x_i} \prod_{a\in\nb{i}} \phi_{i,a}(x_i,0)}
&= \sum_{x_i}\frac{\prod_{a\in\nb{i}} \phi_{i,a}(x_i,0)}{\sum_{\tx_i} \prod_{a\in\nb{i}} \phi_{i,a}(\tx_i,0)}
  \cdot\prod_{a\in\nb{i}}\frac{\phi_{i,a}(x_i,y_{i,a})}{\phi_{i,a}(x_i,0)}.
\end{align*}
Step (b) exploits the setup that $\support(\tilde{g})\subset\setdef*{\vy_\set{E}}{\exists E\in\mathfrak{L}(\set{E})\,\st\ y_{i,a}\!=0\,\forall (i,a)\notin E}$ and the assumption that all of the alphabets are binary.
Step (c) is the result of substituting \eqref{eq:lc:assumption:7}, \eqref{eq:lc:assumption:8}, \eqref{eq:lc:binary:1}, and \eqref{eq:lc:binary:2}.
Step (d) is shown by observing (from~\eqref{eq:def:bi:msg} and~\eqref{eq:msg:BP:fixed:1}) that $Z_{i,a}^{-1}\cdot m_{a\to i}(\conj{x_i}) \cdot m_{i\to a}(\conj{x_i}) = b_i(\conj{x_i})$ for each $(i,a)\in\set{E}$, and the fact that $m_{i\to a}(\rv{X}_i)m_{i\to a}(\conj{\rv{X}_i})$ and $m_{a\to i}(\rv{X}_i)m_{a\to i}(\conj{\rv{X}_i})$ are independent of the value taken by $\rv{X}_i$.
Namely,
\begin{align*}
	&\prod_{a\in\set{F}} \expectationwrt{\prod_{i\in\nb{a},\atop (i,a)\in E}\frac{1}{m_{i\to a}(\rv{X}_i)\cdot m_{i\to a}(\conj{\rv{X}_i})}}{b_a}
	\prod_{i\in\set{V}}\expectationwrt{\prod_{a\in\nb{i},\atop (i,a)\in E}\frac{1}{m_{a\to i}(\rv{X}_i)\cdot m_{a\to i}(\conj{\rv{X}_i})}}{b_i}\\
=   &\prod_{a\in\set{F}} \expectationwrt{\prod_{i\in\nb{a},\atop (i,a)\in E}\frac{1}{m_{i\to a}(0)\cdot m_{i\to a}(1)}}{b_a}
	\prod_{i\in\set{V}}\expectationwrt{\prod_{a\in\nb{i},\atop (i,a)\in E}\frac{1}{m_{a\to i}(0)\cdot m_{a\to i}(1)}}{b_i}\\
=   &\prod_{(i,a)\in E} m_{i\to a}(0)\cdot m_{i\to a}(1)\cdot m_{a\to i}(0)\cdot m_{a\to i}(1)
=   \prod_{(i,a)\in E} Z_{i,a}^2\cdot b_i(0)\cdot b_i(1)\\
=   &\prod_{a\in\set{F}} \expectationwrt{\prod_{i\in\nb{a},\atop (i,a)\in E}\frac{1}{Z_{i,a}\cdot\sqrt{b_i(0)b_i(1)}}}{b_a}
	\prod_{i\in\set{V}}\expectationwrt{\prod_{a\in\nb{i},\atop (i,a)\in E}\frac{1}{Z_{i,a}\cdot\sqrt{b_i(0)b_i(1)}}}{b_i}.
\end{align*}
Finally,~\eqref{eq:lc:binary} can be shown by noticing that for any binary random variable $\rv{X}$ with its \pmf being $\prob_\rv{X}$, it holds that $(-1)^{\conj{\rv{X}}} \prob_\rv{X}(\conj{\rv{X}}) = \rv{X}-\expectation{\rv{X}}$, and $\prob_\rv{X}(0) \cdot \prob_\rv{X}(1) = \var{\rv{X}}$.
\end{proof}
To prove the general case of Theorem~\ref{thm:holant}, \ie,~\eqref{eq:lc:general}, observe that for satisfying $\phi_{i,a}\cdot\hat{\phi}_{i,a}=\hat{\phi}_{i,a}\cdot\phi_{i,a}=I$ given~\eqref{eq:lc:assumption:7} and~\eqref{eq:lc:assumption:8}, it is equivalent to require
\begin{align}
	\label{eq:lc:general:1}
	\expectationwrt{\frac{\phi_{i,a}(\rv{X}_i,y)}{\phi_{i,a}(\rv{X}_i,0)}}{b_i} & = 0 &&\forall y\neq 0,\\
	\label{eq:lc:general:2}
	\expectationwrt{\frac{\hat{\phi}_{i,a}(y,\rv{X}_i)}{\hat{\phi}_{i,a}(0,\rv{X}_i)}}{b_i} & = 0 &&\forall y\neq 0,\\
	\label{eq:lc:general:3}
	\expectationwrt{\frac{\phi_{i,a}(\rv{X}_i,y)}{\phi_{i,a}(\rv{X}_i,0)}
	\cdot \frac{\hat{\phi}_{i,a}(y',\rv{X}_i)}{\hat{\phi}_{i,a}(0,\rv{X}_i)}}{b_i} & = \delta_{y,y'} &&\forall y,y'\neq 0,
\end{align}
which, in turn, is equivalent to
\begin{align}
	\label{eq:lc:general:4}
	\phi_{i,a}(\rv{X}_i,y) & = \theta^{i,a}_y \cdot
		\phi_{i,a}(\rv{X}_i,0) && \text{ for each }y\in\set{Y}_{i,a}\xk{0} \\
	\label{eq:lc:general:5}
	\hat{\phi}_{i,a}(y,\rv{X}_i) & = \eta^{i,a}_y \cdot
		\hat{\phi}_{i,a}(0,\rv{X}_i) && \text{ for each }y\in\set{Y}_{i,a}\xk{0}
\end{align}
for some pair of dual coordinate systems $\{\theta^{i,a}_y\}_{y\in\set{Y}_{i,a}\xk{0}}$ and $\{\eta^{i,a}_y\}_{y\in\set{Y}_{i,a}\xk{0}}$ for the $(\size{\set{Y}_{i,a}}\!-\!1)=(\size{\set{X}_i}]\!-\!1)$-dimensional manifold $\log(\ProbSp(\set{X}_i))\defeq\{\log{p}:p\in\ProbSp(\set{X}_i)\}$ (see, \eg,~\cite[Chapter~2]{amari2000methods}).
Eq~\eqref{eq:lc:general} can be proven by substituting \eqref{eq:lc:assumption:7}, \eqref{eq:lc:assumption:8}, \eqref{eq:lc:general:4}, and \eqref{eq:lc:general:5} into \eqref{eq:holant} and following the same logic as in the proof of~\eqref{eq:lc:binary} above.
We refer to~\cite[Section~IV-C]{mori2015loop} for details.
%*******************************************************************************
%*******************************************************************************
\section{Basic quantum information theory}\label{sec:basic_quantum_theory}
This section provides a brief introduction to basic quantum information theory, including quantum postulates, quantum information measures, and quantum channels.
Much of the contents of this section are based on Nielsen and Chuang~\cite{nielsen2011quantum}.
%wilde2017quantum, watrous2018quantum
%*******************************************************************************
\subsection{Quantum Postulates}\label{sec:quantum_postulates}
The following four postulates are based on our current understanding of physics and provide the framework in which quantum information is considered.
\begin{postulate}[State space]\label{pos:1}
A quantum system is fully described by a unit vector (\aka its state vector\index{state vector} or its state) from a complex Hilbert space (\aka its state space\index{state space}).
\end{postulate}
\begin{postulate}[Evolution]
The evolution\index{evolution} of a closed system is described by a unitary transformation applied to its state vector.
\end{postulate}
\begin{postulate}[Measurement]
A quantum measurement\index{measurement} is described by an indexed set of operators\footnote{As already mentioned in ``Notations'', an operator is a linear transformation in this context.} $\{M_y\in\LinearOp{(\hilbert_\system{A})}\}_{y\in\set{Y}}$ on its state space, where $\sum_{y} M_y^\Herm M_y = I$.
If the state of the system is $\bra{\psi}$ immediately before the measurement, then:
\begin{itemize}
	\item The probability of the measurement outcome being $\cy$ is $p(\cy)=\ket{\psi}{M_\cy}^\Herm {M_\cy}\bra{\psi}$;
	\item Provided the measurement outcome $y$, the state of the system immediately after the measurement is $\frac{M_y\bra{\psi}}{\sqrt{\ket{\psi}M_y^\Herm M_y\bra{\psi}}}$.
\end{itemize}
\end{postulate}
\begin{postulate}[Composition] \label{post:composition}
The state space of a composite system\index{composite system} is the tensor product of the state spaces of each component system.
\end{postulate}
\begin{example}
A \emph{single-qubit system} is a quantum system with 2-dimensional state space and is often referred to as a (single) qubit.
Suppose $\bra{0}$ and $\bra{1}$ form an orthonormal basis of the state space.
Under this basis (\aka \emph{computational basis}), 
\begin{itemize}
	\item Any state vectors can be expressed as $\left[\begin{smallmatrix}a\\b\end{smallmatrix}\right]\defeq a\bra{0}+b\bra{1}$ where $a,b\in\Complex$, and $\abs{a}^2 + \abs{b}^2 = 1$.
	\item (Closed) Evolution on this system can be represented by a $2$-by-$2$ unitary matrix $U$, where the post-evolution state will be $U\cdot \left[\begin{smallmatrix}a\\b\end{smallmatrix}\right]$ given initial state $\left[\begin{smallmatrix}a\\b\end{smallmatrix}\right]$.
	\item A quantum measurement can be described by a set of $2$-by-$2$ matrices $\{M_y\}_y$ such that $\sum_{y}M_y^\Herm M_y=\left[\begin{smallmatrix} 1 & 0 \\ 0 & 1\end{smallmatrix}\right]$, where, given the pre-measurement state being $\left[\begin{smallmatrix}a\\b\end{smallmatrix}\right]$,
	\begin{itemize}
		\item the probability of the outcome being $\cy$ can be expressed as $\left[\begin{smallmatrix}\conj{\vphantom{b}a} & \conj{b}\end{smallmatrix}\right] \cdot M_\cy^\Herm M_\cy \cdot \left[\begin{smallmatrix}a\\b\end{smallmatrix}\right]$,
		\item provided the outcome $y$, the post-measurement state is $p(y)^{-1}\cdot M_y\left[\begin{smallmatrix}a\\b\end{smallmatrix}\right]$. \qedhere
	\end{itemize}
\end{itemize}	 
\end{example}
\begin{example}
A \emph{two-qubit system} is the composition of two single-qubit systems.
According to Postulate~\ref{post:composition}, a basis of its state space can be expressed as
\[
\left\{\bra{00},\bra{01},\bra{10},\bra{11}\right\} \defeq \left\{\bra{0}\tensor\bra{0}, \bra{0}\tensor\bra{1}, \bra{1}\tensor\bra{0}, \bra{1}\tensor\bra{1}\right\}.\qedhere
\]
\end{example}
% Distributions over the state space
\subsubsection{Density Operators}
Like classical systems, there are cases when the state of a quantum system is \emph{not} entirely known but rather known to be in the state $\bra{\psi}_i$ with probability $p_i$ (for some indices $i$)\footnote{
Notice that the state space of a quantum system is continuous.
Thus, the distribution of the state vector should also be continuous in general, in which case \pdfs should be considered.}.
In this case, the \pmf of the measurement outcome \wrt the measurement $\{M_y\}_{y}$ can be expressed as
\begin{equation}\label{eq:density:measure:pmf}
p(y)= \sum_{i}p_i\cdot\ket{\psi_i}M_y^\Herm M_y\bra{\psi_i}
    = \tr\left(M_y\cdot \left(\sum_{i}p_i\cdot\braket{\psi_i}\right) \cdot M_y^\Herm\right).
\end{equation}
% Indistinguishability 
Observe that~\eqref{eq:density:measure:pmf} is a function of $\sum_{i}p_i\cdot\braket{\psi_i}$.
Namely, given another ensemble $\{q_j,\bra{\phi_j}\}_j$, \ie, a setup that the state vector is $\bra{\phi}_j$ with probability $q_j$, such where $\sum_{i}p_i\cdot\braket{\psi_i}=\sum_{j}q_j\cdot\braket{\phi_j}$, it is impossible to distinguish between these two ensembles using measurements.
This motivated physicists to treat $\rho\defeq\sum_{i}p_i\cdot \braket{\psi_i}$, known as the density operator (see Definition~\ref{def:density:operator}), as a \emph{complete} description of the original system.
\begin{definition}[Density Operator] \label{def:density:operator} \index{density operator}
Given a quantum system $\system{A}$ with its state space being $\hilbert_\system{A}$, a density operator (of $\system{A}$) is a positive trace-$1$ linear operator on $\hilbert_\system{A}$.
The set of all density operators on $\hilbert_\system{A}$ is denoted by $\DensOp(\hilbert_\system{A})$.
\end{definition}
\begin{remark}
$\DensOp(\hilbert_\system{A})$ is the set of all possible operators such that $\rho=\sum_{i}p_i\cdot\braket{\psi_i}$ for some ensemble $\{p_i,\bra{\psi_i}\}_i$.
\end{remark}
% Postulates w.r.t for density operators
Using the notations of density operators, Postulates 1--3 can be rewritten as follows.\\
\noindent\textbf{Postulate 1a.} A quantum system is fully described by a density operator on a complex Hilbert space (known as its state space).
Each density operator on the state space describes some valid configurations of the system.\\
\noindent\textbf{Postulate 2a.} The evolution of a closed system is described by some unitary transformation $U$.
Given the density operator of the system being $\rho$, the density operator after the evolution is $U \rho U^\Herm$.\\
\noindent\textbf{Postulate 3a.} A quantum measurement is described by an indexed set of operators $\{M_y\in\LinearOp{(\hilbert_\system{A})}\}_{y\in\set{Y}}$ on its state space, where $\sum_{y} M_y^\Herm M_y=I$.
If the density operator of the system is $\rho$ immediately before the measurement, then:
	\begin{itemize}
		\item The probability of the measurement outcome being $\cy$ is $p(\cy)=\tr\left({M_\cy} \rho {M_\cy}^\Herm\right)$;
		\item Provided the measurement outcome $y$, the density operator immediately after the measurement will be $\frac{M_y \rho M_y^\Herm}{\tr\left(M_y \rho M_y^\Herm\right)}$.
	\end{itemize}
\par
% POVM
By the cyclic property of the trace operation, it suffices to know $\{M_y^\Herm M_y\}_y$ for a measurement, if the post-measurement state/density operator is irrelevant to us.
In this case, a quantum measurement can be specified by a POVM $\{E_y\}_y$ (see Definition~\ref{def:POVM}), and the probability of measurement outcome being $\cy$, given the pre-measurement density operator $\rho$, can be expressed as $p(\cy)=\tr(E_\cy\rho)$.
\begin{definition}[Positive-Operator-Valued Measurement (POVM)]\label{def:POVM}
\index{positive-operator-valued Measurement}
Given a quantum system $\system{A}$ with its state space being $\hilbert_\system{A}$, a positive-operator-valued measurement (POVM) is an indexed set of positive operators $\{E_y\in\PositiveOp{(\hilbert_\system{A})}\}_y$ such that $\sum_{y} E_y = I$.
\end{definition}
% Partial Trace
\subsubsection{Partial Measurement and Partial-trace}
Let $\system{AB}$ be a composite system.
A measurement in the form of $\{M_y\tensor I_\system{B}\}_y$, where $\{M_y\}_{y}$ is some measurement on the subsystem $\system{A}$, is known as a \emph{partial measurement}.
In particular, $\{M_y\tensor I_\system{B}\}_y$ is the unique measurement having the same effect as $\{M_y\}_{y}$ on the system $\system{A}$ while leaving the system $\system{B}$ intact, namely
	\[
	(M_y\tensor I_\system{B}) \cdot
	(\bra{\psi_\system{A}}\tensor\bra{\phi_\system{B}}) =
	(M_y\bra{\psi_\system{A}})\tensor\bra{\phi_\system{B}}
	\quad \forall \bra{\psi_\system{A}}\in\hilbert_\system{A},
	\, \bra{\phi_\system{B}}\in\hilbert_\system{B}.
	\]\par
On the other hand, let $\rho_\system{AB}$ be a density operator for the composite system $\system{AB}$.
We argue that the density operator $\rho_\system{A}\defeq\tr_\system{B}(\rho_\system{AB})$ (\aka \emph{reduced operator}) describes the subsystem $\system{A}$, where $\tr_\system{B}$ is a linear map defined as
\begin{equation}
\tr_\system{B}:\bra{i_\system{A}}\ket{j_\system{A}} \tensor \bra{i'_\system{B}}\ket{j'_\system{B}} \mapsto \bra{i_\system{A}}\ket{j_\system{A}}\cdot\tr(\bra{i'_\system{B}}\ket{j'_\system{B}})
\quad \forall i,j,i',j'
\end{equation}
for \emph{arbitrary}\footnote{To check the well-defineness, one can show that the partial trace $\tr_\system{B}$ is the unique linear map such that $\tr_\system{B}(\sigma_\system{A}\tensor\sigma_\system{B})=\sigma_\system{A}\cdot\tr(\sigma_\system{B})$ for all operators $\sigma_\system{A}$ and $\sigma_\system{B}$ on $\hilbert_\system{A}$ and $\hilbert_\system{B}$, respectively.} bases $\{\bra{k_\system{A}}\}_k$ and $\{\bra{k_\system{B}}\}_k$ for $\hilbert_\system{A}$ and $\hilbert_\system{B}$, respectively.
In particular, given $\rho_\system{AB}$, the reduced density operator $\rho_\system{A}\defeq\tr_\system{B}(\rho_\system{AB})$ is the \emph{unique} density operator such that
\begin{equation}
\tr\left(M_y \rho_\system{A} M_y^\Herm\right)
=\tr\left((M_y \otimes I_\system{B}) \rho_\system{AB} (M_y \otimes I_\system{B})^\Herm\right)
\end{equation}
for any measurement $\{M_y\}_y$ on $\system{A}$.
We refer to~\cite[Box~2.6]{nielsen2011quantum} for details.
\begin{example}\label{ex:joint:system:1}
Let $\system{A}$, $\system{B}$ each be a single qubit system.
Suppose the state vector of their composite system $\system{AB}$ is $\bra{\psi_\system{AB}}=\frac{\bra{00}+\bra{11}}{\sqrt{2}}$.
In this case, the matrix representations\footnote{In the remainder of this thesis, we may refer to an operator/state vector and their matrix representations (\wrt certain computational basis) interchangeably.} of $\rho_\system{AB}$ and $\rho_\system{A}$ are
\[
\rho_\system{AB}=\braket{\psi_\system{AB}}=
\frac{1}{2}\left[
\begin{smallmatrix}1 & 0 &0 & 1 \\ 0 & 0 & 0 & 0 \\ 0 & 0 & 0 & 0 \\ 1 & 0 &0 & 1 \end{smallmatrix}
\right]
\]
and $\rho_\system{A}=\tr_\system{B}(\rho_\system{AB})=\frac{1}{2}\left[\begin{smallmatrix}1 & 0 \\ 0 & 1 \end{smallmatrix}\right]=\frac{1}{2}\braket{0}+\frac{1}{2}\braket{1}$, respectively.
\end{example}
%*******************************************************************************
\subsection{Quantum Information Measures and Some Inequalities}
This section reviews some basic elements of quantum information theory, including some of the common quantum information measures and some famous inequalities.
\begin{definition}[von Neumann Entropy] \index{von Neumann entropy}
Given a quantum system $\system{A}$ described by the density operator $\rho_\system{A}$, the von Neumann entropy of $\system{A}$ or $\rho_\system{A}$ is defined as
\[
\qEntropy(\system{A}) = \qEntropy(\rho_\system{A})\defeq -\tr\left(\rho_\system{A}\cdot\log{\rho_\system{A}}\right),
\]
where $\log$ stands for the matrix logarithm, \ie, $\log{\rho_\system{A}}\defeq\sum_{k=1}^\infty \frac{(-1)^{k+1}}{k}(\rho_\system{A}-I)^k$.
\end{definition}
\begin{remark}\label{rem:decompose:logA:AlogA}
For a $d$-by-$d$ PD matrix $A=UDU^\Herm$, where $U$ is unitary, and $D=\diag(\lambda_1,\lambda_2,\ldots,\lambda_n) > 0$ is diagonal, one can show that 
\begin{align}
\log{A} = U\cdot\log{D}\cdot U^\Herm &= U\cdot\left[\begin{smallmatrix} \log{\lambda_1} & & & \\ & \log{\lambda_2} & & \\ & & \ddots & \\ & & & \log{\lambda_n} \end{smallmatrix}\right]\cdot U^\Herm,\\
\label{eq:decompose:AlogA}
A\cdot\log{A} = U\cdot D\log{D}\cdot U^\Herm &= U\cdot\left[\begin{smallmatrix} \lambda_1\log{\lambda_1} & & & \\ & \lambda_2\log{\lambda_2} & & \\ & & \ddots & \\ & & & \lambda_n\log{\lambda_n} \end{smallmatrix}\right]\cdot U^\Herm.
\end{align}
By a continuity argument ($\lim_{\lambda\to 0}\lambda\log{\lambda}=0$),~\eqref{eq:decompose:AlogA} also holds for PSD matrices $A$.
\end{remark}
\begin{theorem}
Let $\rho$ be a density operator, and suppose a spectral decomposition of $\rho$ is $\rho=\sum_{i}p_i\braket{\psi_i},$ where $\{p_i\}_{i}$ is some \pmf and where $\{\bra{\psi_i}\}_{i}$ are orthonormal \wrt each other, then $\qEntropy(\rho)=\entropy(p)=-\sum_{i}p_i\log{p_i}$.
\end{theorem}
\begin{proof}
This is a direct result of Remark~\ref{rem:decompose:logA:AlogA}.
We omit the details.
\end{proof}
\begin{corollary}\label{cor:qEntropy:basics}
Let $\rho$ be a density operator on a $d$-dimensional Hilbert space.
Then, it holds that $0\leqslant \qEntropy(\rho) \leqslant \log{d}$, where $\qEntropy(\rho)=0$ if and only if $\rho$ is in a pure state\footnote{A density operator $\rho$ is said to be in a \emph{pure state} if $\rank(\rho)=1$, otherwise it is said to be in a \emph{mixed state}.} and where $\qEntropy(\rho)=\log{d}$ if and only if $\rho=I/d$.
\end{corollary}
\begin{definition}[Quantum Information Divergence] \index{quantum information divergence}
Given density operators $\rho$ and $\sigma$ (on the same Hilbert space), the \emph{quantum information divergence} between $\rho$ and $\sigma$ is defined to be
\[
\infdiv{\rho}{\sigma}\defeq\tr(\rho\log{\rho})-\tr(\rho\log{\sigma}).\qedhere
\]
\end{definition}
\begin{theorem}[Klein's Inequality~\cite{klein1931quantenmechanischen}]
For all $\rho,\sigma\in\DensOp(\hilbert)$, it holds that $\infdiv{\rho}{\sigma})\geqslant 0$, and equality holds if and only if $\rho=\sigma$.
\end{theorem}
\begin{proof}
Omitted. (See, \eg,~\cite[Theorem~11.7]{nielsen2011quantum} for details.)
\end{proof}
The von Neumann entropy is a continuous function, as shown by the following inequality.
\begin{theorem}[Fannes' Inequality~\cite{fannes1973continuity}]
Let $\rho,\sigma\in\DensOp(\hilbert)$ be density operators on some $d$-dimensional Hilbert space $\hilbert$.
Then,
\[
\abs{\qEntropy(\rho) - \qEntropy(\sigma)}
\leqslant \norm{\rho-\sigma}_1 \cdot\log{d} - \norm{\rho-\sigma}_1\log\left(\norm{\rho-\sigma}_1\right)
\leqslant \norm{\rho-\sigma}_1 \cdot\log{d} + \frac{1}{e}.
\]
\end{theorem}
\begin{proof}
Omitted. (See, \eg,~\cite[Theorem~11.6]{nielsen2011quantum} for details.)
\end{proof}
\begin{definition}[Conditional entropy, mutual information]
\index{quantum conditional entropy}
\index{quantum mutual information}
Given a joint quantum system $\system{AB}$ described by the density operator $\rho_\system{AB}$, the conditional entropy of $\system{A}$ given $\system{B}$ is defined as
\[
\qEntropy(\system{A}|\system{B})\defeq \qEntropy(\system{AB})-\qEntropy(\system{B}),
\]
and the mutual information between the systems $\system{A}$ and $\system{B}$ is defined as
\[
\qmutualInfo(\system{A}:\system{B})\defeq \qEntropy(\system{A})+\qEntropy(\system{B})-\qEntropy(\system{AB}),
\]
where the reduced operators $\rho_\system{A}=\tr_\system{B}(\rho_\system{AB})$ and $\rho_\system{B}=\tr_\system{A}(\rho_\system{AB})$ are used for $\qEntropy(\system{A})$ and $\qEntropy(\system{B})$, respectively.
\end{definition}
\begin{remark}
In the above definitions, the quantities $\qEntropy(\system{A})$, $\qEntropy(\system{A}|\system{B})$, and $\qmutualInfo(\system{A}:\system{B})$  depend on the density operators of the involved system(s).
If the density operators are not clear from the context, one may specify the density operators being considered by writing $\qEntropy(\system{A})[\rho_\system{A}]$, $\qEntropy(\system{A}|\system{B})[\rho_\system{AB}]$, and $\qmutualInfo(\system{A}:\system{B})[\rho_\system{AB}]$, respectively.
\end{remark}
\begin{example}[Quantum Conditional Entropy can be Negative]
Consider a joint quantum system $\system{AB}$ with its density operator being $\rho_\system{AB}=\braket{\psi_\system{AB}}$, where $\bra{\psi_\system{AB}}=\frac{\bra{00}+\bra{11}}{\sqrt{2}}$ (see Example~\ref{ex:joint:system:1}).
In this case, $\qEntropy(\system{AB})=0$ (see Corollary~\ref{cor:qEntropy:basics}), and $\qEntropy(\system{A})=\qEntropy(\system{B})=1$ bit since $\rho_\system{A}=\rho_\system{B}=\frac{1}{2}I$.
Moreover, one can verify that $\qEntropy(\system{A}|\system{B})=\qEntropy(\system{B}|\system{A})=-1$ bit and that $\qmutualInfo(\system{A}:\system{B}) = 2$ bits.
(Note that we have used binary logarithm in this example.)
\end{example}
We omit the proof of the following two inequalities (see, \eg,~\cite[Section~11.3.4, 11.4]{nielsen2011quantum}).
\begin{theorem}[Subadditivity]
	Given a joint quantum system $\system{AB}$, it holds that
	\begin{align}
		\qEntropy(\system{A},\system{B})&\leqslant \qEntropy(\system{A}) + \qEntropy(\system{B}),\\
		\qEntropy(\system{A},\system{B})&\geqslant \abs{\qEntropy(\system{A}) - \qEntropy(\system{B})}.
	\end{align}
\end{theorem}
\begin{theorem}[Strong subadditivity]
	Given a joint quantum system $\system{ABC}$, it holds that
	\begin{equation}
	\qEntropy(\system{A},\system{B},\system{C}) + \qEntropy(\system{B}) \leqslant
	\qEntropy(\system{A},\system{B}) + \qEntropy(\system{B},\system{C}).
	\end{equation}
\end{theorem}

%*******************************************************************************
\subsection{Quantum Channels and Their Representations}
%*******************************************************************************
%Similar to classical channels, quantum channels are relationships between quantum systems, and can be used for communications.
%In this section, we review the definition of quantum channels, the task of classical and quantum communications over them, and some of the fundamental results regarding the performances (rates/capacities) of such communications.
\begin{definition}[Quantum Channel]\label{def:quantum:channel}
\index{quantum channel}
A quantum channel from system $\system{A}$ to system $\system{B}$ is a complete positive trace-preserving linear transformation from $\LinearOp(\hilbert_\system{A})$ to $\LinearOp(\hilbert_\system{B})$.
Namely, $\operator{N}:\LinearOp(\hilbert_\system{A})\to\LinearOp(\hilbert_\system{B})$ is a quantum channel if 
\begin{itemize}
	\item $\operator{N}$ is linear, namely $\operator{N}(a\cdot\rho+\sigma)=a\cdot\operator{N}(\rho)+\operator{N}(\sigma)$ $\forall\rho,\sigma \in \LinearOp(\hilbert_\system{A})$ and $a\in\Complex$;
	\item $\operator{N}$ is completely positive, namely for any other system $\system{R}$,
	\[
	\left(\id_\system{R}\tensor\operator{N}\right)(\rho_\system{RA})\in\PositiveOp(\hilbert_{\system{RB}})
	\quad \forall \rho_\system{RA}\in\PositiveOp(\hilbert_\system{RA}),
	\]
	where $\id_\system{R}$ is the identity map on $\hilbert_\system{R}$;
	\item $\operator{N}$ is trace-preserving, namely $\tr(\operator{N}(\rho))=\tr(\rho)$ for any $\rho\in\LinearOp(\hilbert_\system{A})$.
\end{itemize}
Additionally, $\operator{N}$ is said to be finite-dimensional if both $\hilbert_\system{A}$ and $\hilbert_\system{B}$ are of finite dimension.
\end{definition}
For the rest of this thesis, all quantum channels involved are assumed to be finite-dimensional, unless stated otherwise.
In the following, we review two representations of quantum channels.
\begin{theorem}[Operator-sum representation / Kraus representation]\label{thm:Kraus}
$\operator{N}$ is a finite-dimensional quantum channel if and only if there exist $M$ operators $\{E_k:\hilbert_\system{A}\to\hilbert_\system{B}\}_{k=1}^M$ such that
\begin{equation}\label{eq:Kraus}
\operator{N}:\rho\mapsto \sum_{k=1}^{M} E_k\rho E_k^\Herm
\end{equation}
and such that $1\leq M\leqslant \dim(\hilbert_\system{A})\cdot\dim(\hilbert_\system{B})$ and $\sum_{k=1}^M E_k^\Herm E_k = I$.
The set of operators $\{E_k\}_{k=1}^M$ is known as an operator-sum representation or a Kraus representation of the quantum channel $\operator{N}$.
\end{theorem}
\begin{proof}[Proof Outline ({see, \eg,~\cite[Theorem~8.1]{nielsen2011quantum}} for details)]
It is relatively straightforward to show that the mapping $\rho\mapsto \sum_{k=1}^{M} E_k\rho E_k^\Herm$ is a quantum channel.
To show any quantum channel can be represented in this way, we introduce an auxiliary system $\system{R}$ with $\dim(\hilbert_\system{R})=\dim(\hilbert_\system{A})$, and define
\[
\sigma \defeq \left( \id_\system{R}\tensor\operator{N} \right) \left( \left( \sum_i \bra{i_\system{R}}\bra{i_\system{A}} \right) \left( \sum_i \bra{i_\system{R}}\bra{i_\system{A}} \right)^\Herm \right),
\]
where $\{\bra{i_\system{R}}\}_i$ and $\{\bra{i_\system{A}}\}_i$ are some orthonormal bases for $\hilbert_\system{R}$ and $\hilbert_\system{A}$, respectively.
Decompose $\sigma$ into $\sigma=\sum_{k}\braket{k_\system{RB}}$, where $\bra{k_\system{RB}}\in\hilbert_\system{RB}$ need not be normalized.
We claim that~\eqref{eq:Kraus} can be satisfied by the operators $\{E_k\}_{k=1}^M$ defined as
\[
E_k \bra{\psi} \defeq \sum_i \inner{i_\system{A}}{\psi} \inner{i_\system{R}}{k_\system{RB}}
\quad \forall \bra{\psi}\in\hilbert_\system{A}
\]
for each $k=1,\ldots,M$.
(Please refer to the entry ``$\inner{\phi_{\system{A}}}{\psi_{\system{AB}}}$'' in the notation table.)
Finally, $M\leqslant \dim(\hilbert_\system{A})\cdot\dim(\hilbert_\system{B})$, since the rank of $\sigma$ is at most $\dim(\hilbert_\system{R}\tensor\hilbert_\system{B})=\dim(\hilbert_\system{R})\cdot\dim(\hilbert_\system{B})=\dim(\hilbert_\system{A})\cdot\dim(\hilbert_\system{B})$.
\end{proof}
\begin{theorem}[Unitary representation/Stinespring representation]
	\label{thm:Stinespring}
	$\operator{N}$ is a finite-dimensional quantum channel if and only if there exists some environment system $\system{env}$ with state space of dimension at most $\dim(\hilbert_\system{A}) \cdot \dim(\hilbert_\system{B})$, and some unitary operator $U$ such that
	\begin{equation}\label{eq:unirary:representation}
	\operator{N}(\rho) = \tr_{\system{env}}\left(U(\rho\tensor\braket{0})U^\Herm\right)
	\qquad \forall \rho\in\LinearOp(\hilbert_\system{A})
	\end{equation}
	for some state vector $\bra{0}$ of $\system{env}$.
\end{theorem}
\begin{proof}[Proof Outline ({see, \eg,~\cite[Box~8.1]{nielsen2011quantum}} for details)]
We use Theorem~\ref{thm:Kraus}.
Without loss of generality, assume $\dim(\hilbert_\system{A})=\dim(\hilbert_\system{B})$.
Let  $\{\bra{i_\system{env}}\}_i$ be an orthonormal basis for $\hilbert_\system{env}$ where $\bra{0_\system{env}}=\bra{0}$.
In this case, by letting $E_k\defeq \ket{k_\system{env}}U\bra{0_\system{env}}$, we have
\[
\tr_{\system{env}}\left(U(\rho\tensor\braket{0})U^\Herm\right) = \sum_{k} E_k\rho E_k^\Herm.
\]
By verifying $\sum_k E_k^\Herm E_k = I$, we have shown that the mapping $\operator{N}$ in~\eqref{eq:unirary:representation} is a quantum channel.
On the other hand, to show that any quantum channel $\operator{N}$ can be represented in this way, it suffices to, given any Kraus representation $\{E_k\}_k$ of some quantum channel, construct a unitary operator $U$ such that $\ket{k_\system{env}}U\bra{0_\system{env}}=E_k$.
This can be done by fixing the first column (of blocks) of $[U]$ to be $[E_k]$'s, and filling up the remaining columns of the matrix while ensuring all columns are orthonormal \wrt each other.
Here, $[E_k]$ is the matrix representation of $E_k$ under some orthonormal basis $\{\bra{i_\system{P}}\}_i$ (for both $\hilbert_\system{A}$ and $\hilbert_\system{B}$); and $[U]$ is the matrix representation of $U$ under the basis $\{\bra{k_\system{env}}\tensor\bra{i_\system{P}}\}_{k,i}$.
\end{proof}
\begin{corollary}[Choi--Jamio{\l}kowski matrix] \label{cor:Choi}
Let $\operator{N}:\LinearOp(\hilbert_\system{A})\to\LinearOp(\hilbert_\system{B})$ be a quantum channel, and let $\{\bra{\ell}\}_{\ell=1}^{d}$ be an orthonormal basis for both $\hilbert_\system{A}$ and $\hilbert_\system{B}$ (where $d=\dim(\hilbert_\system{A})=\dim(\hilbert_\system{B})$).
Let the matrix $W\in\Complex^{d^2\times d^2}$ be defined as
\begin{equation}
W_{(i,j),(i'j')} \defeq \ket{j}\operator{N}(\bra{i}\!\ket{i'})\bra{j'}
	\qquad \forall i,i',j,j'=1,\ldots,d.
\end{equation}
Then the matrix $W$ is PSD.
\end{corollary}
\begin{proof}
This is a direct result of Theorem~\ref{thm:Kraus}.
We omit the details.
\end{proof}
The matrix $W$ defined in Corollary~\ref{cor:Choi} is often known as a \emph{Choi--Jamio{\l}kowski matrix} of the quantum channel $\operator{N}$ (\wrt basis $\{\bra{\ell}\}_{\ell=1}^{d}$).
In particular, it holds that
\begin{equation}
[\operator{N}(\rho)]_{j,j'} = \sum_{i,i'}W_{(i,j),(i'j')}\cdot [\rho]_{i,i'}
	\quad \forall \rho\in\DensOp(\hilbert_\system{A}),
\end{equation}
where $[\rho]$ and $[\operator{N}(\rho)]$ are the matrix representation of the corresponding operators under the basis $\{\bra{\ell}\}_{\ell=1}^{d}$.
\begin{example} Below is a list of examples of quantum channels.
	\begin{description}
		\item[Bit flip channel] $\rho\mapsto E_0\rho E_0^\Herm\!+\!E_1\rho E_1^\Herm$, 
		$E_0\!=\!\sqrt{p}\left[\begin{smallmatrix}1 & 0 \\ 0 & 1\end{smallmatrix}\right]$,
		$E_1\!=\!\sqrt{1\!-\!p}\left[\begin{smallmatrix}0 & 1 \\ 1 & 0\end{smallmatrix}\right]$,
		$p\in(0,1)$.
		\item[Phase flip channel] $\rho\mapsto E_0\rho E_0^\Herm\!+\!E_1\rho E_1^\Herm$, 
		$E_0\!=\!\sqrt{p}\left[\begin{smallmatrix}1 & 0 \\ 0 & 1\end{smallmatrix}\right]$,
		$E_1\!=\!\sqrt{1\!-\!p}\left[\begin{smallmatrix}1 & 0 \\ 0 & -1\end{smallmatrix}\right]$,
		$p\in(0,1)$.
		\item[Amplitude damping] $\rho\mapsto E_0\rho E_0^\Herm\!+\!E_1\rho E_1^\Herm$, 
		$E_0\!=\!\left[\begin{smallmatrix}1 & 0 \\ 0 & \sqrt{1-\gamma}\end{smallmatrix}\right]$,
		$E_1\!=\!\left[\begin{smallmatrix}0 & \sqrt{\gamma} \\ 0 & 0\end{smallmatrix}\right]$,
		$\gamma\in(0,1)$.
		\item[Depolarizing channel] $\rho\mapsto\frac{pI}{2}+(1-p)\rho$, $p\in(0,1)$.
		\qedhere
	\end{description}
\end{example}
%*******************************************************************************
\subsection{Communications over quantum channels}
Quantum channels can be used to convey either classical or quantum information.
In this section, we focus on classical communications.
Though quantum communication is an important topic in quantum information theory, it is beyond our discussion scope.
\par
% General setup
Given (multiple copies of) a quantum channel from system $\system{A}$ to system $\system{B}$, the task of classical communication over this channel is to transmit some classical information from Alice, who has \emph{access} to the copies of system $\system{A}$, to Bob, who has \emph{access} to the copies of system $\system{B}$, namely
\begin{itemize}
	\item Alice is able to fully manipulate the (joint) state (\ie, the density operator) of the copies of $\system{A}$;
	\item Bob can perform any quantum measurement of his choice on the copies of $\system{B}$.
\end{itemize}
\begin{figure}
	\centering
	\begin{tikzpicture}[square/.style={rectangle, draw, minimum size=1cm}]
		\node[square] (Source) {};
		\node[square, minimum height=5cm, right=1cm of Source] (Map) {};
		\path (Source) edge[-latex] node[above=0pt] {$m$} (Map);
		\node[square, right=1cm of Map.north east, anchor=north west] (Enc1) {};
		\node[square, below=.5cm of Enc1] (Enc2) {};
		\node[square, right=1cm of Map.south east, anchor=south west] (Encn) {};
		\path (Enc2) edge[draw=none] node (Encm) {$\vdots$} (Encn);
		\node[square, right=1cm of Enc1] (N1) {$\operator{N}$};
		\node[square, right=1cm of Enc2] (N2) {$\operator{N}$};
		\node[square, right=1cm of Encn] (Nn) {$\operator{N}$};
		\path (N2) edge[draw=none] node (Nm) {$\vdots$} (Nn);
		\node[square, minimum height=5cm, right=1cm of N1.north east, anchor=north west] (Measurement) {};
		\node[square, right=1cm of Measurement] (Destination) {};
		
		\path (Map.east|-Enc1) edge[-latex] node[above=0pt] {$m_1$} (Enc1);
		\path (Map.east|-Enc2) edge[-latex] node[above=0pt] {$m_2$} (Enc2);
		\path (Map.east|-Encn) edge[-latex] node[above=0pt] {$m_n$} (Encn);
		\path (Enc1) edge[-latex] node[above=0pt] {$\rho_{\system{A}_1}^{m_1}$} (N1);
		\path (Enc2) edge[-latex] node[above=0pt] {$\rho_{\system{A}_2}^{m_2}$} (N2);
		\path (Encn) edge[-latex] node[above=0pt] {$\rho_{\system{A}_n}^{m_n}$} (Nn);
		\path (N1) edge[-latex] node[above=0pt] {$\sigma_{\system{B}_1}^{m_1}$} (N1-|Measurement.west);
		\path (N2) edge[-latex] node[above=0pt] {$\sigma_{\system{B}_2}^{m_2}$} (N2-|Measurement.west);
		\path (Nn) edge[-latex] node[above=0pt] {$\sigma_{\system{B}_n}^{m_n}$} (Nn-|Measurement.west);
		\path (Measurement) edge[-latex] node[above=0pt] {$\hat{m}$} (Destination);
		
		% Draw a measurement symbol
		\node[minimum size = 3pt, fill = black, inner sep = 0pt, outer sep = 0pt, below = 6pt of Measurement.center, circle, anchor = center] (m) {};
		\draw[decoration={markings, mark = at position 1 with {\arrow[>=latex]{>}}}, postaction={decorate}, draw = none] (m.center)--([yshift=14pt,xshift=6pt]m.center);
		\draw(m.center)--([yshift=12.6pt, xshift=5.4pt]m.center);
		\draw ([xshift=8pt]m.center) arc (0:180:8pt);
	\end{tikzpicture}
	\caption{Product-state setup for classical communications over a quantum channel.}
	\label{fig:cc:over:qc}
\end{figure}
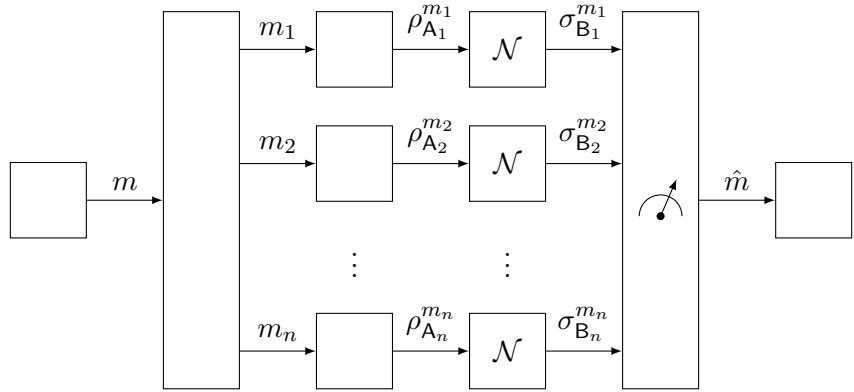
% Product-State Setup
For example, to transmit a message $m\in\set{M}_n$ using $n$ copies of the channel (\ie, $\operator{N}^{\tensor n}$ will be applied):
Immediately before applying the channels, Alice prepares the states of the systems $\system{A}_1^n$ such that its (joint) density operator becomes $\rho_{\system{A}_1}^{m_1}\tensor\rho_{\system{A}_2}^{m_2}\tensor\cdots\tensor\rho_{\system{A}_n}^{m_n}$, where $(m_1,m_2,\ldots,m_n)$ is some function of $m$.\footnote{Note that $\system{A}_1^n \defeq \system{A}_1\tensor\cdots\tensor\system{A}_n$, where $\system{A}_i$ is the copy of $\system{A}$ at the $i$-th channel use. Similar statement also applies to $\system{B}_1^n$.}
Upon receiving the state of the systems $\system{B}_1^n$, Bob estimates the original message as a result of some physical operation on $\system{B}_1^n$, where, for each possible estimate $\hat{m}$, the corresponding effect is described by some POVM element $E_{\hat{m}}$.
An error is said to have occurred if this estimate differs from $m$, and in this setup, this happens with probability $\prob_{e}^{(n)}(m)=1-\tr\left(E_m\cdot\Tensor_{\ell=1}^n \operator{N}(\rho_{\system{A}_\ell}^{m_\ell})\right)$.
If there exists a sequence (indexed by $n$) of encoders $m\mapsto\Tensor_{\ell=1}^n \rho_{\system{A}_\ell}^{m_\ell}$ and POVMs $\{E_{\hat{m}}\}_{\hat{m}\in\set{M}_n}$ such that $\lim_{n\to\infty}\prob_{e}^{(n)}=\mathbf{0}$, a communication rate of $R\defeq\liminf_{n\to\infty}\frac{1}{n}\log{\size{\set{M}_n}}$ is said to have been \emph{achieved}.
If the supremum of all achievable rates and the infimum of all unachievable rates are the same, this number is known as the capacity.
The capacity of this particular setup exists and is often referred to as the \emph{product-state capacity} of the quantum channel $\operator{N}$ (see Theorem~\ref{thm:HSW}).
The product-state capacity is denoted by $C_1$.
\begin{definition}[Holevo capacity]
Given a quantum channel $\operator{N}:\LinearOp(\hilbert_\system{A})\to\LinearOp(\hilbert_\system{B})$, the non-negative real number
\begin{equation}\label{eq:def:holevo:quantity}
	\chi(\operator{N}) \defeq
	\sup_{p_i \geqslant 0, \sum_{i}p_i=1 \atop \rho_i\in\DensOp(\hilbert_\system{A})}
	\left\{\qEntropy\left(\operator{N}\left(\sum_{i}p_i\rho_i\right)\right) - 
	\sum_{i} p_i \qEntropy\left(\operator{N}\left(\rho_i\right)\right)\right\}
\end{equation}
is called the the Holevo capacity of $\operator{N}$.
\end{definition}
\begin{theorem}[HSW Theorem~\cite{holevo1973bounds, holevo1998capacity, schumacher1997sending}] \label{thm:HSW}
The product-state capacity of a (memoryless) quantum channel coincides with its Holevo capacity.
Namely, $C_1(\operator{N})\equiv\chi(\operator{N})$.
\end{theorem}
\begin{proof}
Omitted. (See, \eg,~\cite[Section~12.3.2]{nielsen2011quantum}.)
\end{proof}
As a result of Theorem~\ref{thm:HSW}, computing the product-state capacity is equivalent to maximizing the \emph{Holevo quantity} $\chi(\{\varrho_i\}_i)$, where
\begin{equation}
\chi(\{\varrho_i\}_i) \defeq 
\qEntropy\left(\operator{N}\left(\sum_{i}\varrho_i\right)\right) - \sum_{i} \qEntropy\left(\operator{N}\left(\varrho_i\right)\right)
\end{equation}
and where $\{\varrho_i\}_i$ is an index set of operators such that\footnote{Such a set of positive operators is sometimes referred to as an \emph{ensemble}; however, more often, an ensemble refers to an indexed set of pairs $\{(p_i,\rho_i)\}_i$ where $(p_i)_i$ is a \pmf, and $\{\rho_i\}_i$'s are density operators. These two formalisms are equivalent.} $\varrho_i>0$ for each $i$, and $\tr(\sum_{i}\varrho_i)=1$.
The complexity of solving this optimization problem \emph{exactly} is NP-complete~\cite{beigi2007complexity}.
Fortunately, given any positive tolerance (\wrt the distance from the exact maximizer), an estimated optimal $\{\varrho_i\}_i$ can be computed efficiently in an iterative manner~\cite{nagaoka1998algorithms}.
\par
% Additivity
However, knowing $C_1$ is still steps away from finding the unrestricted capacity for transmitting classical information (\aka classical capacity).
If the channel is \emph{additive} (see Definition~\ref{def:additive:cc}), \ie, joint manipulation of multiple input systems does not increase capacity, it holds that classical capacity $C=C_1$.
Otherwise, the situation is more complicated.
Namely, we need to consider a relaxed setup by allowing Alice to manipulate at most $k$ input systems jointly, \ie, Alice encodes each message $m$ as $\rho_{\system{A}_1^k}^{m_1}\tensor\rho_{\system{A}_{k+1}^{2\cdot k}}^{m_2}\tensor\cdots\tensor\rho_{\system{A}_{(n\!-\!1)\cdot k}^{n\cdot k}}^{m_n}$ for some $n\in\Integers_{>0}$.
The capacity of this relaxed setup is $C_k\defeq\frac{1}{k}C_1(\operator{N}^{\tensor k})$.
Finding the classical capacity is then equivalent to finding the supremum of $\{C_k\}_{k=1}^{\infty}$, \ie, $C=\limsup_{k\to\infty}C_k(\operator{N})$.
\begin{definition}[Additivity of classical capacity] \label{def:additive:cc}
The classical capacity of a quantum channel $\operator{N}$ is said to be \emph{additive} if
\begin{equation}
C_1(\operator{N}) = C_k(\operator{N})
\defeq \frac{1}{k}C_1(\operator{N}^{\tensor k})
\end{equation}
for all positive integers $k$.
\end{definition}
Unfortunately, classical capacities are \emph{not} additive in general~\cite{hastings2009superadditivity}, despite the alternative speculations~\cite{amosov2000some, fukuda2005extending} in the early years of the study.
This makes the computation of classical capacities particularly difficult, since evaluating the Holevo quantity alone would take an exponential amount of time (and storage) as the dimensions of the involved quantum systems grow.
On the other hand, quite a number of additive channels exist, including quantum erasure channels~\cite{bennett1997capacities}, unital qubit channels~\cite{king2002additivity}, depolarizing channels~\cite{king2003capacity}, and transpose depolarizing channels~\cite{datta2006additivity}.
\par
Readers must be reminded that the setup presented in this section is not the only type of classical communication over quantum channels.
Other important setups include entanglement-assist classical communications~\cite{bennett1999entanglement, holevo2002entanglement, shor2004classical} and private (classical) communications~\cite{devetak2005private}.
These discussions are beyond the scope of this thesis, and we omit the details.
%\subsubsection{Quantum communications over quantum channels}
	%Quantum channels can be used to convey quantum information as well.
	%quantum capacity~\cite{devetak2005private,wolf2007quantum}
	%capacities:~\cite{lloyd1997capacity,barnum1998information}\\
	%von Neumann capacity~\cite{adami1997neumann}\\
	%Converse coding theorems~\cite{ogawa1999strong, winter1999coding, bennett2002entanglement, wilde2014strong, wilde2014strong2}
	%Squeezed vacuum channel~\cite{daffer2003quantum}
	%Additivity~\cite{amosov2000some}
%*******************************************************************************
% Chapter 2 Double-Edge Factor Graphs*******************************************
% This chapter is based on the 1st-year Report and the ITW2017 paper
\chapter{Double-Edge Factor Graphs}\label{chapter:DeFGs} 
In this chapter, we consider the problem of representing the dynamics of quantum systems using concepts of factor graphs.
The result is a graphical model called double-edge factor graphs (DeFGs).
The idea is based on the linear-algebraic descriptions of factor graphs~\cite{al2011normal} and can be connected to the methods proposed by Loeliger and Vontobel for representing quantum systems~\cite{loeliger2012factor, loeliger2017factor}.
Related studies also include~\cite{mori2015holographic}, in which he proposed a generalized bipartite graphical model capable of representing quantum systems.
For comparison, the global functions considered here are in the form of
\begin{equation}\label{eq:global:function:DeFG}
	g(\vx,\tvx;\vy) =
	\prod_{a\in\mathcal{F}} f_a(\vx_{\nb{a}},\tvx_{\nb{a}};\vy_{\delta{a}}),
\end{equation}
where for any fixed $\vy_{\delta{a}}$, the matrix $\left[f_a^{\vy_{\delta{a}}}\right]_{\vx_{\nb{a}},\tvx_{\nb{a}}}\defeq f_a^{\vy_{\delta{a}}}(\vx_{\nb{a}},\tvx_{\nb{a}})$ is PSD.
In contrast, in \cite{loeliger2012factor, loeliger2017factor}, the functions $\{f_a\}$'s are  required to be decomposable, \ie,
\begin{equation}\label{eq:decompose}
	f_a(\vx_{\nb{a}},\tvx_{\nb{a}};\vy_{\delta{a}}) =
	\tilde{f}_a(\vx_{\nb{a}}; \vy_{\delta{a}})
	\conj{\tilde{f}_a(\tvx_{\nb{a}}; \vy_{\delta{a}})}
	\quad \exists\tilde{f}_a \quad \forall a.
\end{equation}
Therefore, the scenario we consider here is more general.
We also consider the problem of computing the partition sums of such global functions, \ie, 
\begin{equation}
	Z(\set{G}) \defeq \sum_{\vy}\sum_{\vx,\,\tvx} g(\vx,\tvx;\vy).
\end{equation}
This problem is directly related to inference problems on systems with quantum components.
For solving this problem, we propose a generalized version of the BA algorithms for DeFGs.
Some numerical results and (attempted) analysis of the algorithms are also included in this chapter.
%*******************************************************************************
\section[From {FGs} for Quantum Probabilities to {DeFGs}]{From Factor Graphs for Quantum Probabilities to Double-Edge Factor Graphs}
%*******************************************************************************
\subsection{Factor Graphs for Quantum Probabilities}\label{subsec:FG:QP}
This section reviews the method proposed in~\cite{loeliger2012factor, loeliger2017factor}, where a relaxed version of factor graphs is used to represent ``quantum probabilities''.
We start with the following example.
\begin{example}\label{exp:quantum_1} 
Consider the factor graph in Figure~\ref{fig:quantum_1}.
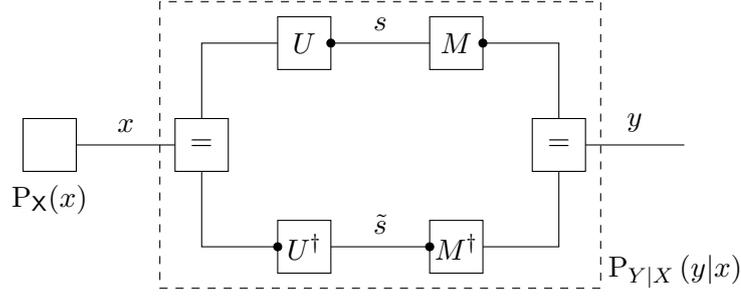
\begin{figure}\centering\centering
\begin{tikzpicture}[
	nodes={draw=none,minimum size=0pt,inner sep=0pt,outer sep=0pt},
	factor/.style={rectangle, minimum size=.7cm, draw},
	node distance=2cm]
	\node[factor] (P) {}; \node [below=5pt of P] {$\prob_\rv{X}(x)$};
	\node[factor] (e1) [right of = P] {$=$};
	\node (m1) [right=1cm of e1] {};
	\node[factor] (U1)[above=1cm of m1] {$U$};
	\node[factor] (U1H)[below=1cm of m1] {$U^\Herm$};
	\node[factor] (BH)[right of=U1] {$M$};
	\node[factor] (B)[right of=U1H] {$M^\Herm$};
	\path (B) -- (BH) coordinate [midway] (m2);
	\node[factor] (e2) [right=1cm of m2] {$=$};
	\node (Y)[right=1.3cm of e2] {};
	\path (P) edge node[above=5pt] {$x$} (e1);
	\draw (e1) |- (U1); \draw[-*] (e1) |- (U1H);
	\path (U1) edge[draw,*-] node[above=5pt] {$s$} (BH); \path (U1H) edge[draw,-*] node[above=5pt] {$\ts$} (B);
	\draw[*-] (BH) -| (e2);\draw (B) -| (e2);
	\path (e2) edge node[above=5pt] {$y$} (Y);\
	\draw [dashed] ([xshift=-.2cm,yshift=+.2cm]e1.west|-U1.north) rectangle ([xshift=+.2cm,yshift=-.2cm]e2.east|-U1H.south);
	\node [anchor=south west,xshift=.1cm] at ([xshift=+.2cm,yshift=-.2cm]e2.east|-U1H.south) {$\prob_{Y|X}\left(y|x\right)$};
\end{tikzpicture}
\caption{FG describing a simple quantum system.}
\label{fig:quantum_1}
\end{figure}
Here, some of the local functions are specified by matrices.
The dots at the end of some of the edges specify which variable is the first index of the matrix.
This setup depicts a cq-channel followed by a unitary evolution and then a measurement.
In this case, the global function is
\begin{align}\label{eq:defg:ex1:gf}
	g(s,\ts;x,y)&=\prob_\rv{X}(x)\cdot U_{s,x}U^\Herm_{x,\ts} \cdot M_{y,s}M^\Herm_{\ts,y}\\
	&=\prob_\rv{X}(x)\cdot U(s,x)\conj{U(\ts,x)}\cdot M(y,s)\conj{M(y,\ts)}, \label{eq:defg:ex1:gf:2}
\end{align}
where $U$ and $M$ are both complex unitary matrices.
Note that the letters $U$ and $M$ denote matrices in~\eqref{eq:defg:ex1:gf}, but functions in~\eqref{eq:defg:ex1:gf:2}.
\par
Though we have complex-valued functions as local functions in these cases, the marginals and the partition sum are still real and non-negative due to the symmetric structure.
For example, by closing the dashed box in Figure~\ref{fig:quantum_1}, the resultant function can be expressed as
\begin{align}\label{eq:defg:ex1:cond}
	\prob_{Y|X}(y|x) 
	&= \sum_{s}U(s,x)M(y,s) \cdot \sum_{\ts}\conj{U(\ts,x)M(y,\ts)}\\
	&= \biggabs{\sum_{s}U(s,x)M(y,s)}^2,
\end{align}
which is the conditional probability of $\rv{Y}$ given $\rv{X}$ in this setup.%a cq-channel followed by a unitary evolution and a measurement.
\end{example}
Consider the factor graphs constructed in a symmetric manner where complex functions always appear in conjugate pairs.
In such cases, any factorization with local functions satisfying~\eqref{eq:decompose} can be represented by some factor graph as in Example~\ref{exp:quantum_1}.
A variety of such representable quantum systems can be found in~\cite{loeliger2012factor, loeliger2017factor}.
%*******************************************************************************
\subsection{Double-Edge Factor Graphs (DeFGs)}
\begin{example}[Redrawing of Example~\ref{exp:quantum_1}]\label{fig:ex:DeFG_1}
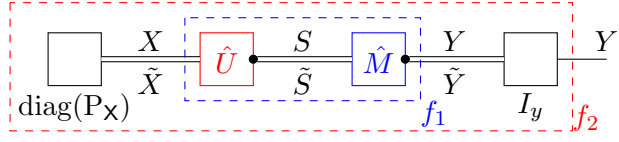
\begin{figure}\centering
\begin{tikzpicture}[
	nodes={draw=none,minimum size=0pt,inner sep=0pt,outer sep=0pt},
	factor/.style={rectangle, minimum size=.7cm, draw},
	node distance=2cm]
	\node[factor] (P) {}; \node [below=.1cm of P] {$\diag(\prob_\rv{X})$};
	\node[factor, red] (u) [right of= P] {$\hat{U}$};
	\node[factor, blue] (b) [right of= u] {$\hat{M}$};
	\node[factor] (Q) [right of= b] {}; \node [below=.1cm of Q] {$I_y$};
	
	\path (P) edge[double] node[above=3pt] {$X$} node[below=3pt] {$\tX$} (u);
	\path (u) edge[*-, double] node[above=3pt] {$S$} node[below=3pt] {$\tS$} (b);
	\path (b) edge[*-,double] node[above=3pt] {$Y$} node[below=3pt] {$\tY$} (Q);
	\path (Q) edge node[above=3pt, pos=1] {$Y$} ([xshift=1cm]Q);

	\draw [blue, dashed] ([xshift=-.2cm,yshift=.2cm]u.north west) rectangle ([xshift=.2cm,yshift=-.2cm]b.south east);
	\node [anchor=north west, blue] at ([xshift=.2cm,yshift=-.2cm]b.south east) {$f_1$};
	\draw [red, dashed] ([xshift=-.5cm,yshift=.4cm]P.north west) rectangle ([xshift=.2cm,yshift=-.6cm]Q.south east);
	\node [anchor=south west, red] at ([xshift=.2cm,yshift=-.6cm]Q.south east) {$f_2$};
\end{tikzpicture}
\caption{Quantum NFG for an elementary quantum system.}
\label{fig:DeFG_1}
\end{figure}
Consider the factor graph in Figure~\ref{fig:DeFG_1} as a redrawing of Figure~\ref{fig:quantum_1}, where each local function is specified by a matrix with the variables associated with upper edges composing the first index of the matrix.
The dots in this graph are used to specify how the variables are \emph{arranged} to compose the indices of the matrices.
For example, the factors labeled by the matrices $\hat{U}$ and $\hat{M}$ are associated with the functions
\begin{align}
	\hat{U}(s,x;\ts,\tx) &\defeq \hat{U}_{(s,x),(\ts,\tx)} = U_{s,x}\cdot U^\Herm_{\tx,\ts} = U(s,x)\conj{U(\ts,\tx)},\\
	\hat{M}(y,s;\ty,\ts) &\defeq \hat{M}_{(y,s),(\ty,\ts)} = M_{y,s}\cdot M^\Herm_{\ts,\ty} = M(y,s)\conj{M(\ty,\ts)},
\end{align}
respectively.
Thus, the global function is
\begin{equation}
	g(x,s,t,\tx,\ts,\ty) = \prob_\rv{X}(x)\cdot \hat{U}(s,x;\ts,\tx) \cdot \hat{M}(y,s;\ty,\ts) \cdot \delta_{y,\ty},
\end{equation}
which is equivalent to~\eqref{eq:defg:ex1:gf}.
In addition, the exterior functions $f_1$, $f_2$ corresponding to the inner and outer boxes in Figure~\ref{fig:DeFG_1} can be expressed as
\begin{align}
	f_1(x,y,\tx,\ty) 
	& = \sum_{s,\ts} U(s,x)\conj{U(\ts,\tx)}\cdot M(y,s)\conj{M(\ty,\ts)}\\
	& = \sum_{s}U(s,x)M(y,s) \cdot \conj{\sum_{\ts}U(\ts,\tx)M(\ty,\ts)}\\
	& = [M\cdot U]_{y,x}\cdot [M\cdot U]^\Herm_{\tx,\ty},\label{eq:foornote:order}\\
	f_2 (y)
	&= \sum_{x}\prob_\rv{X}(x)\cdot f_1(x,y,x,\ty) = \sum_{x} [M\cdot U]_{y,x}\cdot \prob_\rv{X}(x)\cdot [M\cdot U]^\Herm_{x,\ty} \\
	&= \left[M\cdot U\cdot \diag(\prob_\rv{X}) \cdot U^\Herm \cdot M^\Herm\right]_{y,y},
\end{align}
respectively.\footnote{Note that in~\eqref{eq:foornote:order}, we take Hermitian of a matrix first, then extract the indices. Namely, for any complex matrix $A$, $A^\Herm_{i,j}=\conj{A_{j,i}}$.}
One can check that $f_2(y)=\sum_{x}\prob_\rv{X}(x)\cdot\prob_{\rv{Y}|\rv{X}}(y|x)$, where $\prob_{\rv{Y}|\rv{X}}$ was defined earlier in~\eqref{eq:defg:ex1:cond}, \ie, $f_2(y)$ is the \pmf of $Y$ given the \pmf of $X$ being $\prob_\rv{X}$.
One can view $f_1$ as a matrix generalization of $\prob_{\rv{Y}|\rv{X}}$.
In this example, $\diag(f_1)=\prob_{\rv{Y}|\rv{X}}$.
\end{example}
In the above example, we have successfully described a quantum system using a simpler graph with compatible ``closing-the-box'' operations.
Such techniques can be made generic, and we name such graphical models \emph{double-edge factor graphs} (DeFGs).
%*******************************************************************************
\begin{definition}[Double-Edge Factor Graph] \label{def:DeFG} \index{double-edge factor graph}
A double-edge factor graph (DeFG) is a bipartite graph $\set{G}=(\set{V}_1\sqcup\set{V}_2,\set{F},\set{E}_1\sqcup\set{E}_2)$ associated with variable sets $\mathfrak{V}_1$, $\mathfrak{V}_2$, and a factor set $\mathfrak{F}$, where
\begin{itemize}
	\item $\set{E}_1\subset\set{V}_1\times\set{F}$, $\set{E}_2\subset\set{V}_2\times\set{F}$; and for each $a\in\set{F}$, we denote the sets of its neighbors in $\set{V}_1$ and $\set{V}_2$ by $\delta{a}$ and $\nb{a}$, respectively;
	\item $\mathfrak{V}_1=\{\set{X}_i\}_{i\in\set{V}_1}$ and $\mathfrak{V}_2=\{\set{S}_j\}_{j\in\set{V}_2}$ are indexed by $\set{V}_1$ and $\set{V}_2$, respectively, and each element of $\mathfrak{V}_1$ and  $\mathfrak{V}_2$ is a set (\aka alphabets);
	\item $\mathfrak{F}=\{f_a\}_{a\in\set{F}}$ is indexed by $\set{F}$, where for each $a\in\set{F}$ $f_a:\bigtimes_{i\in\nb{a}}\set{S}_j^2\times\bigtimes_{i\in\delta{a}}\set{X}_i \to \Complex$ and where the matrix $[f_a(\vx_{\delta{a}})]_{\vs_\nb{a},\tvs_\nb{a}}\defeq f_a(\vs_\nb{a},\tvs_\nb{a};\vx_{\delta{a}})$ is PSD for each $\vx_{\delta{a}}$.
\end{itemize}
The function $g(\vs,\tvs;\vx) \defeq \prod_{a\in\set{F}} f_a(\vs_\nb{a},\tvs_\nb{a};\vx_{\delta{a}})$ is called the \emph{global function} of $\set{G}$, and in this case, $\set{G}$ is also said to be representing the factorization $g(\vs,\tvs;\vx) \defeq \prod_{a\in\set{F}} f_a(\vs_\nb{a},\tvs_\nb{a};$ $\vx_{\delta{a}})$.
Like (classical) factor graphs, a DeFG is said to be \emph{normal} if the degree of any vertex in $\set{V}_1\sqcup\set{V}_2$ is at most 2.
\end{definition}
\begin{remark}
Similar to factor graphs, in a normal DeFG, we redraw the vertices in $\set{V}_1$ as edges and the vertices in $\set{V}_2$ as double edges, as in Figure~\ref{fig:ex:DeFG_1}.
\end{remark}
\begin{remark}
Any DeFG can be converted into a normal DeFG by properly introducing equality node(s) along with introducing suitable auxiliary variables.
\end{remark}
In the previous example, we have converted a factor graph describing a quantum system into a DeFG.
Such conversions can be done systematically, as described in Table~\ref{tb:CintoDeFG}.
\begin{table}\centering
\begin{tabular}{|p{3.5cm}|p{3.5cm}|m{5.5cm}|}
	\hline
	FG Description & DeFG Description & Remarks\\
	\hline
	\centering\raisebox{-.5\height}{
	\begin{tikzpicture}[nodes = {draw= none, minimum size = 0pt},
		factor/.style={rectangle, minimum size=.5cm, draw},
		node distance=1.7cm]
	\node[factor] (u) {}; \node[below = 0pt of u] {$U$};
	\node[factor] (l) [below=.5cm of u] {}; \node[below = 0pt of l] {$U^\Herm$};
	\draw[*-] (u) -- ([xshift=1cm]u.east); \node[anchor=west] at ([xshift=1cm]u.east) {$s'$};
	\draw     (u) -- ([xshift=-1cm]u.west); \node[anchor=east] at ([xshift=-1cm]u.west) {$s$};
	\draw     (l) -- ([xshift=1cm]l.east); \node[anchor=west] at ([xshift=1cm]l.east) {$\ts'$};
	\draw[*-] (l) -- ([xshift=-1cm]l.west); \node[anchor=east] at ([xshift=-1cm]l.west) {$\ts$};
	\node[fit=(current bounding box),inner ysep=1mm,inner xsep=0]{};
	\end{tikzpicture}
	}& \centering\raisebox{-.5\height}{
	\begin{tikzpicture}[nodes = {draw= none, minimum size = 0pt},
		factor/.style={rectangle, minimum size=.5cm, draw},
		node distance=1.7cm]
	\node[factor] (f) {}; \node[below = 0pt of f] {$\hat{U}$};
	\path (f) edge[draw,*-,double] node[above=0pt] {$s'$} node[below=0pt] {$\ts'$} ([xshift=1cm]f.east);
	\path (f) edge[draw,double] node[above=0pt] {$s$} node[below=0pt] {$\ts$} ([xshift=-1cm]f.west);
	\node[fit=(current bounding box),inner ysep=1mm,inner xsep=0]{};
	\end{tikzpicture}
	} & $\hat{U}((s',s),(\ts',\ts))=U_{s',s}U^\Herm_{\ts,\ts'}$\\
	\hline
	\centering\raisebox{-.5\height}{
	\begin{tikzpicture}[nodes = {draw=none, minimum size = 0pt},
		factor/.style={rectangle, minimum size=.5cm, draw},
		node distance=1.7cm]
	\node[factor] (e) {}; \node at (e) {$=$};
	\node (U)[above right =.5cm and .5cm of e, anchor=west] {$s$};
	\node (UH)[below right =.5cm and .5cm of e, anchor=west] {$\ts$};
	\draw (e) |- (U.west);
	\draw (e) |- (UH.west);
	\path (U.east) edge[draw=none] node (or) [anchor=west] {or~} (UH.east);
	\node[factor, right=0pt of or] (e) {}; \node at (e) {$I$};
	\node (U)[above right =.5cm and .5cm of e, anchor=west] {$s$};
	\node (UH)[below right =.5cm and .5cm of e, anchor=west] {$\ts$};
	\draw[*-] (e) |- (U.west);
	\draw (e) |- (UH.west);
	\node[fit=(current bounding box),inner ysep=1mm,inner xsep=0]{};
	\end{tikzpicture}
	} & \centering\raisebox{-.5\height}{
	\begin{tikzpicture}[nodes = {draw=none, minimum size = 0pt},
		factor/.style={rectangle, minimum size=.5cm, draw},
		node distance=1.7cm]
	\node[factor] (I) {$=$}; \node [below=0pt of I] {$I$};
	\path (I) edge[double] node[above=0pt] {$s$} node[below=0pt] {$\ts$} ([xshift=1.5cm]I);
	\node[fit=(current bounding box),inner ysep=1mm,inner xsep=0]{};
	\end{tikzpicture}
	} & $I(s,\ts)=\delta_{s,\ts}$\\
	\hline
	\centering\raisebox{-.5\height}{
	\begin{tikzpicture}[nodes = {draw= none, minimum size = 0pt},
		factor/.style={rectangle, minimum size=.5cm, draw},
		node distance=1.7cm]
	\node[factor] (P) {}; \node [below=0pt of P] {$\diag(\prob_\rv{X})$};
	\node[factor] (e) [right=.7cm of P] {$=$};
	\node (U)[above right =.5cm and .5cm of e, anchor=west] {$s$};
	\node (UH)[below right =.5cm and .5cm of e, anchor=west] {$\ts$};
	\path (P) edge node[above] {$x$} (e);
	\draw (e) |- (U);
	\draw (e) |- (UH);
	\node[fit=(current bounding box),inner ysep=1mm,inner xsep=0]{};
	\end{tikzpicture}
	} & \centering\raisebox{-.5\height}{
	\begin{tikzpicture}[nodes = {draw= none, minimum size = 0pt},
		factor/.style={rectangle, minimum size=.5cm, draw},
		node distance=1.7cm]
	\node[factor] (P) {}; \node [below=.1cm of P] {$\diag(\prob_\rv{X})$};
	\path (P) edge[double] node[above=0pt] {$s$} node[below=0pt] {$\ts$} ([xshift=1.5cm]P);
	\node[fit=(current bounding box),inner ysep=1mm,inner xsep=0]{};
	\end{tikzpicture}
	} & $\diag(\prob_\rv{X})(s,\ts)= \delta_{s,\ts}\cdot \prob_\rv{X}(s)$ \\
	\hline
	\centering\raisebox{-.5\height}{
	\begin{tikzpicture}[nodes = {draw= none, minimum size = 0pt},
		factor/.style={rectangle, minimum size=.5cm, draw},
		node distance=1.7cm]
	\node[factor] (e) {$=$};
	\node (U)[above left=.5cm and .8cm of e, anchor=east] {$s'$};
	\node (UH)[below left=.5cm and .8cm of e, anchor=east] {$\ts'$};
	\path (e) edge node[above=0cm, pos=1] {$y$} ([xshift=1cm]e);
	\draw (e) |- (U);
	\draw (e) |- (UH);
	\node[fit=(current bounding box),inner ysep=1mm,inner xsep=0]{};
	\end{tikzpicture}
	} & \centering\raisebox{-.5\height}{
	\begin{tikzpicture}[nodes = {draw= none, minimum size = 0pt},
		factor/.style={rectangle, minimum size=.5cm, draw},
		node distance=1.7cm]
	\node[factor] (I) {}; \node [below = 0pt of I] {$I_y$};
	\path (I) edge[double] node[above=0pt, pos=1] {$s$} node[below=0pt, pos=1] {$\ts$} ([xshift=-1cm]I);
	\path (I) edge node[above=0pt, pos=1] {$y$} ([xshift=1cm]I);
	\node[fit=(current bounding box),inner ysep=1mm,inner xsep=0]{};
	\end{tikzpicture}
	} & $I_y(s,\ts) = \delta_{s,\ts} \cdot \delta_{s,y}$ \\
	\hline
\end{tabular}
\caption{Conversions between factor graphs and DeFGs}
\label{tb:CintoDeFG}
\end{table}
%*******************************************************************************
\subsection{``Closing-the-Box'' Operations on DeFGs}
We define the ``closing-the-box'' operations on DeFGs in the same manner as for factor graphs.
\begin{definition}[``Closing-the-Box'' Operations on DeFGs] \label{def:DeFG:CtB} \index{closing-the-box operations!of DeFGs}
Let $\set{G}=(\set{V}_1\sqcup\set{V}_2,\set{F},\set{E}_1\sqcup\set{E}_2)$ be a DeFG as defined in Definition~\ref{def:DeFG}.
Let $\set{G}'=(\set{V}'_1\sqcup\set{V}'_2,\set{F}')$ be a subgraph of $\set{G}$ such that if a variable vertex is in $\set{G}'$, then so do all of its neighbors (all of which are in $\set{F}$).
(We call such a subgraph a \emph{box} of $\set{G}$.)
The ``closing-the-box'' operation (\wrt $\set{G}'$) is to replace $\set{G}'$ in $\set{G}$ by a factor vertex associated with the \emph{exterior function} $f_{\set{G}'}$ of $\set{G}'$, where $f_{\set{G}'}$ is the resultant function by summing the product of the local functions associated with $\set{F}'$ over the variables associated with $\set{V}'_1\sqcup\set{V}'_2$.
Namely, $f_{\set{G}'}$ is defined as
\[
f_{\set{G}'}(\vs_{\cup_{a\in\set{F}'}\nb{a}\xk\set{V}'_2},\tvs_{\cup_{a\in\set{F}'}\nb{a}\xk\set{V}'_2};\vx_{\cup_{a\in\set{F}'}\delta{a}\xk\set{V}'_1}) \defeq \sum_{\vs_{\set{V}'_2},\tvs_{\set{V}'_2},\vx_{\set{V}'_1}} \prod_{a\in\set{F}'} f_a(\vs_\nb{a},\tvs_\nb{a};\vx_{\delta{a}}).\qedhere
\]
\end{definition}
An advantage of DeFGs over factor graphs for quantum probabilities is that a DeFG remains to be DeFG after any ``closing-the-box'' operations.
\begin{theorem}\label{thm:DeFG:CtB}
	Given a DeFG, it remains a DeFG after a ``closing-the-box'' operation.
\end{theorem}
\begin{proof}
Consider the setup in Definition~\ref{def:DeFG:CtB}.
It suffices to show that the matrices associated with the exterior function $[f_{\set{G}'}(\vx)]_{\vs,\tvs}\defeq f_{\set{G}'}(\vs,\tvs;\vx)$ are PSD.
Using mathematical induction, it suffices to show that given any pair of functions $f_a(\vs_\nb{a},\tvs_\nb{a};\vx_{\delta{a}})$ and $f_b(\vs_\nb{b},\tvs_\nb{b};\vx_{\delta{b}})$ such that the matrices $[f_a(\vx_{\delta{a}})]_{\vs_\nb{a},\tvs_\nb{a}}\defeq f_a(\vs_\nb{a},\tvs_\nb{a};\vx_{\delta{a}})$ and $[f_b(\vx_{\delta{b}})]_{\vs_\nb{b},\tvs_\nb{b}}\defeq f_b(\vs_\nb{b},\tvs_\nb{b};\vx_{\delta{b}})$ are PSD, then, for each $i\in\delta{a}\cap\delta{b}$, $j\in\nb{a}\cap\nb{b}$, and $\vx_{\delta{a}\cup\delta{b}\xk i}\in\Tensor_{\ell\in\delta{a}\cup\delta{b}\xk i}\set{X}_\ell$, the matrix $[f(\vx_{\delta{a}\cup\delta{b}\xk i})]$ with its $(\vs_{\nb{a}\cup\nb{b}\xk j},\tvs_{\nb{a}\cup\nb{b}\xk j})$-th entry being
	$
	\sum_{s_j,\ts_j,x_i} f_a(\vs_\nb{a},\tvs_\nb{a};\vx_{\delta{a}})\cdot f_b(\vs_\nb{b},\tvs_\nb{b};\vx_{\delta{b}})
	$
	must also be PSD.
However, notice that given any complex-valued function $h(\vs_{\nb{a}\cup\nb{b}\xk j})$, it holds that
\begin{align*}
&\sum_{\vs_{\nb{a}\cup\nb{b}\xk j},\atop\tvs_{\nb{a}\cup\nb{b}\xk j}}
h(\vs_{\nb{a}\cup\nb{b}\xk j})\cdot
\left[\sum_{s_j,\ts_j,x_i} f_a(\vs_\nb{a},\tvs_\nb{a};\vx_{\delta{a}})\cdot f_b(\vs_\nb{b},\tvs_\nb{b};\vx_{\delta{b}})\right]
\cdot \conj{h(\tvs_{\nb{a}\cup\nb{b}\xk j})} = \\
&\sum_{x_i}\sum_{\vs_{\nb{a}\cup\nb{b}},\atop\tvs_{\nb{a}\cup\nb{b}}}
h(\vs_{\nb{a}\cup\nb{b}\xk j})\!\cdot\!\mathbf{1}(s_j)\cdot
\bigg[f_a(\vs_\nb{a},\tvs_\nb{a};\vx_{\delta{a}})\cdot f_b(\vs_\nb{b},\tvs_\nb{b};\vx_{\delta{b}})\bigg]
\cdot \conj{h(\tvs_{\nb{a}\cup\nb{b}\xk j})\!\cdot\!\mathbf{1}(\ts_j)}
\geq 0,
\end{align*}
since the Hadamard product of two PSD matrices is PSD.\footnote{Here, $\mathbf{1}$ stands for the constant-$1$ function.}
\end{proof}
Similar to factor graphs, we define the partition sum of a DeFG\index{partition sum!of DeFGs} as the result of summing the global function over all of its variables.
Namely, we have the definition
\begin{equation}
	Z(\set{G}) \defeq \sum_{\vs,\tvs,\vx} g(\vs,\tvs;\vx).
\end{equation}
Since the partition sum can be understood as the result of ``closing-the-box'' \wrt the graph itself, we have the following corollary as a direct result of Theorem~\ref{thm:DeFG:CtB}.
\begin{corollary}
The partition sum of a DeFG is real non-negative.
\end{corollary}
%*******************************************************************************
\section{DeFGs and Quantum Systems: Examples}
In this section, we would like to present several examples of DeFGs, especially the DeFGs representing elementary quantum systems.
Some of the examples presented here are related to the examples in~\cite{loeliger2017factor, loeliger2012factor}.
In particular, we would like to illustrate the connection between the marginals of DeFGs and elements from quantum information theory.
%*******************************************************************************
\begin{example}
\begin{figure}\centering
\begin{tikzpicture}[nodes = {draw= none, minimum size = 0pt},
	factor/.style={rectangle, minimum size=.7cm, draw},
	node distance=2cm]
	\node[factor] (P) {}; \node [below=0pt of P] {$\diag(\prob_\rv{X})$};
	\node[factor] (u1) [right of= P] {}; \node [below=0pt of u1] {$\hat{U}_1$};
	\node[factor] (u2) [right of= u1] {}; \node [below=0pt of u2] {$\hat{U}_2$};
	\node (um) [right of = u2] {$\cdots$};
	\node[factor] (un) [right of= um] {}; \node [below=0pt of un] {$\hat{U}_n$};
	\node[factor] (b) [right of= un] {}; \node [below=0pt of b] {$\hat{B}$};
	\node[factor] (Q) [right of= b] {}; \node [below=0pt of Q] {$I_y$};
	
	\path (P) edge[double] node[above] {$s_0$} node[below] {$\ts_0$} (u1);
	\path (u1) edge[*-, double] (u2);
	\path (u2) edge[*-, double] (um);
	\path (um) edge[double] (un);
	\path (un) edge[*-*, double] (b);
	\path (b) edge[double] (Q);
	\path (Q) edge node[above, pos=1] {$y$} ([xshift=1.5cm]Q);
	
	\node(R1A) [above left = .1cm and .3cm of P]{};
	\node(R1C) [below right = .5cm and .25cm of u1]{};
	\draw[dashed, red] (R1A) rectangle (R1C);
	\node [anchor=south east, font=\scriptsize, red] at (R1C) {$\rho_1$};
	
	\node(R2A) [above left = .2cm and .4cm of P]{};
	\node(R2C) [below right = .6cm and .25cm of u2]{};
	\draw[dashed, red] (R2A) rectangle (R2C);
	\node [anchor=south east, font=\scriptsize, red] at (R2C) {$\rho_2$};
	
	\node(R3A) [above left = .3cm and .5cm of P]{};
	\node(R3C) [below right = .7cm and .25cm of un]{};
	\draw[dashed, red] (R3A) rectangle (R3C);
	\node [anchor=south east, font=\scriptsize, red] at (R3C) {$\rho_n$};
	
	\node(R4A) [above left = .3cm and .25cm of b]{};
	\node(R4C) [below right = .7cm and .05cm of Q]{};
	\draw[dotted, blue] (R4A) rectangle (R4C);
	
	\node(R5A) [above left = .4cm and .25cm of un]{};
	\node(R5C) [below right = .8cm and .15cm of Q]{};
	\draw[dotted, blue] (R5A) rectangle (R5C);
	
	\node(R6A) [above left = .5cm and .25cm of u2]{};
	\node(R6C) [below right = .9cm and .25cm of Q]{};
	\draw[dotted, blue] (R6A) rectangle (R6C);
	
	\node(R7A) [above left = .6cm and .25cm of u1]{};
	\node(R7C) [below right = 1cm and .35cm of Q]{};
	\draw[dotted, blue] (R7A) rectangle (R7C);
\end{tikzpicture}
\caption{$n$ steps of unitary evolution followed by a projective measurement \wrt the computation basis.}
\label{fig:DeFG_2}\end{figure}
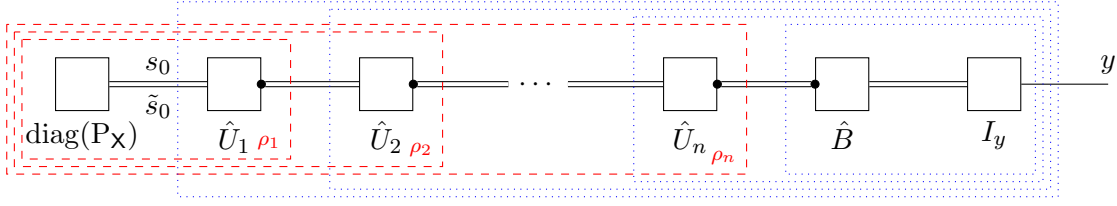
The DeFG in Figure~\ref{fig:DeFG_2} describes a quantum system with $n$ steps of unitary evolutions followed by a projective measurement \wrt the computation basis.
For every $k\in\{1,2,\ldots,n\}$, the exterior function of the k-th dashed box can be expressed as\footnote{Note that we may sometimes group the arguments of some functions (\eg, $\hat{U}_k$ in~\eqref{eq:footnote:group:argumnets}) using parentheses for better readability.}
\begin{align} \label{eq:footnote:group:argumnets}
	\rho_k(s_k,\ts_k) &= \sum_{s_{k-1},\ts_{s-1}} \hat{U}_k((s_k,s_{k-1}),(\ts_k,\ts_{k-1}))\cdot \rho_{k-1}(s_{k-1},\ts_{k-1}).
\end{align}
Here, note that $\hat{U}_k((s_k,s_{k-1}),(\ts_k,\ts_{k-1}))=U_{s_k,s_{k-1}}U^\Herm_{\ts_{k-1},\ts_k}$ for some unitary matrix $U$.
The functions $\{\rho_k\}_{k=1}^n$ are the Schrödinger representation of the system.
Similarly, the exterior functions of the dotted boxes correspond to 
	the Heisenberg representation.
\end{example}
%*******************************************************************************
\begin{example}\label{ex:DeFG:unitary:measurement}
\begin{figure}\centering
\begin{tikzpicture}[nodes = {draw= none, minimum size = 0pt},
	factor/.style={rectangle, minimum size=.7cm, draw},
	node distance=2cm]
	\node[factor] (P) {}; \node [below=0pt of P] {$\diag(\prob_\rv{X})$};
	\node[factor] (u0) [right of= P] {}; \node [below=0pt of u0] {$\hat{U}_0$};
	\node[factor] (A1) [right of= u0] {}; \node [below=0pt of A1] {$\hat{A}_1$};
	\node[factor] (u1) [right of= A1] {}; \node [below=0pt of u1] {$\hat{U}_1$};
	\node[factor] (A2) [right of= u1] {}; \node [below=0pt of A2] {$\hat{A}_2$};
	\node[factor] (I) [right of = A2] {=}; \node [below=0cm of I] {$I$};
	\path (P) edge[double] node[above] {$s_0$} node[below] {$\ts_0$} (u0);
	\path (u0) edge[*-, double] node[above] {$s_1$} node[below] {$\ts_1$} (A1);
	\path (A1) edge[*-, double] node[above] {$s_1'$} node[below] {$\ts_1'$} (u1);
	\path (u1) edge[*-, double] node[above] {$s_2$} node[below] {$\ts_2$} (A2);
	\path (A2) edge[*-, double] node[above] {$s_2'$} node[below] {$\ts_2'$} (I);
	\path (A1) edge node[above, pos=1] {$y_1$} ([yshift=1cm]A1);
	\path (A2) edge node[above, pos=1] {$y_2$} ([yshift=1cm]A2);
\end{tikzpicture}
\caption{A factor graph describing a two-measurement quantum system.}
\label{fig:DeFG_3}
\end{figure}
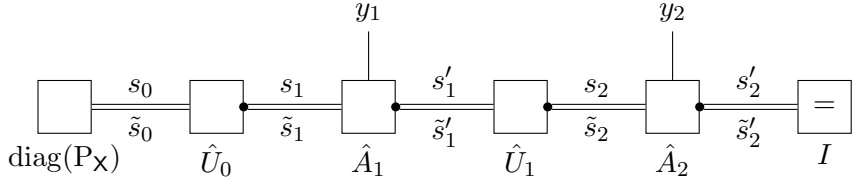
Figure~\ref{fig:DeFG_3} describes a DeFG for a quantum system with two measurements, where $\hat{A}_1$, $\hat{A}_2$ are functions such that $\sum_{y_k}\sum_{s_k'=\ts_k'} A_k(s_k',s_k,\ts_k',\ts_k;y_k)=\delta_{s_k,\ts_k}$ for each $k=1,2$.
\end{example}
%*******************************************************************************
\begin{example}
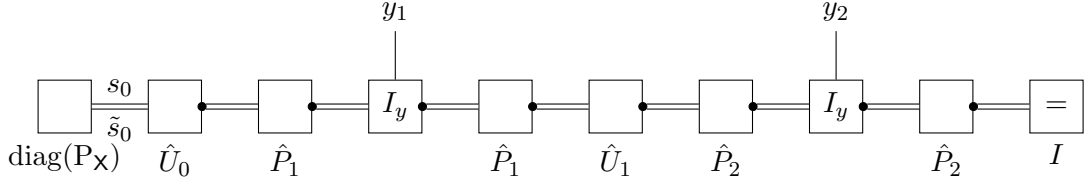
\begin{figure}\centering
\begin{tikzpicture}[nodes = {draw= none, minimum size = 0pt},
	factor/.style={rectangle, minimum size=.7cm, draw},
	node distance=1.45cm]
	\node[factor] (P) {}; \node [below=0pt of P] {$\diag(\prob_\rv{X})$};
	\node[factor] (u0) [right of= P] {}; \node [below=0pt of u0] {$\hat{U}_0$};
	\node[factor] (b1) [right of= u0] {}; \node [below=0pt of b1] {$\hat{P}_1$};
	\node[factor] (A1) [right of= b1] {};\node at (A1) {$I_y$};
	\node[factor] (b1h) [right of= A1] {}; \node [below=0pt of b1h] {$\hat{P}_1$};
	\node[factor] (u1) [right of= b1h] {}; \node [below=0pt of u1] {$\hat{U}_1$};
	\node[factor] (b2) [right of= u1] {}; \node [below=0pt of b2] {$\hat{P}_2$};
	\node[factor] (A2) [right of= b2] {};\node at (A2) {$I_y$};
	\node[factor] (b2h) [right of= A2] {}; \node [below=0pt of b2h] {$\hat{P}_2$};
	\node[factor] (I) [right of = b2h] {=};\node [below=0pt of I] {$I$};
	\path (P) edge[double] node[above] {$s_0$} node[below] {$\ts_0$} (u0);
	\path (u0) edge[*-, double] (b1);
	\path (b1) edge[*-, double] (A1);
	\path (A1) edge[*-, double] (b1h);
	\path (b1h) edge[*-, double] (u1);
	\path (u1) edge[*-, double] (b2);
	\path (b2) edge[*-, double] (A2);
	\path (A2) edge[*-, double] (b2h);
	\path (b2h) edge[*-, double] (I);
	\path (A1) edge node[above, pos=1] {$y_1$} ([yshift=1cm]A1);
	\path (A2) edge node[above, pos=1] {$y_2$} ([yshift=1cm]A2);
\end{tikzpicture}
\caption{A factor graph describing a two-measurement quantum system with projective measure onto 1-dimension eigenspace.}
\label{fig:DeFG_4}
\end{figure}
The DeFG in Figure~\ref{fig:DeFG_4} is a special case of Example~\ref{ex:DeFG:unitary:measurement}, in which both measurements are projective measurements with one-dimensional eigenspaces.
Here, $\hat{P}_k(y_k,s_k,\ty_k,\ts_k)=P^{(k)}_{y_k,s_k}P^{(k)\Herm}_{\ts_k.\ty_k}$ for some projective matrix $P^{(k)}$ for each $k=1,2$.
\end{example}
%*******************************************************************************
\begin{example}
\begin{figure}\centering
\begin{tikzpicture}[nodes = {draw= none, minimum size = 0pt},
	factor/.style={rectangle, minimum size=.7cm, draw},
	Factor/.style={rectangle, minimum width=.7cm, minimum height=1.5cm, draw},
	node distance=2cm]
	\node[factor] (P) {}; \node [below=0pt of P] {$\diag(\prob_\rv{X})$};
	\node[Factor] (u0) [right of= P] {}; \node [below=0pt of u0] {$\hat{U}_0$};
	\node (a1) [right of= u0] {};
	\node[factor] (A1) [above = -.04cm of a1] {}; \node at (A1) {$\hat{A}_1$};
	\node (A1h) [below = .1cm of a1] {};
	\node[Factor] (u1) [right of=  a1] {}; 
	\node [below=0pt of u1] {$\hat{U}_1$};
	\node (a2) [right of= u1] {};
	\node[factor] (A2) [above = -.04cm of a2] {}; \node at (A2) {$\hat{A}_2$};
	\node (A2h) [below = .1cm of a2] {};
	\node[Factor] (u2) [right of= a2] {}; \node [below=0pt of u2] {$\hat{U}_2$};
	\node[factor] (I) [right of = u2] {=}; \node [below=0cm of I] {$I$};

	\path (P) edge[double] node[above] {$s_0$} node[below] {$\ts_0$} (u0);
	\path (A1-|u0.east) edge[*-, double] node[above] {$s_1$} node[below] {$\ts_1$} (A1);
	\path (A1) edge[*-, double] node[above] {$s_1'$} node[below] {$\ts_1'$} (A1-|u1.west);
	\path (A2-|u1.east) edge[*-, double] node[above] {$s_2$} node[below] {$\ts_2$} (A2);
	\path (A2) edge[*-, double] node[above] {$s_2'$} node[below] {$\ts_2'$} (A2-|u2.west);
	\path (u2) edge[*-, double] node[above] {$s_3$} node[below] {$\ts_3$} (I);
	\path (A1h-|u0.east) edge[*-, double] node[above] {$w_1$} node[below] {$\tilde{w}_1$} (A1h-|u1.west);
	\path (A2h-|u1.east) edge[*-, double] node[above] {$w_2$} node[below] {$\tilde{w}_2$} (A2h-|u2.west);
	\path (A1) edge node[above, pos=1] {$y_1$} ([yshift=1cm]A1);
	\path (A2) edge node[above, pos=1] {$y_2$} ([yshift=1cm]A2);
\end{tikzpicture}
\caption{A quantum system with two partial measurements.}
\label{fig:DeFG_5}
\end{figure}
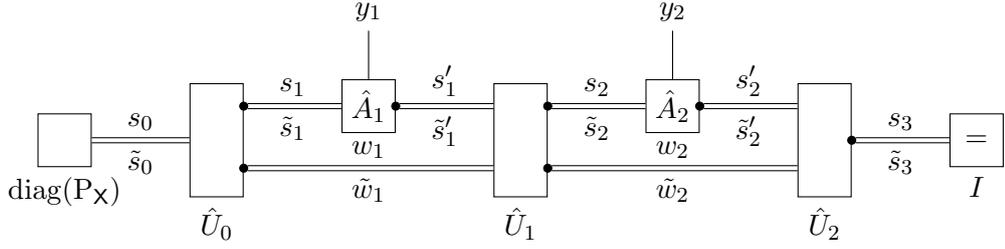
The DeFG in Figure~\ref{fig:DeFG_5} depicts a quantum system with two partial measurements.
Here, $\hat{A}_1$, $\hat{A}_2$ satisfy the same requirement as in Example~\ref{ex:DeFG:unitary:measurement}, and $\hat{U}_1$, $\hat{U}_2$, and $\hat{U}_3$ are functions such that
\begin{align*}
	\hat{U}_1(s_1,w_1,s_0,\ts_1,\tilde{w}_1,\ts_0)
	&= U^{(1)}_{(s_1,w_1),s_0} U^{(1)\Herm}_{\ts_0,(\ts_1,\tilde{w}_1)}\\
	\hat{U}_2(s_2,w_2,s_1',w_1,\ts_2,\tilde{w}_2,\ts_1',\tilde{w}_1)
	&= U^{(2)}_{(s_2,w_2),(s_1',w_1)}U^{(2)\Herm}_{(\ts_1',\tilde{w}_1),(\ts_2,\tilde{w}_2)}\\
	\hat{U}_3(s_3,s_2',w_2,\ts_3,\ts_2',\tilde{w}_2)
	&= U^{(3)}_{s_3,(s_2',w_2)}U^{(3)\Herm}_{(\ts_2',\tilde{w}_2),\ts_3}
\end{align*}
for some unitary matrices $U^{(1)}$, $U^{(2)}$, $U^{(3)}$.
Note that this DeFG contains cycles.
\end{example}
%*******************************************************************************
\begin{table}\centering
\begin{tabular}{|p{3.5cm}|p{3.5cm}|m{5.5cm}|}
	\hline
	FG Description & DeFG Description & Remarks\\
	\hline
	\centering\raisebox{-.5\height}{
	\begin{tikzpicture}[nodes = {draw= none, minimum size = 0pt, font=\small},
		factor/.style={rectangle, minimum size=.5cm, draw},
		node distance=1.7cm]
	\node[factor] (rho) {}; \node at (rho) {$\rho$};
	\node[factor, above right = .15cm and .7cm of rho] (u) {}; \node at (u) {$U$};
	\node[factor, below right = .15cm and .7cm of rho] (l) {}; \node at (l) {$U^\Herm$};
	\draw[*-] (rho) |- (u);
	\draw[-*] (rho) |- (l);
	\draw[*-] (u) -- ([xshift=1cm]u.east);
	\draw     (l) -- ([xshift=1cm]l.east);
	\node[fit=(current bounding box),inner ysep=1mm,inner xsep=0]{};
	\end{tikzpicture}
	} &  \centering\raisebox{-.5\height}{
	\begin{tikzpicture}[nodes = {draw= none, minimum size = 0pt, font=\small},
		factor/.style={rectangle, minimum size=.5cm, draw},
		node distance=1.2cm]
	\node[factor] (rho) {}; \node at (rho) {$\rho$};
	\node[factor, right of= rho] (f) {}; \node at (f) {$\hat{U}$};
	\draw[*-,double] (f) -- ([xshift=1cm]f.east);
	\draw[double] (f) -- (rho);
	\node[fit=(current bounding box),inner ysep=1mm,inner xsep=0]{};
	\end{tikzpicture}
	} & $\hat{U}(s',s,\ts',\ts)=U_{s',s}U^\Herm_{\ts,\ts'}$\\
	\hline
	\centering\raisebox{-.5\height}{
	\begin{tikzpicture}[nodes = {draw= none, minimum size = 0pt, font=\small},
		factor/.style={rectangle, minimum size=.5cm, draw},
		node distance=1.7cm]
	\node[factor] (rho) {}; \node at (rho) {$\rho$};
	\node[factor, above right = .15cm and .7cm of rho] (u) {}; \node at (u) {$E_k$};
	\node[factor, below right = .15cm and .7cm of rho] (l) {}; \node at (l) {$E_k^\Herm$};
	\draw[*-] (rho) |- (u);
	\draw[-*] (rho) |- (l);
	\draw[*-] (u) -- ([xshift=1cm]u.east);
	\draw     (l) -- ([xshift=1cm]l.east);
	\path (u) edge[draw] node[right] {$k$} (l);
	\node[fit=(current bounding box),inner ysep=1mm,inner xsep=0]{};
	\end{tikzpicture}
	} & \centering\raisebox{-.5\height}{
	\begin{tikzpicture}[nodes = {draw= none, minimum size = 0pt, font=\small},
		factor/.style={rectangle, minimum size=.5cm, draw},
		node distance=1.2cm]
	\node[factor] (rho) {}; \node at (rho) {$\rho$};
	\node[factor, right of= rho] (f) {}; \node at (f) {$N$};
	\draw[*-,double] (f) -- ([xshift=1cm]f.east);
	\draw[double] (f) -- (rho);
	\node[fit=(current bounding box),inner ysep=1mm,inner xsep=0]{};
	\end{tikzpicture}
	} & $N(s',s,\ts',\ts)=\sum_{k}E_{s',s}E^\Herm_{\ts,\ts'}$\\
	\hline
	\centering\raisebox{-.5\height}{
	\begin{tikzpicture}[nodes = {draw= none, minimum size = 0pt, font=\small},
		factor/.style={rectangle, minimum size=.5cm, draw}]
	\node[factor] (rho) {}; \node at (rho) {$\rho$};
	\node[factor, above right = .25cm and .7cm of rho] (u) {}; \node at (u) {$M_y$};
	\node[factor, below right = .25cm and .7cm of rho] (l) {}; \node at (l) {$M_y^\Herm$};
	\draw[*-] (rho) |- (u);
	\draw[-*] (rho) |- (l);
	\draw[*-] (u) -- ([xshift=1cm]u.east);
	\draw     (l) -- ([xshift=1cm]l.east);
	\path (u) edge[draw=none] node[factor] (Y) {} (l); \node at (Y) {$I_y$};
	\draw (u) -- (Y) -- (l);
	\node (y) at ([xshift=.7cm,yshift=1.3cm]Y) {$y$};
	\draw (Y) -| (y);
	\node[fit=(current bounding box),inner ysep=1mm,inner xsep=0]{};
	\end{tikzpicture}
	} & \centering\raisebox{-.5\height}{
	\begin{tikzpicture}[nodes = {draw= none, minimum size = 0pt, font=\small},
		factor/.style={rectangle, minimum size=.5cm, draw},
		node distance=1.2cm]
	\node[factor] (rho) {}; \node at (rho) {$\rho$};
	\node[factor, right of= rho] (f) {}; \node at (f) {$A$};
	\draw[*-,double] (f) -- ([xshift=1cm]f.east);
	\draw[double] (f) -- (rho);
	\draw (f) -- ([yshift=1cm]f.north);
	\node[fit=(current bounding box),inner ysep=1mm,inner xsep=0]{};
	\end{tikzpicture}
	} & $A(s',s,\ts',\ts;y) = M_y(s',s)M_y^\Herm(\ts,\ts')$ \\
	\hline
	\centering\raisebox{-.5\height}{
	\begin{tikzpicture}[nodes = {draw= none, minimum size = 0pt, font=\small},
		factor/.style={rectangle, minimum size=.5cm, draw}]
	\node[factor] (rho) {}; \node at (rho) {$\rho$};
	\node[factor, above right = .25cm and .7cm of rho] (u) {}; \node at (u) {$M_y$};
	\node[factor, below right = .25cm and .7cm of rho] (l) {}; \node at (l) {$M_y^\Herm$};
	\draw[*-] ([xshift=.1cm]rho.north) |- (u);
	\draw[*-] ([xshift=-.1cm]rho.north) |- ([xshift=1cm,yshift=.5cm]u.east);
	\draw[-*] ([xshift=.1cm]rho.south) |- (l);
	\draw ([xshift=-.1cm]rho.south) |- ([xshift=1cm,yshift=-.5cm]l.east);
	\draw[*-] (u) -- ([xshift=1cm]u.east);
	\draw     (l) -- ([xshift=1cm]l.east);
	\path (u) edge[draw=none] node[factor] (Y) {} (l); \node at (Y) {$I_y$};
	\draw (u) -- (Y) -- (l);
	\node (y) at ([xshift=.7cm,yshift=1.8cm]Y) {$y$};
	\draw (Y) -| (y);
	\node[fit=(current bounding box),inner ysep=1mm,inner xsep=0]{};
	\end{tikzpicture}
	} & \centering\raisebox{-.5\height}{
	\begin{tikzpicture}[nodes = {draw= none, minimum size = 0pt, font=\small},
		factor/.style={rectangle, minimum size=.5cm, draw},
		node distance=1.2cm]
	\node[factor, minimum height = 1cm] (rho) {}; \node at (rho) {$\rho$};
	\node[factor, anchor=north, xshift=1.2cm] at (rho.north) (f) {}; \node at (f) {$A$};
	\draw[*-,double] (f) -- ([xshift=1cm]f.east);
	\draw[double] (f) -- (f-|rho.east);
	\draw (f) -- ([yshift=1cm]f.north);
	\draw[double] ([yshift=-.25cm]rho.east) -- ([xshift=2.24cm,yshift=-.25cm]rho.east); 
	\node[fit=(current bounding box),inner ysep=1mm,inner xsep=0]{};
	\end{tikzpicture}
	} & Same as above \\
	\hline
	\centering\raisebox{-.5\height}{
	\begin{tikzpicture}[nodes = {draw= none, minimum size = 0pt, font=\small},
		factor/.style={rectangle, minimum size=.5cm, draw}]
	\node[factor] (rho) {}; \node at (rho) {$\rho$};
	\node[factor, right = 1.2cm of rho] (I) {}; \node at (I) {$I$};
	\draw[*-] (rho.north) -- ([yshift=.5cm]rho.north) -| (I.north);
	\draw (I.south) |- ([yshift=-.5cm]rho.south) -- (rho.south);
	\node[fit=(current bounding box),inner ysep=1mm,inner xsep=0]{};
	\end{tikzpicture}
	} &
	\centering\raisebox{-.5\height}{
	\begin{tikzpicture}[nodes = {draw= none, minimum size = 0pt, font=\small},
		factor/.style={rectangle, minimum size=.5cm, draw},
		node distance=1.2cm]
	\node[factor] (rho) {}; \node at (rho) {$\rho$};
	\node[factor, right of= rho] (f) {}; \node at (f) {$I$};
	\draw[double] (f) -- (rho);
	\node[fit=(current bounding box),inner ysep=1mm,inner xsep=0]{};
	\end{tikzpicture}
	} & - \\
	\hline
\end{tabular}
\caption{Representing quantum systems using factor graphs and DeFGs.}
\label{tb:quantum:system:FGs}
\end{table}
As illustrated by the above examples, the dynamics of a quantum system (see Section~\ref{sec:quantum_postulates}) can be described by DeFGs systematically, as summarized in Table~\ref{tb:quantum:system:FGs}.
%*******************************************************************************
\section{Belief-Propagation Algorithms for DeFGs}
In this section, we consider the problem of computing the partition sums of DeFGs.
For acyclic cases, the problem is not so different from that on a factor graph and can be solved by a slightly modified version of the method in Section~\ref{subsec:marginal:acyclic:FGs}.
For DeFGs with cycles, we propose and analyze a generalized version of the belief-propagation algorithms.
\par
Without loss of generality, we assume all the DeFGs involved in the rest of this chapter do not have ``single edges''.
Namely, we only consider global functions in the form
\begin{equation}\label{eq:global:function:DeFG:simplified}
	g(\vs,\tvs) = \prod_{a\in\mathcal{F}} f_a(\vs_{\nb{a}},\tvs_{\nb{a}}).
\end{equation}
%*******************************************************************************
\subsection{Computing the marginals/partition sum of an acyclic DeFG}
Similar to the methods for acyclic factor graphs in Section~\ref{subsec:marginal:acyclic:FGs}, one can also ``shrink'' an acyclic DeFG to a root factor via a sequence of ``closing-the-box'' operations starting from the leaf nodes.
The following example illustrates this idea.
\begin{example}
\begin{figure}\centering
\begin{tikzpicture}[nodes = {draw= none, minimum size = 0pt},
	factor/.style={rectangle, minimum size=.7cm, draw},
	Factor/.style={rectangle, minimum width=.7cm, minimum height=1.5cm, draw},
	node distance=2cm]
	\pgfdeclarelayer{front}
	\pgfdeclarelayer{back}
	\pgfsetlayers{back,main,front}
	\begin{pgfonlayer}{front}
		\node[factor] (A) {a};
		\node[factor] (B) [below left of= A] {b};
		\node[factor] (C) [below right of=A] {c};
		\node[factor] (D) [below left of=B] {d};
		\node[factor] (E) [below left of=C] {e};
		\node[factor] (F) [below right of=C] {f};
		\path (A) edge[double, thick] (B);
		\path (A) edge[double, thick] (C);
		\path (B) edge[double, thick] (D);
		\path (C) edge[double, thick] (E);
		\path (C) edge[double, thick] (F);
	\end{pgfonlayer}
	\begin{pgfonlayer}{main}
		\path (B)--(D) coordinate [midway] (bd);
		\draw[rotate=45, dashed, fill opacity=.4,draw opacity=1, fill=blue!40]
			(bd) ellipse (1.6cm and .8cm);
		\path (C)--(E) coordinate [midway] (ce);
		\draw[rotate=45, dashed, fill opacity=.4,draw opacity=1, fill=blue!40]
			(ce) ellipse (1.6cm and .8cm);
		\path (E)--(F) coordinate [midway] (ef);
		\draw[dashed, fill opacity=.4, draw opacity=1, fill=blue!40]
			([yshift=.4cm]ef) ellipse (2.2cm and 1.6cm);
		\path (A)--(D) coordinate [midway] (ad);
		\path (A)--(B) coordinate [midway] (ab);
		\draw[rotate=45, dashed, fill opacity=.4, draw opacity=1, fill=blue!40]
			([xshift=-.1cm]ad) ellipse (2.7cm and .9cm);
		\path (B)--(C) coordinate [midway] (bc);
		\path (D)--(E) coordinate [midway] (de);
		\draw[dashed, fill opacity=.4, draw opacity=1, fill=blue!40]
			([yshift=-.5cm]bc) ellipse (4.6cm and 2.5cm);
	\end{pgfonlayer}
	\begin{pgfonlayer}{front}
		\node[red] at ([xshift=-.2cm,yshift=.2cm]bd) {\textbf{i}};
		\node[red] at ([xshift=-.2cm,yshift=.2cm]ce) {\textbf{ii}};
		\node[red] at ([yshift=-.4cm]ef) {\textbf{iii}};
		\node[red] at ([xshift=-.2cm,yshift=.2cm]ab) {\textbf{iv}};
		\node[red] at ([yshift=-1cm]de) {\textbf{v}};
	\end{pgfonlayer}
\end{tikzpicture}
\caption{Computing the partition sum of a normal \emph{acyclic} DeFG.}
\label{fig:SPA_1}
\end{figure}
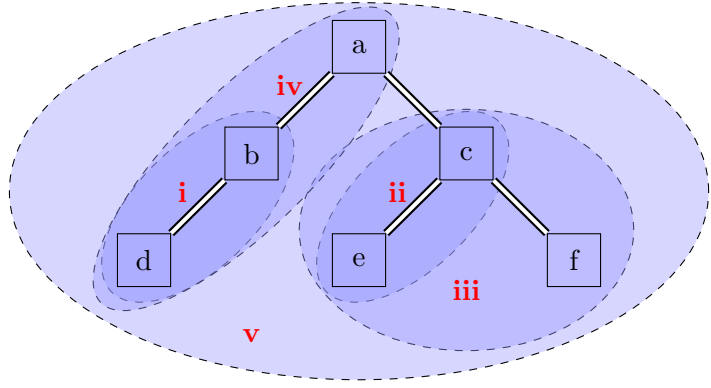
By taking a sequence of ``closing-the-box'' operations for each adjacent factor, the normal DeFG in Figure~\ref{fig:SPA_1} can be ``shrunk'' to a null graph with a single local constant function.
The roman numbers indicate the sequence of the ``closing-the-box'' operations.
The constant obtained in the end is the partition sum.
\end{example}
We summarize the method above as Algorithm~\ref{alg:BP:acyclic:DeFG}.
Without loss of generality, we have assumed all the leaf vertices are from $\set{F}$, since we can always append a constant-$1$ local function to a leaf vertex from $\set{V}$ without changing the partition sum.
\begin{algorithm}
\caption{Belief-Propagation Algorithm for Acyclic DeFGs}
\label{alg:BP:acyclic:DeFG}
\begin{algorithmic}[1]
\Require{An acyclic DeFG $\left(\set{G}=(\set{V},\set{F},\set{E}\in\set{V}\times\set{F}), \mathfrak{V}=\{\set{S}_j\}_{j\in\set{V}}, \mathfrak{F}=\{f_a\}_{a\in\set{F}}\right)$ with all leaf vertices belonging to $\set{F}$, a root vertex $r\in\mathcal{F}$}
\Ensure{The partition sum $Z(\mathcal{G})\defeq\sum_{\vs,\tvs}g(\vs,\tvs)$}
	\State Define the bipartite graph $G'=(\set{V}',\set{F}',\set{E}')\gets \set{G}$;
	\While{$\set{F}'\neq\{r\}$}
	\ForAll{$a\in\set{F}'$ being a leaf and $j$ being its parent in $\set{G}'$}
		\State $m_{a\to j}(s_j,\ts_j) \gets \sum_{\vs_{\nb{a}\xk j},\tvs_{\nb{a}\xk j}} f_a(\vs_\nb{a},\tvs_\nb{a}) \cdot \prod_{k\in\nb{a}\xk{j}} m_{k\to a}(s_k,\ts_k)$;
		\State $\set{F}'\gets \set{F}'\xk a$, $\set{E}'\gets \set{E}'\xk (\nb{a}\times a)$;
		\Comment{Remove $a$ from $\set{G}'$.}
	\EndFor
	\ForAll{$j\in\set{V}'$ being a leaf and $a$ being its parent in $\set{G}'$}
		\State $m_{j\to a}(s_j,\ts_j)\gets\prod_{c\in\nb{j}\xk{a}} m_{c\to j}(s_j,\ts_j)$;
		\State $\set{V}'\gets \set{V}'\xk j$, $\set{E}'\gets \set{E}'\xk (j\times\nb{j})$;
		\Comment{Remove $j$ from $\set{G}'$.}
	\EndFor
	\EndWhile
	\State $Z(\set{G})\gets \sum_{\vs_\nb{r},\tvs_\nb{r}}f_r(\vs_\nb{r},\tvs_\nb{r}) \cdot \prod_{j\in\nb{r}}m_{j\to r}(s_j,\ts_j)$;
\end{algorithmic}
\end{algorithm}
\par
In analogy to Theorem~\ref{thm:BP:tree}, we have the following theorem for acyclic DeFGs.
\begin{theorem}\label{thm:BP:DeFG:tree}
	Given an acyclic DeFG $\left(\set{G}=(\set{V},\set{F},\set{E}\in\set{V}\times\set{F}), \{\set{S}_j\}_{j\in\set{V}}, \{f_a\}_{a\in\set{F}}\right)$, there exists a unique set of messages $\{m_{j\to a},m_{a\to j}:\set{S}_j^2\to\Complex\}_{(j,a)\in\set{E}}$ such that
	\begin{align}
		\label{eq:msg:acyclic:DeFG:ja}
		m_{j\to a}(s_j,\ts_j)&=\prod_{c\in\nb{j}\xk{a}} m_{c\to j}(s_j,\ts_j) \\
		\label{eq:msg:acyclic:DeFG:aj}
		m_{a\to j}(s_j,\ts_j)&=\sum_{\vs_{\nb{a}\xk j},\tvs_{\nb{a}\xk j}} f_a(\vs_\nb{a},\tvs_\nb{a}) \cdot \prod_{k\in\nb{a}\xk{j}} m_{k\to a}(s_k,\ts_k).
	\end{align}
	Moreover, in this case, 
	\begin{align}
		\sum_{\vs_{\set{V}\xk{j}},\tvs_{\set{V}\xk{j}}}\prod_{a\in\set{F}} f_a(\vs_\nb{a},\tvs_\nb{a}) &= \prod_{c\in\nb{j}} m_{c\to j}(s_j,\ts_j) && \forall j\in\set{V},\\
		\sum_{\vs_{\set{V}\xk\nb{a}},\tvs_{\set{V}\xk\nb{a}}}\prod_{\hat{a}\in\set{F}} f_{\hat{a}}(\vs_\nb{\hat{a}},\tvs_\nb{\hat{a}}) &= f_a(\vs_\nb{a},\tvs_\nb{a})\cdot\prod_{j\in\nb{a}} m_{j\to a}(s_j,\ts_j) && \forall a\in\set{F}.
	\end{align}
\end{theorem}
\begin{proof}
	By treating $(s_j,\ts_j)$ as a single variable for each $j$, the theorem is essentially the same as Theorem~\ref{thm:BP:tree}.
	We omit the details.
\end{proof}
\begin{remark}
By Theorem~\ref{thm:DeFG:CtB}, the matrices corresponding to the messages defined in~\eqref{eq:msg:acyclic:DeFG:ja} and~\eqref{eq:msg:acyclic:DeFG:aj} are PSD.
\end{remark}
%*******************************************************************************
\subsection{Belief-Propagation Algorithms for DeFG and BP Fixed Points}
Similar to BP algorithms for factor graphs, we define BP algorithms\index{belief-propagation algorithm!for DeFGs} as a heuristic generalization based on the messaging-passing rules~\eqref{eq:msg:acyclic:DeFG:ja} and~\eqref{eq:msg:acyclic:DeFG:aj}.
Namely, we consider an iterative method with the updating rules
\begin{align}
\label{eq:msg:DeFG:ja}
	m_{j\to a}^{(t)}(s_j,\ts_j)&\propto \prod_{c\in\nb{j}\xk{a}} m_{c\to j}^{(t)}(s_j,\ts_j), \\
\label{eq:msg:DeFG:aj}
	m_{a\to j}^{(t)}(s_j,\ts_j)&\propto \sum_{\vs_{\nb{a}\xk j},\tvs_{\nb{a}\xk j}} f_a(\vs_\nb{a},\tvs_\nb{a}) \cdot \prod_{k\in\nb{a}\xk{j}} m_{k\to a}^{(t-1)}(s_k,\ts_k),
\end{align}
where the initial messages $\{m_{j\to a}^{(0)}\}_{(j,a)\in\set{E}}$ are some constant functions.
% Schedules
Same as the BP algorithms for factor graphs, the updating sequence of the messages in~\eqref{eq:msg:DeFG:ja} and~\eqref{eq:msg:DeFG:aj} (\aka schedule) can be cleverly designed to suit different scenarios.
In this thesis, we choose to focus on the synchronous schedule (\aka flooding schedule).
Algorithm~\ref{alg:BP:DeFG} lists BP algorithm for DeFGs with the flooding schedule.
\begin{algorithm}
\caption{Belief-Propagation Algorithm for DeFGs (Flooding Schedule with Timeout)}
\label{alg:BP:DeFG}
\begin{algorithmic}[1]
\Require{A DeFG $\left(\set{G}=(\set{V},\set{F},\set{E}\in\set{V}\times\set{F}), \mathfrak{V}=\{\set{S}_j\}_{j\in\set{V}}, \mathfrak{F}=\{f_a\}_{a\in\set{F}}\right)$, $\epsilon>0$}
\Ensure{Messages $\{m_{j\to a},m_{a\to j}:\set{S}_j^2\to\Complex\}_{(j,a)\in\set{E}}$, $\mathsf{FLAG}\in\{\mathrm{completed},\mathrm{timeout}\}$.}
	\ForAll{$(j,a)\in\set{E}$}
		\State $m_{j\to a}(s_j,\ts_j)\gets 1/\abs{\set{S}}^2$ for each $(s_j,\ts_j)\in\set{S}_j^2$;
		%\State $m_{a\to j}(s_j,\ts_j)\gets 1/\abs{\set{S}}^2$ for each $(s_j,\ts_j)\in\set{S}_j^2$;
	\EndFor
	\State $t\gets 0$;
	\Do
		\State $t\gets t+1$;
		\ForAll{$(j,a)\in\set{E}$}
			%\State $m_{j\to a}^{(t)}(s_j,\ts_j)\defpropto\prod_{c\in\nb{j}\xk{a}} m_{c\to j}^{(t-1)}(s_j,\ts_j)$ for each $(s_j,\ts_j)\in\set{S}_j^2$;
			\State $m_{a\to j}^{(t)}(s_j,\ts_j)\defpropto \sum_{\vs_{\nb{a}\xk j},\tvs_{\nb{a}\xk j}} f_a(\vs_\nb{a},\tvs_\nb{a}) \cdot \prod_{k\in\nb{a}\xk{j}} m_{k\to a}^{(t-1)}(s_k,\ts_k)$ for each $(s_j,\ts_j)\in\set{S}_j^2$;
		\EndFor
		\ForAll{$(j,a)\in\set{E}$}
		\State $m_{j\to a}^{(t)}(s_j,\ts_j)\defpropto\prod_{c\in\nb{j}\xk{a}} m_{c\to j}^{(t)}(s_j,\ts_j)$ for each $(s_j,\ts_j)\in\set{S}_j^2$;
		\EndFor
	\DoWhile{$\left(\neg\mathsf{timeout}\right) \land \left( \exists(j,a)\in\set{E} \text{ s.t. }\norm{m_{j\to a}^{(t)}-m_{j\to a}^{(t-1)}}_2>\varepsilon\text{ or }\norm{m_{a\to j}^{(t)}-m_{a\to j}^{(t-1)}}_2>\varepsilon\right)$}
	\Comment{$\mathsf{timeout}=\mathsf{false}$ unless the operating time exceeds a pre-selected waiting time.}
	\If{$\mathrm{timeout}$}
        \State{$\mathsf{FLAG}\gets\mathrm{timeout}$;}
	\Else
	    \State{$\mathsf{FLAG}\gets\mathrm{completed}$;}
	    \ForAll{$(j,a)\in\set{E}$}
	    \State{$m_{j\to a}\gets m_{j\to a}^{(t)}$;}
	    \State{$m_{a\to j}\gets m_{a\to j}^{(t)}$;}
	    \EndFor
	\EndIf
\end{algorithmic}
\end{algorithm}
\par
BP fixed points of a DeFG, defined similarly below as those of a factor graph, are of great importance for investigating the properties of the BA algorithms for DeFGs.
\begin{definition}[BP Fixed Points of a DeFG]\label{def:DeFG:BP:fixed:points} \index{BP fixed points!of DeFGs}
Applying Algorithm~\ref{alg:BP:DeFG} to a DeFG, a resulting set of messages $\{m_{j\to a},m_{a\to j}:\set{S}_j^2\to\Complex\}_{(j,a)\in\set{E}}$ is said to be a BP fixed point if 
\begin{align}
\label{eq:fixed:DeFG:ja}
m_{j\to a}(s_j,\ts_j)&\propto\prod_{c\in\nb{j}\xk{a}} m_{c\to j}(s_j,\ts_j), \\
\label{eq:fixed:DeFG:aj}
m_{a\to j}(s_j,\ts_j)&\propto\sum_{\vs_{\nb{a}\xk j},\tvs_{\nb{a}\xk j}} f_a(\vs_\nb{a},\tvs_\nb{a}) \cdot \prod_{k\in\nb{a}\xk{j}} m_{k\to a}(s_k,\ts_k).
\end{align}
In this case, the set $\{m_{j\to a},m_{a\to j}\}_{(j,a)\in\set{E}}$ is also called a set of fixed-point messages.
\end{definition}
\begin{definition} \index{induced partition sum}
Given a set of normalized PSD messages $\{m_{j\to a},m_{a\to j}\}_{(j,a)\in\set{E}}$, not necessarily a BP fixed point, the \emph{induced} partition sum (\wrt the messages) is defined as
\begin{equation}
Z_{\mathsf{induced}}\left(\{m_{j\to a},m_{a\to j}\}_{(j,a)\in\set{E}}\right)=\frac{\prod_{a\in\set{F}}Z_a(\{m_{j\to a}\}_{j\in\nb{a}})\cdot\prod_{j\in\set{V}}Z_j(\{m_{a\to j}\}_{a\in\nb{j}})}{\prod_{(j,a)\in\set{E}}Z_{j,a}(m_{j\to a},m_{a\to j})},
\end{equation}
where
\begin{align*}
Z_a(\{m_{j\to a}\}_{j\in\nb{a}}) &\defeq \sum_{\vs_\nb{a},\tvs_\nb{a}}f_a(\vs_\nb{a},\tvs_\nb{a}) \cdot \prod_{j\in\nb{a}} m_{j\to a}(s_j,\ts_j) && \forall a\in\set{F},\\
Z_j(\{m_{a\to j}\}_{a\in\nb{j}}) &\defeq \sum_{s_j,\ts_j}\prod_{a\in\nb{j}} m_{a\to j}(s_j,\ts_j) && \forall j\in\set{V},\\
Z_{j,a}(m_{j\to a},m_{a\to j}) &\defeq \sum_{s_j,\ts_j} m_{j\to a}(s_j,\ts_j)\cdot m_{a\to j}(s_j,\ts_j) && \forall (j,a)\in\set{E}. \qedhere
\end{align*}
\end{definition}
\begin{proposition}\label{prop:DeFG:Zinduced:stationary}
The function $Z_{\mathsf{induced}}$	is real non-negative, and its stationary points correspond to BP fixed points.
\end{proposition}
\begin{proof}
It suffice to show that the set of equations 
\begin{equation}\label{eq:proof:prop:DeFG:Zinduced:stationary:1}
\begin{aligned}
	\left.\frac{\D}{\D h}\right\rvert_{h=0} Z_{\mathsf{induced}}(m_{c\to k}+h\cdot \eta_{c\to k}) = 0 & \quad\forall \eta_{c\to k}(s_k,\ts_k)\text{ PSD}\\
	\left.\frac{\D}{\D h}\right\rvert_{h=0} Z_{\mathsf{induced}}(m_{k\to c}+h\cdot \eta_{k\to c}) = 0 & \quad\forall \eta_{k\to c}(s_k,\ts_k)\text{ PSD}
\end{aligned}
	\qquad \forall (k,c)\in\set{E}
\end{equation}
is equivalent to the updating rules~\eqref{eq:fixed:DeFG:ja} and~\eqref{eq:fixed:DeFG:aj}.
	
First, suppose $\{m_{j\to a},m_{a\to j}\}_{(j,a)\in\set{E}}$ is a stationary point of $Z_{\mathsf{induced}}$, \ie,~\eqref{eq:proof:prop:DeFG:Zinduced:stationary:1} holds.
By definition of $Z_{\mathsf{induced}}$, one can rewrite the upper set of~\eqref{eq:proof:prop:DeFG:Zinduced:stationary:1} as
\[
\left.\frac{\D}{\D h}\right\rvert_{h=0} Z_{k,c}(m_{c\to k}+h\cdot \eta_{c\to k}) / Z_{k,c} = 
\left.\frac{\D}{\D h}\right\rvert_{h=0} Z_k(m_{c\to k}+h\cdot \eta_{c\to k}) / Z_k.
\]
Substituting the definitions of $Z_k$ and $Z_{k,c}$, the above is equivalent to 
\begin{equation}\label{eq:proof:prop:DeFG:Zinduced:stationary:2}
\sum_{s_k,\ts_k} m_{k\to c}(s_k,\ts_k)\cdot\eta_{c\to k}(s_k,\ts_k) = \frac{Z_{k,c}}{Z_k}\cdot \sum_{s_k,\ts_k} \tilde{m}_{k\to c}(s_k,\ts_k)\cdot\eta_{c\to k}(s_k,\ts_k),
\end{equation}
where $\tilde{m}_{k\to c}(s_k,\ts_k)\defeq\prod_{a\in\nb{k}\xk{c}}m_{a\to k}(s_k,\ts_k)$.
Since this holds for all PSD functions $\eta_{c\to k}$, linear algebra must hold that $m_{k\to c}\propto\tilde{m}_{k\to c}$, which recovers~\eqref{eq:fixed:DeFG:ja}.
Similarly, the lower set of equations of~\eqref{eq:proof:prop:DeFG:Zinduced:stationary:1} can be rewritten as
\[
\left.\frac{\D}{\D h}\right\rvert_{h=0} Z_{k,c}(m_{k\to c}+h\cdot \eta_{k\to c}) / Z_{k,c} = 
\left.\frac{\D}{\D h}\right\rvert_{h=0} Z_c(m_{k\to c}+h\cdot \eta_{k\to c}) / Z_c,
\]
which is, in turn, equivalent to 
\begin{equation}\label{eq:proof:prop:DeFG:Zinduced:stationary:3}
\sum_{s_k,\ts_k}m_{c\to k}(s_k,\ts_k)\cdot\eta_{k\to c}(s_k,\ts_k) = 
\frac{Z_{k,c}}{Z_c}\cdot\sum_{s_k\ts_k}\tilde{m}_{c\to k}(s_k,\ts_k)\cdot\eta_{k\to c}(s_k,\ts_k),
\end{equation}
where $\tilde{m}_{c\to k}(s_k,\ts_k)\defeq\sum_{\vs_{\nb{c}\xk k},\tvs_{\nb{c}\xk k}}f_c(\vs_\nb{c},\tvs_\nb{c})\cdot\prod_{j\in\nb{c}\xk k}m_{j\to c}(s_j,\ts_j)$.
Again, since this also holds for all PSD functions $\eta_{k\to c}$, we have $m_{c\to k}\propto\tilde{m}_{c\to k}$, which recovers~\eqref{eq:fixed:DeFG:aj}.
	
Second, suppose $\{m_{j\to a},m_{a\to j}\}_{(j,a)\in\set{E}}$ is a BP fixed point, \ie,~\eqref{eq:fixed:DeFG:ja} and~\eqref{eq:fixed:DeFG:aj} hold.
To show~\eqref{eq:proof:prop:DeFG:Zinduced:stationary:1}, it suffices to verify~\eqref{eq:proof:prop:DeFG:Zinduced:stationary:2} and~\eqref{eq:proof:prop:DeFG:Zinduced:stationary:3}.
The latter can be shown rather straightforwardly as soon as one notices that
\begin{align*}
\frac{Z_{k,c}}{Z_k} &= \frac{\sum_{s_k,\ts_k}m_{k\to c}(s_k,\ts_k)\cdot m_{c\to k}(s_k,\ts_k)}{\sum_{s_k,\ts_k}\left(\prod_{a\in\nb{k}\xk{c}}m_{a\to k}(s_k,\ts_k)\right)\cdot m_{c\to k}(s_k,\ts_k)}
= \frac{m_{k\to c}(s_k,\ts_k)}{\prod_{a\in\nb{k}\xk{c}}m_{a\to k}(s_k,\ts_k)}\\
\frac{Z_{k,c}}{Z_c} &= \frac{\sum_{s_k,\ts_k}m_{c\to k}(s_k,\ts_k)\cdot m_{k\to c}(s_k,\ts_k)}{\sum_{s_k,\ts_k}\left(\sum_{\vs_{\nb{c}\xk k},\tvs_{\nb{c}\xk k}}f_c(\vs_\nb{c},\tvs_\nb{c})\cdot\prod_{j\in\nb{c}\xk k}m_{j\to c}(s_j,\ts_j)\right)\cdot m_{k\to c}(s_k,\ts_k)}\\
&= \frac{m_{c\to k}(s_k,\ts_k)}{\sum_{\vs_{\nb{c}\xk k},\tvs_{\nb{c}\xk k}}f_c(\vs_\nb{c},\tvs_\nb{c})\cdot\prod_{j\in\nb{c}\xk k}m_{j\to c}(s_j,\ts_j)}
\end{align*}
for all $s_k$, $\ts_k$ since $m_{k\to c}$ is proportional to $\prod_{a\in\nb{k}\xk{c}}m_{a\to k}$, and $m_{c\to k}$ is proportional to $\sum_{\vs_{\nb{c}\xk k},\tvs_{\nb{c}\xk k}}f_c\cdot\prod_{j\in\nb{c}\xk k}m_{j\to c}$.
\end{proof}
One \emph{must} note that $Z_{\mathsf{induced}}$ is conceptually different from $Z_{\mathsf{B}}$ for factor graphs.
In particular, $Z_{\mathsf{induced}}$ is a function of messages, instead of some local beliefs or marginals.
More importantly, in the classical case, $Z_{\mathsf{induced}}=Z_{\mathsf{B}}$ for messages corresponding to the belief minimizing $\bethe$.
However, although $\bethe$ can be generalized to DeFGs by analytical continuation, it is not clear how to define a similar version of $Z_{\mathsf{B}}$ based on the minimum of $\bethe$ while maintaining the connections to BP algorithms.
Thus, it is currently unclear how $Z_{\mathsf{induced}}$ is related to the partition sum $Z$ (for general DeFGs).
Nevertheless, for acyclic DeFGs, we still have the following corollary.
\begin{corollary}
For an acyclic DeFG $\set{G}$, its BP fixed point is unique.
Also, in this case, $Z_{\mathsf{induced}}$ has only one stationary point, and at that point, $Z_{\mathsf{induced}}=Z(\set{G})$.
\end{corollary}
\begin{proof}
This is a direct result of Theorem~\ref{thm:BP:DeFG:tree} and Proposition~\ref{prop:DeFG:Zinduced:stationary}.
\end{proof}
%*******************************************************************************
\subsection{Holographic Transformations and Loop Calculus for DeFGs}
In this section, we consider the generalization of the holographic transformation and the loop calculus expansion from factor graphs~\cite{mori2015loop, chernyak2007loop} to DeFGs.

The generalization of the holographic transformation for DeFG is rather straightforward.
The following theorem is a direct generalization of Theorem~\ref{thm:holant}.
\begin{theorem}[Holant Theorem for DeFGs] \label{thm:DeFG:holant} \index{Holant theorem}
Let $\set{G}=(\set{V},\set{F},\set{E})$ be a DeFG representing the factorization $g(\vs,\tvs) = \prod_{a\in\mathcal{F}} f_a(\vs_{\nb{a}},\tvs_{\nb{a}})$.
Given the Hermitian functions\footnote{Here, a function $f:\bigtimes_{k}\set{X}_k^2\to\Complex$ is said to be Hermitian if $f(x_1,\ldots,x_k;\tx_1,\ldots,\tx_k)=\conj{f(\tx_1,\ldots,\tx_k;x_1,\ldots,x_k)}$ for all $x_1,\ldots,x_k,\tx_1,\ldots,\tx_k$.} $\{\hat{\phi}_{j,a}:\set{T}_{j,a}^2\times\set{S}_j^2\to\Complex\}_{(j,a)\in\set{E}}$ and $\{\phi_{j,a}:\set{S}_j^2\times\set{T}_{j,a}^2\to\Complex\}_{(i,a)\in\set{E}}$ such that (where the sets $\set{T}_{j,a}$ and $\set{S}_j$ are finite)
\begin{equation}\label{eq:holographic:requirement:DeFG}
	\sum_{t,\tit} \phi_{j,a}(s,\ts;t,\tit)\cdot\hat{\phi}_{j,a}(t,\tit;s',\ts') = \delta_{s,s'}\cdot\delta_{\ts,\ts'},
\end{equation}
the partition sum $Z(\set{G})$ can be expressed as
\begin{equation}
Z(\set{G}) \defeq \sum_{\vs,\tvs} \prod_{a\in\set{F}} f_a(\vs_\nb{a},\tvs_\nb{a})
= \sum_{\vt,\tvt} \prod_{a\in\set{F}} \hat{f}_a(\vt_{\nb{a},a},\tvt_{\nb{a},a}) \prod_{j\in\set{V}} \hat{h}_j(\vt_{j,\nb{j}},\tvt_{i,\nb{j}}),
\end{equation}
where 
\begin{align}
	\label{eq:DeFG:holographic:1}
	\hat{f}_a(\vt_{\nb{a},a},\tvt_{\nb{a},a}) &\defeq \sum_{\vs_\nb{a},\tvs_\nb{a}} f_a(\vs_\nb{a},\tvs_\nb{a}) \cdot \prod_{j\in\nb{a}} \hat{\phi}_{j,a}(t_{j,a},\tit_{j,a};s_j,\ts_j),\\
	\label{eq:DeFG:holographic:2}
	\hat{h}_j(\vt_{j,\nb{j}},\tvt_{i,\nb{j}}) &\defeq \sum_{s_j,\ts_j} \prod_{a\in\nb{j}} \phi_{j,a}(s_j,\ts_j;t_{j,a},\tit_{j,a}).
\end{align}
\end{theorem}
The idea behind the above theorem is the same as that of Theorem~\ref{thm:holant}, and we omit the proof.
Similar to the case of factor graphs, the holographic transform\index{holographic transformation!of DeFGs} of $\set{G}$ (\wrt $\{\hat{\phi}_{j,a},\phi_{j,a}\}_{(j,a)}$) is defined to be the DeFG $\hat{\set{G}}=(\set{E},\set{F}\cup\set{V},\setdef{(e,e_1),(e,e_2)}{e\in\set{E}})$ representing the factorization 
\begin{equation}
\hat{g}(\vt,\tvt) = \prod_{a\in\set{F}} \hat{f}_a(\vt_{\nb{a},a},\tvt_{\nb{a},a}) \prod_{j\in\set{V}} \hat{h}_j(\vt_{j,\nb{j}},\tvt_{i,\nb{j}}).
\end{equation}
\par
Following the idea of the method of loop calculus (recall Section~\ref{subsec:FG:LoopCalculus}), we consider a specific holographic transform such that
\begin{align}
	\label{eq:DeFG:lc:assumption:1}
	\hat{f}_a(\vt_{\nb{a},a},\tvt_{\nb{a},a}) 
		&= 0,\text{ if } \wt(\vt_{\nb{a},a}\tensor\tvt_{\nb{a},a}) = 1,\\
	\label{eq:DeFG:lc:assumption:2}
	\hat{h}_j(\vt_{j,\nb{j}},\tvt_{i,\nb{j}})
		&= 0,\text{ if } \wt(\vt_{j,\nb{j}}\tensor\tvt_{i,\nb{j}}) = 1,
\end{align}
where we have assumed that each alphabet $\set{T}_{j,a}$ contains an elements labeled as $0$.\footnote{Note that $\wt(\vx_1^n\tensor\tvx_1^n)\defeq\sum_{i=1}^n 1_{(x_i,\tx_i)\neq(0,0)}$, where $1_\text{true}=1$ and $1_\text{false}=0$.}
In this case, the support of $\hat{g}$ is limited to those $(t_{j,a},\tit_{j,a})_{(j,a)\in\set{E}}$ such that the edges $(j,a)$ corresponding to $(t_{j,a},\tit_{j,a})\neq\vect{0}$ form a \emph{generalized loop} (see~\eqref{eq:def:extended:loops}).
Substituting~\eqref{eq:DeFG:holographic:1} and~\eqref{eq:DeFG:holographic:2} into~\eqref{eq:DeFG:lc:assumption:1} and~\eqref{eq:DeFG:lc:assumption:2}, respectively, we can rewrite the latter as
\begin{align}
\label{eq:DeFG:lc:assumption:3}
\sum_{s_j,\ts_j}\left(\sum_{\vs_{\nb{a}\xk j},\tvs_{\nb{a}\xk j}} f_a(\vs_\nb{a},\tvs_\nb{a}) \cdot \prod_{k\in\nb{a}\xk j} \hat{\phi}_{k,a}(0,0;s_k,\ts_k)\right) \cdot \hat{\phi}_{j,a}(t_{j,a},\tit_{j,a};s_j,\ts_j) &= 0 \\
\label{eq:DeFG:lc:assumption:4}
\sum_{s_j,\ts_j}\left(\prod_{c\in\nb{j}\xk a} \phi_{j,c}(s_j,\ts_j;0,0)\right)\cdot\phi_{j,a}(s_j,\ts_j;t_{j,a},\tit_{j,a}) &= 0
\end{align}
for each $(j,a)\in\set{E}$ and $(t_{j,a},\tit_{j,a})\neq\vect{0}$.
Since all functions involved in~\eqref{eq:DeFG:lc:assumption:3} and~\eqref{eq:DeFG:lc:assumption:4} are Hermitian, the \LHS of these equations are nothing more than the Frobenius inner product between matrices.
In other words,
\begin{align}
\label{eq:DeFG:lc:assumption:3a}
\left[\sum_{\vs_{\nb{a}\xk j},\tvs_{\nb{a}\xk j}}\!\! f_a(\vs_\nb{a},\tvs_\nb{a}) \cdot\!\! \prod_{k\in\nb{a}\xk j} \hat{\phi}_{k,a}(0,0;s_k,\ts_k)\right]_{s_j,\ts_j}\!\!
&\perp \bigg[\hat{\phi}_{j,a}(t_{j,a},\tit_{j,a};s_j,\ts_j)\bigg]_{s_j,\ts_j}\\
\label{eq:DeFG:lc:assumption:4a}
\left[\prod_{c\in\nb{j}\xk a} \phi_{j,c}(s_j,\ts_j;0,0)\right]_{s_j,\ts_j}\!\!
&\perp \bigg[\phi_{j,a}(s_j,\ts_j;t_{j,a},\tit_{j,a})\bigg]_{s_j,\ts_j}
\end{align}
for each $(j,a)\in\set{E}$ and $(t_{j,a},\tit_{j,a})\neq\vect{0}$.
Considering~\eqref{eq:holographic:requirement:DeFG}, we have 
\begin{align}
\label{eq:DeFG:lc:assumption:5}
\phi_{j,a}(\cdot,\cdot;0,0) &\propto
\left[\sum_{\vs_{\nb{a}\xk j},\tvs_{\nb{a}\xk j}}\!\! f_a(\vs_\nb{a},\tvs_\nb{a}) \cdot\!\! \prod_{k\in\nb{a}\xk j} \hat{\phi}_{k,a}(0,0;s_k,\ts_k)\right]_{s_j,\ts_j},\\
\label{eq:DeFG:lc:assumption:6}
\hat{\phi}_{j,a}(0,0;\cdot,\cdot) &\propto
\left[\prod_{c\in\nb{j}\xk a} \phi_{j,c}(s_j,\ts_j;0,0)\right]_{s_j,\ts_j}.
\end{align}
Compare above with Definition~\ref{def:DeFG:BP:fixed:points}.
Eq~\eqref{eq:DeFG:lc:assumption:5} and~\eqref{eq:DeFG:lc:assumption:6} are equivalent to the existence of some fixed-point messages $\{m_{j\to a}, m_{a\to j}\}_{(i,a)}$ such that
\begin{align}
	\label{eq:DeFG:lc:assumption:7}
	\phi_{j,a}(s_j,\ts_j\cdot;0,0) &= c_{j,a}\cdot m_{a\to j}(s_j,\ts_j)\\
	\label{eq:DeFG:lc:assumption:8}
	\hat{\phi}_{j,a}(0,0;s_j,\ts_j) &= \hat{c}_{j,a}\cdot m_{j\to a}(s_j,\ts_j)
\end{align}
for some constants $c_{j,a}$, $\hat{c}_{j,a}$ such that $c_{j,a}\cdot \hat{c}_{j,a}=Z_{j,a}^{-1}(m_{j\to a},m_{a\to j})$ for each $(j,a)\in\set{E}$.
Retracing the steps above, given any BP fixed point $\{m_{j\to a}, m_{a\to j}\}_{(j,a)\in\set{E}}$ of the DeFG, one can construct a holographic transform satisfying~\eqref{eq:DeFG:lc:assumption:7} and~\eqref{eq:DeFG:lc:assumption:8}; and thus  satisfying~\eqref{eq:DeFG:lc:assumption:1} and~\eqref{eq:DeFG:lc:assumption:2}.
Additionally, in such a case, for each $a\in\set{F}$ and $j\in\set{V}$, we have 
\begin{equation}
\hat{f}_a(\vect{0},\vect{0}) = Z_a \prod_{j\in\nb{a}} \hat{c}_{j,a},\quad
\hat{h}_j(\vect{0},\vect{0}) = Z_j \prod_{c\in\nb{j}} c_{j,c},
\end{equation}
and thus,
\begin{equation}
Z_{\mathsf{induced}}(\{m_{j\to a},m_{a\to j}\}_{(j,a)\in\set{E}})=
	\prod_{a\in\set{F}} \hat{f}_a(\vect{0},\vect{0}) \cdot \prod_{j\in\set{V}} \hat{h}_j(\vect{0},\vect{0}).
\end{equation}

\begin{theorem}[Loop Calculus Expansion for DeFGs] \index{loop calculus!of DeFGs}
Consider a DeFG representing the factorization $g(\vs,\tvs) = \prod_{a\in\mathcal{F}} f_a(\vs_{\nb{a}},\tvs_{\nb{a}})$.
Let $\{m_{j\to a}, m_{a\to j}\}_{(i,a)}$ be a set of fixed-point messages.
Then, it holds that 
\begin{equation}\label{eq:DeFG:lc}
Z(\set{G})=Z_{\mathsf{induced}}(\{m_{j\to a},m_{a\to j}\}_{(j,a)\in\set{E}}) \cdot \sum_{E\in\mathfrak{L}(\set{E})}\mathcal{K}(E),
\end{equation}
where the set of \emph{generalized loops} $\mathfrak{L}(\set{E})$ has been defined in~\eqref{eq:def:extended:loops} and where $\mathcal{K}\left(E\right)$ is some function of $E$ such that $\mathcal{K}(\emptyset)=1$.
\end{theorem}
\begin{proof}
Consider a holographic transform \wrt $\{\hat{\phi}_{j,a},\phi_{j,a}\}_{(j,a)}$ as established throughout~\eqref{eq:DeFG:lc:assumption:1} to~\eqref{eq:DeFG:lc:assumption:8}.
Under such a setup,
\begin{fleqn}\begin{equation*}
	Z(\set{G}) = \sum_{(\vt,\tvt)\in\support(\hat{g})} \prod_{a\in\set{F}} \hat{f}_a(\vt_{\nb{a},a},\tvt_{\nb{a},a}) \prod_{j\in\set{V}} \hat{h}_j(\vt_{j,\nb{j}},\tvt_{i,\nb{j}}),
\end{equation*}\end{fleqn}
\begin{fleqn}\begin{equation*}\phantom{Z(\set{G})}
	= Z_{\mathsf{induced}} \cdot \sum_{(\vt,\tvt)\in\support(\hat{g})} \prod_{a\in\set{F}} \frac{\hat{f}_a(\vt_{\nb{a},a},\tvt_{\nb{a},a})}{\hat{f}_a(\vect{0},\vect{0})} \prod_{j\in\set{V}} \frac{\hat{h}_j(\vt_{j,\nb{j}},\tvt_{i,\nb{j}})}{\hat{h}_j(\vect{0},\vect{0})},
\end{equation*}\end{fleqn}
\begin{fleqn}\begin{equation*}\phantom{Z(\set{G})}
	= Z_{\mathsf{induced}} \cdot \sum_{(\vt,\tvt)\in\support(\hat{g})}
	\begin{aligned}[t]
	&\prod_{a\in\set{F}} \left[\sum_{\vs_\nb{a},\tvs_\nb{a}}b_a(\vs_\nb{a},\tvs_\nb{a})\cdot\prod_{j\in\nb{a}}\frac{\hat{\phi}_{j,a}(t_{j,a},\tit_{j,a};s_j,\ts_j)}{\hat{\phi}_{j,a}(0,0;s_j,\ts_j)}\right]\cdot\\
	&\prod_{j\in\set{V}} \left[\sum_{s_j,\ts_j}b_j(s_j,\ts_j)\cdot\prod_{a\in\nb{j}}\frac{\phi_{j,a}(s_j,\ts_j;t_{j,a},\tit_{j,a})}{\phi_{j,a}(s_j,\ts_j;0,0)}\right],
	\end{aligned}
\end{equation*}\end{fleqn}
\begin{fleqn}\begin{equation*}\phantom{Z(\set{G})}
	= Z_{\mathsf{induced}} \cdot \sum_{E\in\mathfrak{L}(\set{E})} \sum_{(\vt_E,\tvt_E)\neq\vect{0}}
	\begin{aligned}[t]
	&\prod_{a\in\set{F}} \left[\sum_{\vs_\nb{a},\tvs_\nb{a}}b_a(\vs_\nb{a},\tvs_\nb{a})\cdot\prod_{j\in\nb{a}:\:(j,a)\in E}\frac{\hat{\phi}_{j,a}(t_{j,a},\tit_{j,a};s_j,\ts_j)}{\hat{\phi}_{j,a}(0,0;s_j,\ts_j)}\right]\cdot\\
	&\prod_{j\in\set{V}} \left[\sum_{s_j,\ts_j}b_j(s_j,\ts_j)\cdot\prod_{a\in\nb{j}:\:(j,a)\in E}\frac{\phi_{j,a}(s_j,\ts_j;t_{j,a},\tit_{j,a})}{\phi_{j,a}(s_j,\ts_j;0,0)}\right],
	\end{aligned}
\end{equation*}\end{fleqn}
where
\begin{align}
b_a(\vs_\nb{a},\tvs_\nb{a}) &\defeq \frac{1}{Z_a}\cdot f_a(\vs_\nb{a},\tvs_\nb{a})\cdot \prod_{j\in\nb{a}} m_{j \to a} (s_j,\ts_j),\\
b_j(s_j\ts_j) &\defeq \frac{1}{Z_j}\cdot \prod_{a\in\nb{j}} m_{a\to j}(s_j,\ts_j).
\end{align}
Finally, by letting 
\begin{equation}\label{eq:lc:DeFG:K}
\begin{aligned}
\mathcal{K}(E)\defeq \sum_{(\vt_E,\tvt_E)\neq\vect{0}}
&\prod_{a\in\set{F}} \left[\sum_{\vs_\nb{a},\tvs_\nb{a}}b_a(\vs_\nb{a},\tvs_\nb{a})\cdot\prod_{j\in\nb{a}:\:(j,a)\in E}\frac{\hat{\phi}_{j,a}(t_{j,a},\tit_{j,a};s_j,\ts_j)}{\hat{\phi}_{j,a}(0,0;s_j,\ts_j)}\right]\cdot\\
&\prod_{j\in\set{V}} \left[\sum_{s_j,\ts_j}b_j(s_j,\ts_j)\cdot\prod_{a\in\nb{j}:\:(j,a)\in E}\frac{\phi_{j,a}(s_j,\ts_j;t_{j,a},\tit_{j,a})}{\phi_{j,a}(s_j,\ts_j;0,0)}\right],
\end{aligned}
\end{equation}
Eq.~\eqref{eq:DeFG:lc} can be justified.
\end{proof}
%Note that the above theorem does not provide an expansion of $\mathcal{K}(E)$ as in Theorem~\ref{thm:loop:calculus}.
Unfortunately, even for the simplest case where $\set{T}_{j,a}$ is binary for each $(j,a)\in\set{E}$,~\eqref{eq:lc:DeFG:K} is pretty complicated.
However, the theorem does \emph{suggest} that $Z_{\mathsf{induced}}$ at a BP fixed point for DeFGs with a relatively smaller number of cycles tends to be closer to the partition sum.
%*******************************************************************************
\section{Numerical Examples}
In this section, we discuss various examples of DeFGs.
In particular, we compare the induced partition sum $Z_{\mathsf{induced}}$ with the exact partition sum $Z$.
(The DeFGs in this section have a modest size so that the exact partition sums are tractable.)
Moreover, for the first example, we also make some analytical statements.
\par
\begin{example}\label{numerical:ex:DeFG:1}
\begin{figure}
\begin{subfigure}[c]{0.4\columnwidth}\centering
\begin{tikzpicture}[nodes = {draw= none, minimum size = 0pt},
	factor/.style={rectangle, minimum size=1cm, draw},
	node distance=3cm]
	\node[factor] (f1) {}; \node[right=0pt of f1] {$f_1$};
	\node[factor, left of = f1] (f2) {}; \node[left=0pt of f2] {$f_2$};
	\node[factor, below of = f2] (f3) {}; \node[left=0pt of f3] {$f_3$};
	\node[factor, right of = f3] (f0) {}; \node[right=0pt of f0] {$f_0$};
	\path (f0) edge[double] node[left]{$s_1$} node[right]{$\ts_1$} (f1);
	\path (f1) edge[double] node[above]{$s_2$} node[below]{$\ts_2$} (f2);
	\path (f2) edge[double] node[left]{$s_3$} node[right]{$\ts_3$} (f3);
	\path (f3) edge[double] node[above]{$s_0$} node[below]{$\ts_0$} (f0);
\end{tikzpicture}
\caption{DeFG in Example~\ref{numerical:ex:DeFG:1}.}
\label{fig:numerical:ex:DeFG:1:a}
\end{subfigure}
\begin{subfigure}[c]{0.6\columnwidth}\centering
\begin{tikzpicture}[nodes = {draw=none, minimum size = 0pt, inner sep = 0pt, outer sep = 0pt}]
	\node (plot) {\includegraphics[width=6cm]{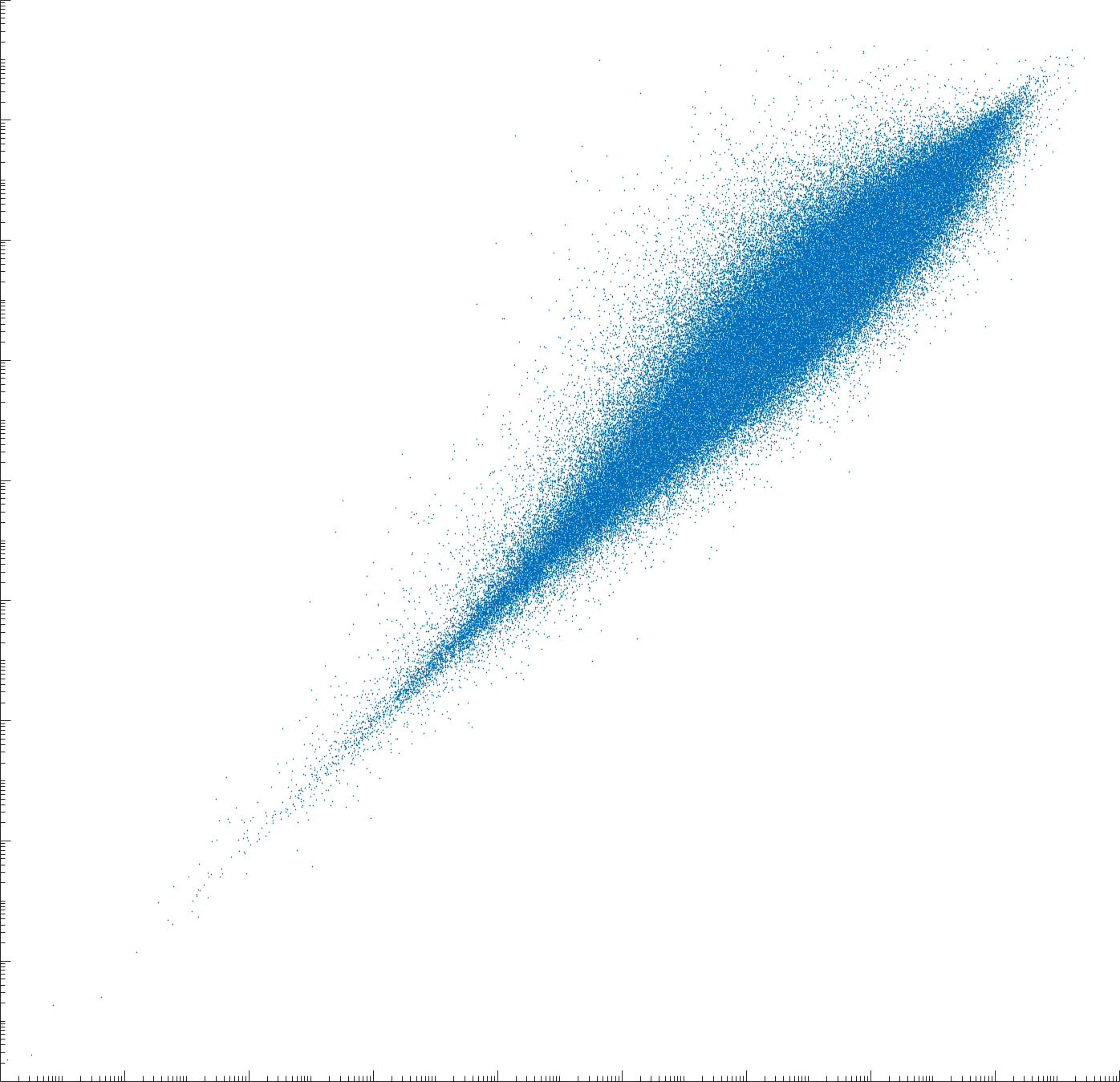}};
	\node (o) at (plot.south west) {};
	\node (xend) at (plot.south east) {};
	\node (yend) at (plot.north west) {};
	\path (o) edge[draw=none]
		node[anchor=north, yshift=-15pt] {$Z$}
		node[pos=0, anchor=north, yshift=-2pt, font=\scriptsize] {$10^{-12}$}
		node[pos=.3333, anchor=north, yshift=-2pt, font=\scriptsize] {$10^{-6}$}
		node[pos=.6667, anchor=north, yshift=-2pt, font=\scriptsize] {$10^{0}$}
		node[pos=1, anchor=north, yshift=-2pt, font=\scriptsize] {$10^{6}$}
	(xend);
	\path (o) edge[draw=none]
		node[anchor=south, yshift=25pt, sloped] {$Z_{\mathsf{induced}}$}
		node[pos=0, anchor=east, xshift=-2pt, font=\scriptsize] {$10^{-12}$}
		node[pos=.3333, anchor=east, xshift=-2pt, font=\scriptsize] {$10^{-6}$}
		node[pos=.6667, anchor=east, xshift=-2pt, font=\scriptsize] {$10^{0}$}
		node[pos=1, anchor=east, xshift=-2pt, font=\scriptsize] {$10^{6}$}
	(yend);
	\draw[red,dashed] (o) -- (plot.north east);
\end{tikzpicture}
\caption{Numerical demonstration in Example~\ref{numerical:ex:DeFG:1}.}
\label{fig:numerical:ex:DeFG:1:b}
\end{subfigure}
\caption{Plots for Example~\ref{numerical:ex:DeFG:1}.}
\label{fig:numerical:ex:DeFG:1}
\end{figure}
Consider a normal DeFG whose topology is an $n$-cycle ($n>1$) and where all variables take on values in the same finite alphabet $\set{S}$.
(Figure~\ref{fig:numerical:ex:DeFG:1:a} shows such a DeFG for $n=4$.)
Let $f$ be a complex-valued PD matrix of size $\size{\set{S}}^2\times\size{\set{S}}^2$ with its $((s_0,s_1),(\ts_0,\ts_1))$-th  entries being $f(s_0,s_1;\ts_0,\ts_1)$.
For $i\in\{0,1,\ldots,n-1\}$, we define the local function $f_i$ to be $f_i(s_i,s_{i+1};\ts_i,\ts_{i+1})\defeq f(s_i,s_{i+1};\ts_i,\ts_{i+1})$.
(All indices are modulo $n$.)
\par
To proceed, it is convenient to define the complex-valued matrix $F$ of size $\size{\set{S}}^2\times\size{\set{S}}^2$ with its $((s_0,\ts_0),(s_1,\ts_1))$-th entries being $f(s_0,s_1;\ts_0,\ts_1)$.
\par
For $i\in\{0,1,\ldots,n-1\}$, let $m^{(t)}_{i\to i+1}(s_{i+1},\ts_{i+1})$ and $m^{(t)}_{i\to i-1}(s_i,\ts_i)$ be the belief-propagation messages from $f_i$ to $f_{i+1}$ and $f_{i-1}$ at time index $t$, respectively.
Clearly,
\begin{align}
\label{eq:numerical:ex:DeFG:1:1}
m^{(t)}_{i\to i+1}(s_{i+1},\ts_{i+1}) &\propto \sum_{s_i,\ts_i} m^{(t-1)}_{i-1\to i}(s_i,\ts_i) \cdot F_{(s_i,\ts_i),(s_{i+1},\ts_{i+1})}, \\
\label{eq:numerical:ex:DeFG:1:2}
m^{(t)}_{i\to i-1}(s_i,\ts_i) &\propto \sum_{s_{i+1},\ts_{i+1}} F_{(s_i,\ts_i),(s_{i+1},\ts_{i+1})} \cdot m^{(t)}_{i+1\to i}(s_{i+1},\ts_{i+1}).
\end{align}
We assume $m^{(0)}_{i\to i+1}(s_{i+1},\ts_{i+1})\defeq\delta_{s_{i+1},\ts_{i+1}}$ and $m^{(0)}_{i\to i-1}(s_i,\ts_i)\defeq\delta_{s_i,\ts_i}$ for each $i\in\{0,1,\ldots,n-1\}$.
\par
Due to the properties of $F$ (which are induced from that of $f$), the mappings specified in~\eqref{eq:numerical:ex:DeFG:1:1} and~\eqref{eq:numerical:ex:DeFG:1:2} are completely positive.
Using generalizations of Perron--Frobenius theory (see~\cite{evans1977spectral, schrader2000perron}), one can make the following statements:
\begin{itemize}
	\item For each $i$, the messages $m^{(t)}_{i\to i+1}$ converge to a PD matrix as $t\to\infty$.
	\item For each $i$, the messages $m^{(t)}_{i\to i-1}$ converge to a PD matrix as $t\to\infty$.
	\item The eigenvalue of $F$ with the largest absolute value is real and is unique. Denote it by $\lambda_0$.
	\item When the messages converge, the induced partition sum $Z_{\mathsf{induced}} = \lambda_0^n$.
\end{itemize}
Compare this result with the exact partition sum, which is
\begin{equation}
Z = \sum_{k=1}^{\size{\set{S}}^2-1} \lambda_k^n = \lambda_0^n \cdot \left( 1 +\sum_{k=1}^{\size{\set{S}}^2-1} \left(\frac{\lambda_k}{\lambda_0}\right)^n\right),
\end{equation}
where $\lambda_0,\ldots,\lambda_{\size{\set{S}}^2-1}$ are the eigenvalues of $F$.
We see that the smaller the ratios $(\lambda_k/\lambda_0)^n$ for each $k$ are, the better the approximation is.
\par
For $n=4$ and $\size{\set{S}}=2$, Figure~\ref{fig:numerical:ex:DeFG:1:b} plots a million instances of the values of $Z$ and $Z_{\mathsf{induced}}$ \wrt the randomly generated matrices $F=U\cdot D\cdot U^\Herm$, where the unitary matrix $U$ is randomly generated according to the Haar measure and where $D$ is a diagonal matrix with each of its diagonal entries drawn from the square of the standard normal distribution in an {i.i.d.}~fashion.
We see that, very often, the ratio $Z_{\mathsf{induced}}/Z$ is rather close to $1$.
\end{example}
%*******************************************************************************
\begin{example}\label{numerical:ex:DeFG:2}
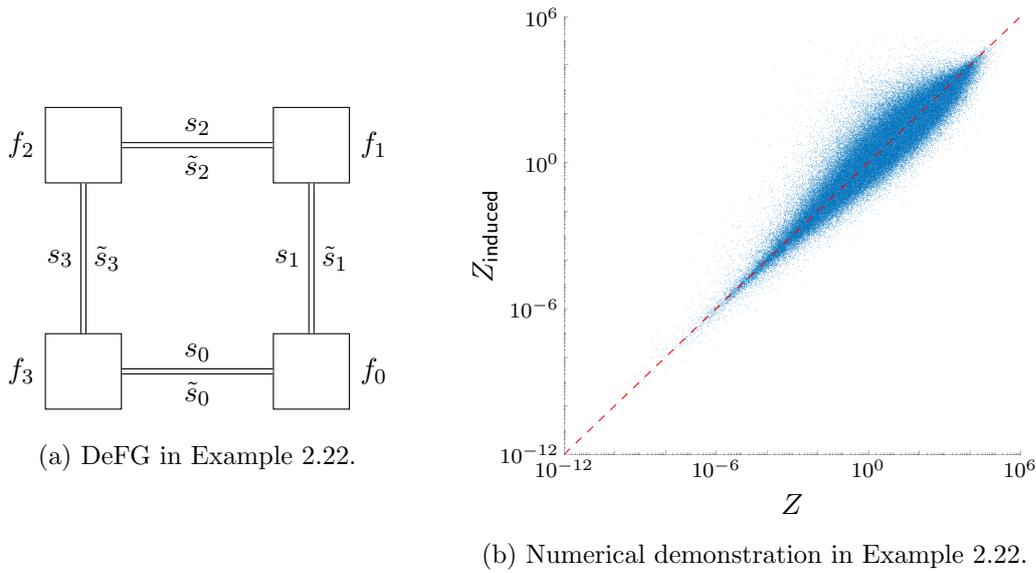
\begin{figure}
\begin{subfigure}[c]{0.4\columnwidth}\centering
\begin{tikzpicture}[nodes = {draw= none, minimum size = 0pt},
	factor/.style={rectangle, minimum size=1cm, draw},
	node distance=3cm]
	\node[factor] (f1) {}; \node[right=0pt of f1] {$f_1$};
	\node[factor, left of = f1] (f2) {}; \node[left=0pt of f2] {$f_2$};
	\node[factor, below of = f2] (f3) {}; \node[left=0pt of f3] {$f_3$};
	\node[factor, right of = f3] (f0) {}; \node[right=0pt of f0] {$f_0$};
	\path (f0) edge[double] node[left]{$s_1$} node[right]{$\ts_1$} (f1);
	\path (f1) edge[double] node[above]{$s_2$} node[below]{$\ts_2$} (f2);
	\path (f2) edge[double] node[left]{$s_3$} node[right]{$\ts_3$} (f3);
	\path (f3) edge[double] node[above]{$s_0$} node[below]{$\ts_0$} (f0);
	\path (f2) edge[double] node[above right=-5pt]{$s_4$} node[below left=-5pt]{$\ts_4$} (f0);
\end{tikzpicture}
\caption{DeFG in Example~\ref{numerical:ex:DeFG:2}.}
\label{fig:numerical:ex:DeFG:2:a}
\end{subfigure}
\begin{subfigure}[c]{0.6\columnwidth}\centering
\begin{tikzpicture}[nodes = {draw=none, minimum size = 0pt, inner sep = 0pt, outer sep = 0pt}]
	\node (plot) {\includegraphics[width=6cm]{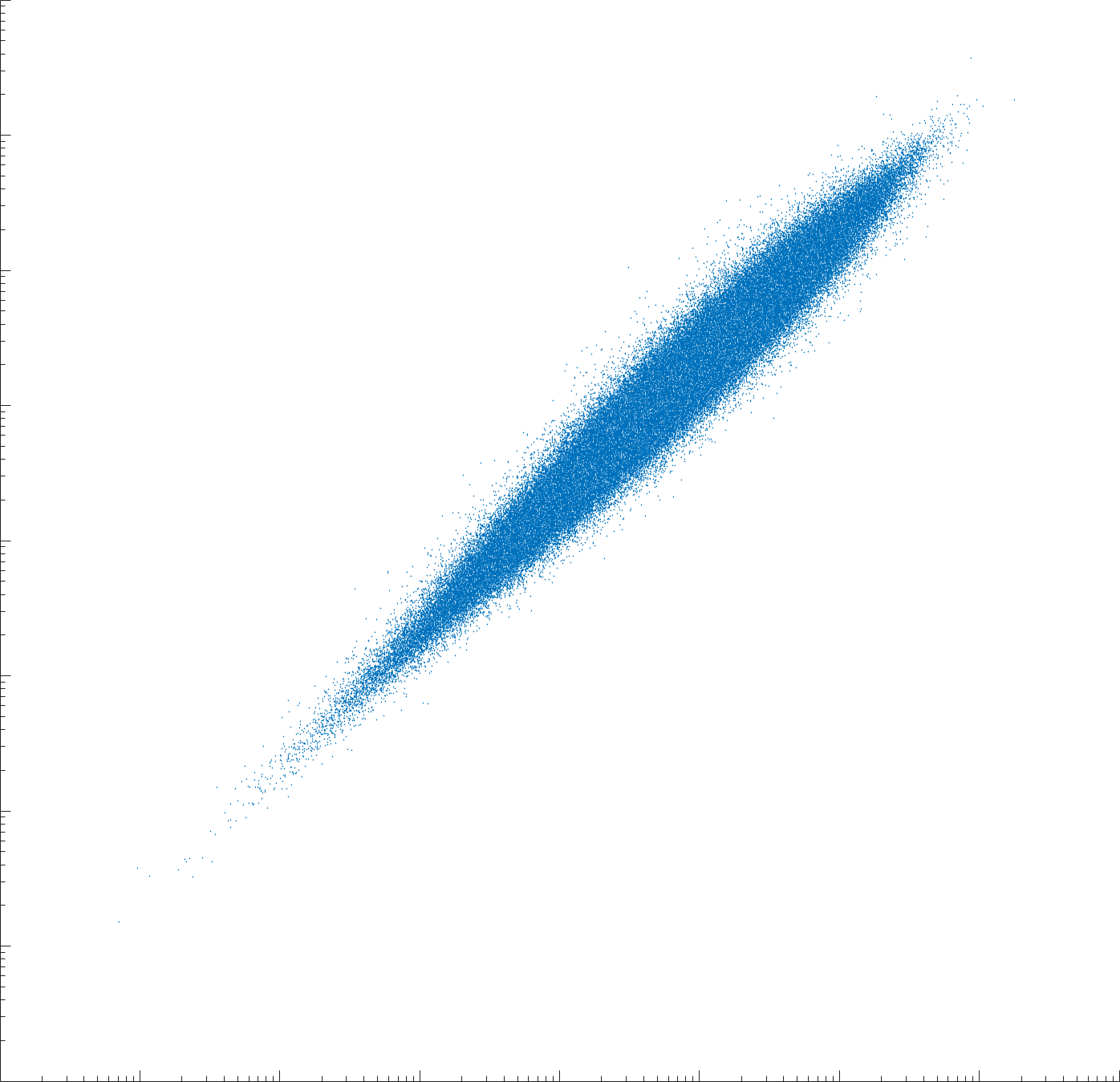}};
	\node (o) at (plot.south west) {};
	\node (xend) at (plot.south east) {};
	\node (yend) at (plot.north west) {};
	\path (o) edge[draw=none]
		node[anchor=north, yshift=-15pt] {$Z$}
		node[pos=0, anchor=north, yshift=-2pt, font=\scriptsize] {$10^{-4}$}
		node[pos=.5, anchor=north, yshift=-2pt, font=\scriptsize] {$10^{0}$}
		node[pos=1, anchor=north, yshift=-2pt, font=\scriptsize] {$10^{4}$}
	(xend);
	\path (o) edge[draw=none]
		node[anchor=south, yshift=25pt, sloped] {$Z_{\mathsf{induced}}$}
		node[pos=0, anchor=east, xshift=-2pt, font=\scriptsize] {$10^{-4}$}
		node[pos=.5, anchor=east, xshift=-2pt, font=\scriptsize] {$10^{0}$}
		node[pos=1, anchor=east, xshift=-2pt, font=\scriptsize] {$10^{4}$}
	(yend);
	\draw[red,dashed] (o) -- (plot.north east);
\end{tikzpicture}
\caption{Numerical demonstration in Example~\ref{numerical:ex:DeFG:2}.}
\label{fig:numerical:ex:DeFG:2:b}
\end{subfigure}
\caption{Plots for Example~\ref{numerical:ex:DeFG:2}.}
\label{fig:numerical:ex:DeFG:2}
\end{figure}
Consider the DeFG in Figure~\ref{fig:numerical:ex:DeFG:2:a}.
For $\size{\set{S}}=2$, Figure~\ref{fig:numerical:ex:DeFG:2:b} plots a million instances of the values of $Z$ and $Z_{\mathsf{induced}}$ \wrt randomly generated local functions.
In contrast to Example~\ref{numerical:ex:DeFG:1}, where all local functions were the same for every instantiation, in this example, all local functions are generated independently.
We observe that the ratio $Z_{\mathsf{induced}}/Z$ is reasonably close to~$1$ but typically larger than~$1$.
\end{example}
%*******************************************************************************
\begin{example}\label{numerical:ex:DeFG:3}
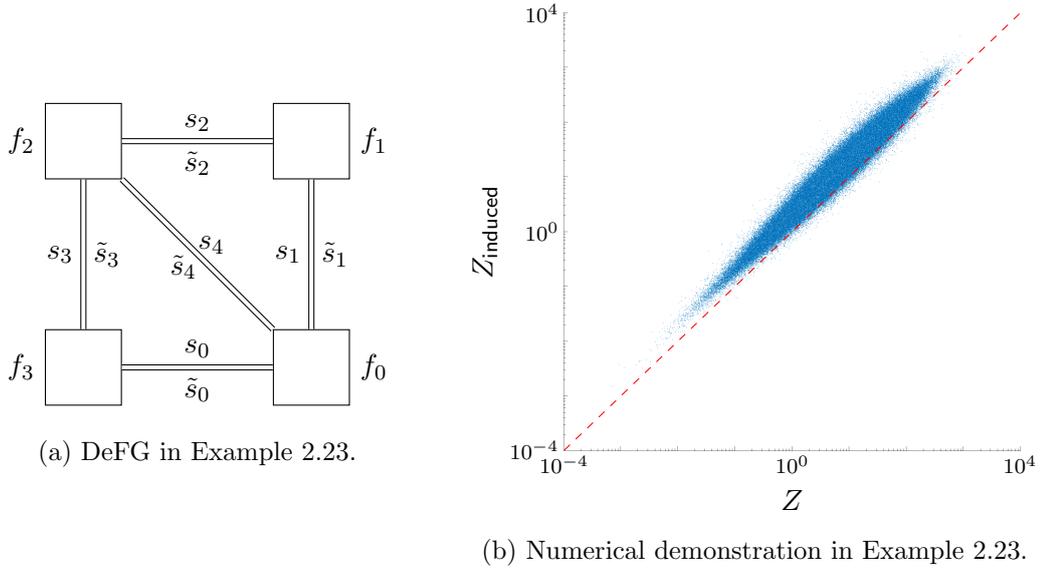
\begin{figure}
\begin{minipage}{0.4\columnwidth}
\begin{subfigure}[c]{\linewidth}\centering
\begin{tikzpicture}[nodes = {draw= none, minimum size = 0pt},
	factor/.style={rectangle, minimum size=.5cm, draw}]
	\node[factor] (f) {}; \node at (f) {$f_{i,j}$};
	\node[factor, below left = .2cm and 1.2cm of f] (fL) {};
	\node[left=0pt of fL] {$f_i^L$};
	\node[yshift=-.9cm] (L1) at (f.west) {};
	\node[yshift=-1.5cm] (L2) at (f.west) {};
	\node[factor, above right = .2cm and 1.5cm of f] (fR) {};
	\node[right=0pt of fR] {$f_j^R$};
	\node[yshift=.9cm] (R1) at (f.east) {};
	\node[yshift=1.5cm] (R2) at (f.east) {};
	\path (fL) edge node[above]{$x_{i,j}^{L}$} (f);
	\path (fL) edge[draw=none] coordinate[midway] (Lm1) (L1);
	\draw (fL) -- (Lm1); \draw[dashed] (Lm1) -- (L1);
	\path (fL) edge[draw=none] coordinate[midway] (Lm2) (L2);
	\draw (fL) -- (Lm2); \draw[dashed] (Lm2) -- (L2);
	\path (fR) edge node[below]{$x_{i,j}^{R}$} (f);
	\path (fR) edge[draw=none] coordinate[midway] (Rm1) (R1);
	\draw (fR) -- (Rm1); \draw[dashed] (Rm1) -- (R1);
	\path (fR) edge[draw=none] coordinate[midway] (Rm2) (R2);
	\draw (fR) -- (Rm2); \draw[dashed] (Rm2) -- (R2);
\end{tikzpicture}
\caption{FG in Example~\ref{numerical:ex:DeFG:3}.}
\label{fig:numerical:ex:DeFG:3:a}
\end{subfigure}
\begin{subfigure}[c]{\linewidth}\centering
\begin{tikzpicture}[nodes = {draw= none, minimum size = 0pt},
	factor/.style={rectangle, minimum size=.5cm, draw}]
	\node[factor] (f) {}; \node at (f) {$f_{i,j}$};
	\node[factor, below left = .2cm and 1.2cm of f] (fL) {};
	\node[left=0pt of fL] {$f_i^L$};
	\node[yshift=-.9cm] (L1) at (f.west) {};
	\node[yshift=-1.5cm] (L2) at (f.west) {};
	\node[factor, above right = .2cm and 1.5cm of f] (fR) {};
	\node[right=0pt of fR] {$f_j^R$};
	\node[yshift=.9cm] (R1) at (f.east) {};
	\node[yshift=1.5cm] (R2) at (f.east) {};
	\path (fL) edge[draw=none] coordinate[midway] (Lm1) (L1);
	\draw[double] (fL) -- (Lm1); \draw[dashed,double] (Lm1) -- (L1);
	\path (fL) edge[draw=none] coordinate[midway] (Lm2) (L2);
	\draw[double] (fL) -- (Lm2); \draw[dashed,double] (Lm2) -- (L2);
	\path (fR) edge[draw=none] coordinate[midway] (Rm1) (R1);
	\draw[double] (fR) -- (Rm1); \draw[dashed,double] (Rm1) -- (R1);
	\path (fR) edge[draw=none] coordinate[midway] (Rm2) (R2);
	\draw[double] (fR) -- (Rm2); \draw[dashed,double] (Rm2) -- (R2);
	\path (fL) edge[double] node[above,pos=.2]{$s_{i,j}^{L}$}
		node[below,pos=.8]{$\ts_{i,j}^{L}$} (f);
	\path (fR) edge[double] node[above,pos=.8]{$s_{i,j}^{R}$}
		node[below,pos=.2]{$\ts_{i,j}^{R}$} (f);
\end{tikzpicture}
\caption{DeFG in Example~\ref{numerical:ex:DeFG:3}.}
\label{fig:numerical:ex:DeFG:3:b}
\end{subfigure}
\end{minipage}
\begin{minipage}{0.6\columnwidth}
\begin{subfigure}[c]{\linewidth}\centering
\begin{tikzpicture}[nodes = {draw=none, minimum size = 0pt, inner sep = 0pt, outer sep = 0pt}]
	\node (plot) {\includegraphics[width=6cm]{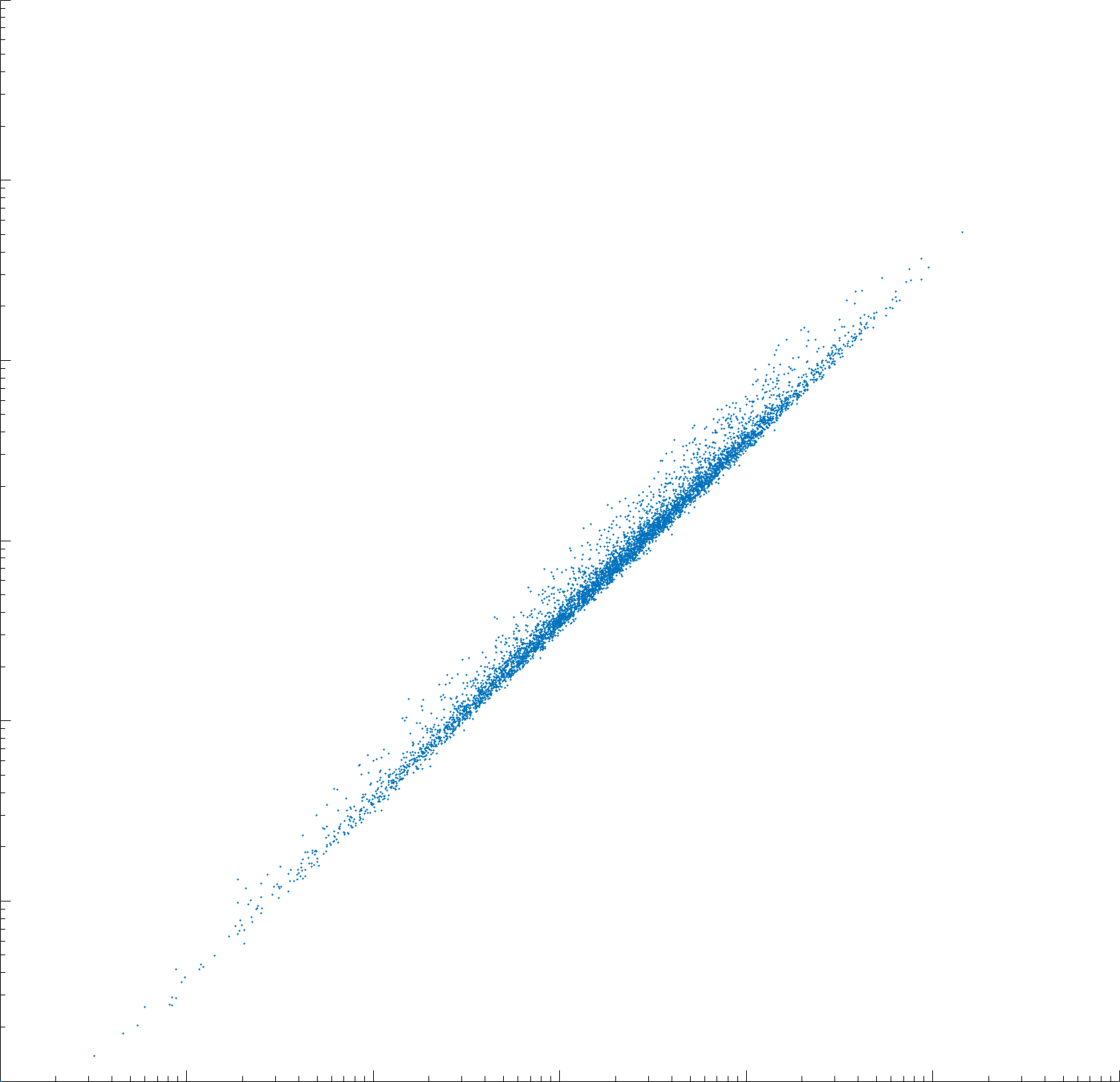}};
	\node (o) at (plot.south west) {};
	\node (xend) at (plot.south east) {};
	\node (yend) at (plot.north west) {};
	\path (o) edge[draw=none]
		node[anchor=north, yshift=-15pt] {$Z$}
		node[pos=0, anchor=north, yshift=-2pt, font=\scriptsize] {$10^{0}$}
		node[pos=.3333, anchor=north, yshift=-2pt, font=\scriptsize] {$10^{2}$}
		node[pos=.6667, anchor=north, yshift=-2pt, font=\scriptsize] {$10^{4}$}
		node[pos=1, anchor=north, yshift=-2pt, font=\scriptsize] {$10^{6}$}
	(xend);
	\path (o) edge[draw=none]
		node[anchor=south, yshift=25pt, sloped] {$Z_{\mathsf{induced}}$}
		node[pos=0, anchor=east, xshift=-2pt, font=\scriptsize] {$10^{0}$}
		node[pos=.3333, anchor=east, xshift=-2pt, font=\scriptsize] {$10^{2}$}
		node[pos=.6667, anchor=east, xshift=-2pt, font=\scriptsize] {$10^{4}$}
		node[pos=1, anchor=east, xshift=-2pt, font=\scriptsize] {$10^{6}$}
	(yend);
	\draw[red,dashed] (o) -- (plot.north east);
\end{tikzpicture}
\caption{Numerical demonstration in Example~\ref{numerical:ex:DeFG:3}.}
\label{fig:numerical:ex:DeFG:3:c}
\end{subfigure}
\end{minipage}
\caption{Plots for Example~\ref{numerical:ex:DeFG:3}.}
\label{fig:numerical:ex:DeFG:3}
\end{figure}
Let $\boldsymbol{\theta}$ be a complex-valued matrix of size $n\times n$ with its $(i,j)$-th entries being $\theta_{i,j}$.
The permanent~\cite{minc1984permanents} of $\boldsymbol{\theta}$ is defined to be $\perm(\boldsymbol{\theta})\defeq\sum_{\sigma}\prod_{i=1}^n \theta_{i,\sigma(i)}$, where the summation is over all $n!$ permutations of the set $\{1,\ldots,n\}$.
Ryser's algorithm, one of the most efficient algorithms for computing $\perm(\boldsymbol{\theta})$ for general matrices $\boldsymbol{\theta}$ exactly, requires $\Theta(n\cdot 2^n)$ arithmetic operations~\cite{ryser1963combinatorial}, and so the exact computation of the permanent is intractable, even for moderate values of $n$.
(Note that even the computation of the permanent of matrices containing only zeros and ones is \#P-complete~\cite{valiant1979complexity}.)
\par
One can formulate an NFG whose partition sum equals $\perm(\boldsymbol{\theta})$ (see, \eg,~\cite[Figure~1]{vontobel2013bethe}).
That factor graph is a complete bipartite graph with $n$ factor vertices on the left and $n$ factor vertices on the right.
Here, Figure~\ref{fig:numerical:ex:DeFG:3:a} shows a slightly modified version of that NFG.
All variables take values in the set $\set{X}=\{0,1\}$.
Moreover, for each $i$, the function $f_i^L$ is defined to be
\begin{equation}
f_i^L(\{x_{i,j}^L\}_{j=1}^{n}) \defeq \begin{cases}
 	1 & \text{exactly one of }\{x_{i,j}^L\}_{j=1}^{n}\text{ equals }1 \\
 	0 & \text{otherwise}\end{cases}
\end{equation}
for each $j$, the function $f_j^R$ is defined analogously; and for each $(i,j)$, the function $f_{i,j}$ is defined to be
\begin{equation}
f_{i,j}(x_{i,j}^L,x_{i,j}^R) \defeq \delta_{x_{i,j}^L,x_{i,j}^R} \cdot
	\begin{cases}
 	\theta_{i,j} & \text{ if }x_{i,j}^L=1\\
 	1 & \text{ if }x_{i,j}^L=0
	\end{cases}.
\end{equation}
\par
In this example, we consider the following rather natural generalization to the DeFG in Figure~\ref{fig:numerical:ex:DeFG:3:b}.
Assume that for each $(i,j)$, $\tilde{\boldsymbol{\theta}}_{i,j}$ is a PSD matrix of size $2\times 2$ with its $(k,\ell)$-th entries being $\tilde{\theta}_{i,j}(k,\ell)$.
With this, for each $i$, the function $\tilde{f}_i^L$ is defined to be
\begin{equation}
\tilde{f}_i^L(\{s_{i,j}^L,\ts_{i,j}^L\}_{j=1}^{n}) \defeq 
f_i^L(\{s_{i,j}^L\}_{j=1}^{n})\cdot f_i^L(\{\ts_{i,j}^L\}_{j=1}^{n})
\end{equation}
for each $j$, the function $\tilde{f}_j^R$ is defined analogously; and for each $(i,j)$, the function $\tilde{f}_{i,j}$ is defined to be
\begin{equation}
\tilde{f}_{i,j}(s_{i,j}^L,s_{i,j}^R;\ts_{i,j}^L,\ts_{i,j}^R)\defeq \delta_{s_{i,j}^L,s_{i,j}^R}\cdot\delta_{\ts_{i,j}^L,\ts_{i,j}^R} \cdot\tilde{\theta}_{i,j}(s_{i,j}^L,\ts_{i,j}^L).
\end{equation}
(One can easily verify that these local functions satisfy the requirement of a DeFG.)
Finally, let $\tilde{Z}$ be the partition sum of this DeFG.
\par
This DeFG definition has the following two important special cases:
\begin{itemize}
	\item If $\tilde{\boldsymbol{\theta}}_{i,j}=\left(\begin{smallmatrix} 1 & 0 \\ 0 & \theta_{i,j} \end{smallmatrix}\right)$ for each $(i,j)$, then $\tilde{Z} = \perm(\boldsymbol{\theta})$.
	\item If $\tilde{\boldsymbol{\theta}}_{i,j}=\left(\begin{smallmatrix} 1 \\ \theta_{i,j} \end{smallmatrix}\right) \cdot \left(\begin{smallmatrix} 1 & \overline{\theta_{i,j}} \end{smallmatrix}\right) = \left(\begin{smallmatrix} 1 & \overline{\theta_{i,j}} \\ \theta_{i,j} & \abs{\theta_{i,j}}^2 \end{smallmatrix}\right)$ for each $(i,j)$, then $\tilde{Z} = \perm(\boldsymbol{\theta}) \cdot \perm(\conj{\boldsymbol{\theta}}) = \abs{\perm(\boldsymbol{\theta})}^2$, where $\conj{\boldsymbol{\theta}}$ denotes the matrix whose entries are the complex-conjugate values of the entries of $\boldsymbol{\theta}$.
		(Note that such partition sums are of interest in quantum information processing~\cite{aaronson2013computational}, where $\boldsymbol{\theta}$ are certain types of square matrices over the complex numbers.
		We refer to\cite{aaronson2013computational} for details.)
\end{itemize}
\par
In our simulations, we considered the following setup.
Namely, for every $(i,j)\in\{1,\ldots,n\}^2$, we independently generate $\tilde{\boldsymbol{\theta}}_{i,j}$ as follows:
$\tilde{\theta}_{i,j}(0,0)\defeq 1$;
$\tilde{\theta}_{i,j}(1,0)$ is picked uniformly from the unit circle in the complex plane;
$\tilde{\theta}_{i,j}(0,1) \defeq \conj{\tilde{\theta}_{i,j}(1,0)}$;
$\tilde{\theta}_{i,j}(1,1)$ is picked uniformly (and independently of the other entries) from the real line interval $[1.10, 11.10]$.
Figure~\ref{fig:numerical:ex:DeFG:3:c} plots $5000$ instances of the values of 
  $Z_{\mathsf{induced}}$ and $Z$ for the case when $n$ is $5$.
We observe that the ratio $Z_{\mathsf{induced}}/ Z$ is concentrated around a value smaller than $1$.
\end{example}
%*******************************************************************************
% Chapter 3 Quantum Factor Graphs***********************************************
\chapter{Quantum Factor Graphs}\label{chapter:QFGs}
In this chapter, we consider a generalization of factor graphs known as quantum factor graphs\index{quantum factor graph} (QFGs)~\cite{Leifer08}. % proposed in~\cite{Leifer08} ?
This graphical model is a direct generalization of the bifactor networks proposed in the same paper, in which they considered generalized ``factorizations'' of positive operators like
\begin{equation}\label{eq:QFG:global}
\rho \defeq \bigstar_{a \in \set{F}}\rho_a = \exp\left(\sum_{a\in\set{F}}\log{\rho_a}\right),
\end{equation}
where, for each $a\in\set{F}$, $\rho_a$ is some positive operators and where $\star$ is some associative and commutative binary operator on PSD matrices (see~\eqref{eq:def:star} and~\eqref{eq:def:star:2}).
In this case, the generalized partition sum is defined to be
\begin{equation}\label{eq:QFG:partition}
Z(\set{G})\defeq\tr(\rho) = \tr\left(\exp\left(\sum_{a\in\set{F}}\log{\rho_a}\right)\right).
\end{equation}
Notice that the (classical) factor graphs are a special case of QFGs where all the involved local operators $\{\rho_a\}_{a\in\set{F}}$ are diagonal.
In~\cite{Leifer08}, Leifer and Poulin also proposed and studied a generalized belief-propagation algorithm for bifactor networks.

In our study, we are interested in the scenarios without commutativity constraints as those in~\cite{Leifer08}.
In such cases, many vital expressions that hold precisely for factor graphs only hold approximately for QFGs, provided the local operators are chosen randomly according to some distributions.
We give some analytical and numerical characterizations of these approximations.
In particular, we study how the ``closing-the-box'' operations and Bethe's approximation can be generalized to QFGs in an approximate manner.
%*******************************************************************************
\section{The $\star$-Product and Quantum Factor Graphs (QFGs)}
This section reviews the definition of the $\star$-product and quantum factor graphs.
We discuss various properties of the $\star$-product, particularly the (lack of) distributivity of $\star$ over (partial) trace operations.
\subsection{Definition of the $\star$-product and QFGs}
As already mentioned in~\eqref{eq:QFG:global}, for QFGs, we consider the factorization of some positive operator into the $\star$-product~\cite{warmuth2005Bayes},\footnote{Note that in~\cite{warmuth2005Bayes} this operation is denoted by $\odot$. However, in this thesis, $\odot$ denotes the Hadamard product.} of some local operators.
Given (strictly) PD operators $\rho,\sigma\in\StPositiveOp(\hilbert)$, their $\star$-product is defined to be 
\begin{equation}\label{eq:def:star}\index{$\star$-product}
 	\rho\star\sigma\defeq \exp(\log{\rho}+\log{\sigma}),
\end{equation}
where $\exp$ and $\log$ stand for matrix exponential and logarithm, respectively.
For PSD operators, we consider the Lie product formula\index{Lie product formula}~\cite{bhatia2013Matrix} and rewrite~\eqref{eq:def:star} as
\begin{equation}\label{eq:def:star:2}
  \rho\star\sigma = \lim_{n\to \infty} \left(\rho^{\frac{1}{n}}\sigma^{\frac{1}{n}}\right)^n.
\end{equation}
Using a continuity argument, one can generalize the $\star$-product to PSD operators $\rho,\sigma\in \PositiveOp(\hilbert)$.
(For the proof of the convergence of the \RHS of~\eqref{eq:def:star:2} for PSD operators $\rho$ and $\sigma$, see, \eg,~\cite{bhatia2013Matrix} and~\cite[Theorem 1.2]{simon2005functional}.)
The $\star$-product is associative and commutative, \ie, 
\begin{align}
  \rho_1 \star \rho_2 &= \rho_2 \star \rho_1 
  &&\forall \rho_1, \rho_2 \in \PositiveOp(\hilbert),\\
  (\rho_1 \star \rho_2) \star \rho_3 &= \rho_1 \star (\rho_2 \star \rho_3) 
  &&\forall \rho_1, \rho_2, \rho_3 \in \PositiveOp(\hilbert).
\end{align}
\par
As a convention, given $\rho$ and $\sigma$ acting on different Hilbert spaces, we treat the expression $\rho\star\sigma$ as the $\star$-product of the \emph{embedded} operators of $\rho$ and $\sigma$ on some smallest \emph{common} Hilbert space.
For example, given $\rho_\system{AB}\in\PositiveOp(\hilbert_\system{A}\tensor\hilbert_\system{B})$ and $\rho_\system{BC}\in\PositiveOp(\hilbert_\system{B}\tensor\hilbert_\system{C})$, the convention
\begin{equation}\label{eq:blowupConvention}
  \rho_\system{AB} \star \rho_\system{BC} \defeq (\rho_\system{AB} \tensor I_\system{C}) \star (I_\system{A} \tensor \rho_\system{BC})
\end{equation}
holds, where $I_\system{C}$ and $I_\system{A}$ are the identity operators on $\hilbert_\system{A}$ and $\hilbert_\system{C}$, respectively;
and $\rho_\system{AB} \star \rho_\system{BC}$ is an operator acting on $\hilbert_\system{A} \tensor \hilbert_\system{B} \tensor \hilbert_\system{C}$.
This convention also applies to the expression $\bigstar_{a \in \set{F}}\rho_a$, where some equations similar to~\eqref{eq:blowupConvention} are applied recursively.
As a reminder, note that $\log(\rho\tensor I)\equiv\log{\rho}\tensor I$.
\begin{definition}[Quantum Factor Graph~\cite{Leifer08}] \label{def:QFG} \index{quantum factor graph}
A quantum factor graph (QFG) is a bipartite graph $\set{G}=(\set{V},\set{F},\set{E}\in\set{V}\times\set{F})$ associated with a variable set $\mathfrak{V}$ and a factor set $\mathfrak{F}$, where
\begin{itemize}
	\item $\mathfrak{V}=\{\hilbert_i\}_{i\in\set{V}}$ is indexed by $\set{V}$, and each element of $\mathfrak{V}$ is a Hilbert space;
	\item $\mathfrak{F}=\{\rho_a\}_{a\in\set{F}}$ is indexed by $\set{F}$, and $\rho_a \in \PositiveOp(\Tensor_{i\in\nb{a}}\hilbert_i)$ for each $a\in\set{F}$.
\end{itemize}
The operator $\rho \defeq \bigstar_{a \in \set{F}}\rho_a$ is called the \emph{global operator} of $\set{G}$, and in this case, $\set{G}$ is also said to be representing the factorization $\rho = \bigstar_{a \in \set{F}}\rho_a$.
Similar to factor graphs, a QFG is said to be \emph{normal} if the degree of any vertex in $\set{V}$ is at most 2.
\end{definition}
Note that the global operator $\rho$ is always a PSD operator on $\Tensor_{i\in \set{V}} \hilbert_i$ (i.e., $\rho\in\PositiveOp{\Tensor_{i\in \set{V}} \hilbert_i}$).
Moreover, if all the local operators are \emph{strictly} positive, \ie, $\rho_a \in \StPositiveOp(\hilbert_a)$ for each $a$, so is $\rho$.
\begin{remark}
Similar to NFGs, in a normal QFG, we redraw the vertices in $\set{V}$ as edges (see \RHS of Figure~\ref{fig:QFG:1} and Figure~\ref{fig:QFG:2}).
\end{remark}
%*******************************************************************************
\begin{example}\label{ex:QFG:1}
\begin{figure}\centering
\begin{subfigure}[c]{0.49\textwidth}\centering
\begin{tikzpicture}[node distance=2.2cm,
	var/.style={circle, minimum size=.5cm, draw},
	factor/.style={rectangle, minimum size=.7cm, draw}]
	\node[var] (A) {}; \node at (A) {$\system{A}$};
	\node[factor, left=1cm of A] (rho) {$\rho_\system{A}$};
	\node[factor, above right = .8660cm and .5cm of A] (sigma) {$\sigma_\system{A}$};
	\node[factor, below right = .8660cm and .5cm of A] (tau) {$\tau_\system{A}$};
	\draw[line width=2pt] (A) -- (rho);
	\draw[line width=2pt] (A) -- (sigma);
	\draw[line width=2pt] (A) -- (tau);
\end{tikzpicture}
\end{subfigure}
\begin{subfigure}[c]{0.49\textwidth}\centering
\begin{tikzpicture}[node distance=2.2cm,
	var/.style={circle, minimum size=.5cm, draw},
	factor/.style={rectangle, minimum size=.7cm, draw}]
	\node[factor] (pi) {$\pi$};
	\node[factor, left=1cm of pi] (rho) {$\rho_\system{A}$};
	\node[factor, above right = .8660cm and .5cm of pi] (sigma) {$\sigma_\system{B}$};
	\node[factor, below right = .8660cm and .5cm of pi] (tau) {$\tau_\system{C}$};
	\path (pi) edge[line width=2pt] node[above]{$\system{A}$} (rho);
	\path (pi) edge[line width=2pt] node[below, anchor=north west]{$\system{B}$} (sigma);
	\path (pi) edge[line width=2pt] node[below, anchor=north east]{$\system{C}$} (tau);
\end{tikzpicture}
\end{subfigure}
\caption{QFG in Example~\ref{ex:QFG:1}.}
\label{fig:QFG:1}
\end{figure}
The \LHS of Figure~\ref{fig:QFG:1} depicts a QFG with a single variable vertex and three factor vertices.
Here, all local factors $\rho_\system{A}$, $\sigma_\system{A}$, and $\tau_\system{A}$ act on the same Hilbert space $\hilbert_\system{A}$, and the global operator is given by $\rho_\set{G}\defeq\rho_\system{A}\star\sigma_\system{A}\star\tau_\system{A}$.
On the \RHS of Figure~\ref{fig:QFG:1}, we have a normal QFG, which is equivalent to the QFG on the left.
Here, the new-introduced variable vertices $\system{B}$ and $\system{C}$ are associated with the Hilbert spaces $\hilbert_\system{B}$ and $\hilbert_\system{C}$, each being of the same dimension as that of $\hilbert_\system{A}$.
The local operators $\sigma_\system{B}$ and $\tau_\system{C}$ are the same as $\sigma_\system{A}$ and $\tau_\system{A}$, except that they act on $\hilbert_\system{B}$ and $\hilbert_\system{C}$, respectively.
The new-introduced operator $\pi$ acts on $\hilbert_\system{A}\tensor\hilbert_\system{B}\tensor\hilbert_\system{C}$ and is defined as
\begin{equation}\label{eq:QFG:convert:normal}
\pi\defeq\exp\left(\log{\sigma_\system{A}}\!\tensor\! I_\system{BC} + \log{\tau_\system{A}}\!\tensor\! I_\system{BC} - I_\system{A}\!\tensor\!\log{\sigma_\system{B}}\!\tensor\! I_\system{C} - I_\system{AB}\!\tensor\!\log{\tau_\system{C}}\right).
\end{equation}
In this case, the global operator of the QFG on the \RHS in Figure~\ref{fig:QFG:1} can be expressed as
\begin{align*}
\tilde{\rho} &\defeq \rho_\system{A} \star \sigma_\system{B} \star \tau_\system{C}\star \pi\\
&=\exp\left(\log{\rho_\system{A}}\tensor I_\system{BC} + \log{\sigma_\system{A}}\tensor I_\system{BC} + \log{\tau_\system{A}}\tensor I_\system{BC}\right)\\
&=\rho_\set{G}\tensor I_\system{BC}. \qedhere
\end{align*}
\end{example}
\begin{remark}
As illustrated in Example~\ref{ex:QFG:1}, any QFG can be converted into an equivalent normal QFG by introducing additional local operators.
However, in contrast to factor graphs or DeFGs, the additional local operator introduced depends on the original local operators (see, \eg,~\eqref{eq:QFG:convert:normal}).
In the remaining part of this chapter, we will focus solely on normal QFGs.
\end{remark}
\par
%*******************************************************************************
\subsection{Approximated Distributivity of the $\star$-Product over (Partial) Trace}
For a factor graph or a DeFG, the ``closing-the-box'' operation is a handy visualization of the distributive law of addition over multiplication.
However, for a QFG, the $\star$-product does not distribute over the partial trace operations in general, as illustrated in the example below.
\begin{example}\label{ex:QFG:2}
\begin{figure}\centering
\begin{subfigure}[c]{0.49\textwidth}\centering
\begin{tikzpicture}[node distance=2.2cm,
	var/.style={circle, minimum size=.5cm, draw},
	factor/.style={rectangle, minimum size=.7cm, draw}]
	\node[factor] (a) {$\rho_a$};
	\node[factor] (b)[right of=a] {$\rho_b$};
	\node[factor] (c)[right of=b] {$\rho_c$};
	\path (a) edge[draw=none] node[var] (A) {} (b); \node at (A) {$\system{A}$};
	\path (b) edge[draw=none] node[var] (B) {} (c); \node at (B) {$\system{B}$};
	\draw[line width=2pt] (a) -- (A) -- (b) -- (B) -- (c);
	\draw[dashed] ([xshift=-.5cm,yshift=.5cm]b) rectangle ([xshift=.5cm,yshift=-.5cm]c);
	\draw[dashed] ([xshift=-.5cm,yshift=.7cm]a) rectangle ([xshift=.7cm,yshift=-.7cm]c);
\end{tikzpicture}
\end{subfigure}
\begin{subfigure}[c]{0.49\textwidth}\centering
\begin{tikzpicture}[node distance=2.2cm,
	factor/.style={rectangle, minimum size=.7cm, draw}]
	\node[factor] (a) {$\rho_a$};
	\node[factor] (b)[right of=a] {$\rho_b$};
	\node[factor] (c)[right of=b] {$\rho_c$};
	\draw[line width=2pt] (a) -- (b) -- (c);
	\draw[dashed] ([xshift=-.5cm,yshift=.5cm]b) rectangle ([xshift=.5cm,yshift=-.5cm]c);
	\draw[dashed] ([xshift=-.5cm,yshift=.7cm]a) rectangle ([xshift=.7cm,yshift=-.7cm]c);
\end{tikzpicture}
\end{subfigure}
\caption{QFG in Example~\ref{ex:QFG:2}.}
\label{fig:QFG:2}
\end{figure}
Consider the QFG in Figure~\ref{fig:QFG:1} with the variable vertices $\set{V} =\{\system{A},\system{B}\}$, the factor vertices $\set{F} = \{a,b,c\}$, and the local operators $\rho_a\in\PositiveOp(\hilbert_\system{A})$, $\rho_b\in \PositiveOp(\hilbert_\system{A}\tensor\hilbert_\system{B})$, and $\rho_c\in\PositiveOp(\hilbert_\system{B})$.
This QFG represents the factorization $\rho=\rho_a\star\rho_b\star\rho_c$, and the partition sum is $Z = \tr(\rho_a\star\rho_b\star\rho_c)$.
However, in general, we cannot compute $Z$ \emph{in steps}, \ie, 
\begin{equation}\label{eq:notwork}
\tr(\rho_a\star\rho_b\star\rho_c)
=\tr_\system{A}\left(\tr_\system{B}(\rho_a\star\rho_b\star\rho_c)\right)
\not\equiv \tr_\system{A}\left(\rho_a\star\tr_\system{B}(\rho_b\star\rho_c)\right).
\end{equation}
In particular, let $\rho_a=\frac{1}{2}\left[\begin{smallmatrix}+1 & -1\\ -1 & +1 \end{smallmatrix}\right]$, $\rho_b=\diag([0,1,1,0])$, and $\rho_c=\left[\begin{smallmatrix} 1 & 0 \\ 0 & 1 \end{smallmatrix}\right]$, then $\tr_\system{A}(\tr_\system{B}(\rho_a\star\rho_b\star\rho_c))=0$, but $\tr_\system{A}(\rho_a\star\tr_\system{B}(\rho_b\star\rho_c))=1$.
(Here, we have used~\eqref{eq:def:star:2}.)
\end{example}
\par
%*******************************************************************************
However, there do exist situations where $\tr_\system{A}(\rho_a\star\tr_\system{B}(\rho_b\star\rho_c))$ \emph{approximates} $\tr_\system{A}(\tr_\system{B}(\rho_a\star\rho_b\star\rho_c))$ reasonably well.
This subsection aims to understand when this happens so that an approximated notion of the ``closing-the-box'' operations can be salvaged.
%*******************************************************************************
\begin{lemma}\label{lem:partialTraceBound}
Given $\rho_a\in\StPositiveOp(\hilbert_\system{A})$, $\rho_b\in \StPositiveOp(\hilbert_\system{A}\tensor\hilbert_\system{B})$, and $\rho_c\in\PositiveOp(\hilbert_\system{B})$, we have
\begin{equation}\label{eq:partialTraceBound}
  S\left(\kappa(\rho_a)\right)^{-1} \leqslant
  \frac{\tr_\system{A}(\rho_a\star\tr_\system{B}(\rho_b\star\rho_c))}
	   {\tr(\rho_a\star\rho_b\star\rho_c)}
  \leqslant S\left(\kappa(\rho_a)\right),
\end{equation}
where $\kappa(\rho_a)\geqslant 1$ is the condition number of the operator $\rho_a$, and $S(\cdot)$ is the Specht ratio function defined as 
\begin{equation} \label{eq:Specht} \index{Specht ratio function}
  S(r) \defeq \frac{(r-1)\cdot r^{\frac{1}{r-1}}}{e\cdot \log{r}} \quad \forall r>0.
\end{equation}
\end{lemma}
%*******************************************************************************
\begin{proof}
Consider the Golden--Thompson inequality\index{Golden--Thompson inequality} and its reverse~\cite{bourin2007Reverse}, \ie, 
\begin{equation}\label{eq:Golden-Thompson}
  \tr\left(e^{V+W}\right) \leqslant \tr\left(e^V \cdot e^W\right)
  \leqslant S(\alpha)\cdot \tr\left(e^{V+W}\right),
\end{equation}
where $V$ and $W$ are Hermitian operators, and $\alpha$ is the condition number of $e^V$.
For arbitrary (strict) PD operators $\rho$ and $\sigma$, since both $\log{\rho}$ and $\log{\sigma}$ are Hermitian, by substituting $V=\log{\rho}$ and $W=\log{\sigma}$ into~\eqref{eq:Golden-Thompson}, we have  
\begin{equation}\label{eq:Golden-Thompson-log}
	\tr(\rho\star\sigma) \leqslant \tr(\rho\cdot\sigma)
  	\leqslant S(\kappa(\rho))\cdot\tr(\rho\star\sigma) \quad \forall \rho,\sigma>0.
\end{equation}
\par
To prove the first inequality in~\eqref{eq:partialTraceBound}, we have
\begin{align*}
\tr(\rho_a\star\rho_b\star\rho_c)
&\overset{\text{(a)}}{\leqslant} \tr(\rho_a\cdot(\rho_b\star\rho_c))
= \tr_\system{A}(\rho_a\cdot\tr_\system{B}(\rho_b\star\rho_c))\\
&\overset{\text{(b)}}{\leqslant} S(\kappa(\rho_a))\cdot \tr_\system{A}(\rho_a\star\tr_\system{B}(\rho_b\star\rho_c))
= S(\kappa(\rho_a))\cdot \tr(\rho_a\star\rho_b\star\rho_c),
\end{align*}
where we have applied the first and second inequalities in~\eqref{eq:Golden-Thompson-log} to step~(a) and step~(b), respectively.
Similarly, to show the second inequality in~\eqref{eq:partialTraceBound}, we have
\begin{align*}
\tr_\system{A}(\rho_a \star \tr_\system{B}(\rho_b\star\rho_c))
&\overset{\text{(a)}}{\leqslant} \tr_\system{A}(\rho_a \cdot \tr_\system{B}(\rho_b\star\rho_c))
= \tr(\rho_a \cdot (\rho_b\star\rho_c))\\
&\overset{\text{(b)}}{\leqslant} S(\kappa(\rho_a)) \cdot \tr(\rho_a\star(\rho_b\star\rho_c))
= S(\kappa(\rho_a))\cdot \tr(\rho_a\star\rho_b\star\rho_c),
\end{align*}
where, again, we have used~\eqref{eq:Golden-Thompson-log} in step~(a) and step~(b).
\end{proof}
%*******************************************************************************
Lemma~\ref{lem:partialTraceBound} indicates that $\tr_\system{A}(\rho_a\star\tr_\system{B}(\rho_b\star\rho_c))$ is expected to approximate $\tr(\rho_a\star\rho_b\star\rho_c)$ reasonably well when $\rho_a$ or $\rho_b\star\rho_c$ is proportionally close to the identity matrix.
This fact is further elaborated in the following theorem, where $\rho_a$ and $\rho_b\star\rho_c$ are close to the identity matrix in a linear fashion, \ie, $\rho_a=I+tX$ and $\rho_b\star\rho_c=I+tY$ for some Hermitian matrices $X$ and $Y$, and some real number $t$ in a small neighborhood of zero.
Another approach to study such approximations is to consider $\rho_a=e^{tX}$ and $\rho_b\star\rho_c=e^{tY}$, again, for some Hermitian matrices $X$ and $Y$, and some real number $t$ in a small neighborhood of zero.
We present the second approach in Appendix~\ref{app:exp:approx}.
%*******************************************************************************
\begin{theorem}\label{thm:approx:distri}
Consider finite-dimensional Hilbert spaces $\hilbert_\system{A}$ and $\hilbert_\system{B}$.
Given $X \in \HermitianOp(\hilbert_\system{A})$ and $Y\in\HermitianOp(\hilbert_\system{A} \tensor \hilbert_\system{B})$, it holds that
\begin{equation}\label{eq:approx:distri}
\tr\left((I+t X)\star(I+t Y)\right) = \tr_\system{A} \left( (I+t X) \star \tr_\system{B}(I+t Y) \right) + O(t^4),
\end{equation}
where the real number $t$ is in a neighborhood of $0$ such that $I+tX$ and $I+tY$ are always positive definite.
In other words, $\tr_\system{A} \left( (I+t X) \star \tr_\system{B}(I+t Y) \right)$ approximates $\tr\left((I+t X)\star(I+t Y)\right)$ when $t$ is small, and the error is of fourth order in $t$.
\end{theorem}
%*******************************************************************************
\begin{proof}
The theorem is proven using the Taylor series expansion.
Recall that, for a Hermitian matrix $V$ and a positive definite matrix $W$ with spectral radius less than $1$, we have 
\begin{align*}
\exp{V} &= I + V + \frac{1}{2} V^2 + \frac{1}{3!} V^3 + \cdots + \frac{1}{n!} V^n + \cdots,\\
\log{(I+W)} &= B - \frac{1}{2} W^2 + \frac{1}{3} W^3 - \cdots + \frac{(-1)^{n-1}}{n} W^n + \cdots.
\end{align*}
To simplify the proof, we introduce the \emph{normalized} (partial) trace functions
\[
  \overline{\tr}(\rho) \defeq \frac{\tr(\rho)}{\tr(I)},\qquad
  \overline{\tr}_\system{A}(\rho) \defeq \frac{\tr_\system{A}(\rho)}{\tr_\system{A}(I)}.
\]
Using the above notations, we rewrite~\eqref{eq:approx:distri} as 
\[
\overline{\tr}\left((I+t X)\star(I+t Y)\right) = \overline{\tr}_\system{A} \left( (I+t X) \star \overline{\tr}_\system{B}(I+t Y) \right) + O(t^4).
\]
Let $\tilde{X}\defeq X\tensor I \in \HermitianOp(\hilbert_\system{A} \tensor \hilbert_\system{B})$, and consider the Taylor series expansion of the terms (without $O(t^4)$) on both sides of the above equation:
\begin{align}
\label{eq:taylor:LHS:ADT}
\text{LHS} &= \begin{aligned}[t]
1 &+ t\cdot \overline{\tr}(\tilde{X}+Y) + t^2\cdot\frac{\overline{\tr}(\tilde{X}Y + Y\tilde{X})}{2}\\
&+ t^4 \cdot \frac{\overline{\tr}\left( \tilde{X}Y\tilde{X}Y - \tilde{X}^2Y^2 \right)}{12} + O(t^5),
\end{aligned}\\
\label{eq:taylor:RHS:ADT}
\text{RHS} &= \begin{aligned}[t]
1 &+ t\cdot \overline{\tr}_\system{A}( X + \overline{\tr}_\system{B}(Y) ) + t^2 \cdot \frac{\overline{\tr}_\system{A} \left( X\overline{\tr}_\system{B}(Y) + \overline{\tr}_\system{B} (Y)X \right)}{2}\\
&+ t^4 \cdot \frac{\overline{\tr}_\system{A} \left( X\overline{\tr}_\system{B}(Y)X\overline{\tr}_\system{B}(Y) - X^2\overline{\tr}_\system{B}(Y)^2 \right)}{12}+O(t^5).	
\end{aligned}
\end{align}
Note that $\overline{\tr}_\system{B}(\tilde{X}\cdot Z) = X\cdot\overline{\tr}_\system{B}(Z)$ for any operator $Z\in \HermitianOp(\hilbert_\system{A}\tensor\hilbert_\system{B})$.
Thus, one can easily check that~\eqref{eq:taylor:LHS:ADT} and~\eqref{eq:taylor:RHS:ADT} agree up to $t^3$.
(Note that the coefficients of $t^3$ are both $0$ for~\eqref{eq:taylor:LHS:ADT}
  and~\eqref{eq:taylor:RHS:ADT}.)
In particular, we have
\[
\overline{\tr}\left((I+t X)\star(I+t Y)\right) = 
\begin{aligned}[t]
&\overline{\tr}_\system{A} \left( (I+t X) \star \overline{\tr}_\system{B}(I+t Y) \right)\\
&\hspace{-60pt}+ \frac{\overline{\tr}_\system{A}\left( X \cdot [\overline{\tr}_2(Y),X]\cdot \overline{\tr}_2(Y) \right) - \overline{\tr}\left( \tilde{X} \cdot[\tilde{X},Y] \cdot Y \right)}{12} \cdot t^4 + O(t^5),
\end{aligned}
\]
where $[V,W]\defeq VW-WV$ for matrices $V,W$.
\end{proof}
%*******************************************************************************
Using the same method, we have the following approximation as well.
\begin{proposition}\label{prop:approx:distri}
Under the same setup as in Theorem~\ref{thm:approx:distri}, it holds that
\begin{equation}\label{eq:approx:distri:prop}
\tr_\system{B}\left((I+t X)\star(I+t Y)\right) = (I+t X) \star \tr_\system{B}(I+t Y) + O(t^3).
\end{equation}
\end{proposition}
%*******************************************************************************
\begin{proof}
We use the same method as that in the proof of Theorem~\ref{thm:approx:distri}.
In particular, we rewrite~\eqref{eq:approx:distri:prop} as
\[
\overline{\tr}_\system{B}\left((I+t X)\star(I+t Y)\right) = (I+t X) \star \overline{\tr}_\system{B}(I+t Y) + O(t^3)
\]
and expand each side (without $O(t^3)$) of the above equation as
\begin{align}
\label{eq:taylor:LHS:ADP}
\text{LHS}
&= \begin{aligned}[t]
I &+  t \cdot \overline{\tr}_\system{B}(\tilde{X}+Y) + t^2 \cdot \frac{\overline{\tr}_\system{B}(\tilde{X}Y+Y\tilde{X})}{2} \\
  &+ t^3 \cdot \frac{\overline{\tr}_\system{B}\left( 2\tilde{X}Y\tilde{X} + 2Y\tilde{X}Y - \tilde{X}^2 Y - Y^2\tilde{X} - \tilde{X}Y^2 - Y \tilde{X}^2 \right)}{12} + O(t^4),\!
  \end{aligned}\\
\label{eq:taylor:RHS:ADP}
\text{RHS}
&= \begin{aligned}[t]
I &+ t\cdot (X+\overline{\tr}_\system{B}(Y)) + t^2 \cdot \frac{X\overline{\tr}_\system{B}(Y) + \overline{\tr}_\system{B}(Y)X}{2}\\
&+ t^3 \cdot \begin{aligned}[t]
	\bigg(&\frac{2X\overline{\tr}_\system{B}(Y)X - X^2\overline{\tr}_\system{B}(Y) - X \overline{\tr}_\system{B}(Y)^2}{12} \\
	&+\frac{2\overline{\tr}_\system{B}(Y)X\overline{\tr}_\system{B}(Y) - \overline{\tr}_\system{B}(Y)^2X - \overline{\tr}_\system{B}(Y) X^2}{12}\bigg) + O(t^4).
	\end{aligned}
\end{aligned}
\end{align}
Comparing~\eqref{eq:taylor:LHS:ADP} with~\eqref{eq:taylor:RHS:ADP}, we have
\[
\begin{aligned}
&\overline{\tr}_\system{B}\left((I+t X)\star(I+t Y)\right) = (I+t X) \star \overline{\tr}_\system{B}(I+t Y) + \\
&t^3\cdot\frac{2\left(\overline{\tr}_\system{B}(Y\tilde{X}Y) - \overline{\tr}_\system{B}(Y)X\overline{\tr}_\system{B}(Y)\right) + \left(\overline{\tr}_\system{B}(Y)^2 - \overline{\tr}_\system{B}(Y^2)\right)X + X\left(\overline{\tr}_\system{B}(Y)^2 - \overline{\tr}_\system{B}(Y^2)\right)}{12}\\
& + O(t ^4),
\end{aligned}
\]
which concludes the proof.
\end{proof}
\begin{remark}
Theorem~\ref{thm:approx:distri} and Proposition~\ref{prop:approx:distri} also hold if the local operators are \emph{proportionally} close to the identity matrix.
Namely, it holds that
\begin{align}
\tr\left(a(I+t X)\star b(I+t Y)\right) & = \tr_\system{A} \left( a(I+t X) \star \tr_\system{B}b(I+t Y) \right) + O(t^4)\\
\tr_\system{B}\left(a(I+t X)\star b(I+t Y)\right) &= \left( a(I+t X) \star \tr_\system{B} b(I+t Y) \right) + O(t^3)
\end{align}
for any real numbers $a$, $b$.
\end{remark}
%*******************************************************************************
\begin{corollary}\label{cor:QFG:chain}
\begin{figure}\centering
\begin{tikzpicture}[node distance=2cm, nodes={draw=none},
	factor/.style={rectangle, minimum size=1cm, draw}]
	\node[factor] (1) {}; \node at (1) {$\rho_1$};
	\node[factor] (2)[right of=1] {}; \node at (2) {$\rho_2$};
	\node[factor] (3)[right of=2] {}; \node at (3) {$\rho_3$};
	\node[factor] (4)[right =2cm of 3] {}; \node at (4) {$\rho_{n\!-\!1}$};
	\node[factor] (5)[right of=4] {}; \node at (5) {$\rho_n$};
	\path (3) edge[draw=none] node[midway] (A) {$\cdots$} (4);
	\draw[line width = 2pt] (1) -- (2) -- (3) -- (4) -- (5);
	\draw[dashed] ([xshift=-.7cm,yshift=.7cm]1) rectangle ([xshift=.7cm,yshift=-.7cm]2);
	\draw[dashed] ([xshift=-.9cm,yshift=.9cm]1) rectangle ([xshift=.7cm,yshift=-.9cm]3);
	\draw[dashed] ([xshift=-1.1cm,yshift=1.1cm]1) rectangle ([xshift=.7cm,yshift=-1.1cm]4);
	\draw[dashed] ([xshift=-1.3cm,yshift=1.3cm]1) rectangle ([xshift=.7cm,yshift=-1.3cm]5);
\end{tikzpicture}
\caption{A chain QFG. ($n\geqslant 3$).}
\label{fig:QFG:chain}
\end{figure}
Let $n\geqslant 3$.
Consider a chain QFG as in Figure~\ref{fig:QFG:chain}, where $\rho_1 \in \StPositiveOp(\hilbert_{1})$, $\rho_n \in \StPositiveOp(\hilbert_{n-1})$, and $\rho_k \in \StPositiveOp(\hilbert_{k-1}\tensor\hilbert_{k})$ for each $k=2,\ldots,n-1$.
Suppose that for each $k$, $\rho_k\propto I+t\cdot H_k$ for some Hermitian operator $H_k$ and some real number $t$ close to $0$, then Theorem~\ref{thm:approx:distri} implies the estimate
\begin{equation}\label{eq:chain:QFG}
\tr\left(\bigstar_{k=1}^{n}\rho_k\right) = 
\tr_{n\!-\!1}\left( \tr_{n\!-\!2} \left(\cdots\tr_1(\rho_1\star\rho_2) \star \cdots \star \rho_{n\!-\!1}\right)\star\rho_{n}\right) + O(t^4).
\end{equation}
\end{corollary}
\begin{proof}
We use mathematical induction to prove~\eqref{eq:chain:QFG}.
For $n=3$,~\eqref{eq:chain:QFG} is the same as~\eqref{eq:approx:distri}.
Assume that~\eqref{eq:chain:QFG} is true when $n=m$, for some arbitrary integer $m\geqslant 3$.
Under this hypothesis, consider the case when $n=m+1$, \ie, 
\begin{align*}
\tr\left(\bigstar_{k=1}^{m+1}\rho_k\right)
&= \tr\left(\rho_1\star\cdots\star\rho_{m-1}\star(\rho_m\star\rho_{m+1})\right)\\
&\overset{\text{(a)}}{=} \tr_{m,m\!-\!1}\left( \tr_{m\!-\!2}\left( \cdots \tr_1(\rho_1\star\rho_2) \cdots \star\rho_{m\!-\!1} \right) \star(\rho_{m}\star\rho_{m\!+\!1}) \right) + O(t^4)\\
&\overset{\text{(b)}}{=} \tr_{m}\left(\tr_{m\!-\!1}\left( \tr_{m\!-\!2}\left( \cdots \tr_1(\rho_1\star\rho_2) \cdots \star\rho_{m\!-\!1} \right) \star\rho_{m} \right) \star\rho_{m\!+\!1} \right) + O(t^4),
\end{align*}
where we have used the induction hypothesis in step (a) and Theorem~\ref{thm:approx:distri} in step (b).
Thus, we have shown that~\eqref{eq:chain:QFG} holds for $n=m+1$ provided the same expression holding for $n=m$.
By mathematical induction, the corollary is proven.
\end{proof}
%*******************************************************************************
Via Theorem~\ref{thm:approx:distri} and Proposition~\ref{prop:approx:distri}, we have established the following \emph{approximate distributive laws} of the $\star$-product over the trace and the partial trace functions.
Namely, 
\begin{align}
\label{eq:distributive:star}
\tr(\rho_a\star\rho_b) &\approx \tr_{\nb{b}\xk\set{I}}(\rho_a\star\tr_{\set{I}}(\rho_b)),\\
\label{eq:distributive:star:partial}
\tr_{\set{I}}(\rho_a\star\rho_b) &\approx \rho_a\star\tr_{\set{I}}(\rho_b),
\end{align}
where $\partial{a}\subsetneq\partial{b}$ and where $\set{I}\subset\partial{b}\setminus\partial{a}$ and where $\rho_a \in \StPositiveOp{ \Tensor_{ i\in \partial{a}}\hilbert_i}$ and $\rho_b \in \StPositiveOp{ \Tensor_{ i\in \partial{b}}\hilbert_i}$ are proportionally close to the identity matrix $I$.
\par
%*******************************************************************************
\begin {figure}\centering
\begin{tikzpicture}[every axis/.append style={font=\footnotesize},
	every mark/.append style={scale=.9}]
\begin{semilogyaxis}[width = \textwidth, height = .618\textwidth,
	xlabel={Relative Error $\eta$},
	ylabel={Frequency Density}, ylabel shift = -4 pt,
	xmin=0.002, xmax=0.2, ymin=0.0001,
	legend style={font=\small}, legend cell align={left},
	mark repeat={5}]
	\addplot[dashed] table[x=error,y=freq_N.25] {1155001913-cao_michael_xuan-202104-phd-Data_approx_distributive.txt};
	\addlegendentry{$\abs{\mathcal{N}(\mu,\sigma^2)}$ distributed Eigenvalues ($\mu=1$, $\sigma=0.25$)}
	\addplot[solid] table[x=error,y=freq_N.5] {1155001913-cao_michael_xuan-202104-phd-Data_approx_distributive.txt};
	\addlegendentry{$\abs{\mathcal{N}(\mu,\sigma^2)}$ distributed eigenvalues ($\mu=1$, $\sigma=0.5$)}
	\addplot[solid,mark=+] table[x=error,y=freq_N1] {1155001913-cao_michael_xuan-202104-phd-Data_approx_distributive.txt};
	\addlegendentry{$\abs{\mathcal{N}(\mu,\sigma^2)}$ distributed Eigenvalues ($\mu=1$, $\sigma=1$)}
	\addplot[solid,mark=o] table[x=error,y=freq_U01] {1155001913-cao_michael_xuan-202104-phd-Data_approx_distributive.txt};
	\addlegendentry{Uniformly distributed Eigenvalues ($a=0$, $b=1$)}
\end{semilogyaxis}\end{tikzpicture}
\caption{Distribution of $\eta$ for random positive operators $\rho_a$ and $\rho_b$ under different randomization schemes.}
\label{fig:approx:distributive:star}
\end{figure}
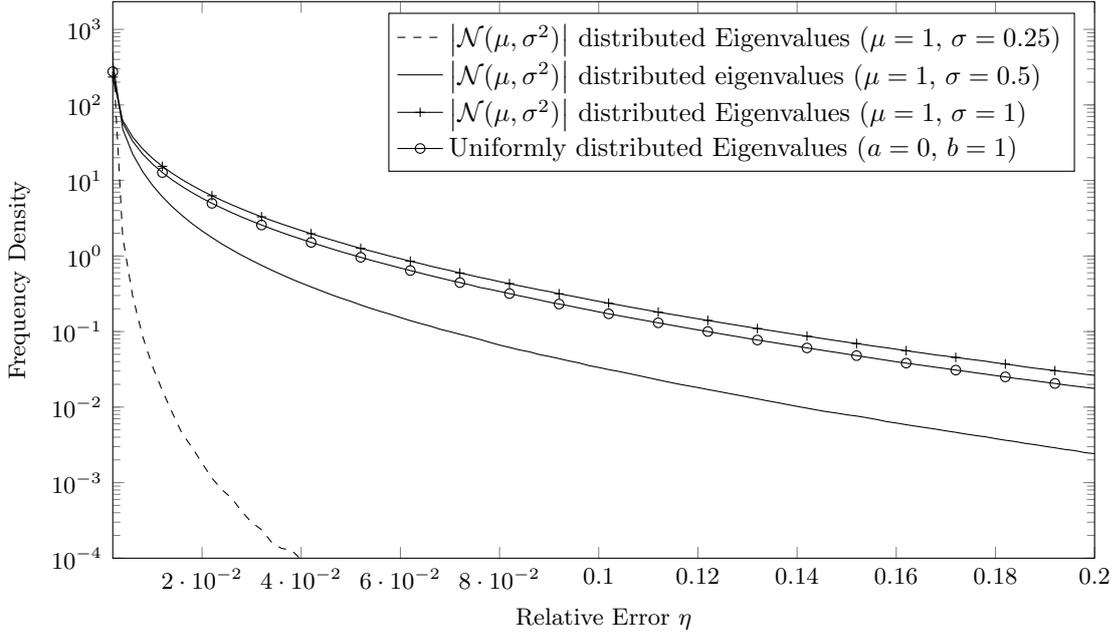
It is worthwhile to make a brief numerical comparison between $\tr_\system{A}( \rho_a\star \tr_\system{B}(\rho_b))$ and $\tr(\rho_\system{a}\star\rho_b)$.
In Figure~\ref{fig:approx:distributive:star}, we randomly generate $\rho_a\in\StPositiveOp(\hilbert_\system{A})$ and $\rho_b\in\StPositiveOp(\hilbert_\system{A}\tensor\hilbert_\system{B})$, where $\hilbert_\system{A}=\hilbert_\system{B}=\Complex^2$, and plot statistical distribution of the relative error as in 
\begin{equation}
\eta \defeq \frac{\abs{\tr_\system{A}( \rho_a\star \tr_\system{B}(\rho_b)) - \tr(\rho_a\star\rho_b)}}{\tr(\rho_a\star\rho_b)}.
\end{equation} 
Here, $\rho_a\defeq U_\system{A}^\Herm\Lambda_\system{A}U_\system{A}$ and $\rho_b\defeq U_\system{AB}^\Herm\Lambda_\system{AB}U_\system{AB}$, where the matrices $U_\system{A}$ and $U_\system{AB}$ are, respectively, random unitary matrices in $\Complex^2$ and $\Complex^4$ generated according to the Haar measures over unitary matrices of corresponding sizes and where $\Lambda_\system{A}$ and $\Lambda_\system{AB}$ are positive diagonal matrices with {i.i.d.}~entries (these entries are eigenvalues of $\rho_a$ and $\rho_b$).
The distributions used to generate the eigenvalues are marked in the legend.
For example, in the legend of Figure~\ref{fig:approx:distributive:star}, $\abs{\mathcal{N}(\mu,\sigma^2)}$ stands for the random variable distributed according to the absolute value of a Gaussian random variable with mean $\mu$ and variance $\sigma^2$.
%*******************************************************************************
\section{Quantum Belief-Propagation Algorithms}
In this section, we propose an approximated version of the ``closing-the-box'' operations for QFGs.
The approximation is based on the approximated distributivity discussed in the previous section.
The quantum belief-propagation algorithm, which was first proposed in~\cite{Leifer08}, can then be understood as a natural extension of the approximated ``closing-the-box'' operations.
At the end of this section, we discuss the holographic transformations of QFGs.
\subsection{Approximated ``Closing-the-Box'' Operations on QFGs}
Motivated by the approximated distributive law of the $\star$-product over partial trace operations (see~\eqref{eq:distributive:star:partial}), we define the ``closing-the-box'' operations on QFGs as the process to replace the box with the result of the partial trace of the $\star$-product of the local operators \emph{in} the box \wrt the Hilbert spaces \emph{in} the box.
\begin{definition}[``Closing-the-Box'' Operations on QFGs] \label{def:QFG:CtB} \index{closing-the-box operations!of QFGs}
Let $\set{G}=(\set{V},\set{F},\set{E})$ be a QFG as defined in Definition~\ref{def:QFG}.
Let $\set{G}'=(\set{V}',\set{F}')$ be a subgraph of $\set{G}$ such that if a variable vertex is in $\set{G}'$, so do all of its neighbors (all of which are in $\set{F}$).
(We call such a subgraph a \emph{box} of $\set{G}$.)
The ``closing-the-box'' operation (\wrt $\set{G}'$) is to replace $\set{G}'$ in $\set{G}$ by a factor vertex associated with the \emph{exterior operator} $\rho_{\set{G}'}$ of $\set{G}'$,
where $\rho_{\set{G}'}$ is the resultant operator by taking the partial trace of the $\star$-product of the local operators associated with $\set{F}'$ over the Hilbert spaces associated with $\set{V}'$.
Namely, 
\[
\rho_{\set{G}'} \defeq \tr_{i\in\set{V}'}\left(\bigstar_{a\in\set{F}'} \rho_a \right).\qedhere
\]
\end{definition}
Since the matrix exponential of Hermitian matrices is always PSD, a ``closing-the-box'' operation of a QFG always ends up with another QFG.
As discussed in Corollary~\ref{cor:QFG:chain}, for a chain QFG, its partition sum can be estimated via a sequence of ``closing-the-box'' operations.
This can be further generalized to all acyclic QFGs (see Theorem~\ref{thm:BP:QFG:tree} below).
%*******************************************************************************
\begin{theorem}\label{thm:BP:QFG:tree}
Given an acyclic QFG $(\set{G}=(\set{V},\set{F},\set{E}\in\set{V}\times\set{F}), \mathfrak{V}=\{\hilbert_i\}_{i\in\set{V}}, \mathfrak{F}=\{\rho_a\}_{a\in\set{F}})$, there exists a sequence of ``closing-the-box'' operations, tracing out one Hilbert space at each step, that shrinks the QFG to a constant\footnote{More precisely, the word ``constant'' here refers to a trivial QFG representing a constant operator on a 1-dimensional Hilbert space.} $\hat{Z}$.
If all local operators are proportionally close to $I$, \ie, for each $a$, $\rho_a\propto I+t\cdot H_a$ for some Hermitian operator $H_a$ and some small number $t$, we have 
\begin{equation}
	\hat{Z} = Z(\set{G}) + O(t^4).	
\end{equation}
\end{theorem}
\begin{proof}
The first part of the theorem is a direct result of acyclicity:
Pick a vertex from $\set{V}$ as a root of $\set{G}$.
Since $\set{G}$ is a tree, there must exist another vertex $v\in\set{V}$ with no descendent from $\set{V}$ (note that $\set{G}$ is bipartite).
Closing the box encompassing $v$ and all of its neighbors will result in a tree with one less vertex from $\set{V}$.
 The process can be repeated until there is only one vertex from $\set{V}$ left in the resultant graph, which can be then ``closed'' into a trivial QFG.
 \par
 The second part of the theorem is a direct generalization of Corollary~\ref{cor:QFG:chain} and can be justified using mathematical induction similar to that in the corollary's proof.
 We omit the details.
\end{proof}
However, due to the approximating nature of the ``closing-the-box'' operations for QFGs, we do not have similar results as in Theorem~\ref{thm:BP:DeFG:tree} or Theorem~\ref{thm:BP:tree} for QFGs.
%*******************************************************************************
The process to estimate the partition sum, as mentioned in Theorem~\ref{thm:BP:QFG:tree}, is summarized in Algorithm~\ref{alg:BP:acyclic:QFG}.
Here, without loss of generality, we have assumed all the leaf vertices are from $\set{F}$, since we can always append an identity factor vertex to a leaf vertex from $\set{V}$ without changing the partition sum.
Note that, in Algorithm~\ref{alg:BP:acyclic:QFG}, we describe the ``closing-the-box'' operations as some message-passing rules, \ie, 
\begin{align}
	\label{eq:msg:acyclic:QFG:update:ia}
	m_{i\to a} &\defeq \bigstar_{c\in\nb{i}\xk{a}} m_{c\to i},\\
	\label{eq:msg:acyclic:QFG:update:ai}
	m_{a\to i} &\defeq \tr_{\nb{a}\xk i}\left( \rho_a \star \Tensor_{k\in\nb{a}\xk{i}} m_{k\to a} \right),
\end{align}
where for each $(i,a)\in\set{E}$, $m_{a\to i}$, $m_{i\to a}$ are some PSD operators acting on $\hilbert_i$.
\begin{algorithm}
\caption{Belief-Propagation Algorithm for Acyclic QFGs}
\label{alg:BP:acyclic:QFG}
\begin{algorithmic}[1]
\Require{An acyclic QFG $\left(\set{G}=(\set{V},\set{F},\set{E}\in\set{V}\times\set{F}), \mathfrak{V}=\{\hilbert_i\}_{i\in\set{V}}, \mathfrak{F}=\{\rho_a\}_{a\in\set{F}}\right)$ with all leaf vertices belonging to $\set{F}$, a root vertex $r\in\mathcal{F}$; and for each $a\in\set{F}$, $\rho_a\propto I + t\cdot H_a$ for some Hermitian operator $H_a$ and some small number $t$}
\Ensure{An estimated partition sum $\hat{Z}=Z(\set{G})+O(t^4)$}
	\State Define the bipartite graph $G'=(\set{V}',\set{F}',\set{E}')\gets \set{G}$;
	\While{$\set{F}'\neq\{r\}$}
	\ForAll{$a\in\set{F}'$ being a leaf and $i$ being its parent in $\set{G}'$}
		\State $m_{a\to i} \gets \tr_{\nb{a}\xk i}\left( \rho_a \star \Tensor_{k\in\nb{a}\xk{i}} m_{k\to a} \right)$;
		%\Comment{Note that $\Tensor_{k\in\nb{a}\xk{i}} m_{k\to a} \equiv \bigstar_{k\in\nb{a}\xk{i}} m_{k\to a}$}
		\State $\set{F}'\gets \set{F}'\xk a$, $\set{E}'\gets \set{E}'\xk (\nb{a}\times a)$;
		\Comment{Remove $a$ from $\set{G}'$.}
	\EndFor
	\ForAll{$i\in\set{V}'$ being a leaf and $a$ being its parent in $\set{G}'$}
		\State $m_{i\to a}\gets\bigstar_{c\in\nb{i}\xk{a}} m_{c\to i}$;
		\State $\set{V}'\gets \set{V}'\xk i$, $\set{E}'\gets \set{E}'\xk (i\times\nb{j})$;
		\Comment{Remove $i$ from $\set{G}'$.}
	\EndFor
	\EndWhile
	\State $\hat{Z}\gets \tr\left(\rho_r \star \Tensor_{i\in\nb{r}}m_{i\to r} \right)$;
\end{algorithmic}
\end{algorithm}
%*******************************************************************************
\subsection{Belief-Propagation Algorithms for QFGs}
Similar to factor graphs and DeFGs, for generic QFGs, we define BP algorithms\index{belief-propagation algorithm!for QFGs} as a heuristic generalization of the message-passing rules~\eqref{eq:msg:acyclic:QFG:update:ai} and~\eqref{eq:msg:acyclic:QFG:update:ia}.
Namely, we consider an iterative method with the updating rules
\begin{align}
	\label{eq:msg:QFG:update:ia}
	m_{i\to a}^{(t)} &\propto \bigstar_{c\in\nb{i}\xk{a}} m_{c\to i}^{(t)},\\
	\label{eq:msg:QFG:update:ai}
	m_{a\to i}^{(t)} &\propto \tr_{\nb{a}\xk i}\left( \rho_a \star \Tensor_{k\in\nb{a}\xk{i}} m_{k\to a}^{(t-1)} \right), 
\end{align}
where the initial messages $\{m_{i\to a}^{(0)}\}_{(i,a)\in\set{E}}$ are proportional to identity operators.
The resultant BP algorithms for QFGs share a similar idea as that of the quantum belief-propagation (QBP) algorithm~\cite{Leifer08}.
However, the latter was derived via a different approach and required the target QFG to be \emph{bifactor}, \ie, each vertex in $\set{F}$ has a degree at most~$2$.
Similar to BP algorithms for factor graphs and DeFGs, different updating sequences of the messages in~\eqref{eq:msg:QFG:update:ia} and~\eqref{eq:msg:QFG:update:ai} (\aka schedules) exist.
However, as usual, we will focus on the synchronous schedule (\aka flooding schedule) in this thesis.
Algorithm~\ref{alg:BP:QFG} lists BP algorithm for QFGs with the flooding schedule.
\begin{algorithm}
\caption{Belief-Propagation Algorithm for QFGs (Flooding Schedule with Timeout)}
\label{alg:BP:QFG}
\begin{algorithmic}[1]
\Require{A QFG $\left(\set{G}=(\set{V},\set{F},\set{E}\in\set{V}\times\set{F}), \mathfrak{V}=\{\hilbert_i\}_{i\in\set{V}}, \mathfrak{F}=\{\rho_a\}_{a\in\set{F}}\right)$ with all leaf vertices belonging to $\set{F}$, a root vertex $r\in\mathcal{F}$; and for each $a\in\set{F}$, $\rho_a \propto I + t\cdot H_a$ for some Hermitian operator $H_a$ and some small number $t$; $\epsilon>0$}
\Ensure{Messages $\{m_{i\to a},m_{a\to i}\in\PositiveOp(\hilbert_i)\}_{(i,a)\in\set{E}}$, $\mathsf{FLAG}\in\{\mathrm{completed},\mathrm{timeout}\}$.}
	\ForAll{$(i,a)\in\set{E}$}
		\State $m_{i\to a}\gets I$;
		%\State $m_{a\to i}\gets I$;
		\Comment{$I$ is the identity operator on $\hilbert_i$.}
	\EndFor
	\State $t\gets 0$;
	\Do
		\State $t\gets t+1$;
		\ForAll{$(i,a)\in\set{E}$}
			%\State $m_{i\to a}^{(t)} \defpropto \bigstar_{c\in\nb{i}\xk{a}} m_{c\to i}^{(t-1)}$;
			\State $m_{a\to i}^{(t)} \defpropto \tr_{\nb{a}\xk i}\left( \rho_a \star \Tensor_{k\in\nb{a}\xk{i}} m_{k\to a}^{(t-1)} \right)$;
		\EndFor
		\ForAll{$(i,a)\in\set{E}$}
			\State $m_{i\to a}^{(t)} \defpropto \bigstar_{c\in\nb{i}\xk{a}} m_{c\to i}^{(t)} \equiv \Tensor_{c\in\nb{i}\xk{a}} m_{c\to i}^{(t)}$;
		\EndFor
	\DoWhile{$\left(\neg\mathsf{timeout}\right) \land \left( \exists(i,a)\in\set{E} \text{ s.t. }\norm{m_{i\to a}^{(t)}-m_{i\to a}^{(t-1)}}_2>\varepsilon\text{ or }\norm{m_{a\to i}^{(t)}-m_{a\to i}^{(t-1)}}_2>\varepsilon\right)$}
	\Comment{$\mathsf{timeout}=\mathsf{false}$ unless the operating time exceeds a pre-selected waiting time.}
	\If{$\mathrm{timeout}$}
        \State{$\mathsf{FLAG}\gets\mathrm{timeout}$;}
	\Else
	    \State{$\mathsf{FLAG}\gets\mathrm{completed}$;}
	    \ForAll{$(i,a)\in\set{E}$}
	    \State{$m_{i\to a}\gets m_{i\to a}^{(t)}$;}
	    \State{$m_{a\to i}\gets m_{a\to i}^{(t)}$;}
	    \EndFor
	\EndIf
\end{algorithmic}
\end{algorithm}
\par
%*******************************************************************************
We define the BP fixed points of a QFG similarly to those of a factor graph or a DeFG.
\begin{definition}[BP Fixed Points of a QFG]\label{def:QFG:BP:fixed:points} \index{BP fixed points! of QFGs}
Given a QFG as provided to Algorithm~\ref{alg:BP:QFG}, a set of messages $\{m_{i\to a},m_{a\to i}\in\PositiveOp(\hilbert_i)\}_{(i,a)\in\set{E}}$ is said to be a BP fixed point if 
\begin{align}
\label{eq:fixed:QFG:ia}
m_{i\to a} &\propto \bigstar_{c\in\nb{i}\xk{a}} m_{c\to i}, \\
\label{eq:fixed:QFG:ai}
m_{a\to i} &\propto \tr_{\nb{a}\xk i}\left( \rho_a \star \Tensor_{k\in\nb{a}\xk{i}} m_{k\to a} \right).
\end{align}
In this case, the set $\{m_{i\to a},m_{a\to i}\}_{(i,a)\in\set{E}}$ is also called a set of fixed-point messages.
\end{definition}
At a BP fixed point, we estimate the partition sum using the induced partition sum (see definition below).
We will discuss an interpretation of the BP fixed points and the corresponding induced partition sums in Section~\ref{sec:QFG:bethe}.
\begin{definition}
Given a set of normalized PSD messages $\{m_{i\to a},m_{a\to i}\}_{(i,a)\in\set{E}}$, not necessarily a BP fixed point, the \emph{induced} partition sum (\wrt the messages) is defined as
\begin{equation}
Z_{\mathsf{induced}}\left(\{m_{i\to a},m_{a\to i}\}_{(i,a)\in\set{E}}\right)=\frac{\prod_{a\in\set{F}}Z_a(\{m_{i\to a}\}_{i\in\nb{a}})\cdot\prod_{i\in\set{V}}Z_i(\{m_{a\to i}\}_{a\in\nb{i}})}{\prod_{(j,a)\in\set{E}}Z_{i,a}(m_{i\to a},m_{a\to i})},
\end{equation}
where
\begin{align*}
Z_a(\{m_{i\to a}\}_{i\in\nb{a}}) &\defeq \tr\left( \rho_a \star \Tensor_{i\in\nb{a}} m_{i\to a} \right) && \forall a\in\set{F},\\
Z_i(\{m_{a\to i}\}_{a\in\nb{i}}) &\defeq \tr\left( \bigstar_{a\in\nb{i}} m_{a\to i} \right) && \forall i\in\set{V},\\
Z_{i,a}(m_{i\to a},m_{a\to i}) &\defeq \tr\left( m_{i\to a}\star m_{a\to i} \right) && \forall (i,a)\in\set{E}. \qedhere
\end{align*}
\end{definition}
%*******************************************************************************
\subsection{Holographic Transformations of QFGs}
The following theorem is a generalization of Theorem~\ref{thm:holant} for QFGs.
The holographic transformations of QFGs is a direct result of this theorem.
\begin{theorem}[Holant Theorem for QFGs]
Let $\set{G}=(\set{V},\set{F},\set{E})$ be a QFG representing the factorization $\rho=\bigstar_{a\in\set{F}} \rho_a$.
Given Hermitian operators $\{\hat{\phi}_{i,a}\in\HermitianOp(\hilbert_{i,a}\tensor\hilbert_i)\}_{i,a}\in\set{E}$ and $\{\phi_{i,a}\in\HermitianOp(\hilbert_i\tensor\hilbert_{i,a})\}_{i,a}\in\set{E}$ such that
\begin{equation}
\tr_{\hilbert_{i,a}}	 \left(\phi_{i,a}\star\hat{\phi}_{i,a}\right) = I_{\hilbert_i}
\end{equation}
for each $(i,a)\in\set{E}$, and supposing $\{\rho_a\}_{a\in\set{F}}$, $\{\hat{\phi}_{i,a},\phi_{i,a}\}_{(i,a)\in\set{E}}$ are proportionally close to the identity operator, the partition sum $Z(\set{G})$ can be approximated as
\begin{equation}
Z(\set{G})\defeq \tr\left( \bigstar_{a\in\set{F}} \rho_a \right) \approx
\tr\left(\Tensor_{a\in\set{F}} \hat{\rho}_a \star \Tensor_{i\in\set{V}}\hat\tau_i \right),
\end{equation}
where
\begin{align}
\hat{\rho}_a &\defeq \tr_{\nb{a}} \left( \rho_a \star \bigstar_{i\in\nb{a}} \hat{\phi}_{i,a} \right) \in \HermitianOp\left( \Tensor_{i\in\nb{a}} \hilbert_{i,a} \right) = \HermitianOp(\hilbert_{\nb{a},a}),\\
\hat{\tau}_i &= \tr_{i} \left(\bigstar_{a\in\nb{i}} \phi_{i,a} \right) \in \HermitianOp\left( \Tensor_{a\in\nb{i}} \hilbert_{i,a} \right) = \HermitianOp(\hilbert_{i,\nb{i}}).
\end{align}
\end{theorem}
\begin{proof}
	The proof is based on the approximated distributivity of $\star$ over trace operation~\eqref{eq:distributive:star}.
	Namely, 
	\begin{fleqn}\begin{equation*}
		\tr\left(\Tensor_{a\in\set{F}} \hat{\rho}_a \star \Tensor_{i\in\set{V}}\hat\tau_i \right)
		\overset{\text{(a)}}{=} \tr\left(\bigstar_{a\in\set{F}} \hat{\rho}_a \star \bigstar_{i\in\set{V}}\hat\tau_i \right)
	\end{equation*}\end{fleqn}
	\begin{fleqn}\begin{equation*}
		\phantom{\tr\left(\Tensor_{a\in\set{F}} \hat{\rho}_a \star \Tensor_{i\in\set{V}}\hat\tau_i \right)}
		= \tr\left(\bigstar_{a\in\set{F}} \tr_{\nb{a}}\left( \rho_a \star \bigstar_{i\in\nb{a}} \hat{\phi}_{i,a} \right) \star \bigstar_{i\in\set{V}} \tr_{i} \left(\bigstar_{a\in\nb{i}} \phi_{i,a} \right) \right)
	\end{equation*}\end{fleqn}
	\begin{fleqn}\begin{equation*}
		\phantom{\tr\left(\Tensor_{a\in\set{F}} \hat{\rho}_a \star \Tensor_{i\in\set{V}}\hat\tau_i \right)}
		\overset{\text{(b)}}{\approx} \tr\left(\bigstar_{(i,a)\in\set{E}}\left( \hat{\phi}_{i,a} \star \phi_{i,a} \right) \star \bigstar_{a\in\set{F}} \rho_a \right)
	\end{equation*}\end{fleqn}
	\begin{fleqn}\begin{equation*}
		\phantom{\tr\left(\Tensor_{a\in\set{F}} \hat{\rho}_a \star \Tensor_{i\in\set{V}}\hat\tau_i \right)}
		\overset{\text{(c)}}{\approx} \tr\left(\bigstar_{(i,a)\in\set{E}} \tr_{\hilbert_{i,a}} \left( \hat{\phi}_{i,a} \star \phi_{i,a} \right) \star \bigstar_{a\in\set{F}} \rho_a \right)
	\end{equation*}\end{fleqn}
	\begin{fleqn}\begin{equation*}
		\phantom{\tr\left(\Tensor_{a\in\set{F}} \hat{\rho}_a \star \Tensor_{i\in\set{V}}\hat\tau_i \right)}
		= \tr\left(\bigstar_{(i,a)\in\set{E}} I_{\hilbert_i} \star \bigstar_{a\in\set{F}} \rho_a \right)
	\end{equation*}\end{fleqn}
	\begin{fleqn}\begin{equation*}
		\phantom{\tr\left(\Tensor_{a\in\set{F}} \hat{\rho}_a \star \Tensor_{i\in\set{V}}\hat\tau_i \right)}
		= \tr\left(I_{\hilbert_\set{V}} \star \bigstar_{a\in\set{F}} \rho_a \right)
		= \tr\left(\bigstar_{a\in\set{F}} \rho_a \right)
		= Z(\set{G}),
	\end{equation*}\end{fleqn}
where step~(a) is due to the convention that $\rho_\system{A}\star \rho_\system{B} \equiv \rho_\system{A}\tensor \rho_\system{B}$ given $\rho_\system{A}$ and $\tau_\system{B}$ acting on isolated Hilbert spaces and where we have used~\eqref{eq:distributive:star} in step~(b) and~(c).
\end{proof}
We define the holographic transform of $\set{G}$ (\wrt $\{\hat{\phi}_{i,a},\phi_{i,a}\}_{(i,a)\in\set{E}}$) to be the QFG $\hat{\set{G}}=(\set{E},\set{F}\cup\set{V},\setdef{(e,e_1),(e,e_2)}{e\in\set{E}})$ representing the factorization
\begin{equation}
\hat{\rho} = \Tensor_{a\in\set{F}} \hat{\rho}_a \star \Tensor_{i\in\set{V}}\hat\tau_i
= \bigstar_{a\in\set{F}} \hat{\rho}_a \star \bigstar_{i\in\set{V}}\hat\tau_i.
\end{equation}
%*******************************************************************************
\section{Generalization of Bethe's approximation for QFGs}\label{sec:QFG:bethe}
Bethe's approximation (see Section~\ref{subsec:FG:Bethe}) is a successful approach to interpret BP algorithms for factor graphs.
In this section, we generalize the method to QFGs.
In particular, we prove a generalized version of Theorem~\ref{thm:bethe:BPA} (\ie,~\cite[Theorem 2]{yedidia2005constructing}).
\subsection{From Local Density Operators to Global Density Operators}
An important result following from Theorem~\ref{thm:approx:distri} and Proposition~\ref{prop:approx:distri} is the ability to construct the global density operators from the local operators under suitable conditions.
This is stated in the following lemma.
Notice that this lemma serves as a quantum analog of Lemma 2 in~\cite{mori2015loop}.
\begin{lemma}\label{lemma:QFG:split}
Consider an acyclic QFG $(\set{V},\set{F},\set{E})$.
Given the density operators $\{\sigma_{a}\in \DensOp(\Tensor_{i\in\nb{a}}\hilbert_i)\}_{a\in\set{F}}$ and $\{\sigma_{i}\in\DensOp(\hilbert_i)\}_{i\in\set{V}}$, all of which are proportionally close to the identity operator, such that
\begin{equation}
\sigma_i = \tr_{\partial{a}\xk i}(\sigma_a) \qquad \forall(i,a)\in\set{E},
\end{equation}
there exists a global density operator $\sigma \in \DensOp(\Tensor_{i\in\set{V}}\hilbert_i)$ such that
\begin{align}
\label{eq:QFG:marginal:a}
\tr_{\set{V}\xk \nb{a}}(\sigma) & \approx \sigma_a && \forall a\in\set{F},\\
\label{eq:QFG:marginal:i}
\tr_{\set{V}\xk i}(\sigma) & \approx \sigma_i && \forall i\in\set{V},
\end{align}
where the approximations ``$\:\approx$'' in the above equations are based on~\eqref{eq:distributive:star}.
\end{lemma}
\begin{proof}
Let $\sigma$ be a density operator acting on $\Tensor_{i\in\set{V}}\hilbert_i$ such that
\begin{equation}
\sigma \propto \exp\left( \sum_{a\in\set{F}} \log{\sigma_a} - \sum_{i\in\set{V}}(d_i-1)\log{\sigma_i} \right),
\end{equation}
where $d_i=\deg(i)$ for each $i\in\set{V}$.
We claim that both~\eqref{eq:QFG:marginal:a} and~\eqref{eq:QFG:marginal:i} hold for this $\sigma$.
To verify~\eqref{eq:QFG:marginal:a}, we have
\begin{fleqn}\begin{equation*}
\tr_{\set{V}\xk\nb{a}}(\sigma)
\propto \tr_{\set{V}\xk\nb{a}}\left( \exp\left( \sum_{a\in\set{F}} \log{\sigma_a} - \sum_{i\in\set{V}}(d_i-1)\log{\sigma_i} \right) \right)
\end{equation*}\end{fleqn}
\begin{fleqn}\begin{equation*}\phantom{\tr_{\set{V}\xk\nb{a}}(\sigma)}
\overset{\text{(a)}}{=} \tr_{\set{V}\xk\nb{a}}\left( \exp\left( \log{\sigma_a} + \sum_{n=0}^N \sum_{i\in\nb{^na},\atop c\in\nb{^*_a i}} \left( \log{\sigma_c} - \log{\sigma_i} \right) \right) \right)
\end{equation*}\end{fleqn}
\begin{fleqn}\begin{equation*}\phantom{\tr_{\set{V}\xk\nb{a}}(\sigma)}
= \tr_{\set{V}\xk\nb{a}}\left( \exp\left( \log{\sigma_a} + \sum_{n=0}^{N-1} \sum_{i\in\nb{^na},\atop c\in\nb{^*_a i}} \left( \log{\sigma_c} - \log{\sigma_i} \right) \right)\star \bigstar_{i\in\nb{^Na},\atop c\in\nb{^*_a i}} \left(\sigma_c \star \sigma_i^{-1}\right) \right)
\end{equation*}\end{fleqn}
\begin{fleqn}\begin{equation*}\phantom{\tr_{\set{V}\xk\nb{a}}(\sigma)}
\overset{\text{(b)}}{\approx} \tr_{\set{V}\xk\nb{a}}\left( \exp\left( \log{\sigma_a} + \sum_{n=0}^{N-1} \sum_{i\in\nb{^na},\atop c\in\nb{^*_a i}} \left( \log{\sigma_c} - \log{\sigma_i} \right) \right)\star \bigstar_{i\in\nb{^Na},\atop c\in\nb{^*_a i}} \left(\tr_{\nb{c}\xk{i}}(\sigma_c) \star \sigma_i^{-1}\right) \right)
\end{equation*}\end{fleqn}
\begin{fleqn}\begin{equation*}\phantom{\tr_{\set{V}\xk\nb{a}}(\sigma)}
= \tr_{\set{V}\xk\nb{a}}\left( \exp\left( \log{\sigma_a} + \sum_{n=0}^{N-1} \sum_{i\in\nb{^na},\atop c\in\nb{^*_a i}} \left( \log{\sigma_c} - \log{\sigma_i} \right) \right) \right)
\end{equation*}\end{fleqn}
\begin{fleqn}\begin{equation*}\phantom{\tr_{\set{V}\xk\nb{a}}(\sigma)}
\approx \cdots \approx \tr_{\set{V}\xk\nb{a}}\left( \exp\left( \log{\sigma_a} \right) \right) = \sigma_a,
\end{equation*}\end{fleqn}
where $\nb{^na}$ denotes the set of vertices in $\set{V}$ reachable from $a\in\set{F}$ after walking through $n$ vertices in $\set{F}$ (without backtracking), where $\nb{^*_a i}$ denotes the set of the neighbors of $i$ excluding the vertex through which $a$ reaches $i$, and where $N$ is the largest integer such that $\nb{^Na}$ is nonempty.
Step~(a) is due to the tree structure.
Step~(b) and the ``$\approx$'' on the last line follow directly from Proposition~\ref{prop:approx:distri}.
\par
We omit the verification of~\eqref{eq:QFG:marginal:a} since the process is very similar.
\end{proof}
It is worth noting that if we had defined the global density operator $\tilde\sigma$ as
\begin{equation}
\tilde\sigma \defeq \exp\left( \sum_{a\in\set{F}} \log{\sigma_a} - \sum_{i\in\set{V}} (d_i-1) \log{\sigma_i} \right),
\end{equation}
we would have encountered the problem that $\tr(\tilde\sigma)\neq 1$, despite that a similar result holds for acyclic (classical) factor graphs, \ie,
\begin{equation}
\sum_{\vx} \prod_{a\in\set{F}} \frac{b_a(\vx_\nb{a})} {\prod_{i\in\nb{a}} b_i(x_i)}\prod_{i\in\set{V}} b_i(x_i) = 1,
\end{equation}
given the \pmfs $\{b_a\}_{a\in\set{F}}$, $\{b_i\}_{i\in\set{V}}$ such that $\sum_{\vx_{\nb{a}\xk{i}}} b_a(\vx_\nb{a}) = b_i(x_i)$ for all $(i,a)\in\set{E}$.
However, approximately speaking (using~\eqref{eq:distributive:star}), we have
\begin{equation}\label{eq:tilde:sigma:tr:1}
\tr(\tilde\sigma) = \tr_{\nb{a}}\left( \tr_{\set{V}\xk\nb{a}}\left( \exp\left( \sum_{a\in\set{F}} \log{\sigma_a} - \sum_{i\in\set{V}}(d_i-1) \log{\sigma_i} \right)\! \right)\! \right)
\approx \tr_{\nb{a}}(\sigma_a) = 1.
\end{equation}
Thus, we can treat $\tilde{\sigma}$ as an approximate global density operator.
%*******************************************************************************
\subsection{Free Energies of QFGs}
We define quantum analogies of the Helmholtz free energy and the Gibbs free energy function as follows.
\begin{definition} \label{def:QFG:energy} \index{free energy!quantum Helmholtz} \index{free energy!quantum Gibbs}
Given a QFG $\set{G}$ representing the factorization $\rho=\bigstar_{a\in\set{F}} \rho_a$, we define the \emph{quantum Helmholtz free energy} and the \emph{quantum Gibbs free energy function} \wrt the density operator $\sigma\in\DensOp(\hilbert_\set{V})$ as
\begin{align}
	\helmholtz &\defeq - \log{Z(\set{G})}, \\
    \gibbs(\sigma) &\defeq - \sum_{a \in \set{F}} \tr\left( \sigma \cdot (\log{\rho_a} \tensor I_{\hilbert_{\set{V}\xk\nb{a}}}) \right) - \qEntropy(\sigma),
\end{align}
where $\qEntropy(\sigma)$ is the von Neumann entropy of $\sigma$.
\end{definition}
\begin{theorem}\label{thm:QFG:Gibbs:geq:Helmholtz}
$\gibbs(\sigma)$ is lower bounded by $\helmholtz$.
In particular, 
\begin{equation}
\gibbs(\sigma) = \helmholtz + \infdiv{\sigma}{\tilde\rho},
\end{equation}
where $\tilde\rho\defeq\rho/Z(\set{G})\propto \bigstar_{a\in\set{F}}\rho_a$ is a density operator.
\end{theorem}
\begin{proof}
This is a result of the direct computation, \ie,
\begin{fleqn}\begin{equation*}
\gibbs(\sigma) - \helmholtz = - \sum_{a \in \set{F}} \tr\left( \sigma \cdot (\log{\rho_a} \tensor I_{\hilbert_{\set{V}\xk\nb{a}}}) \right) - \qEntropy(\sigma) + \log(Z(\set{G}))
\end{equation*}\end{fleqn}
\begin{fleqn}\begin{equation*}\phantom{\gibbs(\sigma) - \helmholtz}
= - \tr\left( \sigma \cdot \left(\sum_{a\in\set{F}} \log{\rho_a} - \log{Z(\set{G})} \cdot I \right) \right) + \tr\left(\sigma\cdot \log{\sigma} \right)
\end{equation*}\end{fleqn}
\begin{fleqn}\begin{equation*}\phantom{\gibbs(\sigma) - \helmholtz}
= - \tr\left( \sigma \cdot \log\left(\exp{\left(\sum_{a\in\set{F}} \log{\rho_a} - \log{Z(\set{G})} \cdot I \right)} \right) \right) + \tr\left(\sigma\cdot \log{\sigma} \right)
\end{equation*}\end{fleqn}
\begin{fleqn}\begin{equation*}\phantom{\gibbs(\sigma) - \helmholtz}
= - \tr\left( \sigma \cdot \log{\tilde\rho} \right) + \tr\left(\sigma\cdot \log{\sigma} \right)
= \infdiv{\sigma}{\tilde\rho}.
\end{equation*}\end{fleqn}
By the Klein's inequality, we have $\gibbs(\sigma) \geq \helmholtz$, with equality if and only if $\sigma = \tilde\rho$.
\end{proof}
%*******************************************************************************
Theorem~\ref{thm:QFG:Gibbs:geq:Helmholtz} allows us to compute $\helmholtz$ via minimizing $\gibbs(\sigma)$ over all possible $\sigma\in\DensOp(\hilbert_\set{V})$.
However, such an optimization problem is, in general, not tractable due to the enormous dimension of the density operator $\sigma$.
In analogy to the Bethe free energy (see Definition~\ref{def:bethe:energy}), we propose the following quantum Bethe free energy as an approximation to $\gibbs$.
\begin{definition}[Quantum Bethe Free Energy]\label{def:quantum:bethe:energy} \index{free energy!quantum Bethe}
Given a QFG $\set{G}$ representing the factorization $\rho=\bigstar_{a\in\set{F}} \rho_a$, the \emph{quantum Bethe free energy function} is the function
\begin{equation}\label{eq:quantum:bethe:energy}
	\begin{aligned}
	\bethe(\{\sigma_a\}_{a \in \set{F}}, \{\sigma_i\}_{i \in \set{V}})) \defeq
	&-\sum_{a\in\set{F}} \tr\left( \sigma_a\cdot \log{\rho_a} \right) +  \sum_{a\in\set{F}} \tr\left( \sigma_a\cdot \log{\sigma_a} \right)\\
	&-\sum_{i\in\set{V}} (d_i-1) \cdot \tr\left( \sigma_i\cdot \log{\sigma_i} \right), 
	\end{aligned}
\end{equation}
where the domain of $\bethe$ is
\[
\set{L}(\set{G}) \defeq \setdef*{\big(\{\sigma_a\}_{a\in\set{F}},\{\sigma_i\}_{i\in\set{V}}\big)}
	{\begin{array}{ll}
	\sigma_a\in\DensOp(\Tensor_{i\in\nb{a}}\hilbert_i) &\forall a\in\set{F}\\
	\sigma_i\in\DensOp(\hilbert_i) &\forall i\in\set{V}\\
	\tr_{\nb{a}\xk{i}} \sigma_a = \sigma_i &\forall (i,a)\in\set{E}
	\end{array}}.\qedhere
\]
\end{definition}
For acyclic QFGs, the quantum Bethe free energy approximates the Gibbs free energy, as shown in the following theorem.
\begin{theorem}\label{thm:quantum:Bethe:Gibbs}
Consider an acyclic QFG $(\set{V},\set{F},\set{E})$ representing the factorization $\rho = \bigstar_{a\in\set{F}} \rho_a$.
Let $\sigma$ be a global density operator, and for each $a\in\set{F}$, let $\sigma_a \defeq \tr_{\set{V}\xk\nb{a}} (\sigma)$, and for each $i\in\set{V}$, let $\sigma_i \defeq \tr_{\set{V}\xk i}(\sigma)$.
Supposing $\{\rho_a\}_{a\in\set{F}}$ and $\sigma$ are proportionally close to the identity operators, \ie, $\rho_a\propto I + t\cdot H_a$ and $\sigma\propto I + t\cdot H$ for some Hermitian operators $H_a$ and $H$ and some small number $t$, then
\begin{equation}\label{eq:quantum:Bethe:Gibbs}
\gibbs(\sigma) = \bethe(\{\sigma_a\}_{a\in\set{F}}, \{\sigma_i\}_{i\in\set{V}})) + O(t^3).
\end{equation}
\end{theorem}
\begin{proof}
By definition of partial trace, it holds that 
\[
\tr\left( \sigma \cdot (\log{\rho_a} \tensor I_{\hilbert_{\set{V}\xk\nb{a}}}) \right) = \tr\left( \sigma_a\cdot \log{\rho_a} \right).
\]
Considering the definitions of $\gibbs$ and $\bethe$, it suffices to show  
\[
\sum_{a\in\set{F}} \tr\left( \sigma_a\cdot \log{\sigma_a} \right) -\sum_{i\in\set{V}} (d_i-1) \cdot \tr\left( \sigma_i\cdot \log{\sigma_i} \right) - \tr\left(\sigma\cdot\log{\sigma}\right) = O(t^3).
\]
By letting $\tilde\sigma \defeq \exp\left( \sum_{a\in\set{F}} \log{\sigma_a} - \sum_{i\in\set{V}} (d_i-1) \log{\sigma_i} \right)$, the LHS of above can be rewritten as
\[\begin{aligned}
\tr\left(\sigma\cdot\log{\tilde\sigma}\right) - \tr\left(\sigma\cdot\log{\sigma}\right)
= \tr\left(\sigma\cdot\log{\frac{\tilde\sigma}{\tr\left(\tilde\sigma\right)}}\right) - \tr\left(\sigma\cdot\log{\sigma}\right) + \tr\left(\tilde\sigma\right), \\
= \tr\left(\sigma\cdot\log{\frac{\tilde\sigma}{\tr\left(\tilde\sigma\right)}}\right) - \tr\left(\sigma\cdot\log{\sigma}\right) + O(t^4),
\end{aligned}\]
where we have used~\eqref{eq:tilde:sigma:tr:1} and Theorem~\ref{thm:approx:distri} in the very last step.
\par
Now, consider the quantum exponential family as in Example~\ref{example:QEFwrtG} in Appendix~\ref{app:QEF}.
By Lemma~\ref{lemma:QFG:split}, we know that 
\begin{align*}
\tr_{\set{V}\xk\nb{a}} (\sigma) - \tr_{\set{V}\xk\nb{a}}\left( \frac{\tilde\sigma}{\tr\left(\tilde\sigma\right)} \right)
&= \sigma_a - \tr_{\set{V}\xk\nb{a}}\left( \frac{\tilde\sigma}{\tr\left(\tilde\sigma\right)} \right) = O(t^3), \\
\tr_{\set{V}\xk i} (\sigma) - \tr_{\set{V}\xk i}\left( \frac{\tilde\sigma}{\tr\left(\tilde\sigma\right)} \right)
&= \sigma_i - \tr_{\set{V}\xk i}\left( \frac{\tilde\sigma}{\tr\left(\tilde\sigma\right)} \right) =  O(t^3).
\end{align*}
In terms of their dual parameters, this can be written as $\boldeta(\sigma) = \boldeta(\frac{\tilde\sigma}{\tr\left(\tilde\sigma\right)}) + t^3 \cdot \Delta\boldeta + O(t^4)$ for some real vector $\Delta\boldeta$ (see Appendix~\ref{app:QEF}).
Define the function $f(\boldeta): \boldeta \mapsto - \tr\left( \sigma \cdot \log{\sigma(\boldeta)} \right)$.
In this case, we have
\begin{align*}
\lim_{t \to 0} \frac{\tr\left(\sigma\cdot\log{\frac{\tilde\sigma}{\tr\left(\tilde\sigma\right)}}\right) - \tr\left(\sigma\cdot\log{\sigma}\right)}{t^3}
& = \lim_{t \to 0} \frac{f(\boldeta(\frac{\tilde\sigma}{\tr\left(\tilde\sigma\right)})) - f(\boldeta(\sigma))}{t^3}\\
& = \lim_{t \to 0} \frac{f(\boldeta(\frac{\tilde\sigma}{\tr\left(\tilde\sigma\right)})) - f(\boldeta(\frac{\tilde\sigma}{\tr\left(\tilde\sigma\right)}) + t^3 \cdot \Delta\boldeta + O(t^4))}{t^3}\\
& = \grad{f}^\transp \cdot \Delta\boldeta,
\end{align*}
where the differentiability of $f$ is justified in Appendix~\ref{app:differentiability}.
Since $\grad{f}^\transp \cdot \Delta\boldeta$ is a finite number, we conclude that $\tr\left(\sigma\cdot\log{\frac{\tilde\sigma}{\tr\left(\tilde\sigma\right)}}\right) - \tr\left(\sigma\cdot\log{\sigma}\right) = O(t^3)$.
Thus, we have $\tr\left(\sigma\cdot\log{\tilde\sigma}\right) - \tr\left(\sigma\cdot\log{\sigma}\right) = O(t^3)$, which finishes the proof.
\end{proof}
%*******************************************************************************
\subsection{Correspondence between the Stationary Points of $\bethe$ and the BP Fixed Points}
Motivated by Theorem~\ref{thm:quantum:Bethe:Gibbs}, which allows us to treat the Bethe free energy $\bethe$ as an approximation to the Gibbs free energy $\gibbs$, we define the following optimization problem as an approximated version of the Gibbs minimization problem.
\begin{definition}[Constrained Quantum Bethe Minimization Problem]
Given a QFG $(\set{V},\set{F},\set{E})$ representing the factorization $\rho=\bigstar_{a\in\set{F}} \rho_a$, the optimization problem
\begin{equation}\begin{aligned}
\min\quad & \bethe(\{\sigma_a\}_{a \in \set{F}}, \{\sigma_i\}_{i \in \set{V}})) \\
\st\quad & \sigma_a\in\DensOp(\Tensor_{i\in\nb{a}}\hilbert_i) &&\forall a\in\set{F}\\
& \sigma_i\in\DensOp(\hilbert_i) &&\forall i\in\set{V}\\
& \tr_{\nb{a}\xk{i}} \sigma_a = \sigma_i &&\forall (i,a)\in\set{E}
\end{aligned}\end{equation}
is called the \emph{constrained Bethe minimization problem}.
\end{definition}
One must note that even for acyclic QFGs, $\bethe$ is merely an approximation to $\gibbs$ under suitable conditions.
Thus, the constrained Bethe minimization problem does not guarantee the recovery of the quantum Helmholtz free energy $\helmholtz$, even for acyclic cases.
In contrast, for acyclic factor graphs, $\bethe$ and $\gibbs$ share the same minimum, which equals the Helmholtz free energy $\helmholtz$.
In other words, the BP algorithm for acyclic factor graphs is exact, but BP algorithms for acyclic QFGs are not.
\par
%*******************************************************************************
In the remainder of this section, we show how the quantum BP algorithms, particularly the BP fixed points, are connected to the constrained Bethe minimization problem.
\begin{theorem}\label{thm:bethe:QBPA}
    Given a QFG $((\set{V},\set{F},\set{E}),\{\hilbert_i\}_{i\in\set{V}},\{\rho_a\}_{a\in\set{F}})$ with (straightly) PD local operators, $\left(\{\sigma_a\}_{a\in\set{F}},\{\sigma_i\}_{i\in\set{V}}\right)\in\set{L}(\set{G})$ is an interior stationary point of $\bethe$ if 
    \begin{align}
    	\label{eq:def:qba:msg}
    	\sigma_a &\propto \rho_a \star \bigstar_{i\in\nb{a}} m_{i\to a},\\
    	\label{eq:def:qbi:msg}
    	\sigma_i &\propto \bigstar_{a\in\nb{a}} m_{a\to i},
    \end{align}
    where $\{m_{i\to a},  m_{a\to i}\in \DensOp(\hilbert_i)\}_{(i,a)\in\set{E}}$ are (straightly) PD messages such that
	\begin{align}
	\label{eq:msg:QFG:update:2:ia}
	m_{i\to a} &\propto \bigstar_{c\in\nb{i}\xk{a}} m_{c\to i}, \\
	\label{eq:msg:QFG:update:2:ai}
	m_{a\to i} &\propto \tr_{\nb{a}\xk i}\left( \rho_a \star \Tensor_{k\in\nb{a}} m_{k\to a} \right) \star m_{i\to a}^{-1}.
	\end{align}
    Conversely,~\eqref{eq:msg:QFG:update:2:ia} and~\eqref{eq:msg:QFG:update:2:ai} hold if $\left(\{\sigma_a\}_{a\in\set{F}},\{\sigma_i\}_{i\in\set{V}}\right)$ defined according to~\eqref{eq:def:qba:msg} and~\eqref{eq:def:qbi:msg} form an internal stationary point of $\bethe$.
\end{theorem}
\begin{proof}
This theorem is a quantum analogy of Theorem~\ref{thm:bethe:BPA}.
Part of the ideas of this proof originated from~\cite{yedidia2005constructing}.
\par
Suppose $\left(\{\sigma_a\}_{a\in\set{F}},\{\sigma_i\}_{i\in\set{V}}\right)\in\set{L}(\set{G})$ is an \emph{interior} stationary point of the constrained Bethe minimization problem.
In this case, the Lagrangian of this problem can be expressed as
\[
L \defeq \bethe + \sum_{a\in\set{F}} \gamma_a \cdot \left( \tr\left( \sigma_a \right)-1 \right) + \sum_{i\in\set{V}} \gamma_i \cdot \left( \tr\left( \sigma_i \right) - 1 \right) + \sum_{(i,a)\in\set{E}} \tr\left( \lambda_{i,a} \cdot \left(\sigma_i - \tr_{\partial{a}\xk i}\left( \sigma_a \right) \right) \right),
\]
where $\{\gamma_a\in\Reals\}_{a\in\set{F}}$, $\{\gamma_i\in\Reals\}_{i\in\set{V}}$, and $\{\lambda_{i,a}\in \LinearOp(\hilbert_i)\}_{(i,a)\in\set{E}}$ are Lagrangian dual variables.
At the stationary point, it must hold that
\begin{align}
\label{eq:kkt:1}
\frac{\partial{L}}{\partial{\gamma_a}} &=0 &&\forall a\in\set{F},\\
\label{eq:kkt:2}
\frac{\partial{L}}{\partial{\gamma_i}} &=0 &&\forall i\in\set{V},\\
\label{eq:kkt:3}
\left.\frac{\D}{\D t}\right|_{t=0} L(\lambda_{i,a} + t\cdot H_{i,a}) &= 0 &&\forall H_{i,a}\in\HermitianOp(\hilbert_i),\, \forall (i,a)\in\set{E},\\
\label{eq:kkt:4}
\left.\frac{\D}{\D t}\right|_{t=0} L(\sigma_a + t\cdot H_a) &= 0 &&\forall H_a\in\HermitianOp(\hilbert_\nb{a}),\, \forall a\in\set{F},\\
\label{eq:kkt:5}
\left.\frac{\D}{\D t}\right|_{t=0} L(\sigma_i + t\cdot H_i) &= 0, &&\forall H_i\in\HermitianOp(\hilbert_i),\, \forall i\in\set{V}.
\end{align}
Eqs.~\eqref{eq:kkt:1},~\eqref{eq:kkt:2},~\eqref{eq:kkt:3} are equivalent to the constraints of the problem.
By Spectral theorem and the first-order perturbation theory~\cite{Kato1966},~\eqref{eq:kkt:4} and~\eqref{eq:kkt:5} can be expanded as
\begin{align*}
-\tr\left( H_a \cdot \log{\rho_a} \right) + \tr\left( H_a \cdot \left( I+\log{\sigma_a} \right) \right) + \tr(H_a) \cdot \gamma_a - \sum_{i\in\nb{a}} \tr\left( \lambda_{i,a} \cdot \tr_{\partial{a}\xk i} H_a \right) & = 0,\\
(1-d_i) \cdot \tr\left( H_i \cdot \left( I+\log{\sigma_i} \right) \right) + \tr(H_i) \cdot \gamma_i + \sum_{a\in\nb{i}} \tr\left( H_i \cdot \lambda_{i,a} \right) & = 0.
\end{align*}
Solving the above equations for $\{\sigma_a\}_{a\in\set{F}}$ and $\{\sigma_i\}_{i\in\set{V}}$, respectively, we have
\begin{align}
\label{eq:kkt:4b}
\sigma_a &= \exp\left( \log{\rho_a} + \sum_{i\in\nb{a}} \lambda_{i,a} - (1+\gamma_a) \cdot I \right) && \forall a\in\set{F},\\
\label{eq:kkt:5b}
\sigma_i &= \exp\left( \frac{1}{d_i-1} \cdot \left( \sum_{a\in\nb{i}} \lambda_{i,a} + \left(1+\gamma_i\right) \cdot I \right) \right) && \forall i\in\set{V}.
\end{align}
Eqs.~\eqref{eq:def:qba:msg} and~\eqref{eq:def:qbi:msg} can be shown by making the identification that 
\begin{equation}\label{eq:QBP:lambda}
\lambda_{i,a} = \log{m_{i\to a}} = \sum_{c\in\nb{i}\xk a}\log{m_{c\to i}} \qquad \forall (i,a)\in\set{E}.
\end{equation}
In this case,~\eqref{eq:msg:QFG:update:2:ai} can be derived from the constraints that $\tr_{\nb{a}\xk{i}} \sigma_a = \sigma_i$ for each $(i,a){\in\set{E}}$.
\par
Conversely, suppose there exist some PD messages $\{m_{i\to a},  m_{a\to i}\in \DensOp(\hilbert_i)\}_{(i,a)\in\set{E}}$ satisfying~\eqref{eq:msg:QFG:update:2:ia} and~\eqref{eq:msg:QFG:update:2:ai}.
Let $\left(\{\sigma_a\}_{a\in\set{F}},\{\sigma_i\}_{i\in\set{V}}\right)\in\set{L}(\set{G})$ be defined according to~\eqref{eq:def:qba:msg} and~\eqref{eq:def:qbi:msg}.
By choosing $\{\gamma_a\}_{a\in\set{F}}$, $\{\gamma_i\}_{i\in\set{V}}$, and $\{\lambda_{i,a}\}_{(i,a)\in\set{E}}$ by such that~\eqref{eq:kkt:4b}, \eqref{eq:kkt:5b}, and~\eqref{eq:QBP:lambda} hold simultaneously, we have satisfied~\eqref{eq:kkt:4} and~\eqref{eq:kkt:5}.
Eqs.~\eqref{eq:kkt:1} and~\eqref{eq:kkt:2} also hold since $\{\sigma_a\}_{a\in\set{F}}$ and $\{\sigma_i\}_{i\in\set{V}}$ are density operators.
Finally,~\eqref{eq:kkt:3} holds since $\tr_{\nb{a}\xk{i}} \sigma_a = \sigma_i$ for each $(i,a){\in\set{E}}$, which in turn can be derived from~\eqref{eq:msg:QFG:update:2:ai}.
Thus, we have verified $\left(\{\sigma_a\}_{a\in\set{F}},\{\sigma_i\}_{i\in\set{V}}\right)$ to be a stationary point.
\end{proof}
\par
%*******************************************************************************
Under suitable conditions, by Proposition~\ref{prop:approx:distri},~\eqref{eq:msg:QFG:update:2:ai} can be approximated as
\begin{equation}\label{eq:msg:QFG:update:3:ai}
m_{a\to i} \propto \tr_{\nb{a}\xk i}\left( \rho_a \star \Tensor_{k\in\nb{a}} m_{k\to a} \right) \star m_{i\to a}^{-1}
\approx \tr_{\nb{a}\xk i}\left( \rho_a \star \Tensor_{k\in\nb{a}\xk{i}} m_{k\to a} \right),
\end{equation}
which is exactly the second half of the BP fixed point condition (recall Definition~\ref{def:QFG:BP:fixed:points}).
Therefore, Theorem~\ref{thm:bethe:QBPA} provides an interpretation of the quantum BP algorithm as an iterative method for finding a stationary point of the constrained Bethe minimization problem.
Unfortunately, there is no guarantee of convergence of such an iterative method.
Moreover, even for acyclic QFGs, we have made multiple approximations to link the BP fixed points to the partition sum (including Theorem~\ref{thm:quantum:Bethe:Gibbs} and Eq.~\eqref{eq:msg:QFG:update:3:ai}).
Despite such concerns, as illustrated in the next section, the quantum BP algorithm still shows a rather promising performance in some numerical applications.
\par
%*******************************************************************************
The following corollary generalizes Corollary~\ref{cor:Zb:decompose} and is a direct result of Theorem~\ref{thm:bethe:QBPA}.
\begin{corollary}
Consider a QFG $\set{G}$ representing the factorization $\rho=\bigstar_{a\in\set{F}} \rho_a$.
For a collection of positive messages $\{m_{i\to a}, m_{a\to i}\}_{i,a}$ satisfying~\eqref{eq:msg:QFG:update:2:ia} and~\eqref{eq:msg:QFG:update:2:ai}, we have
\begin{align}
Z_{\mathsf{B}}(\{\sigma_a\}_{a\in\set{F}},\{\sigma_i\}_{i\in\set{V}})
&\defeq \exp\left(-\bethe(\{\sigma_a\}_{a\in\set{F}}, \{\sigma_i\}_{i\in\set{V}})\right)\\
&= Z_{\mathsf{induced}}(\{m_{i\to a},m_{a\to i}\}_{(i,a)\in\set{E}}),
\end{align}
where the density operators $\{\sigma_a\}_{a\in\set{F}}$ and $\{\sigma_i\}_{i\in\set{V}}$ are defined according to~\eqref{eq:def:qba:msg} and~\eqref{eq:def:qbi:msg}.
\end{corollary}
%*******************************************************************************
\section{Numerical Example}\label{sec:QFG:numerical}
\begin{figure}\centering
\begin{tikzpicture}[factor/.style={rectangle, minimum size=1cm, draw}]
	\node[draw=none,inner sep = 0pt, outer sep = 0pt] (O) {};
	\foreach \t in {1,...,6}{
		\pgfmathsetmacro\posX{2*cos(60*\t)};
		\pgfmathsetmacro\posY{2*sin(60*\t)};
		\node[factor] (P\t) at ([xshift=-\posX cm, yshift=\posY cm]O) {$\rho_\t$};
	}
	\foreach \t in {1,...,6}{
		\pgfmathtruncatemacro\a{mod(\t,6)+1};
		\pgfmathtruncatemacro\b{mod(\t+1,6)+1};
		\path (P\t) edge [line width = 2pt] (P\a);
		\path (P\t) edge [line width = 2pt] (P\b);
	}
\end{tikzpicture}
\caption{QFG in Section~\ref{sec:QFG:numerical}.}
\label{fig:QFG:numerical}
\end{figure}
\begin{figure}\centering
\begin{tikzpicture}[nodes = {draw=none, minimum size = 0pt, inner sep = 0pt, outer sep = 0pt}]
	% Scatter Plot
	\node (plot) {\includegraphics[width=6cm]{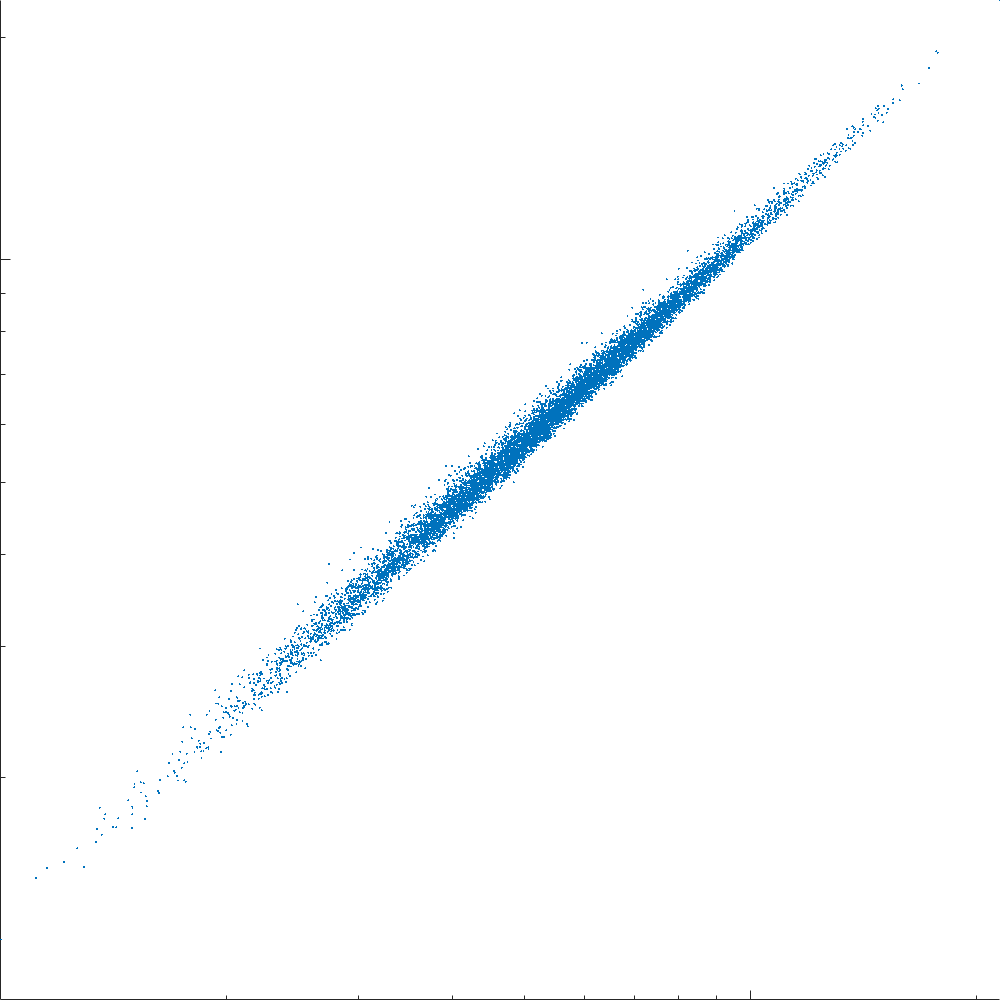}};
	\node (o) at (plot.south west) {};
	\node (xend) at (plot.south east) {};
	\node (yend) at (plot.north west) {};
	\path (o) edge[draw=none]
		node[anchor=north, yshift=-15pt] {$Z$}
		node[pos=0, anchor=north, yshift=-2pt, font=\scriptsize] {$10^{1}$}
		node[pos=.7505, anchor=north, yshift=-2pt, font=\scriptsize] {$10^{2}$}
	(xend);
	\path (o) edge[draw=none]
		node[anchor=south, yshift=20pt, sloped] {$Z_{\mathsf{induced}}$}
		node[pos=0, anchor=east, xshift=-2pt, font=\scriptsize] {$10^{1}$}
		node[pos=.7505, anchor=east, xshift=-2pt, font=\scriptsize] {$10^{2}$}
	(yend);
	\draw[red,dashed] (o) -- (plot.north east);
	% Histogram
	\node[right = 35pt of plot.east, anchor=west] (plotH) {\includegraphics[width=6cm]{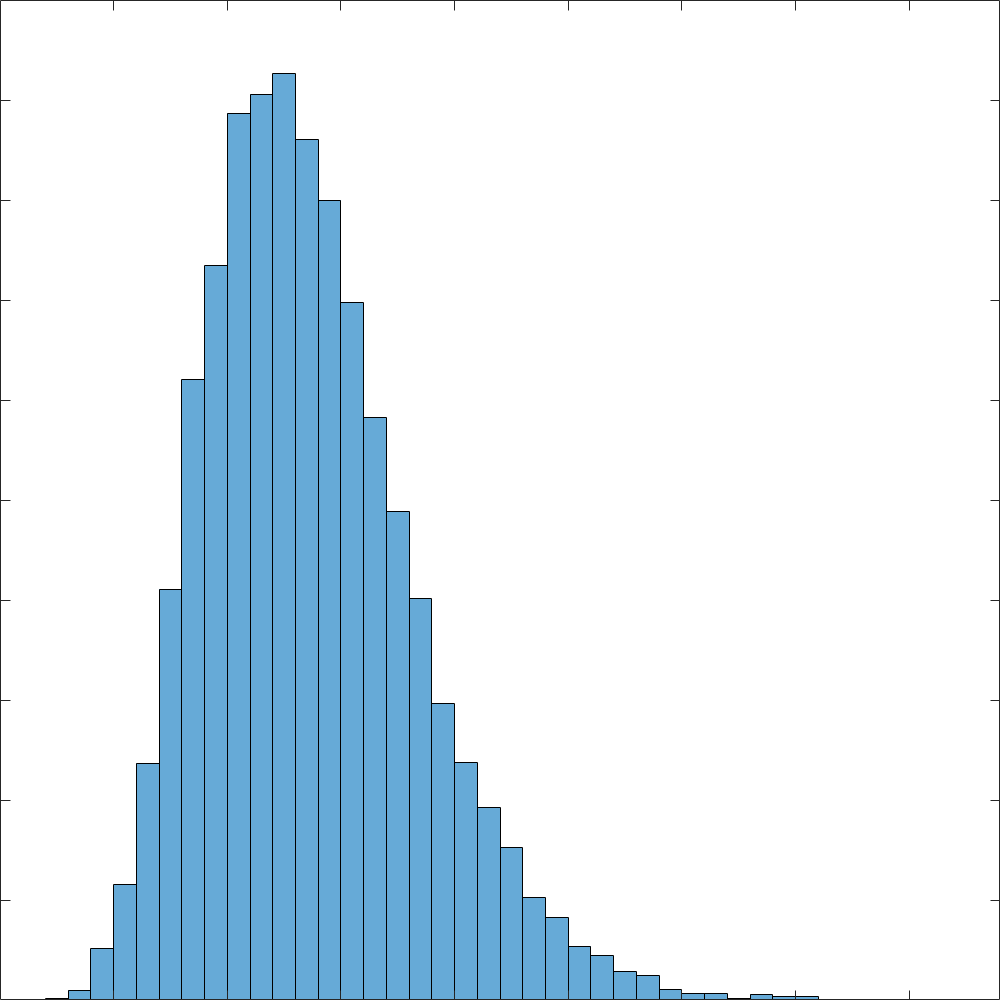}};
	\node (o) at (plotH.south west) {};
	\node (xend) at (plotH.south east) {};
	\node (yend) at (plotH.north west) {};
	\path (o) edge[draw=none]
		node[anchor=north, yshift=-15pt] {Relative Error}
		node[pos=0, anchor=north, yshift=-2pt, font=\scriptsize] {$0$}
		node[pos=.1111, anchor=north, yshift=-2pt, font=\scriptsize] {$.05$}
		node[pos=.2222, anchor=north, yshift=-2pt, font=\scriptsize] {$.10$}
		node[pos=.3333, anchor=north, yshift=-2pt, font=\scriptsize] {$.15$}
		node[pos=.4444, anchor=north, yshift=-2pt, font=\scriptsize] {$.20$}
		node[pos=.5555, anchor=north, yshift=-2pt, font=\scriptsize] {$.25$}
		node[pos=.6666, anchor=north, yshift=-2pt, font=\scriptsize] {$.30$}
		node[pos=.7777, anchor=north, yshift=-2pt, font=\scriptsize] {$.35$}
		node[pos=.8888, anchor=north, yshift=-2pt, font=\scriptsize] {$.40$}
	(xend);
	\path (o) edge[draw=none]
		node[anchor=south, yshift=20pt, sloped] {Number of Appearances}
		node[pos=0, anchor=east, xshift=-2pt, font=\scriptsize] {$0$}
		node[pos=.1, anchor=east, xshift=-2pt, font=\scriptsize] {$100$}
		node[pos=.2, anchor=east, xshift=-2pt, font=\scriptsize] {$200$}
		node[pos=.3, anchor=east, xshift=-2pt, font=\scriptsize] {$300$}
		node[pos=.4, anchor=east, xshift=-2pt, font=\scriptsize] {$400$}
		node[pos=.5, anchor=east, xshift=-2pt, font=\scriptsize] {$500$}
		node[pos=.6, anchor=east, xshift=-2pt, font=\scriptsize] {$600$}
		node[pos=.7, anchor=east, xshift=-2pt, font=\scriptsize] {$700$}
		node[pos=.8, anchor=east, xshift=-2pt, font=\scriptsize] {$800$}
		node[pos=.9, anchor=east, xshift=-2pt, font=\scriptsize] {$900$}
		node[pos=1, anchor=east, xshift=-2pt, font=\scriptsize] {$1000$}
	(yend);
\end{tikzpicture}
\caption{The partition sum $Z$ and the induced partition sum $Z_{\mathsf{induced}}$ produced by Algorithm~\ref{alg:BP:QFG}. The plot consists of 10,000 instances based on the QFG in Figure~\ref{fig:QFG:numerical}, where each local factors $\{\rho_i\}_{i=1}^6$ are generated independently as $\rho_i\gets U\cdot \mbox{diag}(\boldsymbol\lambda_1^{16}) \cdot U$, and where (for each $\rho_i$) $U$ is some $16$-by-$16$ unitary matrix (independently) randomly generated according to Haar measure, and where each $\{\lambda_k\}_{k=1}^{16}$ are independently uniformly distributed on $[0,1]$.}
\label{fig:QFG:numerical:1}
\end{figure}
\begin{figure}\centering
\begin{tikzpicture}[nodes = {draw=none, minimum size = 0pt, inner sep = 0pt, outer sep = 0pt}]
	% Scatter Plot
	\node (plot) {\includegraphics[width=6cm]{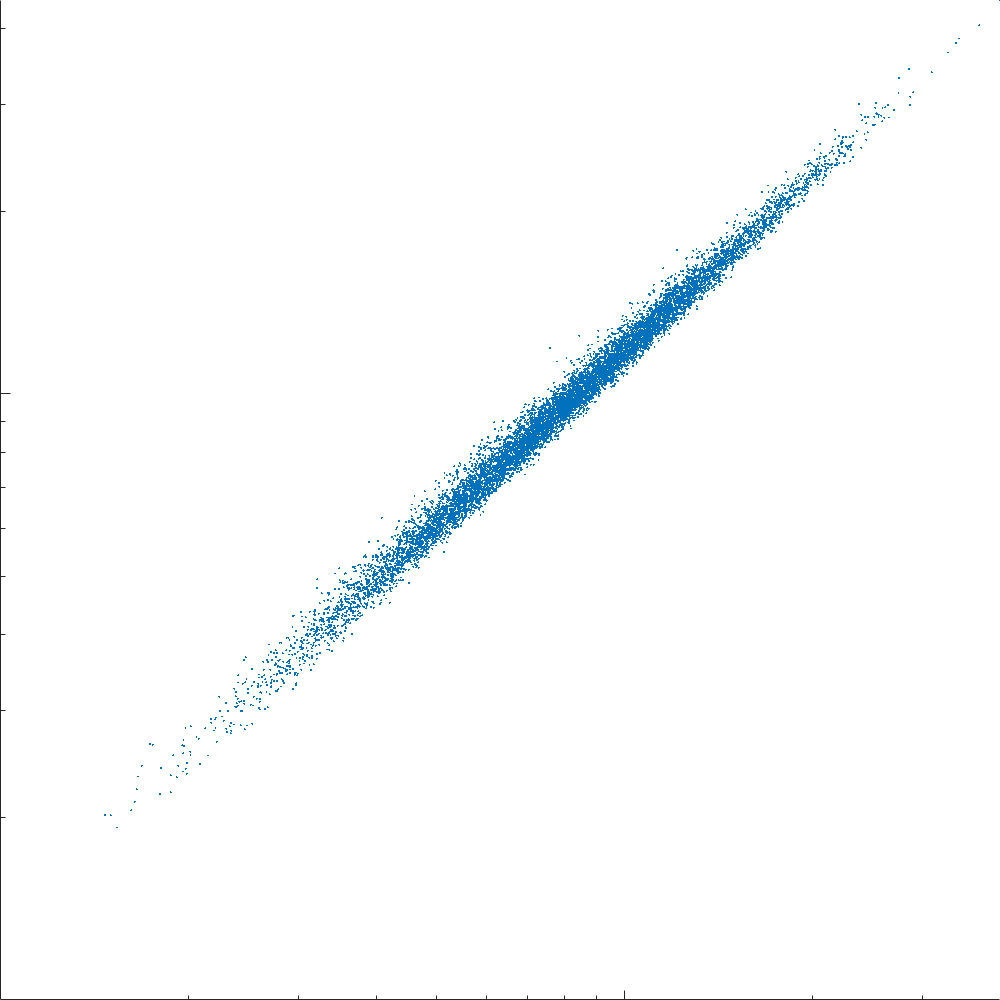}};
	\node (o) at (plot.south west) {};
	\node (xend) at (plot.south east) {};
	\node (yend) at (plot.north west) {};
	\path (o) edge[draw=none]
		node[anchor=north, yshift=-15pt] {$Z$}
		node[pos=0, anchor=north, yshift=-2pt, font=\scriptsize] {$10^{3}$}
		node[pos=.6225, anchor=north, yshift=-2pt, font=\scriptsize] {$10^{4}$}
	(xend);
	\path (o) edge[draw=none]
		node[anchor=south, yshift=20pt, sloped] {$Z_{\mathsf{induced}}$}
		node[pos=0, anchor=east, xshift=-2pt, font=\scriptsize] {$10^{3}$}
		node[pos=.6225, anchor=east, xshift=-2pt, font=\scriptsize] {$10^{4}$}
	(yend);
	\draw[red,dashed] (o) -- (plot.north east);
	% Histogram
	\node[right = 35pt of plot.east, anchor=west] (plotH) {\includegraphics[width=6cm]{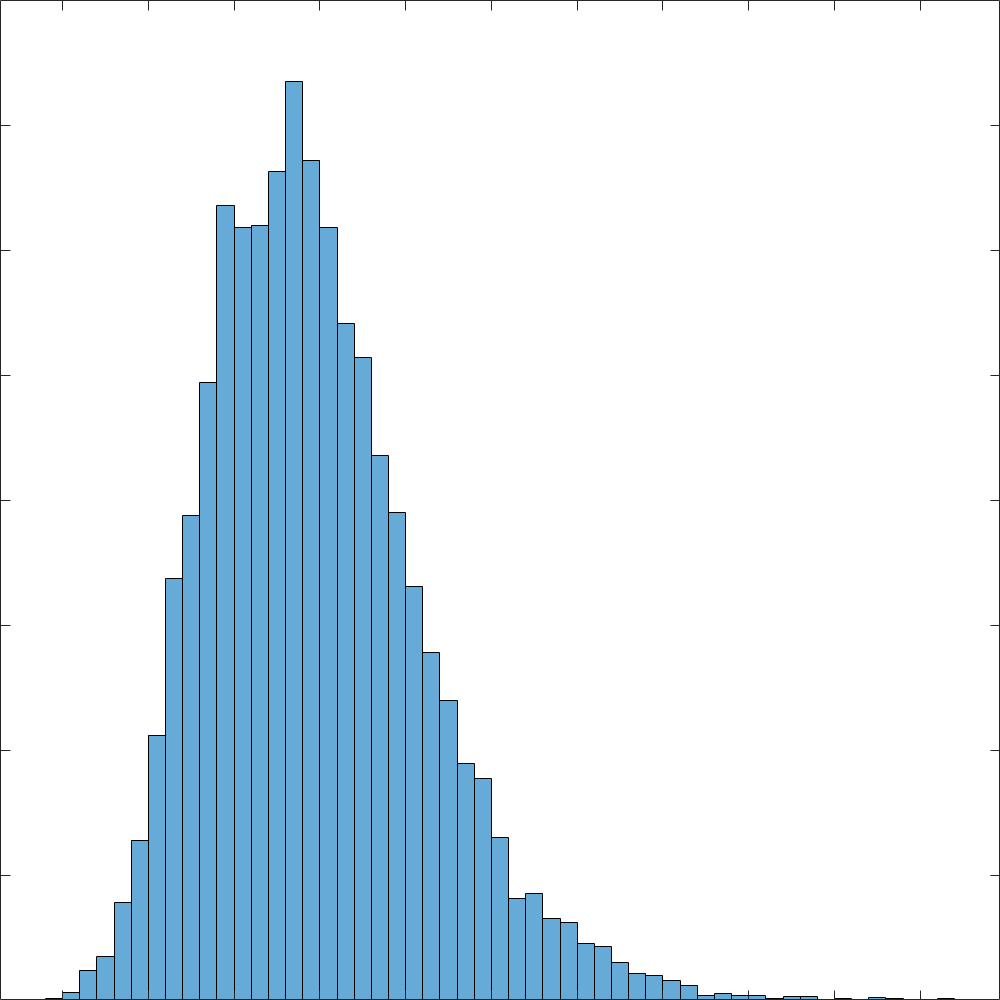}};
	\node (o) at (plotH.south west) {};
	\node (xend) at (plotH.south east) {};
	\node (yend) at (plotH.north west) {};
	\path (o) edge[draw=none]
		node[anchor=north, yshift=-15pt] {Relative Error}
		node[pos=0, anchor=north, yshift=-2pt, font=\scriptsize] {$0$}
		node[pos=.0833, anchor=north, yshift=-2pt, font=\scriptsize] {$.05$}
		node[pos=.1667, anchor=north, yshift=-2pt, font=\scriptsize] {$.10$}
		node[pos=.2500, anchor=north, yshift=-2pt, font=\scriptsize] {$.15$}
		node[pos=.3333, anchor=north, yshift=-2pt, font=\scriptsize] {$.20$}
		node[pos=.4167, anchor=north, yshift=-2pt, font=\scriptsize] {$.25$}
		node[pos=.5000, anchor=north, yshift=-2pt, font=\scriptsize] {$.30$}
		node[pos=.5833, anchor=north, yshift=-2pt, font=\scriptsize] {$.35$}
		node[pos=.6667, anchor=north, yshift=-2pt, font=\scriptsize] {$.40$}
		node[pos=.7500, anchor=north, yshift=-2pt, font=\scriptsize] {$.45$}
		node[pos=.8333, anchor=north, yshift=-2pt, font=\scriptsize] {$.50$}
		node[pos=.9167, anchor=north, yshift=-2pt, font=\scriptsize] {$.55$}
	(xend);
	\path (o) edge[draw=none]
		node[anchor=south, yshift=20pt, sloped] {Number of Appearances}
		node[pos=0, anchor=east, xshift=-2pt, font=\scriptsize] {$0$}
		node[pos=.125, anchor=east, xshift=-2pt, font=\scriptsize] {$100$}
		node[pos=.250, anchor=east, xshift=-2pt, font=\scriptsize] {$200$}
		node[pos=.375, anchor=east, xshift=-2pt, font=\scriptsize] {$300$}
		node[pos=.500, anchor=east, xshift=-2pt, font=\scriptsize] {$400$}
		node[pos=.625, anchor=east, xshift=-2pt, font=\scriptsize] {$500$}
		node[pos=.750, anchor=east, xshift=-2pt, font=\scriptsize] {$600$}
		node[pos=.875, anchor=east, xshift=-2pt, font=\scriptsize] {$700$}
		node[pos=1, anchor=east, xshift=-2pt, font=\scriptsize] {$800$}
	(yend);
\end{tikzpicture}
\caption{The partition sum $Z$ and the induced partition sum $Z_{\mathsf{induced}}$ produced by Algorithm~\ref{alg:BP:QFG}. The plot consists of 10,000 instances based on the QFG in Figure~\ref{fig:QFG:numerical}, where each local factors $\{\rho_i\}_{i=1}^6$ are generated independently as $\rho_i\gets U\cdot \mbox{diag}(\boldsymbol\lambda_1^{16}) \cdot U$, and where (for each $\rho_i$) $U$ is some $16$-by-$16$ unitary matrix (independently) randomly generated according to Haar measure, and where each $\{\lambda_k\}_{k=1}^{16}$ are independently distributed according to $\abs{\mathcal{N}(1,1)}$. ($\mathcal{N}(\mu,\sigma)$ denotes the Gaussian distribution with mean $\mu$ and standard variation $\sigma$.)}
\label{fig:QFG:numerical:2}
\end{figure}
In this section, we consider the QFG as in Figure.~\ref{fig:QFG:numerical}, where each local operator $\{\rho_a\}_{a=1,\ldots,6}$ is generated in the same fashion as $\rho_a$ and $\rho_b$ for Figure~\ref{fig:approx:distributive:star}.
We are interested in comparing the induced partition sum $Z_{\mathsf{induced}}$ at BP fixed points with the actual partition sum $Z(\set{G})$.
In particular, in Figure~\ref{fig:QFG:numerical:1}~and~\ref{fig:QFG:numerical:2}, we have plotted the statistics of the relative error $\eta\defeq\abs{(Z_{\mathsf{induced}}-Z)/Z)}$ \wrt different distributions of the eigenvalues.
The plots are based on $10,000$ simulations.
%*******************************************************************************
% Chapter 4 Case Study**********************************************************
\chapter[IRs of Quantum Channels with Memory]{Bounding and Estimating the Classical Information Rate of Quantum Channels with Memory}\label{chapter:QCwM}
%*******************************************************************************
In this chapter, we consider the transmission rate of classical information over a finite-dimensional quantum channel with memory.
A \emph{quantum channel with memory}~\cite{bowen2004quantum, kretschmann2005quantum, caruso2014quantum} is a quantum channel (from input system $\system{A}$ to output system $\system{B}$) equipped with a memory system $\system{S}$; namely it is a CPTP map from the set of density operators on $\hilbert_\system{A}\tensor\hilbert_\system{S}$ to the set of  density operators on $\hilbert_\system{B}\tensor\hilbert_{\system{S}'}$.
The system $\system{S}$ can be understood either as a state of the channel (as illustrated in Figure~\ref{fig:interpret:1}), or as a part of the environment that does not decay between consecutive channel uses (as illustrated in Figure~\ref{fig:interpret:2}).
Interesting examples of quantum channels with  memory include spin chains~\cite{bose2003quantum} and fiber optic links~\cite{ball2004exploiting}.
\begin{figure}[t]\centering
\begin{subfigure}[b]{0.41\textwidth}\centering
\begin{tikzpicture}[nodes/.style={draw=none. inner sep= 0pt, outer sep= 0pt},
    smallcircle/.style={draw, minimum size = 12pt, circle}]
	\node[draw, minimum width = 30pt, minimum height = 22.5pt] (C) {$\operator{N}$};
	\node[smallcircle, below = 18pt of C] (M) {};
	\node[left = 3pt of M.center, anchor=center, smallcircle] {};
	\node[right = 3pt of M.center, anchor=center, smallcircle] {};
	\draw (C.west|-M.south) -- (C.east|-M.south);
	\draw (C.west|-M.south) arc (270:90:18pt);
	\draw (C.east|-M.south) arc (-90:90:18pt);
	\draw ([yshift=5pt]C.west) -- ([yshift=5pt, xshift=-16pt]C.west) node[pos=1, left](A) {$\system{A}$};
	\draw ([yshift=5pt]C.east) -- ([yshift=5pt, xshift=16pt]C.east) node[pos=1, right](B) {$\system{B}$};
	\node[yshift=16pt] at (A|-M.south) {$\system{S}$};
	\node[yshift=16pt] at (B|-M.south) {$\system{S}'$};
	\node[below=0pt of M, font=\small] {quantum memory};
\end{tikzpicture}
\caption{Memory as the state of the channel.}
\label{fig:interpret:1}
\end{subfigure}
~
\begin{subfigure}[b]{0.5\textwidth}\centering
\begin{tikzpicture}[nodes/.style={draw=none, inner sep=0pt, outer sep=0pt}]
	\node[draw, minimum width = 30pt, minimum height = 75pt] (U) {$\operator{U}$};
	\draw ([yshift=27pt]U.west) -- ([yshift=27pt,xshift=-22.5pt]U.west) node[left, pos=1] (A) {$\system{A}$};
	\draw (U.west) -- ([xshift=-22.5pt]U.west) node[left, pos=1] (E) {$\bra{0}$};
	\draw[latex-] ([yshift=-27pt]U.west) -- ([yshift=-27pt,xshift=-22.5pt]U.west) -- ([yshift=-48pt,xshift=-22.5pt]U.west) -- ([yshift=-48pt,xshift=22.5pt]U.east) -- ([yshift=-27pt,xshift=22.5pt]U.east) -- ([yshift=-27pt]U.east);
	\draw (U.east) -- ([xshift=13.5pt]U.east) node [pos=1, right, draw, minimum size=18pt] (M) {};
	\node[minimum size = 2.25pt, fill = black, inner sep = 0pt, outer sep = 0pt, below = 4.5pt of M.center, circle, anchor = center] (m) {};
	\draw[decoration={markings, mark = at position 1 with {\arrow[>=latex]{>}}}, postaction={decorate}, draw = none] (m.center)--([yshift=10.5pt,xshift=4.5pt]m.center);
	\draw(m.center)--([yshift=9.45pt, xshift=4.05pt]m.center);
	\draw ([xshift=6pt]m.center) arc (0:180:6pt);
	\draw ([yshift=27pt]U.east) -- ([yshift=27pt,xshift=42pt]U.east) node[right, pos=1] (B) {$\system{B}$};
	\draw[line width = 2pt] (M.east) -- ([xshift=20.25pt]M.east) node[pos=1, minimum size = 7.5pt, inner sep = 0pt, outer sep = 0pt] (Eend) {};
	\draw[line width = 2pt] (Eend.south west) -- (Eend.north east);
	\draw[line width = 2pt] (Eend.north west) -- (Eend.south east);

	\node[yshift=-27pt] at (A|-U.west) (S) {$\system{S}$};
	\node[yshift=-27pt] at (B|-U.east) {$\system{S}^\prime$};

	\path[draw=none] (E) edge[draw=none] node[xshift=-18pt](env) {$\Bigg\{$} (S);
	\node[left=0pt of env, anchor = center, rotate=90, font = \small] {environment};
	\node[right=123pt of env] {$\Bigg\}$};
\end{tikzpicture}
\caption{Memory as undecayed partial environment.}
\label{fig:interpret:2}
\end{subfigure}
\caption{Interpretations of quantum channels with memory.}
\end{figure}
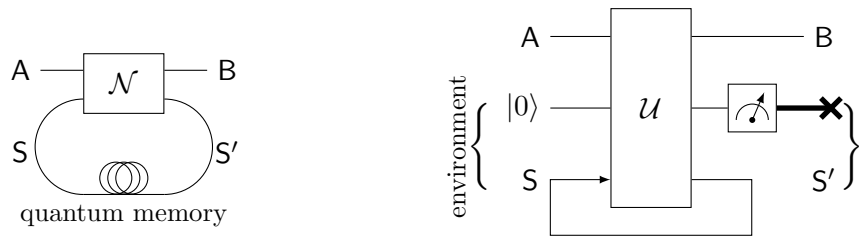
\par
%*******************************************************************************
Classical communication over such channels is accomplished by encoding classical data into some density operators before the transmission and applying measurements to the outputs of the channel.
In the most generic case, an ensemble and a measurement on the \emph{joint} input and output systems of multiple channels uses can be used for encoding and decoding, respectively.
The scenario involving a $k$-channel joint ensemble and a $k$-channel joint measurement is depicted in Figure~\ref{fig:generic_memory}, where
\begin{itemize}
\item the encoding process $\mathcal{E}$ is described by some ensemble $\{\prob_\rv{X}(x), \rho_{\system{A}_1^k}^{(x)}\}_{x\in\set{X}}$ on the joint input system $(\system{A}_1, \ldots, \system{A}_k)$, with $\set{X}$ being the input alphabet, $\prob_\rv{X}(x)$ being the input distribution, and $\rho_{\system{A}_1^k}^{(x)}$ being the density operator on the input systems $\system{A}_1^k$ corresponding to the classical input $x$;
\item the decoding process $\mathcal{D}$ is described by some positive-operator valued measurement (POVM) $\{\Lambda_{\system{B}_1^k}^{(y)}\}_{y\in\set{Y}}$ on the joint output system $(\system{B}_1,\ldots,\system{B}_k)$, with $\set{Y}$ being the output alphabet;
\item the classical input and output are represented by some random variables $\rv{X}$ and $\rv{Y}$, respectively.
\end{itemize}
For comparison, Figure~\ref{fig:generic_memoryless} shows the corresponding memoryless setup.
%*******************************************************************************
The above arrangement results in a (classical) channel from $\rv{X}$ to $\rv{Y}$, whose rate of transmission is given by
\begin{equation}\label{eq:capacity:jointk}
\infoRate(\mathcal{E},\operator{N}^{\boxtimes k},\mathcal{D})= \limsup_{n\to\infty}\frac{1}{n} \mutualInfo(\rv{X}_1^{n};\rv{Y}_1^{n}),
\end{equation}
where we use the above transmission scheme $n$ times consecutively (as depicted in Figure~\ref{fig:generic_memory_multiple}) and where 
\[
\begin{aligned}
\operator{N}^{\boxtimes k} \defeq
&\left( \operator{N}_{\system{A}_k\system{S}_{k-1} \to \system{B}_k\system{S}_k} \otimes \id_{\system{B}_1^{k-1} \to \system{B}_1^{k-1}} \right) \circ
\left(\id_{\system{A}_k \to \system{A}_k} \otimes \operator{N}_{\system{A}_{k-1}\system{S}_{k-2} \to \system{B}_{k-1}\system{S}_{k-1}} \otimes \id_{\system{B}_1^{k-2} \to \system{B}_1^{k-2}} \right)\\ 
&\circ \cdots \circ
\Big( \id_{\system{A}_2^k \to \system{A}_2^k} \otimes \operator{N}_{\system{A}_1\system{S}_0 \to \system{B}_1\system{S}_1} \Big).
\end{aligned}
\]
Here, $\mutualInfo$ stands for the mutual information.
As a fundamental result, this quantity can be simplified to $\mutualInfo(\rv{X};\rv{Y})$ for the memoryless case~\cite{shannon1948mathematical, cover2012elements}.
Optimizing $\infoRate(\mathcal{E},\operator{N}^{\boxtimes k},\mathcal{D})$ over $\mathcal{E}$ and $\mathcal{D}$ (with $k\to\infty$) yields the classical capacity of the quantum channel $\operator{N}$, namely
\begin{equation}
\capacity(\operator{N}) = \limsup_{k}\frac{1}{k}\sup_{\mathcal{E},\mathcal{D}}
    \infoRate(\mathcal{E},\operator{N}^{\boxtimes k},\mathcal{D}).
\end{equation}
\begin{figure}[t]\centering
\begin{subfigure}[b]{.49\textwidth}\centering
\begin{tikzpicture}[nodes/.style={draw=none, inner sep=0pt, outer sep= 0pt},
	bignode/.style={minimum height = 100pt, minimum width = 20pt, draw},
	smallnode/.style={minimum size = 15pt, draw}]
\node[bignode] (cq) {$\mathcal{E}$};
\node[bignode, right = 70pt of cq] (qc) {$\mathcal{D}$};
\path (cq.north) edge[draw=none] node[smallnode, anchor=north] (T1) {$\operator{N}$} (qc.north);
\path (cq.east|-T1) edge node[above] {$\system{A}_1$} (T1);
\path (T1) edge node[above] {$\system{B}_1$} (qc.west|-T1);
\path (cq.south) edge[draw=none] node[smallnode, anchor=south] (Tk) {$\operator{N}$} (qc.south);
\path (cq.east|-Tk) edge node[below] {$\system{A}_k$} (Tk);
\path (Tk) edge node[below] {$\system{B}_k$} (qc.west|-Tk);
\path (T1.center) edge[draw=none] node[smallnode, pos=.3333, anchor=center] (T2) {$\operator{N}$} node[pos=.6667, anchor=center] (Tdots) {$\vdots$} (Tk.center);
\path (cq.east|-T2) edge node[above] {$\system{A}_2$} (T2);
\path (T2) edge node[above] {$\system{B}_2$} (qc.west|-T2);
\path (cq.west) edge[line width=2pt] node[pos=1, above] {$\rv{X}$} ([xshift=-20pt]cq.west);
\path (qc.east) edge[line width=2pt] node[pos=1, above] {$\rv{Y}$} ([xshift=20pt]qc.east);
\draw[draw=none] ([yshift=10pt]T1.north) -- (T1.north) node[pos=0, above=0pt, color=white] {$\system{S}_0$};
\draw[draw=none] ([yshift=-10pt]Tk.south) -- (Tk.south) node[pos=0, below=0pt, color=white] {$\system{S}_k$};
\end{tikzpicture}
\settowidth{\somelength}{(a)}
\setlength{\somelength}{\dimexpr\textwidth-\somelength\relax}
\caption{\parbox[t]{\somelength}{Using the memoryless quantum channel $\operator{N}$ consecutively.}}
\label{fig:generic_memoryless}
\end{subfigure}
\begin{subfigure}[b]{.49\textwidth}\centering
\begin{tikzpicture}[nodes/.style={draw=none, inner sep=0pt, outer sep= 0pt},
	bignode/.style={minimum height = 100pt, minimum width = 20pt, draw},
	smallnode/.style={minimum size = 15pt, draw}]
\node[bignode] (cq) {$\mathcal{E}$};
\node[bignode, right = 70pt of cq] (qc) {$\mathcal{D}$};
\path (cq.north) edge[draw=none] node[smallnode, anchor=north] (T1) {$\operator{N}$} (qc.north);
\path (cq.east|-T1) edge node[above] {$\system{A}_1$} (T1);
\path (T1) edge node[above] {$\system{B}_1$} (qc.west|-T1);
\path (cq.south) edge[draw=none] node[smallnode, anchor=south] (Tk) {$\operator{N}$} (qc.south);
\path (cq.east|-Tk) edge node[below] {$\system{A}_k$} (Tk);
\path (Tk) edge node[below] {$\system{B}_k$} (qc.west|-Tk);
\path (T1.center) edge[draw=none] node[smallnode, pos=.3333, anchor=center] (T2) {$\operator{N}$} node[pos=.6667, anchor=center] (Tdots) {$\vdots$} (Tk.center);
\path (cq.east|-T2) edge node[above] {$\system{A}_2$} (T2);
\path (T2) edge node[above] {$\system{B}_2$} (qc.west|-T2);
\path (cq.west) edge[line width=2pt] node[pos=1, above] {$\rv{X}$} ([xshift=-20pt]cq.west);
\path (qc.east) edge[line width=2pt] node[pos=1, above] {$\rv{Y}$} ([xshift=20pt]qc.east);

\draw([yshift=10pt]T1.north) -- (T1.north) node[pos=0, above=0pt] {$\system{S}_0$};
\path (T1) edge node[right] {$\system{S}_1$} (T2);
\path (T2) edge node[pos=.75, right] {$\system{S}_2$} (Tdots);
\path (Tdots) edge node[pos=.25, right] {$\system{S}_{k-1}$} (Tk);
\draw ([yshift=-10pt]Tk.south) -- (Tk.south) node[pos=0, below=0pt] {$\system{S}_k$};
\end{tikzpicture}
\settowidth{\somelength}{(b)}
\setlength{\somelength}{\dimexpr\textwidth-\somelength\relax}
\caption{\parbox[t]{\somelength}{Using the quantum channel with memory $\operator{N}$ consecutively.}}
\label{fig:generic_memory}
\end{subfigure}
\begin{subfigure}[b]{0.96\columnwidth}\centering
\begin{tikzpicture}[nodes/.style={draw=none, inner sep=0pt, outer sep=0pt},
	smallnode/.style={minimum size = 15pt, draw}]
\node[smallnode] (cq1) {$\mathcal{E}$};
\node[smallnode, below =40pt of cq1.center, anchor=center] (cq2)
     {$\mathcal{E}$};
\node[below =40pt of cq2.center, anchor=center] (cqdots)
     {$\vdots$};
\node[smallnode, below =40pt of cqdots.center, anchor = center] (cqn)
     {$\mathcal{E}$};
\node[smallnode, right = 150pt of cq1] (qc1) {$\mathcal{D}$};
\node[smallnode, right = 150pt of cq2] (qc2) {$\mathcal{D}$};
\node[below = 40pt of qc2.center, anchor = center] (qcdots) {$\vdots$};
\node[smallnode, right = 150pt of cqn] (qcn) {$\mathcal{D}$};

\path (cq1) edge[draw=none] node[smallnode] (T1) {$\operator{N}^{\boxtimes k}$} (qc1);
\path (cq1) edge node[above] {$\system{A}_1^k$} (T1);
\path (T1) edge node[above] {$\system{B}_1^k$} (qc1);

\path (cq2) edge[draw=none] node[smallnode] (T2) {$\operator{N}^{\boxtimes k}$} (qc2);
\path (cq2) edge node[above] {$\system{A}_{k+1}^{2k}$} (T2);
\path (T2) edge node[above] {$\system{B}_{k+1}^{2k}$} (qc2);
\path (cqdots) edge[draw=none] node (Tdots) {$\vdots$} (qcdots);
\path (cqn) edge[draw=none] node[smallnode] (Tn) {$\operator{N}^{\boxtimes k}$} (qcn);
\path (cqn) edge node[below] {$\system{A}_{(n\!-\!1)\cdot k+1}^{n\cdot k}$} (Tn);
\path (Tn) edge node[below] {$\system{B}_{(n\!-\!1)\cdot k+1}^{n\cdot k}$} (qcn);

\path (cq1.west) edge[line width=2pt] node[pos=1, above] {$\rv{X}_1$} ([xshift=-20pt]cq1.west);
\path (cq2.west) edge[line width=2pt] node[pos=1, above] {$\rv{X}_2$} ([xshift=-20pt]cq2.west);
\path (cqn.west) edge[line width=2pt] node[pos=1, above] {$\rv{X}_n$} ([xshift=-20pt]cqn.west);
\path (qc1.east) edge[line width=2pt] node[pos=1, above] {$\rv{Y}_1$} ([xshift=20pt]qc1.east);
\path (qc2.east) edge[line width=2pt] node[pos=1, above] {$\rv{Y}_2$} ([xshift=20pt]qc2.east);
\path (qcn.east) edge[line width=2pt] node[pos=1, above] {$\rv{Y}_n$} ([xshift=20pt]qcn.east);

\draw ([yshift=10pt]T1.north) -- (T1.north) node[pos=0, above=0pt] {$\system{S}_0$};
\path (T1) edge node[right] {$\system{S}_k$} (T2);
\path (T2) edge node[right] {$\system{S}_{2k}$} (Tdots);
\path (Tdots) edge node[right] {$\system{S}_{(n-1)\cdot k}$} (Tn);
\draw ([yshift=-10pt]Tn.south) -- (Tn.south) node[pos=0, below=0pt] {$\system{S}_{n\cdot k}$};
\end{tikzpicture}
\caption{Generic classical communication corresponding to~\eqref{eq:capacity:jointk}.}
\label{fig:generic_memory_multiple}
\end{subfigure}
\caption{Classical communications over quantum channels.}
\label{fig:ccq}
\end{figure}
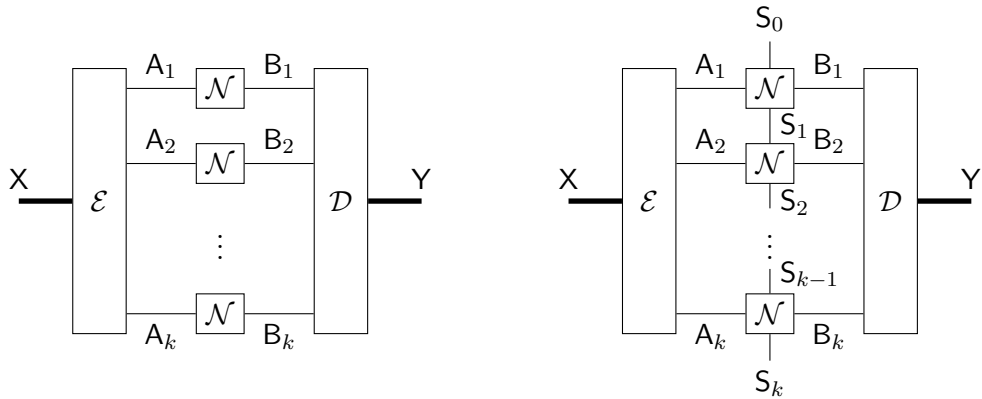
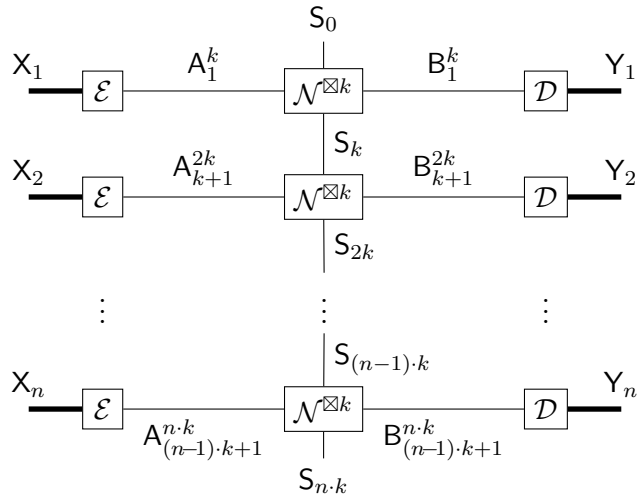
\par
%*******************************************************************************
In this chapter, we are interested in computing and bounding the information rate as in~\eqref{eq:capacity:jointk} for finite-dimensional quantum channels with memory using only separable input ensembles and local output measurements, \ie, the case $k=1$, which is depicted in Figure~\ref{fig:ccq2}.
This restriction is equivalent to the scenario where no quantum computing device is present at the sending or receiving end or where our manipulation of the channel is limited to a single-channel use.
% The above sentence is exactly the same as a sentence in the preface.
The difficulty of the problem lies with the presence of the quantum memory.
In the simplest situation, the memory system exhibits classical properties under certain ensembles and measurements.
In this case, the resulting classical communication setup is equivalent to a finite-state-machine channel (FSMC)~\cite{gallager1968information}.
Though the evaluation of the information rate of an FSMC is nontrivial in general, efficient stochastic methods for estimating and bounding this quantity have been developed~\cite{arnold2006simulation, sadeghi2009optimization}.
\begin{figure}[t]\centering
\begin{tikzpicture}[nodes/.style={draw=none, inner sep=0pt, outer sep=0pt},
	smallnode/.style={minimum size = 15pt, draw},
	bignode/.style={minimum height = 150pt, minimum width = 20pt, draw}]
	\node[bignode] (M) {}; \node[rotate=90] at (M) {Classical Encoder};
	\node[bignode, right=200pt of M] (Mp) {}; \node[rotate=90] at (Mp) {Classical Decoder};
	\path (M.north) edge[draw=none]
		node[smallnode, anchor=north, pos=.25] (cq1) {$\mathcal{E}$}
		node[smallnode, anchor=north, pos=.5] (T1) {$\operator{N}$}
		node[smallnode, anchor=north, pos=.75] (qc1) {$\mathcal{D}$} (Mp.north);
	\path (M.south) edge[draw=none]
		node[smallnode, anchor=south, pos=.25] (cqn) {$\mathcal{E}$}
		node[smallnode, anchor=south, pos=.5] (Tn) {$\operator{N}$}
		node[smallnode, anchor=south, pos=.75] (qcn) {$\mathcal{D}$} (Mp.south);
	\path (cq1.center) edge[draw=none]
		node[smallnode, anchor=center, pos=.3333] (cq2) {$\mathcal{E}$}
		node[anchor=center, pos=.6667] (cqdots) {$\vdots$} (cqn.center);
	\path (T1.center) edge[draw=none]
		node[smallnode, anchor=center, pos=.3333] (T2) {$\operator{N}$}
		node[anchor=center, pos=.6667] (Tdots) {$\vdots$} (Tn.center);
	\path (qc1.center) edge[draw=none]
		node[smallnode, anchor=center, pos=.3333] (qc2) {$\mathcal{D}$}
		node[anchor=center, pos=.6667] (qcdots) {$\vdots$} (qcn.center);	
	\path (M.east|-cq1) edge[line width=2pt] node[above] {$\rv{X}_1$} (cq1.west);
	\path (M.east|-cq2) edge[line width=2pt] node[above] {$\rv{X}_2$} (cq2.west);
	\path (M.east|-cqn) edge[line width=2pt] node[above] {$\rv{X}_n$} (cqn.west);
	\path (cq1) edge node[above] {$\system{A}_1$} (T1);
	\path (cq2) edge node[above] {$\system{A}_2$} (T2);
	\path (cqn) edge node[above] {$\system{A}_n$} (Tn);
	\path (T1) edge node[above] {$\system{B}_1$} (qc1);
	\path (T2) edge node[above] {$\system{B}_2$} (qc2);
	\path (Tn) edge node[above] {$\system{B}_n$} (qcn);
	\path (qc1.east) edge[line width=2pt] node[above] {$\rv{Y}_1$} (Mp.west|-qc1);
	\path (qc2.east) edge[line width=2pt] node[above] {$\rv{Y}_2$} (Mp.west|-qc2);
	\path (qcn.east) edge[line width=2pt] node[above] {$\rv{Y}_n$} (Mp.west|-qcn);
	\path (M.west) edge[line width=2pt] node[above, pos=1] {$\rv{U}$} ([xshift=-20pt]M.west);
	\path (Mp.east) edge[line width=2pt] node[above, pos=1] {$\hat{\rv{U}}$} ([xshift=20pt]Mp.east);
	\path ([yshift=10pt]T1.north) edge node[pos=0, above=0pt] {$\system{S}_0$} (T1.north);
	\path (T1) edge node[right] {$\system{S}_1$} (T2);
	\path (T2) edge node[right] {$\system{S}_2$} (Tdots);
	\path (Tdots) edge node[right] {$\system{S}_n$} (Tn);
	\path ([yshift=-10pt]Tn.south) edge node[pos=0, below=0pt] {$\system{S}_{n+1}$} (Tn.south);
\end{tikzpicture}
\caption{Classical communication over a quantum channel with memory using a separable ensemble and local measurements.}
\label{fig:ccq2}
\end{figure}
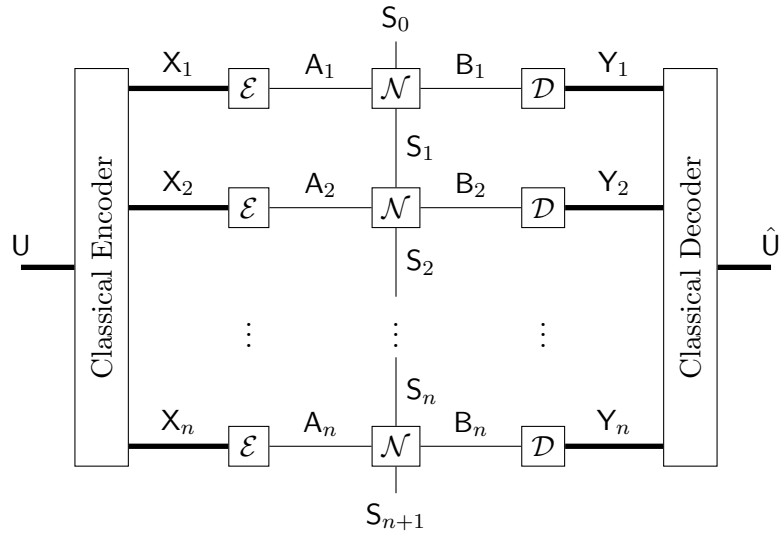
\par
%*******************************************************************************
Our work is highly inspired by~\cite{arnold2006simulation}, where the authors considered the information rate of FSMCs.
In particular, for an indecomposable FSMC~\cite{gallager1968information} with its channel law described by \pmf $W$, its information rate, which is independent of the initial channel state, is given by
\begin{equation}\label{eq:def:fsmc:ir:1}
\infoRate_W(Q) = \lim_{n\to\infty} \frac{1}{n} \mutualInfo(\rv{X}_1^n;\rv{Y}_1^n),
\end{equation}
where $\rv{X}_1^n = \left(\rv{X}_1,\ldots,\rv{X}_n\right)$ is the channel input process characterized by some sequence of distributions $\{Q^{(n)}\}_{n}$ and where $\rv{Y}_1^n = \left(\rv{Y}_1,\ldots,\rv{Y}_n\right)$ is the channel output process.
Although, except for very special cases, there are no single-letter or other simple expressions for information rate available, efficient stochastic techniques have been developed for estimating the information rate for \emph{stationary} and \emph{ergodic} input processes $\{Q^{(n)}\}_{n}$~\cite{arnold2006simulation, sharma2001entropy, pfister2001achievable}.
(For these techniques, under mild conditions, the numerical estimate of the information rate converges with probability one to the true value when the length of the channel input sequence goes to infinity.)
In this chapter, we extend such techniques to quantum channels with memory; in particular, we use \emph{factor graphs} for quantum probabilities~\cite{loeliger2017factor} (see Section~\ref{subsec:FG:QP}) for estimating quantities of interest.
These graphical models are useful for visualizing the relevant computations and for providing a clear comparison between the setup considered in this chapter and its classical counterparts in~\cite{arnold2006simulation} and~\cite{sadeghi2009optimization}.\footnote{
Clearly, the graphical models that we use are very similar to tensor networks (see, for example, the discussion in Appendix~\ref{app:figures}~A of~\cite{loeliger2017factor}).
A benefit of the graphical models that we use (including the corresponding terminology), is that they are compatible with the graphical models that are being used in classical information processing.}
\par
%*******************************************************************************
Our work is also partially inspired by~\cite{sadeghi2009optimization}, where the authors proposed upper and lower bounds based on some so-called auxiliary FSMCs, which are often lower-complexity approximations of the original FSMC.
They also provided efficient methods for optimizing these bounds.
Such techniques have been proven useful for FSMCs with large state spaces, when the above-mentioned information rate estimation techniques can be overly time-consuming.
Interestingly enough, the lower bounds represent achievable rates under mismatched decoding, where the decoder bases its computations not on the original FSMC but on the auxiliary FSMC~\cite{ganti2000mismatched}.
(See~\cite{sadeghi2009optimization} for a more detailed discussion of this topic and for further references.)
In this chapter, we also consider auxiliary channels and their induced bounds.
However, the auxiliary channels of our interest are chosen from a larger set of channels called \emph{quantum-state channels}, which will be defined in Section~\ref{QCwM:sec:3:QCM}.
We also propose a method for optimizing these bounds.
In particular, our method for optimizing the lower bound is ``data-driven'' in the sense that only  the input/output sequences of the original channel are needed, but not the mathematical model of the original channel.
\par
%*******************************************************************************
One must note that even if we can efficiently compute or bound the information rate, it is still a long way to go to compute the classical capacity of a quantum channel with memory.
On the one hand, maximizing $\infoRate(\mathcal{E},\operator{N},\mathcal{D})$ is a difficult problem.
(The analogous classical problems have been addressed in~\cite{arimoto1972algorithm},~\cite{blahut1972computation}, and~\cite{vontobel2008generalization}.)
On the other hand, due to the superadditivity property~\cite{hastings2009superadditivity} of quantum channels, which happens to be more common for quantum channels with memory~\cite{macchiavello2004transition, karimipour2006entanglement, lupo2010transitional} (compared with memoryless quantum channels), it is inevitable to consider joint ensembles on input systems and joint measurements on output systems across multiple channel uses.
\par
%*******************************************************************************
The rest of this chapter is organized as follows.
Section~\ref{sec:2:FSMC} reviews the method of estimating the information rate of an FSMC.
Section~\ref{QCwM:sec:3:QCM} models the classical communication scheme over a quantum channel with memory, and defines the notion of quantum-state channels as an equivalent description.
A graphical notation for representing such channels is also presented in this section.
Section~\ref{QCwM:sec:4:IR} estimates the information rate of such channels.
Section~\ref{QCwM:sec:5:UBLB} considers the upper and lower bounds induced by auxiliary quantum-state channels, and presents methods for optimizing them.
Section~\ref{QCwM:sec:6:example} contains numerical examples.
%*******************************************************************************
%*******************************************************************************
\section[Review of FSMCs]{Review of (Classical) Finite-State Machine Channels:
         Information Rate, its Estimation, and Bounds}
\label{sec:2:FSMC}
In this section, we review the methods developed in~\cite{arnold2006simulation} for estimating the information rate of a (classical) FSMC, and the auxiliary-channel-induced upper and lower bounds studied in~\cite{sadeghi2009optimization}.
As we will see, the development in later sections about quantum channels will have many similarities, but also some important differences.
We emphasize that this section is a \emph{brief review} of~\cite{arnold2006simulation} and~\cite{sadeghi2009optimization} for the purpose of introducing necessary tools and ideas for later sections.
%*******************************************************************************
\subsection{Finite-State Machine Channels (FSMCs) and their Graphical Representation}
A (time-invariant) finite-state machine channel (FSMC) consists of an input alphabet $\set{X}$, an output alphabet $\set{Y}$, a state alphabet $\set{S}$, all of which are finite, and a channel law $W(y,s'|x,s)$, where the latter equals the probability of receiving $y\in\set{Y}$ and ending up in state $s'\in\set{S}$ given the channel input $x\in\set{X}$ and the previous channel state $s\in\set{S}$.
The relationship among the input, output, and state processes $\rv{X}_1^n,\rv{Y}_1^n,\rv{S}_0^n$ of $n$-channel uses can be described by the conditional \pmf
\begin{equation}\label{eq:FSMC:n:conditional}
W(\vy_1^n,\vs_1^n|\vx_1^n,s_0)  \defeq \prob_{\rv{Y}_1^n,\rv{S}_1^n|\rv{X}_1^n,\rv{S}_0}(\vy_1^n,\vs_1^n|\vx_1^n,s_0) = \prod_{\ell=1}^{n} W(y_\ell,s_\ell|x_\ell,s_{\ell-1}),
\end{equation}
where $x_\ell\in\set{X}$, $y_\ell\in\set{Y}$, and $s_\ell\in\set{S}$ for each $\ell$.
%*******************************************************************************
\begin{example}[Gilbert--Elliott channels]\label{example:GEC}
A notable class of examples of FSMCs are the Gilbert--Elliott channels~\cite{mushkin1989capacity}, which behave like a binary symmetric channel (BSC) with cross-over probability $p_s$ controlled by the channel state $s\in\{``\mathrm{b}",``\mathrm{g}"\}$, where usually $\bigl| \pbad - \frac{1}{2} \bigr| < \bigl| \pgood - \frac{1}{2} \bigr|$.
The state process itself is a first-order stationary ergodic Markov process that is independent of the input process.\footnote{
The independence of the state process on the input process is a particular feature of the Gilbert--Elliott channel.
In general, the state process of a finite-state channel can depend on the input process.}
(For more details, see, \eg, the discussions in~\cite{sadeghi2009optimization}.)
\end{example}
%*******************************************************************************
Given an input process $\{Q^{(n)}\}_{n}$ and an initial state \pmf $\prob_{\rv{S}_0}(s_0)$, we can write down the joint \pmf of  $(\rv{X}_1^n,\rv{Y}_1^n,\rv{S}_0^n)$ as 
\begin{equation} \label{eq:FSMC:glabal:distribution}
g(\vx_1^n,\vy_1^n,\vs_0^n) \defeq \prob_{\rv{X}_1^n,\rv{Y}_1^n,\rv{S}_0^n}(\vx_1^n,\vy_1^n,\vs_0^n) = \prob_{\rv{S}_0}(s_0) \!\cdot\! Q^{(n)}(\vx_1^n) \!\cdot\!
    \prod_{\ell=1}^{n} W(y_\ell,s_\ell|x_\ell,s_{\ell-1}).
\end{equation}
The factorization of $g(\vx_1^n,\vy_1^n,\vs_0^n)$ as shown in~\eqref{eq:FSMC:glabal:distribution} can be visualized with the help of an NFG as in Figure~\ref{fig:FMSC:high:level:1}.
\begin{figure}\centering
\begin{tikzpicture}[
    factor/.style={rectangle, minimum width=1cm, minimum height=.7cm, draw},
    sfactor/.style={rectangle, minimum size=.4cm, draw}]
	\node[sfactor] (S) {$\prob_{\rv{S}_0}$};
	\node[factor] (E1) [right=.7cm of S] {$W$};
	\draw (S) -- (E1) node[above=-.05cm,midway] {$s_0$};
	\node[sfactor] (X1) [above=1.1cm of E1] {$Q$};
	\draw (X1) -- (E1) node[right=-.05cm, midway] {$x_1$};
	\draw (E1.south) -- ([yshift=-.6cm]E1.south) node[right=-.05cm] {$y_1$};
	
	\node[factor] (E2) [right=1cm of E1] {$W$};
	\draw (E1) -- (E2) node[above=-.05cm,midway] {$s_1$};
	\node[sfactor] (X2) [above=1.1cm of E2] {$Q$};
	\draw (X2) -- (E2) node[right=-.05cm, midway] {$x_2$};
	\draw (E2.south) -- ([yshift=-.6cm]E2.south) node[right=-.05cm] {$y_2$};

	\node[factor,draw=none] (Edummy) [right=1cm of E2] {$\cdots$};
	\node[sfactor,draw=none] (Xdummy) [above=1.1cm of Edummy] {$\cdots$};
	
	\node[factor] (En) [right=1cm of Edummy] {$W$};
	\draw (E2) -- (Edummy) node[above=-.05cm,midway] {$s_2$};
	
	\draw (Edummy) -- (En) node[above=-.05cm,midway] {$s_{n-1}$};
	\node[sfactor] (Xn) [above=1.1cm of En] {$Q$};
	\draw (Xn) -- (En) node[right=-.05cm, midway] {$x_n$};
	\draw (En.south) -- ([yshift=-.6cm]En.south) node[right=-.05cm] {$y_n$};

	\node[sfactor, inner sep=0pt] (ee) [right=.5cm of En] {$1$};
	\draw (En) -- (ee) node[above=-.05cm, midway] {$s_n$};
	
	\begin{pgfonlayer}{bg}
		\draw[dashed, blue, line width=1.5pt, fill=blue!10] ([xshift=-.8cm,yshift=0.45cm]X1) rectangle ([xshift=.7cm,yshift=-0.45cm]Xn);
		\draw[dashed, red, line width=1.5pt, fill= red!10] ([xshift=-.7cm,yshift=0.55cm]S) rectangle ([xshift=.6cm,yshift=-0.55cm]ee);
	\end{pgfonlayer}
\end{tikzpicture}
\caption[Channel with a classical state: closing the top box yields the input process $Q^{(n)}$, closing the bottom box yields the joint channel law $W(\vy_1^n|\vx_1^n)$.]{Channel with a classical state: closing the {\color{blue}top} box yields the input process $Q^{(n)}$, closing the {\color{red}bottom} box yields the joint channel law $W(\vy_1^n|\vx_1^n)$.}
\label{fig:FMSC:high:level:1}
\end{figure}
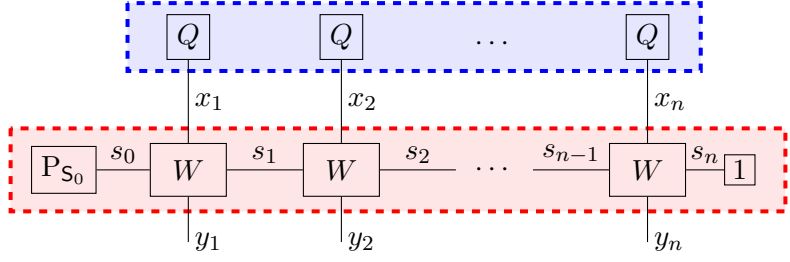
In particular:
\begin{enumerate}[label=\alph*)]
\item \label{closing:the:box} The part of the factor graph inside the {\color{red}bottom} box represents $W\left(\vy_1^n,\vs_1^n|\vx_1^n,s_0\right)$, \ie, the probability of obtaining $\vy_1^n$ and $\vs_1^n$ given $\vx_1^n$ and $s_0$.
	After applying the ``closing-the-box'' operation (see Section~\ref{subsec:marginal:acyclic:FGs}) \wrt the {\color{red}bottom} box, we obtain the joint channel law $W(\vy_1^n|\vx_1^n) \defeq \sum_{\vs_0^n} \prob_{\rv{S}_0}(s_0) \cdot W(\vy_1^n,\vs_1^n|\vx_1^n,s_0) $.
\item The part of the factor graph inside the {\color{blue}top} box represents the input process $Q^{(n)}(\vx_1^n)$.
	Here, for simplicity, the input process is an {i.i.d.}~process characterized by the \pmf $Q$, \ie, $Q^{(n)}(\vx_1^n) = \prod_{\ell=1}^{n} Q(x_{\ell})$.
\item The marginal function $g(\vx_1^n,\vy_1^n) \defeq \sum_{\vs_0^n} g(\vx_1^n,\vy_1^n,\vs_0^n)$is the marginal \pmf of $\vx_1^n$ and $\vy_1^n$.
  The marginal function $g(\vs_0^n) \defeq \sum_{\vx_1^n,\vy_1^n} g(\vx_1^n,\vy_1^n, \vs_0^n)$ is the marginal \pmf of $\vs_0^n$.
  Other marginal \pmfs can be obtained similarly.
\end{enumerate}
Using the ``closing-the-box'' operations, such factor graph representations can be useful in computing a number of quantities of interests.
For example, to prove that~\eqref{eq:FSMC:n:conditional} is indeed a valid conditional \pmf, it suffices to show that
\begin{equation}\label{eq:verfy:FSMC:conditional:distribution}
\sum_{\vs_1^n,\vy_1^n} W(\vy_1^n,\vs_1^n|\cvx_1^n,\cs_0) = 1
	\quad \forall \cvx_1^n\in\set{X}^{n},\,\cs_0\in\set{S},
\end{equation}
which can be verified via a sequence of ``closing-the-box'' operations as shown in Figure~\ref{fig:FMSC:high:level:2} in the Appendix~\ref{app:figures}.
Such techniques are at the heart of the information-rate-estimation methods as in~\cite{arnold2006simulation}.
The details are reviewed in the next subsection.
%*******************************************************************************
\subsection{Information Rate Estimation} \label{QCwM:sec:FSMC:IR}
The approach of~\cite{arnold2006simulation} for estimating information rate of FSMCs, as reviewed in this section, is based on the Shannon--McMillan--Breiman theorem (see \eg,~\cite{cover2012elements}) and suitable generalizations.
We make the following assumptions.
\begin{itemize}
\item As already mentioned, the derivations in this chapter are for the case where the input process $\rv{X} = (\rv{X}_1,\rv{X}_2,\ldots)$ is an {i.i.d.}~process.
	The results can be generalized to other stationary ergodic input processes that can be represented by a finite-state-machine source (FSMS).
	Technically, this is done by defining a new state that combines the source state and the channel state.
\item We assume that the FSMC is indecomposable, which roughly means that in the long term the behavior of the channel is independent of the initial channel state distribution $\prob_{\rv{S}_0}$ (see~\cite[Section~4.6]{gallager1968information} for the exact definition).
	For such channels and stationary ergodic input processes, the information rate $\infoRate_W$ in~\eqref{eq:def:fsmc:ir:1} is well defined.
\end{itemize}
Let $W(\vy_1^n|\vx_1^n)$ be the joint channel law of an FSMC satisfying the assumptions above.
As aforementioned, the information rate of such a channel using the {i.i.d.}~input distribution $\{Q^{(n)}\defeq Q^{\tensor{n}}\}_{n}$ is given by~\eqref{eq:def:fsmc:ir:1}, where the input process $\rv{X}_1^n$ and the output process $\rv{Y}_1^n$ are jointly distributed according to 
\begin{equation} \label{eq:joint:fsmc}
\prob_{\rv{X}_1^n,\rv{Y}_1^n}(\vx_1^n,\vy_1^n) = \prod_{\ell=1}^{n} Q(x_\ell) \cdot W(\vy_1^n|\vx_1^n).
\end{equation}
One can rewrite~\eqref{eq:def:fsmc:ir:1} as
\begin{equation}\label{eq:def:fsmc:ir:2}
    \infoRate_W(Q) = \entropicRate(\rv{X}) + \entropicRate(\rv{Y}) - \entropicRate(\rv{X},\rv{Y}),  
\end{equation}
where the \emph{entropic rates} $\entropicRate(\rv{X})$, $\entropicRate(\rv{Y})$ and $\entropicRate(\rv{X},\rv{Y})$ are defined as
\begin{align}
    \entropicRate(\rv{X}) &\defeq \lim_{n\to\infty} \frac{1}{n} \entropy(\rv{X}_1^n), \\
    \entropicRate(\rv{Y}) &\defeq \lim_{n\to\infty} \frac{1}{n} \entropy(\rv{Y}_1^n), \\
    \entropicRate(\rv{X},\rv{Y}) &\defeq \lim_{n\to\infty} \frac{1}{n} \entropy(\rv{X}_1^n,\rv{Y}_1^n).
\end{align}
\par
%*******************************************************************************
We proceed as in~\cite{arnold2006simulation}.
(For more background information, see the references in~\cite{arnold2006simulation}, in particular~\cite{ephraim2002hidden}.)
Namely, because of~\eqref{eq:def:fsmc:ir:2} and 
\begin{align}
\label{eq:converge:x:1}
-\frac{1}{n}\log{\prob_{\rv{X}_1^n}(\rv{X}_1^n)} &\stackrel{n\to\infty}{\longrightarrow} \entropicRate(\rv{X})
      && \quad \text{in probability},\\
\label{eq:converge:y:1}
-\frac{1}{n}\log{\prob_{\rv{Y}_1^n}(\rv{Y}_1^n)} &\stackrel{n\to\infty}{\longrightarrow} \entropicRate(\rv{Y})
	&& \quad \text{in probability},\\
\label{eq:converge:xy:1}
-\frac{1}{n}\log{\prob_{\rv{X}_1^n,\rv{Y}_1^n}(\rv{X}_1^n,\rv{Y}_1^n)} &\stackrel{n\to\infty}{\longrightarrow} \entropicRate(\rv{X},\rv{Y})
	&& \quad \text{in probability},
\end{align}
by choosing some large number $n$, we have the approximation
\begin{equation}\label{eq:FSMC:ir:estimate:1}
\infoRate_W(Q) \approx -\frac{1}{n}\log{\prob_{\rv{X}_1^n}(\cvx_1^n)} - \frac{1}{n}\log{\prob_{\rv{Y}_1^n}(\cvy_1^n)} + \frac{1}{n}\log{\prob_{\rv{X}_1^n,\rv{Y}_1^n}(\cvx_1^n,\cvy_1^n)}
\end{equation}
where $\cvx_1^n$ and $\cvy_1^n$ are some input and output sequences, respectively, randomly generated according to
\begin{equation}\label{eq:fsmc:joint:xy:distribution}
\prob_{\rv{X}_1^n,\rv{Y}_1^n}(\cvx_1^n,\cvy_1^n) = \sum_{\vs_0^n} \prob_{\rv{S}_0}(s_0)\cdot Q^{(n)}(\cvx_1^n) \cdot W(\cvy_1^n,\vs_1^n|\cvx_1^n,s_0),
\end{equation}
where $W(\cvy_1^n,\vs_1^n|\cvx_1^n,s_0)$ is defined in~\eqref{eq:FSMC:n:conditional}.
Note that $\cvx_1^n$ can be obtained by simulating the input process, and $\cvy_1^n$ can be obtained by simulating the channel for the given input string $\cvx_1^n$.
The latter can be done by keeping track of $\prob_{\rv{Y}_\ell|\rv{X}_1^\ell,\rv{Y}_1^{\ell\!-\!1}}(y_\ell|\cvx_1^\ell,\cvy_1^{\ell\!-\!1})$, which is proportional to $\prob_{\rv{Y}_\ell,\rv{Y}_1^{\ell\!-\!1}|\rv{X}_1^\ell}(y_\ell,\cvy_1^{\ell\!-\!1}|\cvx_1^\ell)$, and can be efficiently calculated by applying suitable ``closing-the-box'' operations as in Figure~\ref{fig:CFSM:channel:simulation:Y} in the Appendix~\ref{app:figures}.
\par 
%*******************************************************************************
We continue by showing how the three terms appearing on the right-hand side of~\eqref{eq:FSMC:ir:estimate:1} can be computed efficiently.
We show it explicitly for the second term, and then outline it for the first and the third term. \par 
In order to efficiently compute the second term on the right-hand side of~\eqref{eq:FSMC:ir:estimate:1}, \ie, $-\frac{1}{n}\log{\prob_{\rv{Y}_1^n}(\cvy_1^n)}$, we consider the \emph{state metric} defined in~\cite{arnold2006simulation} as
\begin{equation}\label{eq:def:classical:channel:state:metric:Y:1}
\muY_{\ell}(s_{\ell}) \defeq \sum_{\vx_1^{\ell}} \sum_{\vs_0^{\ell\!-\!1}} \prob_{\rv{S}_0}(s_0) \cdot Q^{(\ell)}(\vx_1^\ell) \cdot W(\cvy_1^\ell,\vs_1^\ell|\vx_1^\ell,s_0).
\end{equation}
In this case,
\begin{equation} \label{eq:calculate:p:y:1}
\prob_{\rv{Y}_1^n}(\cvy_1^n) = \sum_{s_n} \muY_n(s_n),
\end{equation}
and the calculation of $\muY_\ell(s_{\ell})$ can be done iteratively as
\begin{align}
\muY_\ell(s_{\ell}) & = \sum_{x_{\ell}} \sum_{s_{\ell\!-\!1}} \muY_{\ell\!-\!1}(s_{\ell\!-\!1}) \cdot Q(x_{\ell}|\vx_1^{\ell\!-\!1}) \cdot W(\cy_{\ell},s_{\ell}|x_{\ell},s_{\ell\!-\!1}) \\
\label{eq:recursive:state:metric:Y:1}
& = \sum_{x_{\ell}} \sum_{s_{\ell\!-\!1}} \muY_{\ell\!-\!1}(s_{\ell\!-\!1}) \cdot Q(x_{\ell}) \cdot W(\cy_{\ell},s_{\ell}|x_{\ell},s_{\ell\!-\!1}).
\end{align}
Eq.~\eqref{eq:recursive:state:metric:Y:1} is visualized in Figure~\ref{fig:CFSM:estimate:hY} as applying suitable ``closing-the-box'' operations to the NFG in Figure~\ref{fig:FMSC:high:level:1}. \par 
However, since the value of $\muY_{\ell}(s_\ell)$ tends to zero as $\ell$ grows, such recursive calculations are numerically inconvenient.
A solution is to \emph{normalize} $\muY_{\ell}(s_\ell)$ after each use of~\eqref{eq:recursive:state:metric:Y:1} and to keep track of the scaling coefficients.
Namely,
\begin{equation}\label{eq:recursive:state:metric:Y:2}
\bmuY_{\ell}(s_\ell) \defeq \frac{1}{\lambdaY_\ell} \sum_{x_\ell}\sum_{s_{\ell\!-\!1}} \bmuY_{\ell\!-\!1}(s_\ell) \cdot Q(x_\ell) \cdot W(\cy_\ell,s_\ell|x_\ell,s_{\ell\!-\!1}),
\end{equation}
where the scaling factor $\lambdaY_\ell > 0$ is chosen such that $\sum_{s_\ell} \bmuY_\ell(s_\ell) = 1$.
With this, Eq.~\eqref{eq:calculate:p:y:1} can be rewritten as
\begin{equation} \label{eq:calculate:p:y:2}
\prob_{\rv{Y}_1^n}(\cvy_1^n) = \prod_{\ell=1}^n \lambdaY_{\ell}.
\end{equation}
Finally, we arrive at the following efficient procedure for computing $-\frac{1}{n}\log{\prob_{\rv{Y}_1^n}(\vy_1^n)}$:
\begin{itemize}
\item For $\ell = 1,\ldots, n$, iteratively compute the normalized state
    metric and with that the scaling factors $\lambdaY_{\ell}$.
\item Conclude with the result
	\begin{equation} \label{eq:calculate:p:y:3}
	-\frac{1}{n}\log{\prob_{\rv{Y}_1^n}(\vy_1^n)} = \frac{1}{n} \sum_{\ell=1}^n
    \log(\lambdaY_\ell).
	\end{equation}
\end{itemize} \par 
The third term on the right-hand side of~\eqref{eq:FSMC:ir:estimate:1} can be evaluated by an analogous procedure, where the state metric $\muY_\ell(s_\ell)$ is replaced by the state metric
\begin{equation} \label{eq:def:classical:channel:state:metric:XY:1}
  \muXY_\ell(s_\ell) \defeq \sum_{\vs_0^{\ell\!-\!1}} \prob_{\rv{S}_0}(s_0) \cdot Q^{(\ell)}(\cvx_1^\ell) \cdot W(\cvy_1^\ell,\vs_1^\ell|\cvx_1^\ell).      
\end{equation}
The iterative calculation of $\muXY_{\ell}(s_{\ell})$ is visualized in Figure~\ref{fig:CFSM:estimate:hXY}. \par
Finally, the first term on the right-hand side of~\eqref{eq:FSMC:ir:estimate:1} can be trivially evaluated if $\rv{X}$ is an {i.i.d.}~process, and with a similar approach as above if it is described by an FSMS.
\par
%*******************************************************************************
The above discussion is summarized as Algorithm~\ref{alg:SPA}.
On the side, note that for each $\ell=2,\ldots,n$, the quantities $\lambdaY_\ell$ and $\lambdaXY_\ell$ in the algorithm are the conditional probabilities $\prob_{\rv{Y}_\ell|\rv{Y}_{1}^{\ell\!-\!1}}(\cy_\ell|\cvy_1^{\ell\!-\!1})$ and $\prob_{\rv{X}_\ell\rv{Y}_\ell|\rv{X}_{1}^{\ell\!-\!1}\rv{Y}_{1}^{\ell\!-\!1}}(\cx_{\ell},\cy_\ell|\cvx_1^{\ell\!-\!1},\cvy_1^{\ell\!-\!1})$, respectively.
\begin{algorithm}
\caption{Estimating the Information Rate of an FSMC}
\begin{algorithmic}[1] 
\Require{An indecomposable FSMC channel law $W$, a input distribution $Q$, a positive integer $n$ large enough.}
\Ensure{$\infoRate_W(Q) \approx \entropy(\rv{X}) + \hat\entropicRate(\rv{Y}) - \hat\entropicRate(\rv{X},\rv{Y})$.}
\State Initialize the channel state distribution $\prob_{\rv{S}_0}$ as a uniform distribution over $\set{S}$
\State Generate an input sequence $\cvx_1^n \sim Q^{\tensor n}$
\State Generate a corresponding output sequence $\cvy_1^n$
\State $\bmuY_{0}\gets \prob_{\rv{S}_0}$
\ForEach{$\ell=1,\ldots,n$}
\State $\muY_\ell(s_\ell) \gets \sum_{x_\ell,s_{\ell\!-\!1}} \bmuY_{\ell\!-\!1}(s_{\ell\!-\!1}) \cdot Q(x_\ell) \cdot W(\cy_\ell,s_\ell|x_\ell,s_{\ell\!-\!1})$
\State $\lambdaY_\ell \gets \sum_{s_\ell} \muY_\ell(s_\ell)$;
\State $\bmuY_\ell \gets \muY_\ell / \lambdaY_\ell$
\EndFor
\State $\hat\entropicRate(\rv{Y}) \gets -\frac{1}{n} \sum_{\ell=1}^n \log(\lambdaY_\ell)$
\State $\bmuXY_{0} \gets \prob_{\rv{S}_0}$
\ForEach{$\ell=1,\ldots,n$}
\State $\muXY_\ell(s_\ell) \gets \sum_{s_{\ell\!-\!1}} \bmuXY_{\ell\!-\!1}(s_{\ell\!-\!1}) \cdot Q(\cx_\ell) \cdot W(\cy_\ell,s_\ell|\cx_\ell,s_{\ell\!-\!1})$
\State $\lambdaXY_\ell \gets \sum_{s_\ell} \muXY_\ell(s_\ell)$
\State $\bmuXY_\ell \gets \muXY_\ell/\lambdaXY_\ell$
\EndFor
\State $\hat\entropicRate(\rv{X},\rv{Y})\gets - \frac{1}{n} \sum_{\ell=1}^n \log(\lambdaXY_\ell)$
\State $\entropy(\rv{X}) \gets -\sum_{x} Q(x) \log{Q(x)}$
\State Estimate $\infoRate_W(Q)$ as $\entropy(\rv{X}) + \hat\entropicRate(\rv{Y}) - \hat\entropicRate(\rv{X},\rv{Y})$.
\end{algorithmic}
\label{alg:SPA}
\end{algorithm}
%*******************************************************************************
\subsection{Auxiliary Channels and Bounds on the Information Rate}
\label{QCwM:sec:aux}
As mentioned earlier in this chapter, auxiliary channels\footnote{Technically speaking, an auxiliary channel can be defined as \emph{any} channel with the same input/output alphabet. For example, an auxiliary channel for an FSMC can be just another FSMC with smaller state space; in contrast, in Section~\ref{QCwM:sec:5:UBLB}, an auxiliary channel can also be a quantum-state channel.} are introduced when the state space of the FSMC is too large, making the calculation in Algorithm~\ref{alg:SPA} (practically) intractable.
More precisely, given an auxiliary forward FSMC (AF-FSMC) $\hat{W}(y_\ell,\hat{s}_\ell|x_\ell,\hat{s}_{\ell\!-\!1})$ and an auxiliary backward FSMC (AB-FSMC) $\hat{V}(x_\ell,\hat{s}_\ell|y_\ell,\hat{s}_{\ell\!-\!1})$,  a pair of upper and lower bounds of the information rate is given in~\cite{arnold2006simulation, sadeghi2009optimization} as
\begin{align}
\label{eq:def:IRUB}
\IRUB^{(n)}_{W}(\hat{W}) &\defeq \frac{1}{n} \sum_{\vx_1^n,\vy_1^n}
    Q(\vx_1^n) W(\vy_1^n|\vx_1^n) \log{\frac{W(\vy_1^n|\vx_1^n)}{\QWaux(\vy_1^n)}},\\
\label{eq:def:IRLB}
\IRLB_{W}^{(n)}(\hat{V})  &\defeq \frac{1}{n} \sum_{\vx_1^n,\vy_1^n} Q(\vx_1^n) W(\vy_1^n|\vx_1^n) \log{\frac{\hat{V}(\vx_1^n|\vy_1^n)}{Q(\vx_1^n)}},
\end{align}
where $\QWaux(\vy_1^n) \defeq \sum_{\vx_1^n} Q(\vx_1^n) \cdot \hat{W}(\vy_1^n|\vx_1^n)$.
To see that~\eqref{eq:def:IRUB} and~\eqref{eq:def:IRLB} are, respectively, upper and lower bounds, one can verify that
\begin{align}
\label{eq:IRUB:minus:IR}
\IRUB_{W}(\hat{W}) - \infoRate_{W} &= \frac{1}{n} \infdiv{\QW(\rv{Y}_1^n)}{\QWaux(\rv{y}_1^n)},\\
\label{eq:IR:minus:IRLB}
\infoRate_{W} - \IRLB_{W}(\hat{V}) &= \frac{1}{n} \sum_{\vy_1^n} \QW(\vy_1^n) \cdot \infdiv{V(\rv{X}_1^n|\vy_1^n)}{\hat{V}(\rv{X}_1^n|\vy_1^n)},
\end{align}
where the backward channel $V(\vx|\vy)$ is defined as $V(\vx|\vy)\defeq Q(\vx)W(\vy|\vx)/\QW(\vy)$.
In particular, given an AF-FSMC $\hat{W}$,~\cite{sadeghi2009optimization} considered the induced AB-FSMC $\hat{V}(\vx|\vy)\defeq Q(\vx)\hat{W}(\vy|\vx)/\QWaux(\vy)$.
In this case, 
\begin{equation}
\IRLB_{W}^{(n)}(\hat{V}) = \frac{1}{n} \sum_{\vx_1^n,\vy_1^n} Q(\vx_1^n) W(\vy_1^n|\vx_1^n) \log\frac{\hat{W}(\vy_1^n|\vx_1^n)}{\QWaux(\vy_1^n)}.
\end{equation}
The difference function $\Delta_{W}^{(n)}(\hat{W})$ is defined as
\begin{equation} \label{eq:delta} \begin{aligned}
\Delta_{W}^{(n)}(\hat{W}) &\defeq \IRUB^{(n)}_W(\hat{W}) - \IRLB_{W}^{(n)}(\hat{V})\\
&=\frac{1}{n} \sum_{\vx_1^n,\vy_1^n} Q(\vx_1^n) W(\vy_1^n|\vx_1^n) \log{\frac{W(\vy_1^n|\vx_1^n)}{\hat{W}(\vy_1^n|\vx_1^n)}}\\
&=\frac{1}{n} \infdiv{Q(\rv{X}_1^n) W(\rv{Y}_1^n|\rv{X}_1^n)}{Q(\rv{X}_1^n) \hat{W}(\rv{Y}_1^n|\rv{X}_1^n)}.
\end{aligned}\end{equation}
Apparently, $\Delta_{W}^{(n)}(\hat{W})\geqslant 0$, and equality holds if and only if $\hat{W}(\vy_1^n|\vx_1^n) = W(\vy_1^n|\vx_1^n)$ for all $\vx_1^n$ and $\vy_1^n$ with positive support \wrt $\prob_{\rv{X}_1^n,\rv{Y}_1^n}$ as in~\eqref{eq:joint:fsmc}.
An efficient algorithm for finding a local minimum of the difference function was proposed in~\cite{sadeghi2009optimization}; we refer to~\cite{sadeghi2009optimization} for further details.
%*******************************************************************************
%*******************************************************************************
\section[QCwM \& their Graphical Representation]{Quantum Channel with Memory and their Graphical Representation} \label{QCwM:sec:3:QCM}
In this section, we formalize our notations and modeling of quantum channels with memory~\cite{bowen2004quantum, kretschmann2005quantum, caruso2014quantum} and of classical communications over such channels.
In particular, we will define a class of channels named \emph{quantum-state channels},  which is an alternative description of the classical communications over quantum channels with memory.
In addition, we will introduce several NFGs for representing these channels and processes.
%*******************************************************************************
\subsection{Classical Communication over a Quantum Channel with Memory} \label{QCwM:sec:classical:comm:quantum}
As mentioned earlier in this chapter, we define a quantum channel with memory as follows.
\begin{definition} \label{def:quantum:channel:memory}
A \emph{quantum channel with memory} is a CPTP map:
\begin{equation} \label{eq:def:qcm}
\operator{N}:\DensOp(\hilbert_\system{A}\tensor\hilbert_\system{S}) \to \DensOp(\hilbert_\system{B}\tensor\hilbert_{\system{S}'}),
\end{equation}
where $\system{A}$ is the input system, $\system{B}$ is the output system, $\system{S}$ and $\system{S}'$ are, respectively, the memory systems before and after the channel use.
The Hilbert spaces $\hilbert_\system{A}$, $\hilbert_\system{B}$, and $\hilbert_\system{S} = \hilbert_{\system{S}'}$ are the state spaces corresponding to those systems.
\end{definition}
\par
%*******************************************************************************
We consider classical communication over such channels using some separable input ensemble and local output measurements; namely, the encoder and decoder are, respectively, some classical-to-quantum and quantum-to-classical channels involving a single input or output system.
In particular, given an ensemble $\{\rho^{(x)}_\system{A}\}_{x\in\set{X}}$ and a measurement $\{\Lambda_\system{B}^{(y)}\}_{y\in\set{Y}}$, we define the encoding and decoding function, respectively, as
\begin{align}
\text{Encoding }\mathcal{E}: \prob_\rv{X} &\mapsto \sum_{x\in\set{X}} \prob_\rv{X}(x) \rho^{(x)}_\system{A}
    && \forall\ \prob_\rv{X} \text{ \pmf over }\set{X},\\
\text{Decoding }\mathcal{D}: \sigma_\system{B} &\mapsto \left\{\tr(\Lambda_\system{B}^{(y)} \cdot \sigma_\system{B})\right\}_{y\in\set{Y}}
    && \forall\ \sigma_\system{B}\in\PositiveOp(\hilbert_\system{B}).
\end{align}
We emphasize that in our setup, the ensemble $\{\rho^{(x)}_\system{A}\}_{x\in\set{X}}$ and measurements $\{\Lambda_\system{B}^{(y)}\}_{y\in\set{Y}}$ are given and fixed.
Furthermore, we assume that one does not have access to the memory systems of the channel.
For the case of {i.i.d.}~inputs, the memory system $\system{S}$ before each channel use shall be independent of the input system $\system{A}$, namely, the joint memory-input operator shall take the form of $\rho_\system{A}\tensor\rho_\system{S}$ at each channel input.\footnote{
More generally, for FSMSs, this statement also holds by conditioning on all previous inputs.}
\par
%*******************************************************************************
With this, the probability of receiving $y\in\set{Y}$, given that $x\in\set{X}$  was sent and given that the density operator of the memory system \emph{before} the usage of the channel was $\rho_\system{S}$, equals
\begin{equation}\label{eq:channel:law:1}
\prob_{\rv{Y}|\rv{X};\system{S}}(y|x;\rho_\system{S}) = \tr\left( \Lambda_\system{B}^{(y)} \cdot \tr_{\system{S}'}\left( \operator{N}(\rho^{(x)}_\system{A}\tensor \rho_\system{S}) \right) \right),
\end{equation}
which can also be written as
\begin{equation}\label{eq:channel:law:2}
\prob_{\rv{Y}|\rv{X};\system{S}}(y|x;\rho_\system{S}) = \tr\left( (\Lambda_\system{B}^{(y)}\tensor I_\system{S}) \cdot \operator{N}( \rho^{(x)}_\system{A}\tensor\rho_\system{S}) \right).
\end{equation}
Moreover, assuming that $y$ was observed, the density operator of the
memory system \emph{after} the channel use is given by
\begin{equation} \label{eq:channel:evolution:1}
\rho_{\system{S}'} = \frac{\tr_\system{B}\left( (\Lambda_\system{B}^{(y)} \tensor I_\system{S}) \cdot \operator{N}(\rho^{(x)}_\system{A} \tensor \rho_\system{S})\right)}{\tr\left( (\Lambda_\system{B}^{(y)} \tensor I_\system{S}) \cdot \operator{N}(\rho^{(x)}_\system{A} \tensor \rho_\system{S}) \right)}.
\end{equation}
Notice that the denominator in~\eqref{eq:channel:evolution:1} equals the expressions in~\eqref{eq:channel:law:1} and~\eqref{eq:channel:law:2}.
One should note that, though the input and the memory systems are independent before each channel use (given {i.i.d.}~inputs), the output and the memory systems after each channel use can be correlated or even entangled.
In particular, this translates to the fact that the measurement outcome $y$ can have an influence on the memory system as indicated in~\eqref{eq:channel:evolution:1}.
\par
%*******************************************************************************
Consider using the channel $n$ times consecutively with the above scheme.
The joint channel law, namely the conditional \pmf of the channel outputs $\rv{Y}_1^n$ given the channel inputs $\rv{X}_1^n$ and the initial channel state $\rho_{\system{S}_0}$, can be computed iteratively using~\eqref{eq:channel:law:2} and~\eqref{eq:channel:evolution:1}.
In particular, the joint conditional \pmf can be computed as
\begin{equation}\label{eq:joint:1}
\prob_{\rv{Y}_1^n|\rv{X}_1^n;\system{S}_0}(\vy_1^n|\vx_1^n;\rho_{\system{S}_0}) = \prod_{\ell=1}^n \prob_{\rv{Y}_\ell|\rv{X}_\ell;\system{S}_{\ell\!-\!1}}
                     (y_\ell|x_\ell;\rho_{\system{S}_{\ell\!-\!1}}),
\end{equation}
where we compute the density operators $\{\rho_{\system{S}_\ell}\}_{\ell=1}^n$ iteratively using~\eqref{eq:channel:evolution:1} as
\begin{equation}\label{eq:channel:evolution:2}
\rho_{\system{S}_\ell} = \frac{\tr_\system{B}\left( (\Lambda_\system{B}^{(y_\ell)}\tensor I_\system{S}) \cdot \operator{N}(\rho^{(x_\ell)}_\system{A}\tensor\rho_{\system{S}_{\ell\!-\!1}}) \right)}{\tr\left( (\Lambda_\system{B}^{(y_\ell)}\tensor I_\system{S}) \cdot \operator{N}(\rho^{(x_\ell)}_\system{A} \tensor \rho_{\system{S}_{\ell\!-\!1}}) \right)}.
\end{equation}
%*******************************************************************************
\subsection{Quantum-State Channels} \label{QCwM:sec:QSC}
For each channel-ensemble-measurement configuration ($\operator{N}$, $\{\rho^{(x)}_\system{A}\}_{x\in\set{X}}$, $\{\Lambda_\system{B}^{(y)}\}_{y\in\set{Y}}$) as introduced above, one ends up with a joint conditional \pmf, as  in~\eqref{eq:joint:1}.
However, this relationship is not bijective.
In particular, consider some unitary operators $U_\system{A}$ and $U_\system{B}$ acting on $\hilbert_\system{A}$ and $\hilbert_\system{B}$, respectively.
The following setup induces exactly the same joint conditional \pmf:
\begin{align*}
&\tilde{\operator{N}}: \tilde\rho_\system{AS} \mapsto (U_\system{B} \tensor I_\system{S}) \cdot \operator{N}\left( (U_\system{A} \tensor I_\system{S}) \tilde\rho_\system{AS} (U_\system{A}^\Herm \tensor I_\system{S}) \right) \cdot (U_\system{B}^\Herm \tensor I_\system{S}),\\
&\tilde\rho^{(x)}_\system{A} \defeq U_\system{A}^\Herm \cdot \rho^{(x)}_\system{A} \cdot U_\system{A} &&\quad\forall x\in\set{X},\\
&\tilde\Lambda_\system{B}^{(y)} \defeq U_\system{B}^\Herm \cdot \Lambda_\system{B}^{(y)} \cdot U_\system{B} &&\quad\forall y\in\set{Y}.
\end{align*}
Such redundancy is not only tedious, but also detrimental when we try to compare different channels; in particular, when we try to introduce proper auxiliary channels to approximate the original communication scheme.
\par
%*******************************************************************************
In this subsection, we introduce a class of channels called~\emph{quantum-state channels} to eliminate such redundancies.
In particular, notice that the statistical behavior of the aforementioned communication scheme is fully specified via~\eqref{eq:channel:law:2} and~\eqref{eq:channel:evolution:1}; which are in turn determined by the set of completely positive mappings $\{\operator{N}^{y|x}\}_{x\in\set{X},y\in\set{Y}}$ defined as
\begin{equation}\label{eq:def:qsc}
\operator{N}^{y|x}: \rho_\system{S} \mapsto \tr_\system{B}\left( (\Lambda_\system{B}^{(y)} \tensor I_\system{S}) \cdot \operator{N}(\rho^{(x)}_\system{A} \tensor \rho_\system{S}) \right).
\end{equation}
In this case,~\eqref{eq:channel:law:2},~\eqref{eq:channel:evolution:1}, and~\eqref{eq:joint:1} can be rewritten, respectively, as
\begin{align}
\label{eq:channel:law:3}
\prob_{\rv{Y}|\rv{X};\system{S}}(y|x;\rho_\system{S}) &= \tr\left( \operator{N}^{y|x}(\rho_\system{S}) \right),\\
\label{eq:channel:evolution:3}
\rho_{\system{S}'} &= \operator{N}^{y|x}(\rho_\system{S}) \big/ \tr\left( \operator{N}^{y|x}(\rho_\system{S}) \right),\\
\label{eq:joint:2}
\prob_{\rv{Y}_1^n|\rv{X}_1^n;\system{S}_0}(\vy_1^n|\vx_1^n;\rho_{\system{S}_0} &= \tr\left( \operator{N}^{y_n|x_n} \circ \cdots \circ \operator{N}^{y_1|x_1}(\rho_{\system{S}_0}) \right).
\end{align}
Thus, the operators $\{\operator{N}^{y|x}\}_{x\in\set{X},y\in\set{Y}}$ fully specify the joint conditional \pmf as in~\eqref{eq:joint:2}.
Moreover, such specification is also \emph{unique}; namely, any two sets of channel-ensemble-measurement configuration shall end up with the same joint channel law if and only if the mappings defined in~\eqref{eq:def:qsc} are identical.
This inspires us to make the following definition.
\begin{definition}[Quantum-State Channel] \index{quantum-state channel}
A (finite indexed) set of completely positive operators $\{\operator{N}^{y|x}\}_{x\in\set{X},y\in\set{Y}}$ (acting on the same Hilbert space) is said to be a  \emph{(classical-input classical-output) quantum-state channel} (CC-QSC) if $\sum_{y\in\set{Y}} \operator{N}^{y|x}$ is trace-preserving for each $x \in \set{X}$.
\end{definition}
Given any channel-ensemble-measurement configuration as described in Section~\ref{QCwM:sec:classical:comm:quantum}, one can always define a corresponding CC-QSC by~\eqref{eq:def:qsc}.
On the other hand, as stated in the proposition below, the converse is also true.
%*******************************************************************************
\begin{proposition}\label{prop:quantum:state:channel}
For any CC-QSC $\{\operator{N}^{y|x}\}_{x\in\set{X},y\in\set{Y}}$, there exists some quantum channel with memory $\operator{N}$ as in~\eqref{eq:def:qcm} such that~\eqref{eq:def:qsc} holds with the ensemble $\{\rho_\system{A}^{(x)} = \braket{x}\}_{x\in\set{X}}$ and the measurement $\{\Lambda_\system{B}^{(y)} = \braket{y}\}_{y\in\set{Y}}$.
Here, $\hilbert_\system{A}$ and $\hilbert_\system{B}$ are defined such that $\{\bra{x}\}_x$ and $\{\bra{y}\}_y$ are orthonormal bases of $\hilbert_\system{A}$ and $\hilbert_\system{B}$, respectively.
\end{proposition}
\begin{proof}
It suffices to show that there exists a CPTP map $\operator{N}: \DensOp(\hilbert_\system{A} \tensor \hilbert_\system{S}) \to \DensOp(\hilbert_\system{B} \tensor \hilbert_\system{S})$ such that for all $\rho_\system{S} \in \DensOp(\hilbert_\system{S})$, and $x\in\set{X}$,
\[
\operator{N}: \braket{x}\tensor\rho_\system{S} \mapsto \sum_{y\in\set{Y}} \braket{y}\tensor\operator{N}^{y|x}(\rho_\system{S}).
\]
Such an $\operator{N}$ can be constructed as
\[
\operator{N}: \rho \mapsto \sum_{x,y,k}\left(\bra{y}\!\ket{x} \otimes E^{y|x}_k\right) \cdot \rho \cdot \left( \bra{y}\!\ket{x} \otimes E^{y|x}_k\right)^{\Herm},
\]
where $\left\{E^{y|x}_k\right\}_{k}$ is a Kraus representation of $\operator{N}^{y|x}$.
It remains to check if $\operator{N}$ is a CPTP, which is indeed the case:
\begin{align*}
\sum_{x,y,k} \left(\bra{y}\!\ket{x} \otimes E^{y|x}_k\right)^{\Herm} \cdot \left(\bra{y}\!\ket{x} \otimes E^{y|x}_k\right)
&= \sum_x \sum_{y,k} \braket{x} \otimes (E^{y|x}_k)^\Herm E^{y|x}_k \\
&= \sum_x \braket{x} \otimes I = I. \qedhere
\end{align*}
\end{proof}
%*******************************************************************************
\subsection{Visualization using Normal Factor Graphs} \label{QCwM:subsec:NFGs}
In this subsection, we focus on the computations of~\eqref{eq:channel:law:3},~\eqref{eq:channel:evolution:3}, and~\eqref{eq:joint:2} for the situation where the involved channel $\operator{N}$ is of finite dimension.
In analogy to the FSMCs, we demonstrate how to use NFGs to facilitate and visualize the relevant computations (see Section~\ref{subsec:FG:QP} and~\cite{loeliger2017factor}).
\par
%*******************************************************************************
By Proposition~\ref{prop:quantum:state:channel}, let us consider a CC-QSC $\{\operator{N}^{y|x}\}_{x\in\set{X},y\in\set{Y}}$ acting on $\hilbert_\system{S}$, where $d=\dim(\hilbert_\system{S})$ is finite, and $\{\bra{s}\}_{s\in\set{S}}$ is an orthonormal basis of $\hilbert_\system{S}$.
(Apparently, $\size{\set{S}}=d$.)
Since for each $x$ and $y$, $\operator{N}^{y|x}$ is a completely positive map, there must exist finitely many (\emph{not} necessarily unique) matrices $\{ F_k^{y|x} \in \mathbb{C}^{\set{S}\times\set{S}} \}_{k}$ such that
\begin{equation} \label{eq:Kraus:quantum:state:channel}
\bigl[\operator{N}^{y|x}(\rho_\system{S})\bigr] = \sum_{k} F_k^{y|x} \cdot [\rho_\system{S}] \cdot (F_k^{y|x})^\Herm
	\quad \forall \rho_\system{S} \in \DensOp(\hilbert_\system{S}),
\end{equation}
where $\bigl[\operator{N}^{y|x}(\rho_\system{S})\bigr]$ and $[\rho_\system{S}]$ are, respectively, the matrix representation of the operator $\operator{N}^{y|x}(\rho_\system{S})$ and $\rho_\system{S}$ under $\{\bra{s}\}_{s\in\set{S}}$.
The reason for such matrices $\{F_k^{y|x}\}_{k}$ to exist is the same as for the Kraus operators of CPTP maps (see Theorem~\ref{thm:Kraus}).
Also note that $\sum_{y\in\set{Y}} \mathcal{E}^{y|x}$ is trace-preserving, thus it must hold that
\begin{equation} \label{eq:operator:sum:representation:condition:1}
\sum_{y\in\set{Y}}\sum_{k} (F_k^{y|x})^\Herm F_k^{y|x} = I
	\quad \forall x \in\set{X}.
\end{equation}
Now, define a set of functions $\{W^{y|x}\}_{x\in\set{X},y\in\set{Y}}$ as
\begin{equation}\label{eq:def:channel:function:representation}
W^{y|x}:(s',s,\tilde{s}',\tilde{s}) \mapsto \sum_k F_k^{y|x}(s',s) \conj{F_k^{y|x}(\tilde{s}',\tilde{s})},
\end{equation}
where $s',s,\tilde{s}',\tilde{s}\in\set{S}$ are indices of the corresponding matrices, namely, $F_k^{y|x}(s',s)$ is the $(s',s)$-th entry of the matrix $F_k^{y|x}$.
In this case, one can rewrite~\eqref{eq:channel:law:3},~\eqref{eq:channel:evolution:3} and~\eqref{eq:joint:2}, respectively, into
\begin{align}
\label{eq:channel:law:4}
\prob_{\rv{Y}|\rv{X};\system{S}}(y|x;\rho_\system{S})
&= \sum_{s',\tilde{s}':\:\atop s'=\tilde{s}'} \sum_{s,\tilde{s}} W^{y|x}(s',s,\tilde{s}',\tilde{s})\cdot[\rho_\system{S}]_{s,\tilde{s}}, \\
\label{eq:channel:evolution:4}
[\rho_{\system{S}'}]_{s',\tilde{s}'}
&= \frac{\sum_{s,\tilde{s}} W^{y|x}(s',s,\tilde{s}',\tilde{s}) \cdot [\rho_\system{S}]_{s,\tilde{s}}}{\sum_{s',\tilde{s}':\:\atop s'=\tilde{s}'}\sum_{s,\tilde{s}} W^{y|x}(s',s,\tilde{s}',\tilde{s}) \cdot [\rho_\system{S}]_{s,\tilde{s}}}, \\
\label{eq:joint:3}
\prob_{\rv{Y}_1^n|\rv{X}_1^n;\system{S}_0}(\vy_1^n|\vx_1^n;\rho_{\system{S}_0})
&= \sum_{s_n,\tilde{s}_n:\:\atop s_n=\tilde{s}_n} \sum_{\vs_0^{n\!-\!1},\tilde{\vs}_0^{n\!-\!1}} [\rho_{\system{S}_0}]_{s_0,\tilde{s}_0} \cdot \prod_{\ell=1}^n W^{y_\ell|x_\ell}(s_\ell,s_{\ell\!-\!1},\tilde{s}_\ell,\tilde{s}_{\ell\!-\!1}).
\end{align}
By rearranging the entries of $W^{y|x}$ (for each $x,y$) into a matrix $[W^{y|x}]\in\mathbb{C}^{\set{S}^2\times\set{S}^2}$ as
\begin{equation}\label{eq:qsc:matrix:1}
[W^{y|x}]_{(s',\tilde{s}'),(s,\tilde{s})} \defeq W^{y|x}(s',s,\tilde{s}',\tilde{s}),
\end{equation}
where $(s',\tilde{s}')\in\set{S}^2$ is the first index, and $(s,\tilde{s})\in\set{S}^2$ is the second index of $[W^{y|x}]$, we can simplify~\eqref{eq:channel:law:4},~\eqref{eq:channel:evolution:4}, and~\eqref{eq:joint:3} as
\begin{align}
\label{eq:channel:law:5}
\prob_{\rv{Y}|\rv{X};\system{S}}(y|x;\rho_\system{S})
&= \tr([W^{y|x}] \cdot [\rho_\system{S}]),\\
\label{eq:channel:evolution:5}
[\rho_{\system{S}'}]
&= \frac{[W^{y|x}] \cdot [\rho_\system{S}]}{\tr([W^{y|x}] \cdot [\rho_\system{S}])},\\
\label{eq:joint:4}
\prob_{\rv{Y}_1^n|\rv{X}_1^n;\system{S}_0}(\vy_1^n|\vx_1^n;\rho_{\system{S}_0})
&= \tr\left( [W^{y_{n}|x_{n}}] \cdots [W^{y_{1}|x_{1}}] \cdot [\rho_{\system{S}_0}] \right),
\end{align}
respectively.
Here we treat $[\rho_\system{S}]$ as a length-$d^2$ vector indexed by $(s,\tilde{s})\in\set{S}^2$ in the above equations.
\par
%*******************************************************************************
By considering $\{W^{y|x}\}_{x,y}$ as a function of six variables, we can represent it using a factor vertex of degree six in an NFG as in Figure~\ref{fig:NFG:QSC:single}.
In this case, Eqs.~\eqref{eq:channel:law:4} and~\eqref{eq:channel:law:5} can be visualized as ``closing the {\color{red} outer} box'' in the factor graph.
Similarly,~\eqref{eq:channel:evolution:4} and~\eqref{eq:channel:evolution:5} can be visualized as ``closing the {\color{blue}inner} box''.
The factor graph corresponding to using the channel $n$~times consecutively is depicted in Figure~\ref{fig:NFG:QSC:multiple}, where~\eqref{eq:joint:3} and~\eqref{eq:joint:4} are visualized as closing the {\color{red}outermost} box.
Interestingly, this ``closing-the-box'' operation can be carried out by a sequence of simpler ``closing-the-box'' operations as shown in the figure.
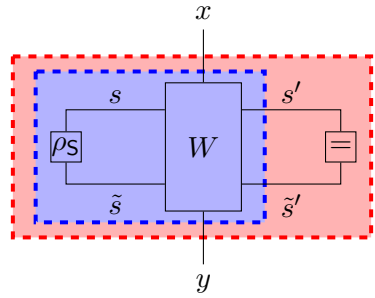
\begin{figure}\centering
\begin{tikzpicture}[node/.style={draw=none},
                    factor/.style={rectangle, minimum width=1cm, minimum height=1.7cm, draw},
                    sfactor/.style={rectangle, minimum size=.4cm, draw}]
    \node[factor] (W) {$W$};
    \node[above left=-.5cm and 1.2cm of W] (sp) {};
    \node[below left=-.5cm and 1.2cm of W] (sp') {};
    \node[above right=-.5cm and 1.2cm of W] (s) {};
    \node[below right=-.5cm and 1.2cm of W] (s') {};
    \node[above=.7cm of W] (X) {$x$};
    \node[below=.7cm of W] (Y) {$y$};
    \node[sfactor, inner sep=0pt, midway] (EW) [right=1.6cm of W] {$=$};
    \node[sfactor, inner sep=0pt, midway] (rho) [left=1.6cm of W] {};
    \node at (rho) {$\rho_\system{S}$};    
    \draw (sp.east-|W.west) -| (rho) node[above=-0.05cm,pos=0.25] {$s$};
    \draw (sp'.east-|W.west) -| (rho) node[below,pos=0.25] {$\tilde{s}$};
    \draw (sp.east-|W.east) -| (EW) node[above=-0.05cm,pos=0.25] {$s'$};
    \draw (sp'.east-|W.east) -| (EW) node[below,pos=0.25] {$\tilde{s}'$};
    \draw (X) -- (W);
    \draw (Y) -- (W);
    \begin{pgfonlayer}{bg}
    \draw[dashed, red, line width=1.5pt,fill= red!30]
        ([xshift=-.7cm,yshift=1.2cm]rho) rectangle
        ([xshift=.4cm,yshift=-1.2cm]EW);
    \draw[dashed, blue, line width=1.5pt,fill= blue!30]
        ([xshift=-.4cm,yshift=1cm]rho) rectangle
        ([xshift=-1cm,yshift=-1cm]EW);
    \end{pgfonlayer}
\end{tikzpicture}
\caption{Representation of $\{W^{y|x}\}_{x,y}$ using an NFG.}
\label{fig:NFG:QSC:single}
\end{figure}
\begin{figure}\centering
\begin{tikzpicture}[
    factor/.style ={rectangle, minimum width=1cm, minimum height=1.7cm, draw},
    sfactor/.style={rectangle, minimum size=.4cm, draw},
    label/.style={anchor=south east, circle, draw, inner sep=.5pt, 
        outer sep=5pt, font=\scriptsize}]
	\node[sfactor] (S) {}; \node at (S) {\resizebox{.375cm}{!}{$\rho_{\system{S}_0}$}};
	\node[factor] (E1) [right=.5cm of S] {$W$};
	\draw (S.north) |- ([yshift=.6cm]E1.west) node[above=-0.05cm,pos=.75] {$s_0$};
	\draw (S.south) |- ([yshift=-.6cm]E1.west) node[below,pos=.75] {$\tilde{s}_0$};
	\draw (E1.north) -- ([yshift=1cm]E1.north) node[right] {$x_1$};
	\draw (E1.south) -- ([yshift=-1cm]E1.south) node[right] {$y_1$};
	\node[factor] (E2) [right=.6cm of E1] {$W$};
	\draw ([yshift=.6cm]E1.east) |- ([yshift=.6cm]E2.west) node[above=-0.05cm,pos=.75] {$s_1$};
	\draw ([yshift=-.6cm]E1.east) |- ([yshift=-.6cm]E2.west) node[below,pos=.75] {$\tilde{s}_1$};
	\draw (E2.north) -- ([yshift=1cm]E2.north) node[right] {$x_2$};
	\draw (E2.south) -- ([yshift=-1cm]E2.south) node[right] {$y_2$};
	\node[factor, draw=none] (Edummy) [right=.6cm of E2] {$\cdots$};
	\draw ([yshift=.6cm]E2.east) |- ([yshift=.6cm]Edummy.west) node[above=-0.05cm,pos=.75] {$s_2$};
	\draw ([yshift=-.6cm]E2.east) |- ([yshift=-.6cm]Edummy.west) node[below,pos=.75] {$\tilde{s}_2$};
	\node[factor] (En) [right=.6cm of Edummy] {$W$};
	\draw ([yshift=.6cm]Edummy.east) |- ([yshift=.6cm]En.west) node[above=-0.05cm,pos=.75] {$s_{\!n-\!1}$};
	\draw ([yshift=-.6cm]Edummy.east) |- ([yshift=-.6cm]En.west) node[below,pos=.75] {$\tilde{s}_{n\!-\!1}$};
	\draw (En.north) -- ([yshift=1cm]En.north) node[right] {$x_n$};
	\draw (En.south) -- ([yshift=-1cm]En.south) node[right] {$y_n$};
	\node[sfactor, draw=none] (Edummy1) [above=-0.50cm of Edummy] {$\cdots$};
	\node[sfactor, draw=none] (Edummy2) [below=-0.45cm of Edummy] {$\cdots$};
	\node[sfactor,draw=none] (Xdummy) at ([yshift=1cm]E1.north-|Edummy) {$\cdots$};
	\node[sfactor,draw=none] (Ydummy) at ([yshift=-1cm]E1.south-|Edummy) {$\cdots$};
	\node[sfactor, inner sep=0pt] (ee) [right=.5cm of En] {$=$};
	\draw ([yshift=.6cm]En.east) -| (ee.north) node[above=-0.05cm,pos=.25] {$s_n$};
	\draw ([yshift=-.6cm]En.east) -| (ee.south) node[below,pos=.25] {$\tilde{s}_n$};
	\begin{pgfonlayer}{bg}
		\draw[dashed, red,line width=1.5pt,fill= red!10] ([xshift=-.7cm,yshift=1.5cm]S) rectangle ([xshift=1.5cm,yshift=-1.6cm]En);
		\draw[dashed, red,line width=1.5pt,fill= red!20] ([xshift=-.6cm,yshift=1.4cm]S) rectangle	([xshift=.6cm,yshift=-1.5cm]En);
		\node[label] at ([xshift=.6cm,yshift=-1.5cm]En) {$n$};
		\draw[dashed, red,line width=1.5pt,fill= red!30] ([xshift=-.5cm,yshift=1.3cm]S) rectangle ([xshift=.6cm,yshift=-1.4cm]E2);
		\node[label] at ([xshift=.6cm,yshift=-1.4cm]E2) {2};
		\draw[dashed, red,line width=1.5pt,fill= red!50] ([xshift=-.4cm,yshift=1.2cm]S) rectangle ([xshift=.6cm,yshift=-1.3cm]E1);
		\node[label] at ([xshift=.6cm,yshift=-1.3cm]E1) {1};
	\end{pgfonlayer}
\end{tikzpicture}
\caption[The joint channel law~\eqref{eq:joint:3} and~\eqref{eq:joint:4} can be visualized as the result of the ``closing of the outermost box'' above, which can in turn be carried out by a sequence of ``closing-the-box'' operations as indicated.]{The joint channel law~\eqref{eq:joint:3} and~\eqref{eq:joint:4} can be visualized as the result of the ``closing of the {\color{blue}outermost} box'' above, which can in turn be carried out by a sequence of ``closing-the-box'' operations as indicated.}
\label{fig:NFG:QSC:multiple}
\end{figure}
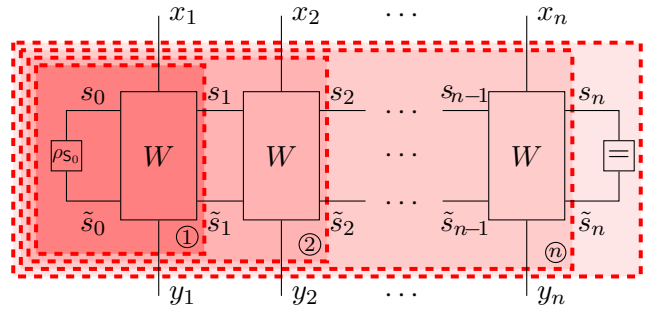
\par
%*******************************************************************************
A number of statistical quantities and density operators of interest can be computed and visualized as ``closing-the-box'' operations on suitable NFGs similar to that of Figure~\ref{fig:NFG:QSC:multiple}.
The following example highlights how quantities of this kind can be computed in such a manner.
\begin{figure}[t]\centering
\resizebox{\columnwidth}{!}{\begin{tikzpicture}[
    node/.style={draw=none},
    factor/.style={rectangle, minimum width=1cm, minimum height=1.7cm, draw},
    sfactor/.style={rectangle, minimum size=.5cm, draw},
    darksolid/.style={rectangle, minimum size=.15cm, draw, fill = black,
    inner sep=0pt, outer sep = 0pt}]
	\node[sfactor] (S) {}; \node at (S) {$\rho_{\system{S}_0}$};
	\node[factor] (E1) [right=.7cm of S] {$W$};
	\draw (S.north) |- ([yshift=.6cm]E1.west) node[above=-0.05cm,pos=.75] {$s_0$};
	\draw (S.south) |- ([yshift=-.6cm]E1.west) node[below,pos=.75] {$\tilde{s}_0$};
	\node[sfactor] (X1) [above=.7cm of E1] {}; \node at (X1) {$Q$};
	\draw (X1) -- (E1) node[right, midway] {$x_1$};
	\draw (E1.south) -- ([yshift=-.8cm]E1.south) node[darksolid]{} node (Y1) [right] {$\cy_1$};
	\node[factor, draw=none] (Edummy1) [right=1cm of E1] {};
	\node at ([yshift=.6cm]E1.east-|Edummy1) {$\cdots$};
	\node at ([yshift=-.6cm]E1.east-|Edummy1) {$\cdots$};
	\draw ([yshift=.6cm]E1.east) |- ([yshift=.6cm]Edummy1.west) node[above=-0.05cm,pos=.75] {$s_1$};
	\draw ([yshift=-.6cm]E1.east) |- ([yshift=-.6cm]Edummy1.west) node[below,pos=.75] {$\tilde{s}_1$};
	\node at (X1-|Edummy1) {$\cdots$};
	\node at (Y1-|Edummy1) {$\cdots$};
	\node[factor] (El0) [right=1cm of Edummy1] {$W$};
	\draw ([yshift=.6cm]Edummy1.east) |- ([yshift=.6cm]El0.west) node[above=-0.05cm,pos=.75] {$s_{\ell-2}$};
	\draw ([yshift=-.6cm]Edummy1.east) |- ([yshift=-.6cm]El0.west) node[below,pos=.75] {$\tilde{s}_{\ell-2}$};
	\node[sfactor] (Xl0) [above=.7cm of El0] {}; \node at (Xl0) {$Q$};
	\draw (Xl0) -- (El0) node[right, midway] {$x_{\ell\!-\!1}$};
	\draw (El0.south) -- ([yshift=-.8cm]El0.south) node[darksolid]{} node[right] {$\cy_{\ell\!-\!1}$};
	\node[factor] (El) [right=1cm of El0] {$W$};
	\draw ([yshift=.6cm]El0.east) |- ([yshift=.6cm]El.west) node[above=-0.05cm,pos=.85] {$s_{\ell\!-\!1}$};
	\draw ([yshift=-.6cm]El0.east) |- ([yshift=-.6cm]El.west) node[below,pos=.85] {$\tilde{s}_{\ell\!-\!1}$};
	\node[sfactor] (Xl) [above=.7cm of El] {}; \node at (Xl) {$Q$};
	\draw (El.north) -- (Xl.south) node[darksolid,midway,red,fill=red]{} node[right,midway] {$x_{\ell}$};
	\draw (El.south) -- ([yshift=-.8cm]El.south) node[darksolid]{} node[right] {$\cy_{\ell}$};
	\node[factor] (El2) [right=1cm of El] {$W$};
	\draw ([yshift=.6cm]El.east) |- ([yshift=.6cm]El2.west) node[above=-0.05cm,pos=.75] {$s_\ell$};
	\draw ([yshift=-.6cm]El.east) |- ([yshift=-.6cm]El2.west) node[below,pos=.75] {$s_\ell'$};
	\node[sfactor] (Xl2) [above=.7cm of El2] {}; \node at (Xl2) {$Q$};
	\draw (Xl2) -- (El2) node[right, midway] {$x_{\ell\!+\!1}$};
	\draw (El2.south) -- ([yshift=-.8cm]El2.south) node[darksolid]{} node[right] {$\cy_{\ell\!+\!1}$};
	\node[factor, draw=none] (Edummy2) [right=1cm of El2] {};
	\node at ([yshift=.6cm]El2.east-|Edummy2) {$\cdots$};
	\node at ([yshift=-.6cm]El2.east-|Edummy2) {$\cdots$};
	\draw ([yshift=.6cm]El2.east) |- ([yshift=.6cm]Edummy2.west) node[above=-0.05cm,pos=.75] {$s_{\ell\!+\!1}$};
	\draw ([yshift=-.6cm]El2.east) |- ([yshift=-.6cm]Edummy2.west) node[below,pos=.75] {$s_{\ell\!+\!1}'$};
	\node at (X1-|Edummy2) {$\cdots$};
	\node at (Y1-|Edummy2) {$\cdots$};
	\node[factor] (En) [right=1cm of Edummy2] {$W$};
	\draw ([yshift=.6cm]Edummy2.east) |- ([yshift=.6cm]En.west) node[above=-0.05cm,pos=.75] {$s_{n\!-\!1}$};
	\draw ([yshift=-.6cm]Edummy2.east) |- ([yshift=-.6cm]En.west) node[below,pos=.75] {$s_{n\!-\!1}'$};
	\node[sfactor] (Xn) [above=.7cm of En] {}; \node at (Xn) {$Q$};
	\draw (Xn) -- (En) node[right, midway] {$x_n$};
	\draw (En.south) -- ([yshift=-.8cm]En.south) node[darksolid]{} node[right] {$\cy_n$};
	\node[sfactor, inner sep=0pt] (ee) [right=.5cm of En] {$=$};
	\draw ([yshift=.6cm]En.east) -| (ee.north) node[above=-0.05cm,pos=.25] {$s_n$};
	\draw ([yshift=-.6cm]En.east) -| (ee.south) node[below,pos=.25] {$s_n'$};
	\begin{pgfonlayer}{bg}
		\draw[dashed, black, line width=1.5pt, fill=yellow!20] ([xshift=-.8cm,yshift=2.4cm]S) rectangle ([xshift=.8cm,yshift=-2.3cm]ee);
		\draw[dashed, blue, line width=1.5pt, fill=blue!20] ([xshift=-.7cm,yshift=2.2cm]El2) rectangle ([xshift=.6cm,yshift=-2.1cm]ee);
		\node[anchor=south east, fill= blue!50] at ([xshift=.6cm,yshift=-2.1cm]ee) {$\lvec{\sigma}_{\system{S}_{\ell}}^{(\cvy_{\ell\!+\!1}^n)}$};
		\draw[dashed, red, line width=1.5pt, fill=red!20] ([xshift=-.6cm,yshift=2.2cm]S) rectangle ([xshift=.8cm,yshift=-2.1cm]El0);
		\node at ([xshift=-.6cm,yshift=2.2cm]S) (tempA) {};
		\node at ([xshift=.8cm,yshift=-2.1cm]El0) (tempB) {};
		\node[anchor=south west,fill=red!50] at (tempA |- tempB) {$\rvec{\sigma}_{\system{S}_{\ell\!-\!1}}^{(\cvy_1^{\ell\!-\!1})}$};
	\end{pgfonlayer}
\end{tikzpicture}}
\caption{Computation of the marginal \pmf $\prob_{\rv{X}_\ell|\rv{Y}_1^n;\system{S}_0}$ using a sequence of ``closing-the-box'' operations.}
\label{fig:qsc:bcjr}
\end{figure}
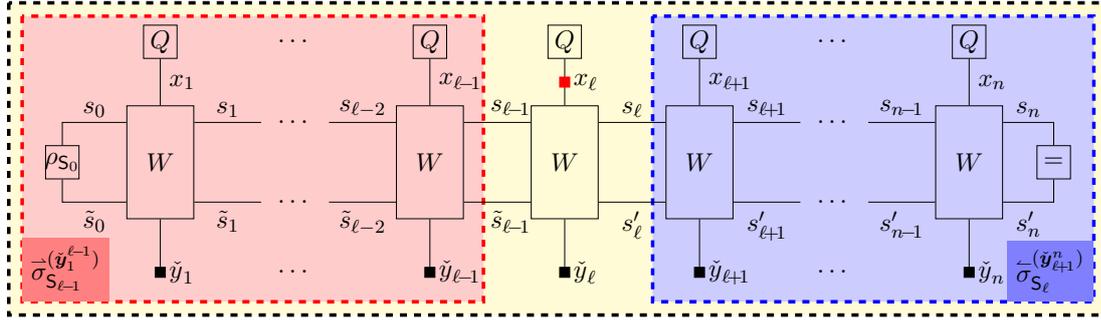
%*******************************************************************************
\begin{example}[BCJR~\cite{bahl1974optimal} decoding for CC-QSCs]
For fixed $\cvy_1^n\in\set{Y}^n$ and a given initial density operator $\rho_{\system{S}_0}$, the conditional probability $\prob_{\rv{X}_\ell|\rv{Y}_1^n;\system{S}_0}(x_\ell|\cvy_1^n;\rho_{\system{S}_0})$ can be computed via
\begin{equation}\label{eq:qsc:bcjr:1}
\prob_{\rv{X}_\ell|\rv{Y}_1^n;\system{S}_0}(\cdot|\cvy_1^n;\rho_{\system{S}_0}) \propto \prob_{\rv{X}_\ell,\rv{Y}_1^n|\system{S}_0}(\cdot,\cvy_1^n|\rho_{\system{S}_0}),
\end{equation}
where the right-hand side of~\eqref{eq:qsc:bcjr:1} is a marginal \pmf defined as
\begin{equation}\label{eq:qsc:bcjr:2}
\prob_{\rv{X}_\ell,\rv{Y}_1^n|\system{S}_0}(x_\ell,\cvy_1^n|\rho_{\system{S}_0}) = \sum_{\vx_{1}^{\ell\!-\!1},\vx_{\ell\!+\!1}^{n}} \sum_{\vs_0^n,\tilde{\vs}_0^n} [\rho_{\system{S}_0}]_{s_0,\tilde{s}_0} \cdot \prod_{i=1}^n Q(x_i) \cdot \prod_{j=1}^n W^{\cy_j|x_j}(s_j,s_{j\!-\!1},\tilde{s}_j,\tilde{s}_{j\!-\!1}),
\end{equation}
where we have assumed that the input process $\rv{X}_1^n$ is {i.i.d.}~characterized by some \pmf $Q$.
The evaluation of~\eqref{eq:qsc:bcjr:2} can be carried out efficiently using a sequence of ``closing-the-box'' operations as visualized in Figure~\ref{fig:qsc:bcjr}.
These operations can be roughly divided into the following three steps:
\begin{enumerate}
\item Closing the {\color{red}left inner} box: this results in an operator $\rvec{\sigma}_{\system{S}_{\ell\!-\!1}}^{(\cvy_1^{\ell\!-\!1})}$ on $\hilbert_{\system{S}_{\ell\!-\!1}}$.
\item Closing the {\color{blue}right inner} box: this results in another operator $\lvec{\sigma}_{\system{S}_{\ell}}^{(\cvy_{\ell\!+\!1}^n)}$ on $\hilbert_{\system{S}_{\ell}}$.
\item Applying the ``closing-the-box'' operation to the yellow box: the result is the marginal \pmf $\prob_{\rv{X}_\ell,\rv{Y}_1^n|\system{S}_0}(x_\ell,\cvy_1^n|\rho_{\system{S}_0})$, from which the desired conditional probability $\prob_{\rv{X}_\ell|\rv{Y}_1^n;\system{S}_0}$ $(x_\ell|\cvy_1^n;\rho_{\system{S}_0})$ can be easily obtained by normalization.
\end{enumerate}
The operators mentioned in 1) and 2) can be computed recursively, using a sequence of ``closing-the-box'' operations.
Namely, one can carry out the computations in 1) consecutively with $\ell=1,2,\ldots,n$; and the computations in 2) consecutively with $\ell=n,n\!-\!1,\ldots,1$.
This provides an efficient way to evaluate $\prob_{\rv{X}_\ell|\rv{Y}_1^n;\system{S}_0}(x_\ell|\cvy_1^n;\rho_{\system{S}_0})$ for each $\ell=1,\ldots,n$; and thus provides an efficient symbol-wise decoding algorithm.
The idea in this example is conceptually identical to that of the BCJR decoding algorithm for an FSMC.
\end{example}
As shown in the above example, very often the desired functions or quantities are based on the same partial results.
The NFG framework is very helpful to visualize these partial results and to show how they can be combined to obtain the desired functions and quantities.
\par
%*******************************************************************************
We emphasize that the functions $\{W^{y|x}\}_{x,y}$ defined in~\eqref{eq:def:channel:function:representation} are unique for a given finite-dimensional CC-QSC $\{\operator{N}^{y|x}\}_{x,y}$; even though such uniqueness does not apply to the Kraus operators $\{F^{y|x}\}_k$ being used to define $\{W^{y|x}\}_{x,y}$.
This can be proven by making the identification that
\begin{equation}\label{eq:qsc:matrix:identification}
\bigl[\operator{N}^{y|x}(\rho_\system{S})\bigr] = [W^{y|x}]\cdot [\rho_\system{S}]
	\qquad \forall \rho_\system{S}\in\DensOp(\hilbert_\system{S})
\end{equation}
for all $x$ and $y$.
Moreover, we argue that the functions $\{W^{y|x}\}_{x,y}$, are an \emph{equivalent} way to specify a CC-QSC, or classical communication over a quantum channel with memory as described at the beginning of this section.
Namely, for any set of complex-valued functions $\{W^{y|x}\}_{x,y}$ on $\set{S}^{4}$ satisfying some constraints to be clarified later, there must exist a unique CC-QSC $\{\operator{N}^{y|x}\}_{x,y}$ such that~\eqref{eq:qsc:matrix:identification} holds; and thus, there must exist some corresponding channel-ensemble-measurement configuration, unique up to its channel law.
As for such constraints, we rearrange the entries of $W^{y|x}$ (for each $x,y$) into another matrix $\llbracket W^{y|x}\rrbracket \in \mathbb{C}^{\set{S}^2 \times \set{S}^2}$ (\aka Choi--Jamio{\l}kowski matrix~\cite{jamiolkowski1972linear}), whose entries are defined as
\begin{equation}\label{eq:qsc:matrix:2}
\llbracket W^{y|x}\rrbracket_{(s',s),(\tilde{s}',\tilde{s})} \defeq W^{y|x}(s',s,\tilde{s}',\tilde{s}),
\end{equation}
where $(s',s)\in\set{S}^2$ is the first index, and $(\tilde{s}',\tilde{s})\in\set{S}^2$ is the second index of $\llbracket W^{y|x}\rrbracket$.
Notice that, $\llbracket W^{y|x}\rrbracket$ is a PSD matrix, and satisfies the equation
\begin{equation} \label{eq:operator:sum:representation:condition:2}
\sum_{y\in\set{Y}} \sum_{s',\tilde{s}':\: s'=\tilde{s}'} \llbracket W^{y|x}\rrbracket_{(s',s),(\tilde{s}',\tilde{s})} = \delta_{s,\tilde{s}}
	\quad \forall x\in\set{X}.
\end{equation}
In this case, the ``equivalence'' can be shown by the following proposition.
\begin{proposition} \label{prop:CFR}
Let $\set{X}$, $\set{Y}$ be finite sets, and $\hilbert_\system{S}$ be a finite-dimensional Hilbert space with an orthonormal basis $\{\bra{s}\}_{s\in\set{S}}$.
For any set of functions
\[
\{W^{y|x}:\set{S}\times\set{S}\times\set{S}\times\set{S} \to \mathbb{C}\}_{x\in\set{X},y\in\set{Y}}
\]
such that their matrix form $\{\llbracket W^{y|x}\rrbracket\}_{x,y}$ consists of PSD matrices and satisfies~\eqref{eq:operator:sum:representation:condition:2}, there must exist a unique CC-QSC $\{\operator{N}^{y|x}\}_{x,y}$ acting on $\hilbert_\system{S}$ such that~\eqref{eq:qsc:matrix:identification} holds.
\end{proposition}
\begin{proof}
The idea of the proof is to consider the eigenvalue decomposition of $\llbracket W^{y|x}\rrbracket$, and reconstruct $\operator{N}^{y|x}$ by following the equations~\eqref{eq:def:channel:function:representation} and~\eqref{eq:Kraus:quantum:state:channel} backwardly.
We omit the details here.
\end{proof}
\par
%*******************************************************************************
Let us conclude this section by pointing out that the functions $\{W^{y|x}\}_{x,y}$, particularly the corresponding NFG, can be constructed from the channel-ensemble-measurement configuration ($\operator{N}$, $\{\rho^{(x)}_\system{A}\}_{x\in\set{X}}$, $\{\Lambda_\system{B}^{(y)}\}_{y\in\set{Y}}$) as in Figure~\ref{fig:CFR}.
This can be justified by checking~\eqref{eq:def:qsc} and~\eqref{eq:qsc:matrix:identification}.
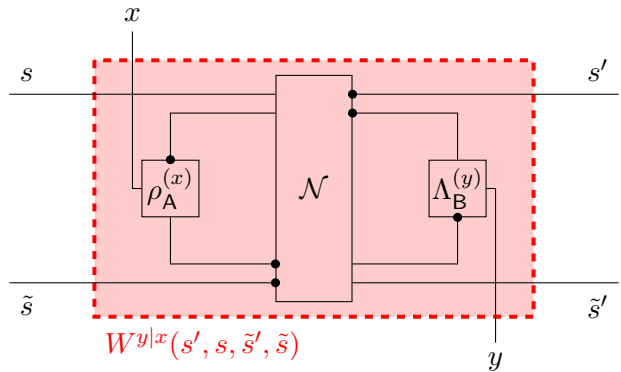
\begin{figure}\centering
\begin{tikzpicture}[
    node/.style={draw=none, inner sep=0pt, outer sep=0pt},
    factor/.style={rectangle, minimum width=1cm, minimum height=3cm, draw},
    sfactor/.style={rectangle, minimum size=.75cm, draw, inner sep=1pt}]
	\node[sfactor] (rhoX) {};
	\node at (rhoX) {$\rho_\system{A}^{(x)}$};
	\node[factor, right=1cm of rhoX] (N) {$\operator{N}$};
	\draw[*-] (rhoX.north) |- ([yshift=1cm]N.west);
	\draw[-*] (rhoX.south) |- ([yshift=-1cm]N.west);
	\node at ([xshift=-.5cm,yshift=2.3cm]rhoX) (X) {$x$};
	\draw (X) |- (rhoX);
	\draw ([xshift=-3.50cm,yshift=1.25cm]N.west) -- ([yshift=1.25cm]N.west) node[above=1pt,pos=0,anchor=south west] {$s$};
	\draw[-*] ([xshift=-3.50cm,yshift=-1.25cm]N.west) -- ([yshift=-1.25cm]N.west) node[below=1pt,pos=0,anchor=north west] {$\tilde{s}$};
	\draw[-*] ([xshift=3.50cm,yshift=1.25cm]N.east) -- ([yshift=1.25cm]N.east) node[above=1pt,pos=0,anchor=south east] {$s'$};
	\draw ([xshift=3.50cm,yshift=-1.25cm]N.east) -- ([yshift=-1.25cm]N.east) node[below=1pt,pos=0,anchor=north east] {$\tilde{s}'$};
	\node[sfactor, right=1cm of N] (M) {};
	\node at (M) {$\Lambda^{(y)}_{\system{B}}$};
	\draw[-*] (M.north) |- ([yshift=1cm]N.east);
	\draw[*-] (M.south) |- ([yshift=-1cm]N.east);
	\node at ([xshift=.5cm,yshift=-2.3cm]M) (Y) {$y$};
	\draw (Y) |- (M);
	\begin{pgfonlayer}{bg}
		\draw[dashed, red, line width=1.5pt, fill= red!20] ([xshift=-1.0cm,yshift=1.7cm]rhoX) rectangle ([xshift=1cm,yshift=-1.7cm]M);
		\node[red,anchor=north west] at ([xshift=-1.0cm,yshift=-1.7cm]rhoX) {$W^{y|x}(s',s,\tilde{s}',\tilde{s})$};
	\end{pgfonlayer}
\end{tikzpicture}
\caption{NFG representation of the channel-ensemble-measurement configuration ($\operator{N}$, $\{\rho^{(x)}_\system{A}\}_{x\in\set{X}}$, $\{\Lambda_\system{B}^{(y)}\}_{y\in\set{Y}}$).}
\label{fig:CFR}
\end{figure}
%*******************************************************************************
%*******************************************************************************
\section{Information Rate and its Estimation}\label{QCwM:sec:4:IR}
In this section, we focus on the information rate of the communication scheme described in Section~\ref{QCwM:sec:3:QCM}.
As defined in~\eqref{eq:capacity:jointk}, the information rate is the limit superior of the average mutual information $\frac{1}{n}\mutualInfo\left(\rv{X}_1^n;\rv{Y}_1^n\right)$ between the input and output processes $\rv{X}_1^n$ and $\rv{Y}_1^n$ as $n$ tends to infinity.
We assume that $\rv{X}_1^n$ is distributed according to some {i.i.d.}~process\footnote{
For more general type of sources, like a finite-state-machine source (FSMS), one can consider ``merging'' the memory of the source into that of the channel, and thus obtaining an equivalent memoryless input process.}
characterized by the \pmf $Q$, \ie, $Q^{(n)}(\vx_1^n) = \prod_{\ell=1}^{n} Q(x_{\ell})$.
In this case, the joint distribution of $(\rv{X}_1^n,\rv{Y}_1^n)$ is given by
\begin{equation}\label{eq:joint:5}
\prob_{\rv{X}_1^n,\rv{Y}_1^n|\system{S}_0}(\vx_1^n,\vy_1^n|\rho_{\system{S}_0}) =  \prod_{\ell=1}^n Q(x_\ell) \cdot \prob_{\rv{Y}_1^n|\rv{X}_1^n;\system{S}_0}(\vy_1^n|\vx_1^n;\rho_{\system{S}_0}),
\end{equation}
where $\prob_{\rv{Y}_1^n|\rv{X}_1^n;\system{S}_0}$ is specified in~\eqref{eq:joint:1},~\eqref{eq:joint:2},~\eqref{eq:joint:3} or~\eqref{eq:joint:4}, depending on which notation we use to specify the channel (see Propositions~\ref{prop:quantum:state:channel} and~\ref{prop:CFR}).
It is obvious that the value of~\eqref{eq:joint:5}, and thus the information rate, depends on the initial density operator $\rho_{\system{S}_0}$.
In this sense, we denote the information rate as a function of the input \pmf $Q$, the CC-QSC $\{\operator{N}^{y|x}\}_{x,y}$ describing the channel, and the initial density operator $\rho_{\system{S}_0}$, namely
\begin{align}
\label{eq:def:qsc:IR:2}
\infoRate(Q,\{\operator{N}^{y|x}\}_{x,y},\rho_{\system{S}_0})
&\defeq \limsup_{n\to\infty} \infoRate^{(n)}(Q,\{\operator{N}^{y|x}\}_{x,y},\rho_{\system{S}_0}),\\
\label{eq:def:qsc:IR:1}
\infoRate^{(n)}(Q,\{\operator{N}^{y|x}\}_{x,y},\rho_{\system{S}_0})
&\defeq \frac{1}{n}\mutualInfo(\rv{X}_1^n;\rv{Y}_1^n)(\rho_{\system{S}_0}).
\end{align}
Here, $\mutualInfo(\rv{X}_1^n;\rv{Y}_1^n)(\rho_{\system{S}_0})$ is the mutual information between $\rv{X}_1^n$ and $\rv{Y}_1^n$; and the latter are jointly distributed according to~\eqref{eq:joint:5}.
The argument $\rho_{\system{S}_0}$ emphasizes the dependency of $\frac{1}{n}\mutualInfo(\rv{X}_1^n;\rv{Y}_1^n)$ on $\rho_{\system{S}_0}$.
\par
%*******************************************************************************
Similar to the case of an FSMC, the dependency of the information rate on the initial density operator usually cannot be ignored.
However, as already mentioned in~Section~\ref{QCwM:sec:FSMC:IR}, for a class of FSMCs, namely the indecomposable FSMCs, it is known that the information rate is independent of the initial channel state~\cite{gallager1968information}.
An indecomposable FSMC, intuitively speaking, is an FSMC whose state distribution, given different initial states, tends to be indistinguishable as $n\to\infty$, independently of the input sequence realized.
A quantum analogy was proposed by Bowen, Devetak, and Mancini~\cite{bowen2005bounds}, where they defined the indecomposable quantum channels with memory, and proved that the quantum entropic bound for such channels is independent of the initial density operator.
\par
%*******************************************************************************
In the remainder of this section we firstly define the indecomposability of CC-QSCs, and prove the independence of the information rate as in~\eqref{eq:def:qsc:IR:2} from the initial density operator.
Secondly, we generalize the methods in Algorithm~\ref{alg:SPA} for estimating such information rate efficiently.
\par
%*******************************************************************************
The definition of an indecomposable CC-QSC in this chapter is similar (but different) and closely related to that of an indecomposable (quantum) channel with memory in~\cite{bowen2005bounds}.
Namely, an indecomposable channel with memory equipped with separable input ensemble and local output measurement will always induce an indecomposable CC-QSC, but not necessarily vice versa.
Moreover, in~\cite{bowen2005bounds} the classical capacity of quantum channels with finite memory was considered, where the capacity is essentially the Holevo bound and where the latter was proven to be achievable~\cite{bowen2004quantum}.
However, in our work, we focus on the situation where the ensemble and the measurement are fixed.
%*******************************************************************************
\subsection{Indecomposable Quantum-State Channel}
\begin{definition}
A CC-QSC $\{\operator{N}^{y|x}\}_{x,y}$ is said to be \emph{indecomposable}\index{indecomposable} if for any initial density operators $\alpha_{\system{S}_0}$ and $\beta_{\system{S}_0}$, the following statement holds: for any $\epsilon>0$, there exists some positive integer $N$ s.t.
\begin{equation}
\norm{\alpha_{\system{S}_n}^{(\vx_1^n)} - \beta_{\system{S}_n}^{(\vx_1^n)}}_1 < \varepsilon \quad \forall n\geqslant N,\ \forall \vx_1^n\in\set{X}^n,
\end{equation}
where
\begin{align}
\alpha_{\system{S}_n}^{(\vx_1^n)}
&\defeq \sum_{\vy_1^n} \operator{N}^{y_n|x_n} \circ \cdots \circ \operator{N}^{y_1|x_1}(\alpha_{\system{S}_0}),\\
\beta_{\system{S}_n}^{(\vx_1^n)}
&\defeq \sum_{\vy_1^n} \operator{N}^{y_n|x_n} \circ \cdots \circ \operator{N}^{y_1|x_1}(\beta_{\system{S}_0}),
\end{align}
and where $\norm{A}_1$ is the trace distance for an operator $A$ on $\hilbert_\system{S}$, \ie, $\norm{A}_1\defeq\frac{1}{2}\tr\sqrt{A^\Herm A}$.
\end{definition}
%*******************************************************************************
\begin{theorem}\label{thm:quantum:indecom}\hspace{-5pt}\footnote{
A similar result regarding indecomposable/forgetful quantum channel with memory can be found in~\cite{kretschmann2005quantum} and~\cite{bowen2005bounds}.}
The information rate of an indecomposable CC-QSC with an {i.i.d.}~input process is independent of the initial density operator.
Namely, if $\{\operator{N}^{y|x}\}_{x,y}$ is indecomposable, then
\begin{equation}
\infoRate^{(n)}(Q,\{\operator{N}^{y|x}\}_{x,y},\alpha_{\system{S}_0}) - \infoRate^{(n)}(Q,\{\operator{N}^{y|x}\}_{x,y},\beta_{\system{S}_0}) \stackrel{n\to\infty}{\longrightarrow} 0
\end{equation}
for any initial density operators $\alpha_{\system{S}_0}$, $\beta_{\system{S}_0}
\in\DensOp(\hilbert_{\system{S}_0})$.
\end{theorem}
%*******************************************************************************
In the proof below, we follow a similar idea as in~\cite{gallager1968information} for indecomposable FSMCs, and as that in~\cite{bowen2005bounds} for indecomposable quantum channels with memory.
\begin{proof}
Let $\system{A}$ and $\system{B}$ be quantum systems described by Hilbert spaces $\hilbert_\system{A}$ and $\hilbert_\system{B}$, respectively, where $\{\bra{x}\}_{x\in\set{X}}$ and $\{\bra{y}\}_{y\in\set{Y}}$ are orthonormal bases of $\hilbert_\system{A}$ and $\hilbert_\system{B}$, respectively.
Let $\system{A}_1^n$ and $\system{B}_1^n$ be $n$ copies of $\system{A}$ and $\system{B}$, respectively.
Let $\rho_{\system{S}_0}$ be some initial density operator; and let the joint density operator on the system $\system{A}_1^n\system{B}_1^n$ be 
\[
\rho_{\system{A}_1^n\system{B}_1^n} \defeq \sum_{\vx_1^n} Q(\vx_1^n) \cdot \braket{\vx_1^n} \tensor \sum_{\vy_1^n} \tr\left( \operator{N}^{\vy_1^n|\vx_1^n}(\rho_{\system{S}_0}) \right) \cdot \braket{\vy_1^n},
\]
where $\operator{N}^{\vy_1^n|\vx_1^n} \defeq \operator{N}^{y_n|x_n} \circ \cdots \circ \operator{N}^{y_1|x_1}$.
In this case, it is not hard to see that
\[
\mutualInfo(\rv{X}_1^n;\rv{Y}_1^n)[\rho_{\system{S}_0}] = \qmutualInfo(\system{A}_1^n;\system{B}_1^n)[\rho_{\system{S}_0}].
\]
In fact, one can easily check that
\begin{align*}
\qEntropy(\system{A}_1^n) &= \entropy(\rv{X}_1^n),\\
\qEntropy(\system{B}_1^n) &= \entropy(\rv{Y}_1^n),\\
\qEntropy(\system{A}_1^n,\system{B}_1^n) &= \entropy(\rv{X}_1^n,\rv{Y}_1^n).
\end{align*}
In particular, $\qEntropy(\system{A}_1^n)$ is independent of the initial density operator $\rho_{\system{S}_0}$.
We also claim that, for each $\rho_{\system{S}_0} \in \DensOp(\hilbert_{\system{S}_0})$ and a positive integer $N < n$, 
\begin{align}
\label{eq:indecomposable:AB}
\qmutualInfo(\system{A}_1^N\system{B}_1^N;\system{A}_{N+1}^n\system{B}_{N+1}^n) 
    &\leqslant 2 \qEntropy(\system{S}_N),\\
\label{eq:indecomposable:B}
\qmutualInfo(\system{B}_1^N;\system{B}_{N+1}^n)
    &\leqslant 2 \qEntropy(\system{S}_N),
\end{align}
where the density operator for $\system{S}_N$ is defined as (depending on $\rho_{\system{S}_0}$)
\[
\rho_{\system{S}_N} \defeq \sum_{\vx_1^N} Q(\vx_1^N) \cdot \sum_{\vy_1^N} \operator{N}^{\vy_1^N|\vx_1^N}(\rho_{\system{S}_0}).
\]
\par
%*******************************************************************************
\textbf{Proof of~\eqref{eq:indecomposable:AB}: }
We define a class of CPTP maps $\{\Phi_a^b: \DensOp(\hilbert_{\system{S}_a}) \to \DensOp(\hilbert_{\system{A}_a^b \system{B}_a^b})\}_{a<b\in\mathbb{N}}$ as
\[
\Phi_a^b: \rho_{\system{S}_a} \mapsto \sum_{\vx_a^b} Q(\vx_a^b) \cdot \braket{\vx_a^b} \tensor \sum_{\vy_a^b} \tr\left( \operator{N}^{\vy_a^b|\vx_a^b}(\rho_{\system{S}_a}) \right) \cdot \braket{\vy_a^b}.
\]
Since the input process $Q$ is {i.i.d.}, we can rewrite $\rho_{\system{A}_1^n\system{B}_1^n}$ for each positive integer $N<n$, as
\[
\rho_{\system{A}_1^n\system{B}_1^n} = \left( I_{\system{A}_1^N\system{B}_1^N} \otimes \Phi_{N+1}^n \right) \left( \rho_{\system{A}_1^N\system{B}_1^N\system{S}_N} \right),
\]
where
\[
\rho_{\system{A}_1^N\system{B}_1^N\system{S}_N} \defeq \sum_{\vx_1^N} Q(\vx_1^N) \braket{\vx_1^N} \tensor \sum_{\vy_1^N} \operator{N}^{\vy_1^N|\vx_1^N}(\rho_{\system{S}_0}) \braket{\vy_1^N}.
\]
Hence, by data processing inequality for quantum mutual information (see \eg,~\cite[Theorem~11.9.4]{wilde2017quantum}), one must have 
\[
\qmutualInfo(\system{A}_1^N\system{B}_1^N;\system{A}_{N+1}^n\system{B}_{N+1}^n) \leqslant \qmutualInfo(\system{A}_1^N\system{B}_1^N;\system{S}_N).
\]
Additionally, by subadditivity of joint entropy, we have
\begin{align*}
\qmutualInfo(\system{A}_1^N\system{B}_1^N;\system{S}_N)
&\defeq \qEntropy(\system{A}_1^N\system{B}_1^N) + \qEntropy(\system{S}_N) - \qEntropy(\system{A}_1^N\system{B}_1^N\system{S}_N)\\
&\leqslant \qEntropy(\system{A}_1^N\system{B}_1^N) + \qEntropy(\system{S}_N) - \abs{\qEntropy(\system{A}_1^N\system{B}_1^N) - \qEntropy(\system{S}_N)}\\
&\leqslant 2\qEntropy(\system{S}_N).
\end{align*}
Combining the above two inequalities, we have proven~\eqref{eq:indecomposable:AB}.
\par
\textbf{Proof of~\eqref{eq:indecomposable:B}: }
\eqref{eq:indecomposable:B} can be shown via the same approach above by considering another class of CPTP maps $\{\Psi_a^b: \DensOp(\hilbert_{\system{S}_a}) \to \DensOp(\hilbert_{\system{B}_a^b})\}_{a<b\in\mathbb{N}}$ as
\[
\Psi_a^b: \rho_{\system{S}_a} \mapsto \sum_{\vx_a^b} Q(\vx_a^b) \cdot \sum_{\vy_a^b} \tr\left( \operator{N}^{\vy_a^b|\vx_a^b}(\rho_{\system{S}_a})\right) \cdot \braket{\vy_a^b}.
\]
We omit the details.
\par
%*******************************************************************************
Now return to the main proof.
Given the initial density operators $\alpha_{\system{S}_0}$, and $\beta_{\system{S}_0}$, we define $\alpha_{\system{A}_1^n\system{B}_1^n}$, $\beta_{\system{A}_1^n\system{B}_1^n}$ and $\alpha_{\system{S}_N}$, $\beta_{\system{S}_N}$ in a similar fashion as we have defined $\rho_{\system{A}_1^n\system{B}_1^n}$ and $\rho_{\system{S}_N}$ based on $\rho_{\system{S}_0}$.
In this case, one obtains 
\begin{align}
\nonumber
&\bigabs{\qEntropy(\alpha_{\system{A}_1^n\system{B}_1^n}) - \qEntropy(\beta_{\system{A}_1^n\system{B}_1^n})} - \bigabs{\qEntropy(\alpha_{\system{A}_{N\!+\!1}^n\system{B}_{N\!+\!1}^n}) - \qEntropy(\beta_{\system{A}_{N\!+\!1}^n\system{B}_{N\!+\!1}^n})}\\
\nonumber
\overset{\text{(a)}}{\leqslant} 
&\biggabs{\qEntropy(\alpha_{\system{A}_1^N\system{B}_1^N}) \!-\! \qEntropy(\beta_{\system{A}_1^N\system{B}_1^N})} \!+\! \biggabs{\qmutualInfo(\system{A}_1^N\system{B}_1^N;\system{A}_{N\!+\!1}^n\system{B}_{N\!+\!1}^n)[\alpha_{\system{A}_1^n\system{B}_1^n}] \!-\! \qmutualInfo(\system{A}_1^N\system{B}_1^N;\system{A}_{N\!+\!1}^n\system{B}_{N\!+\!1}^n)[\beta_{\system{A}_1^n\system{B}_1^n}]}\\
\label{eq:tail:1}
\overset{\text{(b)}}{\leqslant}
& N\cdot\log(\dim{\hilbert_\system{AB}}) + 2\cdot\max\left\{\qEntropy(\alpha_{\system{S}_N}),\qEntropy(\beta_{\system{S}_N})\right\},
\end{align}
where we have used the triangle inequality in step (a), and Corollary~\ref{cor:qEntropy:basics} and~\eqref{eq:indecomposable:AB} in step (b).
Similarly, using~\eqref{eq:indecomposable:B}, one can prove
\begin{equation}\label{eq:tail:2}
\bigabs{\qEntropy(\alpha_{\system{B}_1^n}) - \qEntropy(\beta_{\system{B}_1^n})} - \bigabs{\qEntropy(\alpha_{\system{B}_{N+1}^n}) - \qEntropy(\beta_{\system{B}_{N+1}^n})}
\leqslant N \cdot \log(\dim{\hilbert_\system{B}}) + 2\cdot\max\left\{\qEntropy(\alpha_{\system{S}_N}),\qEntropy(\beta_{\system{S}_N})\right\}.
\end{equation}
By assumption, there exists some positive integer $d$ such that $\max\{\dim{\hilbert_\system{A}},\dim{\hilbert_\system{B}},\dim{\hilbert_\system{S}}\} \leqslant d$.
Thus, we have
\begin{fleqn}\begin{equation*}\phantom{=}
\frac{1}{n}\bigabs{\mutualInfo(\rv{X}_1^n;\rv{Y}_1^n)[\alpha_{\system{S}_0}] - \mutualInfo(\rv{X}_1^n;\rv{Y}_1^n)[\beta_{\system{S}_0}]}
\end{equation*}\end{fleqn}
\begin{fleqn}\begin{equation*}
=\frac{1}{n}\bigabs{\qmutualInfo(\system{A}_1^n;\system{B}_1^n)[\alpha_{\system{S}_0}] - \qmutualInfo(\system{A}_1^n;\system{B}_1^n)[\beta_{\system{S}_0}]}
\end{equation*}\end{fleqn}
\begin{fleqn}\begin{equation*}
=\frac{1}{n}\bigabs{\left( \qEntropy(\alpha_{\system{B}_1^n}) - \qEntropy(\alpha_{\system{A}_1^n\system{B}_1^n}) \right) - \left(\qEntropy(\beta_{\system{B}_1^n}) - \qEntropy(\beta_{\system{A}_1^n\system{B}_1^n})\right)}
\end{equation*}\end{fleqn}
\begin{fleqn}\begin{equation*}\overset{(\mathrm{c})}{\leqslant}
\frac{1}{n} \bigabs{\qEntropy(\alpha_{\system{B}_1^n}) - \qEntropy(\beta_{\system{B}_1^n})} + \frac{1}{n}\bigabs{\qEntropy(\alpha_{\system{A}_1^n\system{B}_1^n}) - \qEntropy(\beta_{\system{A}_1^n\system{B}_1^n})}
\end{equation*}\end{fleqn}
\begin{fleqn}\begin{equation*}\overset{(\mathrm{d})}{\leqslant}
\frac{3N\!+\!4}{n}\cdot\log{d} + \frac{1}{n}\bigabs{\qEntropy(\alpha_{\system{B}_{N+1}^n}) - \qEntropy(\beta_{\system{B}_{N+1}^n})} + \frac{1}{n}\bigabs{\qEntropy(\alpha_{\system{A}_{N+1}^n\system{B}_{N+1}^n}) - \qEntropy(\beta_{\system{A}_{N+1}^n\system{B}_{N+1}^n})}
\end{equation*}\end{fleqn}
\begin{fleqn}\begin{equation*}
=\frac{3N\!+\!4}{n}\cdot\log{d} + \frac{1}{n}\bigabs{\qEntropy(\Psi_{N+1}^n(\alpha_{\system{S}_N})) - \qEntropy(\Psi_{N+1}^n(\beta_{\system{S}_N}))} + \frac{1}{n}\bigabs{\qEntropy(\Phi_{N+1}^n(\alpha_{\system{S}_N})) - \qEntropy(\Phi_{N+1}^n(\beta_{\system{S}_N}))},
\end{equation*}\end{fleqn}
where we have used the triangle inequality in step (c), and~\cite[Theorem~11.8]{nielsen2011quantum}, \eqref{eq:tail:1}, \eqref{eq:tail:2} in step (d).
Using a loose variant of Fannes' inequality~\cite{fannes1973continuity}\footnote{
Namely, we used the inequality $\abs{\qEntropy(\rho)-\qEntropy(\sigma)} \leqslant \log{\dim}\cdot\norm{\rho-\sigma}_1+e^{-1}$.
Note that tighter variants of Fannes' inequality exist, but the above inequality is good enough to prove the desired result.}, we have
\begin{align*}
\bigabs{\qEntropy(\Psi_{N+1}^n(\alpha_{\system{S}_N})) \!-\! \qEntropy(\Psi_{N+1}^n(\beta_{\system{S}_N}))}
&\leqslant (n\!-\!N) \cdot \log{d}\ \cdot \norm{\Psi_{N+1}^n(\alpha_{\system{S}_N}) \!-\! \Psi_{N+1}^n(\beta_{\system{S}_N})}_1 + e^{-1},\\
\bigabs{\qEntropy(\Phi_{N+1}^n(\alpha_{\system{S}_N})) \!-\! \qEntropy(\Phi_{N+1}^n(\beta_{\system{S}_N}))}
&\leqslant 2\!\cdot\!(n\!-\!N)\cdot\log{d} \cdot \norm{\Phi_{N+1}^n(\alpha_{\system{S}_N}) \!-\! \Phi_{N+1}^n(\beta_{\system{S}_N})}_1 + e^{-1}.
\end{align*}
Moreover, by the contractivity of the trace distance, we have,
\begin{align*}
\norm{\Psi_{N+1}^n(\alpha_{\system{S}_N})-\Psi_{N+1}^n(\beta_{\system{S}_N})}_1
& \leqslant \norm{\alpha_{\system{S}_N}-\beta_{\system{S}_N}}_1,\\
\norm{\Phi_{N+1}^n(\alpha_{\system{S}_N})-\Phi_{N+1}^n(\beta_{\system{S}_N})}_1
& \leqslant \norm{\alpha_{\system{S}_N}-\beta_{\system{S}_N}}_1.
\end{align*}
This allows us to bound the difference of the information rate by
\[
\frac{1}{n}\bigabs{\mutualInfo(\rv{X}_1^n;\rv{Y}_1^n)[\alpha_{\system{S}_0}] - \mutualInfo(\rv{X}_1^n;\rv{Y}_1^n)[\beta_{\system{S}_0}]} \leqslant \frac{3N\!+\!4}{n}\cdot\log{d} + \frac{3(n-N)}{n} \cdot \log{d} \cdot \norm{\alpha_{\system{S}_n} - \beta_{\system{S}_n}}_1 + \frac{2}{n\cdot e}.
\]
Finally, because the CC-QSC is indecomposable, for any $\varepsilon>0$, we can choose $N$ large enough such that
\[
\norm{\alpha_{\system{S}_N}-\beta_{\system{S}_N}}_1 < \frac{\varepsilon}{6\cdot\log{d}}
\]
and then choose an integer $M>N$ such that
\[
\frac{3N+4}{M}\cdot\log{d} + \frac{2}{M\cdot e}< \frac{\varepsilon}{2}.
\]
This will ensure that for any $n>M$, we have
\[
\frac{3N+4}{n} \cdot \log{d} + \frac{3(n-N)}{n} \cdot \log{d} \cdot \norm{\alpha_{\system{S}_n} - \beta_{\system{S}_n}}_1 + \frac{2}{n\cdot e} < \varepsilon,
\]
which concludes the proof.
\end{proof}
%*******************************************************************************
\subsection{Estimation of the Information Rate}
The development in this section is very similar to the development in Subsection~\ref{QCwM:sec:FSMC:IR}.
In particular, we follow the same approach as in~\eqref{eq:def:fsmc:ir:2}--\eqref{eq:FSMC:ir:estimate:1}.
This similarity stems from the similarity of the factor graphs in Figures~\ref{fig:FMSC:high:level:1} and~\ref{fig:qsc:bcjr}, and highlights one of the benefits of the factor-graph approach that we take to estimate information rate of quantum channels with memory.\par
We make the following assumptions.
\begin{itemize}
\item As already mentioned, the derivations in this chapter are for the case where the input process $\rv{X}_1^n=(\rv{X}_1,\ldots,\rv{X}_n)$ is an {i.i.d.}~process.
	The results can be generalized to other stationary ergodic input processes that can be represented by a finite-state-machine source (FSMS).
	Technically, this is done by defining a new state that combines the FSMS state and the channel state.
\item We assume that the corresponding quantum-state channel $\{\operator{N}^{y|x}\}_{x\in\set{X},y\in\set{Y}}$ is finite-dimensional and indecomposable.
	We also assume it can be represented by some functions $\{W^{y|x}\}_{x,y}$ as defined in~\eqref{eq:def:channel:function:representation}.
\end{itemize}
The major difference compared with Section~\ref{QCwM:sec:FSMC:IR} is the conditional \pmf $\prob_{\rv{Y}_1^n|\rv{X}_1^n;\system{S}_0}$, and thus the joint \pmf $\prob_{\rv{Y}_1^n,\rv{X}_1^n|\system{S}_0}$ as specified in~\eqref{eq:joint:4} and~\eqref{eq:joint:5}, respectively.
In this case, in order to compute $-\frac{1}{n}\log{\prob_{\rv{Y}_1^n}(\cvy_1^n)}$ and $-\frac{1}{n}\log{\prob_{\rv{X}_1^n\rv{Y}_1^n}(\cvx_1^n,\cvy_1^n)}$ using a similar method as in Section~\ref{QCwM:sec:FSMC:IR}, we consider the state metrics $\{\sigmaY_\ell\}_{\ell=1}^n$ and $\{\sigmaXY_\ell\}_{\ell=1}^n$ (which are operators on $\hilbert_{\system{S}_\ell}$ for each $\ell$) defined \wrt $\cvy_1^n$ and \wrt $\cvx_1^n$ and $\cvy_1^n$, respectively, as
\begin{align}
\sigmaY_\ell
&\defeq \sum_{\vx_1^\ell} Q^{(\ell)}(\vx_1^\ell) \cdot \operator{N}^{\cy_n|x_n} \circ \cdots \circ \operator{N}^{\cy_1|x_1}(\rho_{\system{S}_0}),\\
\sigmaXY_{\ell}
&\defeq \operator{N}^{\cy_n|\cx_n} \circ \cdots \circ \operator{N}^{\cy_1|\cx_1}(\rho_{\system{S}_0}).
\end{align}
In this case, we have $\prob_{\rv{Y}_1^n}(\cvy_1^n) = \tr(\sigmaY_n)$, and $\prob_{\rv{X}_1^n\rv{Y}_1^n}(\cvx_1^n,\cvy_1^n) = \tr(\sigmaXY_n)$.
Notice that $\{\sigmaY_\ell\}_{\ell}$ and $\{\sigmaXY_\ell\}_{\ell}$ can be computed iteratively as
\begin{align}
\label{eq:recursive:quantum:state:metric:Y:1}
[\sigmaY_{\ell}]
&= \sum_{x_\ell} Q(x_\ell) \cdot [W^{y|x}] \cdot [\sigmaY_{\ell-1}],\\
\label{eq:def:quantum:channel:state:metric:XY:1}
[\sigmaXY_{\ell}]
&= [W^{y|x}] \cdot [\sigmaXY_{\ell-1}],
\end{align}
where we treat $[\sigma^\rv{Y}_{\ell}]$ and $[\sigma^\rv{XY}_{\ell}]$ as length-$d^2$ vectors indexed by $(s,\tilde{s})\in\set{S}^2$ in the above two equations.
(See~\eqref{eq:qsc:matrix:1} and~\eqref{eq:joint:4} for notations.)
Moreover, we can also introduce normalizing coefficients $\{\lambda_\ell^\rv{Y}\}_{\ell}$ and $\{\lambda_\ell^\rv{XY}\}_{\ell}$, similar to~\eqref{eq:recursive:state:metric:Y:2}, for the sake of numerical stability.
In the latter case, we have iterative updating rules
\begin{align}
\label{eq:recursive:quantum:state:metric:Y:2}
[\bsigmaY_\ell] &= \frac{1}{\lambdaY_\ell} \cdot \sum_{x_\ell} Q(x_\ell) \cdot 
    [W^{y|x}] \cdot [\bsigmaY_{\ell-1}],\\
\label{eq:def:quantum:channel:state:metric:XY:2}
[\bsigmaXY_\ell] &= \frac{1}{\lambdaXY_\ell} \cdot [W^{y|x}] 
    \cdot [\bsigmaXY_{\ell-1}],
\end{align}
where the scaling factors $\lambdaY_\ell>0$ and $\lambdaXY_\ell>0$ are chosen such that $\tr(\bsigmaY_\ell)=1$ and $\tr(\bsigmaXY_\ell)=1$, respectively.
In addition, one can verify that $\prob_{\rv{Y}_1^n}(\cvy_1^n) = \prod_{\ell=1}^n \lambdaY_\ell$, and $\prob_{\rv{X}_1^n\rv{Y}_1^n}(\cvx_1^n,\cvy_1^n) = \prod_{\ell=1}^n \lambdaXY_\ell$.
\par
%*******************************************************************************
The above discussion is summarized as Algorithm~\ref{alg:SPA:2}.
The computations corresponding to Line 3, 5--9 and 12--16 are visualized in Figures~\ref{fig:QFSM:channel:simulation:Y},~\ref{fig:QFSM:estimate:hY}, and~\ref{fig:QFSM:estimate:hXY} in the Appendix~\ref{app:figures}, respectively.
\begin{algorithm}
\caption{Estimating the Information Rate of a CC-QSC}
\begin{algorithmic}[1]
\Require{An indecomposable CC-QSC $\{\operator{N}^{y|x}\}_{x\in\set{X},y\in\set{Y}}$, which can be represented by functions $\{W^{y|x}\}_{x,y}$, a input~distribution~$Q$, a positive~integer~$n$ large enough.}
\Ensure{$\infoRate^{(n)}(Q,\{\operator{N}^{y|x}\}_{x,y}) \approx \entropy(\rv{X})+\hat\entropicRate(\rv{Y}) - \hat\entropicRate(\rv{X},\rv{Y})$.}
\State Initialize the memory density operator $\rho_{\system{S}_0} \gets \braket{0_\system{S}}$
\State Generate an input sequence $\cvx_1^n \sim Q^{\tensor n}$
\State Generate a corresponding output sequence $\cvy_1^n$
\State $\bsigmaY_0 \gets \rho_{\system{S}_0}$
\ForEach{$\ell=1,\ldots,n$}
\State $[\sigmaY_\ell] \gets \sum_{x_\ell} Q(x_\ell)\cdot [W^{\cy_{\ell}|x}] \cdot [\bsigmaY_{\ell-1}]$
\State $\lambdaY_\ell \gets \tr(\sigmaY_\ell)$
\State $\bsigmaY_{\ell} \gets \sigmaY_\ell/\lambdaY_\ell$
\EndFor
\State $\hat\entropicRate(\rv{Y}) \gets -\frac{1}{n} \sum_{\ell=1}^n \log(\lambdaY_{\ell})$
\State $\bsigmaXY_0 \gets \rho_{\system{S}_0}$
\ForEach{$\ell=1,\ldots,n$}
\State $[\sigmaXY_\ell] \gets [W^{\cy_\ell|\cx_\ell}] \cdot [\bsigmaXY_{\ell-1}]$
\State $\lambdaXY_\ell \gets \tr(\sigmaXY_\ell)$
\State $\bsigmaXY_{\ell} \gets \sigmaXY_\ell/\lambdaXY_\ell$
\EndFor
\State $\hat\entropicRate(\rv{X},\rv{Y}) \gets -\frac{1}{n} \sum_{\ell=1}^n \log(\lambdaXY_{\ell})$
\State $\entropy(\rv{X}) \gets -\sum_{x} Q(x) \log{Q(x)}$
\State Estimate $\infoRate^{(n)}(Q,\{\operator{N}^{y|x}\}_{x,y})$ as $\entropy(\rv{X}) +\hat\entropicRate(\rv{Y}) - \hat\entropicRate(\rv{X},\rv{Y})$.
\end{algorithmic}
\label{alg:SPA:2}
\end{algorithm}
%*******************************************************************************
%*******************************************************************************
\section[IRUB/IRLB and their Optimization]{Information rate upper/lower bounds and their Optimization} \label{QCwM:sec:5:UBLB}
In this section, we consider auxiliary channels and their induced upper and lower bounds on the information rate.
As mentioned in the beginning of this chapter auxiliary channels are often introduced as a low-complexity approximation of the original channel, which are useful in mismatch decoding.
The techniques developed in this section only require the channel input/output data, but not the channel model itself.
This is particularly useful when the channel is only made physically, but not mathematically, available.
In this case, the task of minimizing the difference between the upper and lower bound is equivalent to finding the channel model (within a specified class of channel models) best fitting the \emph{empirical} channel law.
Similarly, minimizing the upper bound corresponds to finding the channel model best fitting the \emph{empirical} channel output distribution, and maximizing the lower bound corresponds to finding the channel model best fitting the \emph{empirical} reverse channel law.
Motivated by the above scenarios, we particularly consider the auxiliary channels chosen from the domain of all CC-QSCs with the same input and output alphabet as the original channel, and acting on a memory system of a certain dimension (which can be different from the memory dimension of the original channel).
Throughout this section, we assume the original channel as described in Section~\ref{QCwM:sec:3:QCM} is indecomposable, and that all the involved Hilbert spaces are of finite dimension, and that the alphabets $\set{X}$ and $\set{Y}$ are finite.
\par
%*******************************************************************************
Suppose we have some auxiliary CC-QSC $\{\hat{\operator{N}}^{y|x}\}_{x,y}$, describable by some functions $\{\hat{W}^{y|x}\}_{x,y}$ as in~\eqref{eq:def:channel:function:representation}.
Let $\hat{P}_{\rv{Y}_1^n|\rv{X}_1^n,\hat{\system{S}}_0}$ denote its joint channel law, similar to~\eqref{eq:joint:2},~\eqref{eq:joint:3}, or~\eqref{eq:joint:4}.
Namely,
\begin{equation}
\hat{P}_{\rv{Y}_1^n|\rv{X}_1^n,\hat{\system{S}}_0}(\vy_1^n|\vx_1^n;\hat{\rho}_{\system{S}_0}) \defeq \tr\left( [\What^{y_{n}|x_{n}}] \cdots [\What^{y_{1}|x_{1}}] \cdot [\hat{\rho}_{\system{S}_0}] \right).
\end{equation}
We follow a similar approach as in~\cite{arnold2006simulation, sadeghi2009optimization}, and define the quantities
\begin{align}
\label{eq:IRUB}\begin{aligned}
\IRUB_{W}^{(n)}(\hat{W}) \defeq &
\frac{1}{n} \sum_{\vx_1^n,\vy_1^n} Q^{(n)}(\vx_1^n) \cdot \prob_{\rv{Y}_1^n|\rv{X}_1^n;\system{S}_0}(\vy_1^n|\vx_1^n;\rho_{\system{S}_0})\\
&\cdot \log\frac{\prob_{\rv{Y}_1^n|\rv{X}_1^n;\system{S}_0}(\vy_1^n|\vx_1^n;\rho_{\system{S}_0})}{\sum_{\cvx_1^n} Q^{(n)}(\cvx_1^n) \hat{P}_{\rv{Y}_1^n|\rv{X}_1^n,\hat{\system{S}}_0}(\vy_1^n|\cvx_1^n;\rho_{\hat{\system{S}}_0})},
\end{aligned}\\
\label{eq:IRLB}\begin{aligned}
\IRLB_{W}^{(n)}(\hat{W}) \defeq &
\frac{1}{n} \sum_{\vx_1^n,\vy_1^n} Q^{(n)}(\vx_1^n) \cdot \prob_{\rv{Y}_1^n|\rv{X}_1^n;\system{S}_0}(\vy_1^n|\vx_1^n;\rho_{\system{S}_0})\\
&\cdot \log\frac{\hat \prob_{\rv{Y}_1^n|\rv{X}_1^n;\system{S}_0}(\vy_1^n|\vx_1^n;\rho_{\system{S}_0})}{\sum_{\cvx_1^n} Q^{(n)}(\cvx_1^n) \hat{P}_{\rv{Y}_1^n|\rv{X}_1^n,\hat{\system{S}}_0}(\vy_1^n|\cvx_1^n;\rho_{\hat{\system{S}}_0})},
\end{aligned}
\end{align}
where $\prob_{\rv{Y}_1^n|\rv{X}_1^n;\system{S}_0}$ is defined in~\eqref{eq:joint:2},~\eqref{eq:joint:3} or~\eqref{eq:joint:4}.
By following similar arguments like those in~\eqref{eq:IRUB:minus:IR} and~\eqref{eq:IR:minus:IRLB}, one can verify that 
\begin{equation}
\IRLB_{W}^{(n)}(\hat{W}) \leqslant \infoRate^{(n)}_{W} \leqslant \IRUB_{W}^{(n)}(\hat{W}),
\end{equation}
where the first inequality holds with equality if and only if $\hat{P}_{\rv{Y}_1^n|\rv{X}_1^n,\hat{\system{S}}_0}(\vy_1^n|\vx_1^n;\rho_{\hat{\system{S}}_0})$ and $\prob_{\rv{Y}_1^n|\rv{X}_1^n;\system{S}_0}(\vy_1^n|\vx_1^n;\rho_{\system{S}_0})$ coincide for all $\vx_1^n$ and $\vy_1^n$ with positive support of $\prob_{\rv{Y}_1^n|\rv{X}_1^n;\system{S}_0}$ and where the second inequalities holds with equality if and only if $\hat{P}_{\rv{Y}_1^n|\hat{\system{S}}_0}(\vy_1^n|\rho_{\hat{\system{S}}_0})$ and $\prob_{\rv{Y}_1^n|\system{S}_0}(\vy_1^n|\rho_{\system{S}_0})$ coincide for all $\vy_1^n$ with positive support of $\prob_{\rv{Y}_1^n|\system{S}_0}$.
Another quantity of interest is the \emph{difference function} defined as 
\begin{equation} \label{eq:DIFF}
\Delta_{W}^{(n)}(\hat{W}) \defeq \IRUB_{W}^{(n)}(\hat{W}) - \IRLB_{W}^{(n)}(\hat{W}).
\end{equation}
Explicit expressions of~\eqref{eq:IRUB},~\eqref{eq:IRLB}, and~\eqref{eq:DIFF} are given by
\begin{align}
\label{eq:IRUB:explicit}
\IRUB^{(n)}_W(\hat{W}) &= \frac{1}{n}\expectationwrt{
\log\frac{\tr\left([W^{\rv{Y}_n|\rv{X}_n}] \cdots [W^{\rv{Y}_1|\rv{X}_1}] \cdot [\rho_{\system{S}_0}]\right)}{\sum_{\vx_1^n} Q^{(n)}(\vx_1^n) \cdot \tr\left([\hat{W}^{\rv{Y}_n|x_n}] \cdots [\hat{W}^{\rv{Y}_1|x_1}] \cdot [\rho_{\hat{\system{S}}_0}]\right)}}{\rv{X}_1^n\rv{Y}_1^n},\\
\label{eq:IRLB:explicit}
\IRLB_{W}^{(n)}(\hat{W}) &= \frac{1}{n}\expectationwrt{
\log\frac{\tr\left([\hat{W}^{\rv{Y}_n|\rv{X}_n}] \cdots [\hat{W}^{\rv{Y}_1|\rv{X}_1}]\cdot [\rho_{\hat{\system{S}}_0}]\right)}{\sum_{\vx_1^n} Q^{(n)}(\vx_1^n) \cdot \tr\left([\hat{W}^{\rv{Y}_n|x_n}]\cdots[\hat{W}^{\rv{Y}_1|x_1}] \cdot [\rho_{\hat{\system{S}}_0}]\right)}}{\rv{X}_1^n\rv{Y}_1^n},\\
\label{eq:DIFF:explicit}
\Delta_{W}^{(n)}(\hat{W}) &= \frac{1}{n}\expectationwrt{
\log\frac{\tr\left([W^{\rv{Y}_n|\rv{X}_n}] \cdots [W^{\rv{Y}_1|\rv{X}_1}] \cdot [\rho_{\system{S}_0}]\right)}{\tr\left([\hat{W}^{\rv{Y}_n|\rv{X}_n}] \cdots [\hat{W}^{\rv{Y}_1|\rv{X}_1}] \cdot [\rho_{\hat{\system{S}}_0}]\right)}}{\rv{X}_1^n\rv{Y}_1^n},
\end{align}
respectively, where $\rv{X}_1^n$ and $\rv{Y}_1^n$ are random variables distributed according to the joint distribution $Q^{(n)}(\vx_1^n)\cdot \prob_{\rv{Y}_1^n|\rv{X}_1^n;\system{S}_0}(\vy_1^n|\vx_1^n;\rho_{\system{S}_0})$ and where $\left\langle\cdot\right\rangle$ stands for the expectation function.
\par
%*******************************************************************************
In the remainder of this section, we propose an algorithm based on the gradient-descent method and the techniques described in Section~\ref{QCwM:subsec:NFGs} and~\ref{QCwM:sec:4:IR} for optimizing the quantities in~\eqref{eq:IRUB},~\eqref{eq:IRLB}, and~\eqref{eq:DIFF}.
In particular, we consider $\{\hat{W}^{y|x}\}_{x,y}$ to be an \emph{interior} point in the domain of CC-QSCs, namely
\begin{itemize}
\item The Choi--Jamio{\l}kowski matrices $\llbracket \hat{W}^{y|x}\rrbracket$, defined similarly as~\eqref{eq:qsc:matrix:2}, are strictly positive definite for each $x$ and $y$,
\item Eq.~\eqref{eq:operator:sum:representation:condition:2} holds by replacing $W^{y|x}$ with $\hat{W}^{y|x}$, namely $\sum_{y\in\set{Y}} \sum_{s',\tilde{s}':\: s'=\tilde{s}'} \llbracket \hat{W}^{y|x}\rrbracket_{(s',s),(\tilde{s}',\tilde{s})} = \delta_{s,\tilde{s}}$ for all $x\in\set{X}$.
\end{itemize}
For any set of functions $\{H^{y|x}:\set{S}^4\to \mathbb{C}\}_{x,y}$ such that $\llbracket H^{y|x}\rrbracket$ (again, defined similarly as~\eqref{eq:qsc:matrix:2}) is Hermitian for each $x$ and $y$ and such that
\begin{equation}\label{eq:tangent:linear:constraint}
\sum_{y\in\set{Y}}\sum_{s',\tilde{s}':\: s'=\tilde{s}'} \llbracket H^{y|x}\rrbracket_{(s',s),(\tilde{s}',\tilde{s})} = 0 
	\quad \forall x\in\set{X},
\end{equation}
the functions $\{\hat{W}^{y|x}+t\cdot H^{y|x}\}_{x,y}$ describe a valid CC-QSC, for all $t$ in some neighborhood of $0$.
In this case, the directional derivatives of the functions $\IRLB_{W}^{(n)}$, $\IRUB_{W}^{(n)}$, and $\Delta_{W}^{(n)}$ at $\{\hat{W}^{y|x}\}_{x,y}$ \emph{along} $\{H^{y|x}\}_{x,y}$ is well defined, and can be expressed as~\eqref{eq:IRUB:dev:2},~\eqref{eq:IRLB:dev:2}, and~\eqref{eq:DIFF:dev:2} on next page, 
%*******************************************************************************
\begin{sidewaysfigure}
\begin{align}
%\label{eq:IRUB:dev:1}
\left.\frac{\D}{\D t}\right|_{t=0} \IRUB_{W}^{(n)}(\hat{W}+tH)
&\propto -\frac{1}{n} \expectationwrt{
\sum_{k=1}^n \sum_{\vx_1^n} Q^{(n)}(\vx_1^n) \cdot \tr\left( [\hat{W}^{\rv{Y}_n|x_n}] \cdots [\hat{W}^{\rv{Y}_{k\!+\!1}|x_{k\!+\!1}}] [H^{\rv{Y}_k|x_k}] [\hat{W}^{\rv{Y}_{k\!-\!1}|x_{k\!-\!1}}] \cdots [\hat{W}^{\rv{Y}_1|x_1}] \cdot [\rho_{\hat{\system{S}}_0}] \right)}{\rv{Y}_1^n} \nonumber\\
\label{eq:IRUB:dev:2}
&= -\frac{1}{n} \sum_{\vx_1^n,\vy_1^n} \prob_{\rv{X}_1^n,\rv{Y}_1^n|\system{S}_0}(\vx_1^n,\vy_1^n|\rho_{\system{S}_0}) \cdot \sum_{k} \sum_{s',s,\tilde{s}',\tilde{s}} \rvec{\varrho}_{\hat{\system{S}}_{k\!-\!1}}^{(\vy_1^{k\!-\!1})}(s,\tilde{s}) \cdot H^{y_k|x_k}(s',s,\tilde{s}',\tilde{s}) \cdot \lvec{\varrho}_{\hat{\system{S}}_k}^{(\vy_{k\!+\!1}^n)}(s',\tilde{s}')\\
%\label{eq:IRLB:dev:1}
\left.\frac{\D}{\D t}\right|_{t=0} \IRLB_{W}^{(n)}(\hat{W}+tH)
&\propto\!\!\begin{aligned}[t]
&-\frac{1}{n}\expectationwrt{
\sum_{k=1}^n \sum_{\vx_1^n} Q^{(n)}(\vx_1^n) \cdot \tr\left([\hat{W}^{\rv{Y}_n|x_n}] \cdots [\hat{W}^{\rv{Y}_{k\!+\!1}|x_{k\!+\!1}}] [H^{\rv{Y}_k|x_k}] [\hat{W}^{\rv{Y}_{k\!-\!1}|x_{k\!-\!1}}] \cdots [\hat{W}^{\rv{Y}_1|x_1}] \cdot [\rho_{\hat{\system{S}}_0}] \right)}{\rv{Y}_1^n}\\
&+\frac{1}{n}\expectationwrt{
\sum_{k=1}^n \tr\left([\hat{W}^{\rv{Y}_n|\rv{X}_n}] \cdots [\hat{W}^{\rv{Y}_{k\!+\!1}|\rv{X}_{k\!+\!1}}] [H^{\rv{Y}_k|\rv{X}_k}] [\hat{W}^{\rv{Y}_{k\!-\!1}|\rv{X}_{k\!-\!1}}] \cdots [\hat{W}^{\rv{Y}_1|\rv{X}_1}] \cdot [\rho_{\hat{\system{S}}_0}] \right)}{\rv{X}_1^n\rv{Y}_1^n}
\end{aligned}\nonumber\\
\label{eq:IRLB:dev:2}
&=\!\!\begin{aligned}[t]
&-\frac{1}{n} \sum_{\vx_1^n,\vy_1^n} \prob_{\rv{X}_1^n,\rv{Y}_1^n|\system{S}_0}(\vx_1^n,\vy_1^n|\rho_{\system{S}_0}) \cdot \sum_{k}\sum_{s',s,\tilde{s}',\tilde{s}} \rvec{\varrho}_{\hat{\system{S}}_{k\!-\!1}}^{\,(\vy_1^{k\!-\!1})}(s,\tilde{s}) \cdot H^{y_k|x_k}(s',s,\tilde{s}',\tilde{s}) \cdot \lvec{\varrho}_{\hat{\system{S}}_k}^{\,(\vy_{k\!+\!1}^n)}(s',\tilde{s}')\\
&+\frac{1}{n} \sum_{\vx_1^n,\vy_1^n} \prob_{\rv{X}_1^n,\rv{Y}_1^n|\system{S}_0}(\vx_1^n,\vy_1^n|\rho_{\system{S}_0}) \cdot \sum_{k}\sum_{s',s,\tilde{s}',\tilde{s}} \rvec{\varrho}_{\hat{\system{S}}_{k\!-\!1}}^{(\vx_1^{k\!-\!1},\vy_1^{k\!-\!1})}(s,\tilde{s}) \cdot H^{y_k|x_k}(s',s,\tilde{s}',\tilde{s})\cdot \lvec{\varrho}_{\hat{\system{S}}_k}^{(\vx_{k\!+\!1}^n,\vy_{k\!+\!1}^n)}(s',\tilde{s}')
\end{aligned}\\
%\label{eq:DIFF:dev:1}
\left.\frac{\D}{\D t}\right|_{t=0} \Delta_{W}^{(n)}(\hat{W}+tH)
&\propto-\frac{1}{n}\expectationwrt{
\sum_{k=1}^n \tr\left( [\hat{W}^{\rv{Y}_n|\rv{X}_n}] \cdots [\hat{W}^{\rv{Y}_{k\!+\!1}|\rv{X}_{k\!+\!1}}] [H^{\rv{Y}_k|\rv{X}_k}] [\hat{W}^{\rv{Y}_{k\!-\!1}|\rv{X}_{k\!-\!1}}] \cdots [\hat{W}^{\rv{Y}_1|\rv{X}_1}] \cdot [\rho_{\hat{\system{S}}_0}]\right)}{\rv{X}_1^n\rv{Y}_1^n}\nonumber\\
\label{eq:DIFF:dev:2}
&=-\frac{1}{n}\sum_{\vx_1^n,\vy_1^n} \prob_{\rv{X}_1^n,\rv{Y}_1^n|\system{S}_0}(\vx_1^n,\vy_1^n|\rho_{\system{S}_0}) \cdot \sum_{k}\sum_{s',s,\tilde{s}',\tilde{s}} \rvec{\varrho}_{\hat{\system{S}}_{k\!-\!1}}^{\,(\vx_1^{k\!-\!1},\vy_1^{k\!-\!1})}(s,\tilde{s}) \cdot H^{y_k|x_k}(s',s,\tilde{s}',\tilde{s}) \cdot \lvec{\varrho}_{\hat{\system{S}}_k}^{\,(\vx_{k\!+\!1}^n,\vy_{k\!+\!1}^n)}(s',\tilde{s}')
\end{align}
\end{sidewaysfigure}
%*******************************************************************************
where we define the messages $\{\rvec{\varrho}_{\system{S}_\ell}^{(\cvy_1^\ell)}\}_\ell$, $\{\lvec{\varrho}_{\system{S}_\ell}^{(\cvy_{\ell\!+\!1}^n)}\}_\ell$, $\{\rvec{\varrho}_{\system{S}_\ell}^{(\cvx_1^\ell,\cvy_1^\ell)}\}_\ell$, and $\{\lvec{\varrho}_{\system{S}_\ell}^{(\cvx_{\ell\!+\!1}^n,\cvy_{\ell\!+\!1}^n)}\}_\ell$ in a recursive manner as 
\begin{align}
[\rvec{\varrho}_{\system{S}_\ell}^{(\cvy_1^\ell)}] &\defeq
\sum_{\vx_1^\ell} Q(\vx_1^\ell) \cdot [\hat{W}^{\cy_\ell|x_\ell}] \cdots [\hat{W}^{\cy_1|x_1}] \cdot [\rho_{\system{S}_0}],\\
[\lvec{\varrho}_{\system{S}_\ell}^{(\cvy_{\ell\!+\!1}^n)}] &\defeq
\sum_{\vx_{\ell\!+\!1}^n} Q(\vx_{\ell+1}^n) \cdot [I_{\system{S}_n}] \cdot [\hat{W}^{\cy_n|x_n}] \cdots [\hat{W}^{\cy_{\ell\!+\!1}|x_{\ell\!+\!1}}],\\
[\rvec{\varrho}_{\system{S}_\ell}^{(\cvx_1^\ell,\cvy_1^\ell)}] &\defeq
[\hat{W}^{\cy_\ell|\cx_\ell}] \cdots [\hat{W}^{\cy_1|\cx_1}] \cdot [\rho_{\system{S}_0}],\\
[\lvec{\varrho}_{\system{S}_\ell}^{(\cvx_{\ell\!+\!1}^n,\cvy_{\ell\!+\!1}^n)}] &\defeq
[I_{\system{S}_n}] \cdot [\hat{W}^{\cy_n|\cx_n}] \cdots [\hat{W}^{\cy_{\ell\!+\!1}|\cx_{\ell\!+\!1}}].
\end{align}
Recall that, in above equations, $[I_{\system{S}_n}]$ is a row vector, and $[\rho_{\system{S}_0}]$ is a column vector.
\par
%*******************************************************************************
By extending the domain of the functions $\IRLB_{W}^{(n)}$, $\IRUB_{W}^{(n)}$, and $\Delta_{W}^{(n)}$ to include \emph{all} PSD matrices $\llbracket \hat{W}^{y|x}\rrbracket$, one can omit the linear constraint~\eqref{eq:tangent:linear:constraint}.
Namely, the ``direction'' $\{\llbracket H^{y|x}\rrbracket\}_{x,y}$ can take any Hermitian matrices.
Using some linear algebra, the gradient w.r.t. $\hat{W}$ of these functions on this \emph{extended} domain can be expressed as 
\begin{align}
\label{eq:IRUB:grad:1}
\left(\grad\IRUB_{W,\mathrm{ext}}^{(n)}(\hat{W})\right)^{y|x}
&\propto -\frac{1}{n} \expectationwrt{
\sum_{k=1}^{n}\delta_{\rv{X_k},x} \cdot \delta_{\rv{Y_k},y} \cdot \rvec{\varrho}_{\hat{\system{S}}_{k\!-\!1}}^{(\rv{Y}_1^{k\!-\!1})} \tensor \lvec{\varrho}_{\hat{\system{S}}_k}^{(\rv{Y}_{k\!+\!1}^n)}}{\rv{X}_1^n\rv{Y}_1^n},\\
\label{eq:IRLB:grad:1}
\left(\grad \IRLB_{W,\mathrm{ext}}^{(n)}(\hat{W})\right)^{y|x}
&\propto -\frac{1}{n} \Bigg\langle \sum_{k=1}^{n} \delta_{\rv{X_k},x} \cdot \delta_{\rv{Y_k},y} \cdot
	\begin{aligned}[t]
	&\Bigg(\rvec{\varrho}_{\hat{\system{S}}_{k\!-\!1}}^{(\rv{Y}_1^{k\!-\!1})} \tensor \lvec{\varrho}_{\hat{\system{S}}_k}^{(\rv{Y}_{k\!+\!1}^n)} - \\
    &\rvec{\varrho}_{\hat{\system{S}}_{k\!-\!1}}^{(\rv{X}_1^{k\!-\!1},\rv{Y}_1^{k\!-\!1})} \tensor \lvec{\varrho}_{\hat{\system{S}}_k}^{(\rv{X}_{k\!+\!1}^n,\rv{Y}_{k\!+\!1}^n)}\Bigg)\Bigg\rangle_{\rv{X}_1^n\rv{Y}_1^n},
	\end{aligned}\\
\label{eq:DIFF:grad:1}
\left(\grad \Delta_{W,\mathrm{ext}}^{(n)}(\hat{W})\right)^{y|x}
&\propto -\frac{1}{n} \expectationwrt{
\sum_{k=1}^{n} \delta_{\rv{X_k},x} \cdot \delta_{\rv{Y_k},y} \cdot \rvec{\varrho}_{\hat{\system{S}}_{k\!-\!1}}^{(\rv{X}_1^{k\!-\!1},\rv{Y}_1^{k\!-\!1})} \tensor \lvec{\varrho}_{\hat{\system{S}}_k}^{(\rv{X}_{k\!+\!1}^n,\rv{Y}_{k\!+\!1}^n)}}{\rv{X}_1^n\rv{Y}_1^n},
\end{align}
respectively.
%*******************************************************************************
For stationary and ergodic input and output processes $(\rv{X}_1^n,\rv{Y}_1^n)$, we can \emph{estimate}~\eqref{eq:IRUB:grad:1} and~\eqref{eq:DIFF:grad:1}, respectively, as
\begin{align}
\label{eq:IRUB:grad:2}
\left(\grad\IRUB^{(n)}_{W,\mathrm{ext}}(\hat{W})\right)^{y|x}
&\overset{\cdot}{\propto} -\frac{1}{n} \sum_{k:\:{\cx_k=x \atop \cy_k=y}} \rvec{\varrho}_{\hat{\system{S}}_{k\!-\!1}}^{(\cvy_1^{k\!-\!1})} \tensor \lvec{\varrho}_{\hat{\system{S}}_k}^{(\cvy_{k\!+\!1}^n)},\\
\label{eq:DIFF:grad:2}
\left(\grad\Delta_{W,\mathrm{ext}}^{(n)}(\hat{W})\right)^{y|x}
&\overset{\cdot}{\propto} -\frac{1}{n} \sum_{k:\:{\cx_k=x \atop \cy_k=y}} \rvec{\varrho}_{\hat{\system{S}}_{k\!-\!1}}^{(\cvx_1^{k\!-\!1},\cvy_1^{k\!-\!1})} \tensor \lvec{\varrho}_{\hat{\system{S}}_k}^{(\cvx_{k\!+\!1}^n,\cvy_{k\!+\!1}^n)},
\end{align}
where $(\cvx_1^n,\cvy_1^n)$ is a realization of the channel input/output processes generated by the original channel model.
The dot in~\eqref{eq:IRUB:grad:2} and~\eqref{eq:DIFF:grad:2} stands for ``approximation''.
Notice that the messages $\rvec{\varrho}_{\system{S}_{k\!-\!1}}^{(\cvy_1^{k\!-\!1})}$, $\lvec{\varrho}_{\system{S}_{k}}^{(\cvy_{k\!+\!1}^n)}$, $\rvec{\varrho}_{\system{S}_{k\!-\!1}}^{(\cvx_1^{k\!-\!1},\cvy_1^{k\!-\!1})}$, and $\lvec{\varrho}_{\system{S}_k}^{(\cvx_{k\!+\!1}^n,\cvy_{k\!+\!1}^n)}$ can be computed iteratively.
Thus,~\eqref{eq:IRUB:grad:2} and \eqref{eq:DIFF:grad:2} provide efficient means to estimate the gradient.
However, due to the extension of the domain, the gradients computed above may not satisfy constraint~\eqref{eq:tangent:linear:constraint}.
This can be compensated using a projection \wrt the linear constraint, which can be solved using linear programming.
On the other hand, the above gradient method may lead to a violation of the PSD condition required by CC-QSCs.
However, since the feasible domain of CC-QSCs is convex and bounded, this can be corrected using convex programming at each step.
\par
%*******************************************************************************
We summarize the above discussion as Algorithm~\ref{alg:grad:1}, which is an iterative gradient-descent method for minimizing $\Delta_{W}^{(n)}$.
Notice that the quantity $\lambda_\ell$ in this case is the conditional probability $\prob_{\rv{X}_\ell\rv{Y}_\ell|\rv{X}_{1}^{\ell\!-\!1}\rv{Y}_{1}^{\ell\!-\!1}}(\cx_{\ell},\cy_\ell|\cvx_1^{\ell\!-\!1},\cvy_1^{\ell\!-\!1})$.
The algorithm for minimizing the upper and lower bounds are similar, and we omit the details.
\begin{algorithm}
\caption{Optimizing the Difference Function $\Delta_{W}^{(n)}(\hat{W})$}
\begin{algorithmic}[1]
\Require{An indecomposable~CC-QSC, an input~distribution~$Q$, a positive~integer~$n$ large enough, an initial~auxiliary~CC-QSC~$\{\hat{W}^{y|x}\}_{x,y}$, step~size~$\gamma>0$.}
\Ensure{$\{\hat{W}^{y|x}\}_{x,y}$, an estimated local minimum point of $\Delta_{W}^{(n)}$.}
\State Initialize the memory density operator $\rho_{\hat{\system{S}_0}} \gets \braket{0}$
\State Generate an input sequence $\cvx_1^n \sim Q^{\tensor n}$
\State Generate a corresponding output sequence $\cvy_1^n$
\Repeat
\State $\rvec{\varrho}_{\hat{\system{S}}_0} \gets \rho_{\hat{\system{S}}_0}$
\ForEach{$\ell=1,\ldots,n$}
\State $[\rvec{\varrho}_{\hat{\system{S}}_\ell}] \gets [\hat{W}^{\cy_\ell|\cx_\ell}] \cdot [\rvec{\varrho}_{\hat{\system{S}}_{\ell\!-\!1}}]$
\State $\lambda_\ell \gets \tr(\rvec{\varrho}_{\hat{\system{S}}_\ell})$
\State $\rvec{\varrho}_{\hat{\system{S}}_\ell} \gets \lambda_\ell^{-1} \cdot \rvec{\varrho}_{\hat{\system{S}}_\ell}$
\EndFor
\State $\lvec{\varrho}_{\hat{\system{S}}_n} \gets I_{\hat{\system{S}}_n}$
\ForEach{$\ell=n,\ldots,1$}
\State $[\lvec{\varrho}_{\hat{\system{S}}_{\ell\!-\!1}}] \gets [\lvec{\varrho}_{\hat{\system{S}}_\ell}] \cdot [\hat{W}^{\cy_\ell|\cx_\ell}]$
\State $\lvec{\varrho}_{\hat{\system{S}}_{\ell\!-\!1}} \gets \left( \tr( \lvec{\varrho}_{\hat{\system{S}}_{\ell\!-\!1}})\right)^{-1} \cdot \lvec{\varrho}_{\hat{\system{S}}_{\ell\!-\!1}}$
\EndFor
\State for each $x,y$, let $\left(\grad\Delta_{W,\mathrm{ext}}^{(n)}(\hat{W})\right)^{y|x} \gets \mathbf{0}$ 
\ForEach{$k=1,\ldots,n$}
	\State $\left( \grad\Delta_{W,\mathrm{ext}}^{(n)}(\hat{W})\right)^{\cy_k|\cx_k} \mathrel{+}= \frac{1}{n} \cdot \frac{\rvec{\varrho}_{\hat{\system{S}}_{k\!-\!1}} \tensor \lvec{\varrho}_{\hat{\system{S}}_k}}{\lambda_k \cdot \tr(\rvec{\varrho}_{\hat{\system{S}}_k} \cdot \lvec{\varrho}_{\hat{\system{S}}_k})}$
\EndFor
\State Project $\{\left( \grad\Delta_{W,\mathrm{ext}}^{(n)}(\hat{W})\right)^{y|x}\}_{x,y}$ onto the subspace satisfying~\eqref{eq:tangent:linear:constraint}; denoting the result by $\left\{\left(\grad\Delta_{W}^{(n)}(\hat{W})\right)^{y|x}\right\}_{x,y}$
\State $\{\hat{W}^{y|x}\}_{x,y}\gets\{\hat{W}^{y|x}\}_{x,y} - \gamma \cdot \left\{ \left( \grad\Delta_{W}^{(n)}(\hat{W}) \right)^{y|x} \right\}_{x,y}$
\State Solve the following convex program \wrt $\{\tilde{W}^{y|x}\}_{x,y}$:
\[\begin{aligned}
\min\quad & \sum\nolimits_{x,y} \tr\left( (\llbracket\tilde{W}^{y|x}\rrbracket - \llbracket\hat{W}^{y|x}\rrbracket) \cdot (\llbracket\tilde{W}^{y|x}\rrbracket - \llbracket\hat{W}^{y|x}\rrbracket)^\Herm \right) \\
\st\quad & \llbracket\tilde{W}^{y|x}\rrbracket \in \mathbb{C}^{\set{S}^2\times\set{S}^2} \text{being PSD} &&\forall x,y\\
& \sum\nolimits_{y\in\set{Y}}\sum\nolimits_{s',\tilde{s}':\: s'=\tilde{s}'} \llbracket\tilde{W}^{y|x}\rrbracket_{(s',s),(\tilde{s}',\tilde{s})} = \delta_{s,\tilde{s}} &&\forall x
\end{aligned}\]
\State $\{\hat{W}^{y|x}\} \gets \{\tilde{W}^{y|x}\}$
\Until{$\{\hat{W}^{y|x}\}_{x,y}$ has converged.}
\end{algorithmic}
\label{alg:grad:1}
\end{algorithm}
%******************************************************************************
%******************************************************************************
\section{Example: Quantum Gilbert--Elliott Channels}\label{QCwM:sec:6:example}
In this section we present some numerical results as a demonstration of the algorithms introduced in this chapter.
In particular, as a generalization of Example~\ref{example:GEC}, we consider a class of quantum channels with memory named the quantum Gilbert--Elliott channels (QGECs), which were introduced in~\cite{cao2017estimating}, and consider their information rate using some separable input ensemble and local output measurement.
Note that the numerical results in this section are based on binary logarithm, and thus the information rate is measured in bits.
\par
%*******************************************************************************
\begin{figure}\centering
\begin{tikzpicture}[every node/.style={transform shape}]
	\node (S) {$\system{S}$};
	\node[below = 10pt of S] (A) {$\system{A}$};
	\node[below = 10pt of A] (E) {$\bra{0}$};
	\path[draw=none] (A) edge[draw=none] node[midway] (AE) {} (E);
	\node[draw, minimum width = 20pt, minimum height = 42pt, right = 20pt of AE] (CBF) {$U^{\bra{s}}$};
	\node[right = 80pt of A] (B) {$\system{B}$};
	\node[right = 52pt of E, draw, minimum size = 12pt] (M) {};
	\node[minimum size = 1.5pt, fill = black, inner sep = 0pt, outer sep = 0pt, below = 3pt of M.center, circle, anchor = center] (m) {};
	\draw [decoration={markings,mark=at position 1 with {\arrow[scale=0.5,>=latex]{>}}},postaction={decorate}, draw = none] (m.center) -- ([yshift=7pt, xshift=3pt]m.center);
	\draw (m.center) -- ([yshift=6.3pt, xshift=2.7pt]m.center);
	\draw ([xshift=4pt]m.center) arc (0:180:4pt);
	\draw[line width = 1.3pt] (M.east) -- ([xshift=13.5pt]M.east) node[pos=1, minimum size = 5pt, inner sep = 0pt, outer sep = 0pt] (Eend) {};
	\draw[line width = 1.3pt] (Eend.south west) -- (Eend.north east);
	\draw[line width = 1.3pt] (Eend.north west) -- (Eend.south east);
	\node[right = 52pt of S, draw, minimum size = 12pt] (Us) {$V_\system{S}$};
	\node at (Us-|B) {$\system{S}'$};
	\draw (A) -- (A-|CBF.west);
	\draw (E) -- (E-|CBF.west);
	\draw (M-|CBF.east) -- (M);
	\draw (B-|CBF.east) -- (B);
	\draw (S) -- (Us);
	\draw[-latex] (Us) -- (Us-|B.west) -- ([yshift=15pt]Us-|B.west) -- ([yshift=15pt]S.east) -- (S.east);
	\draw[-*] (CBF) -- (S-|CBF);
	\begin{pgfonlayer}{bg}
		\draw [fill=red!20,draw=none] ([xshift=-90pt, yshift=18pt]S) rectangle ([xshift=5pt,yshift=-10pt]Us-|B.east);
		\node [left=85pt of S.center, anchor=west, color = red] {memory system};
		\draw [fill=blue!20,draw=none] ([xshift=-90pt, yshift=10pt]A) rectangle ([xshift=5pt,yshift=-10pt]B.east);
		\node [left=85pt of A.center, anchor=west, color = blue] {primary system};
		\draw [fill=yellow!20,draw=none] ([xshift=-90pt, yshift=10pt]E) rectangle ([xshift=5pt,yshift=-10pt]E-|B.east);
		\node [left=85pt of E.center, anchor=west] {enviroment};
	\end{pgfonlayer}
\end{tikzpicture}
\begin{tikzpicture}[every node/.style={transform shape}]
	\node (S) {$\system{S}$};
	\node[below = 10pt of S] (A) {$\system{A}$};
	\node[below = 10pt of A] (E) {$\bra{0}$};
	\path[draw=none] (A) edge[draw=none] node[midway] (AE) {} (E);
	\node[draw, minimum width = 20pt, minimum height = 42pt, right = 20pt of AE] (CBF) {$U^{\bra{s}}$};
	\node[right = 80pt of A] (B) {$\system{B}$};
	\node[right = 52pt of E, draw, minimum size = 12pt] (M) {};
	\node[minimum size = 1.5pt, fill = black, inner sep = 0pt, outer sep = 0pt, below = 3pt of M.center, circle, anchor = center] (m) {};
	\draw [decoration={markings,mark=at position 1 with {\arrow[scale=0.5,>=latex]{>}}},postaction={decorate}, draw = none] (m.center) -- ([yshift=7pt, xshift=3pt]m.center);
	\draw (m.center) -- ([yshift=6.3pt, xshift=2.7pt]m.center);
	\draw ([xshift=4pt]m.center) arc (0:180:4pt);
	\draw[line width = 1.3pt] (M.east) -- ([xshift=13.5pt]M.east) node[pos=1, minimum size = 5pt, inner sep = 0pt, outer sep = 0pt] (Eend) {};
	\draw[line width = 1.3pt] (Eend.south west) -- (Eend.north east);
	\draw[line width = 1.3pt] (Eend.north west) -- (Eend.south east);
	\node[right = 52pt of S, draw, minimum size = 12pt] (Us) {$\tilde{V}_\system{S}$};
	\node at (Us-|B) {$\system{S}'$};
	\draw (A) -- (A-|CBF.west);
	\draw (E) -- (E-|CBF.west);
	\draw (M-|CBF.east) -- (M);
	\draw (B-|CBF.east) -- (B);
	\draw ([yshift=-2pt]S.east) -- ([yshift=-2pt]Us.west);
	\draw[-latex] ([yshift=-2pt]Us.east) -- ([yshift=-2pt]Us.east-|B.west) -- ([yshift=15pt]Us-|B.west) -- ([yshift=15pt]S.east) -- ([yshift=-2pt]S.east);
	\draw[-*] (CBF) -- ([yshift=-2pt]S-|CBF);
	\draw ([xshift=3pt, yshift=2pt]S.east) -- node[pos=0.8] (marker1) {} ([yshift=2pt]Us.west);
	\draw[-latex] ([yshift=2pt]Us.east) -- node[midway] (marker2) {} ([yshift=2pt,xshift=-3pt]Us.east-|B.west) -- ([yshift=12pt,xshift=-3pt]Us-|B.west) -- ([xshift=3pt,yshift=12pt]S.east) -- ([xshift=3pt,yshift=2pt]S.east);
	\draw ([xshift=-2pt, yshift=-2pt]marker1.center) -- ([xshift=2pt, yshift=2pt]marker1.center);
	\draw ([xshift=-2pt, yshift=-2pt]marker2.center) -- ([xshift=2pt, yshift=2pt]marker2.center);
	\begin{pgfonlayer}{bg}
		\draw [fill=red!20,draw=none] ([xshift=-90pt, yshift=18pt]S) rectangle ([xshift=5pt,yshift=-10pt]Us-|B.east);
		\node [left=85pt of S.center, anchor=west, color = red] {memory system};
		\draw [fill=blue!20,draw=none] ([xshift=-90pt, yshift=10pt]A) rectangle ([xshift=5pt,yshift=-10pt]B.east);
		\node [left=85pt of A.center, anchor=west, color = blue] {primary system};
		\draw [fill=yellow!20,draw=none] ([xshift=-90pt, yshift=10pt]E) rectangle ([xshift=5pt,yshift=-10pt]E-|B.east);
		\node [left=85pt of E.center, anchor=west] {enviroment};
	\end{pgfonlayer}
\end{tikzpicture}
\caption{A quantum Gilbert--Elliott channel (LHS), and a variant where the memory system consists of multiple qubits with only one of them controlling $U^{\bra{s}}$ (RHS).}
\label{fig:QGEC}
\end{figure}
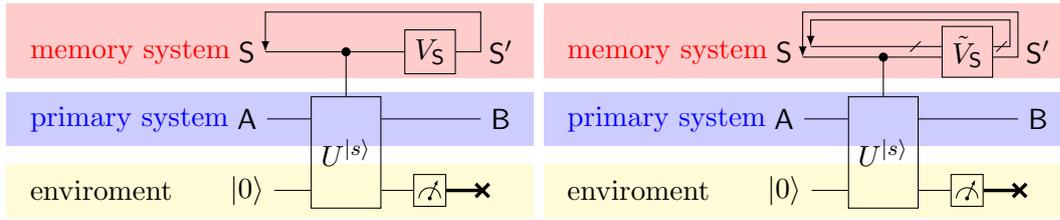
A QGEC is a quantum channel with memory defined by\footnote{
We put the system $\system{S}$ ahead of $\system{A}$ and $\system{B}$ in this example to emphasize the role of $\system{S}$ as a \emph{control} qubit, and also for simplicity reasons.}
\begin{alignat*}{2}
&\operator{N}:\:& \DensOp(\hilbert_\system{S}\tensor\hilbert_\system{A}) 
& \to \DensOp(\hilbert_{\system{S}'}\tensor\hilbert_\system{B})\\
& & \rho_{\system{SA}} & \mapsto (V_\system{S}\tensor I_\system{B}) \cdot \Phi^{\mathrm{CBF}}(\rho_{\system{SA}}) \cdot (V_\system{S}^\Herm\tensor I_\system{B}),
\end{alignat*}
where $\hilbert_\system{A}$, $\hilbert_\system{B}$, and $\hilbert_\system{S}=\hilbert_{\system{S}'}$ are of dimension 2, namely each of them is made up of one qubit, and where $\Phi^{\mathrm{CBF}}$ is the \emph{controlled bit-flip channel} defined by $\Phi^{\mathrm{CBF}}(\rho_\system{SA}) \defeq E_0\rho^{\mathrm{SA}} E_0^\Herm + E_1\rho^{\mathrm{SA}} E_1^\Herm $ with
\[
E_0 \defeq \left[ \begin{smallmatrix}
\!\!\sqrt{1\!-\!\pgood}\!\! & 0                & 0   & 0 \\
0                & \!\!\sqrt{1\!-\!\pgood}\!\! & 0   & 0 \\
0                & 0                & \!\!\sqrt{1\!-\!\pbad}\!\! & 0 \\
0                & 0                & 0   & \!\!\sqrt{1\!-\!\pbad}\!\!
\end{smallmatrix}\right],\ 
E_1 \defeq \left[ \begin{smallmatrix}
0                  & \sqrt{\pgood} & 0     & 0 \\
\sqrt{\pgood} & 0            & 0           & 0 \\
0                  & 0            & 0   & \sqrt{\pbad} \\
0                  & 0            & \sqrt{\pbad} & 0
\end{smallmatrix}\right],
\]
and where $V_\system{S}$ is some unitary operator on $\hilbert_\system{S}$ to be specified later.
The controlled bit-flip channel $\Phi^{\mathrm{CBF}}$ applies a quantum bit-flip channel to the system $\system{A}$ with flipping probability $\pgood$ when the system $\system{S}$ is in the state of $\bra{0}$, and with flipping probability $\pbad$ when the system $\system{S}$ is in the state of $\bra{1}$.
The action of a QGEC is the combined effect of a controlled bit-flip channel and a unitary evolution on $\system{S}$; as depicted in the following circuit diagram in Figure~\ref{fig:QGEC}, where $U^{\bra{s}}$ is a Stinespring representation of $\Phi^{\mathrm{CBF}}$:
\[
U^{\bra{0}} \defeq\left[ \begin{smallmatrix}
\sqrt{1\!-\!\pgood}  & 0
& 0                          & -\sqrt{\pgood}  \\
0                            & \sqrt{1\!-\!\pgood}
&\sqrt{\pgood}      & 0 \\
0                            & -\sqrt{\pgood}
& \sqrt{1\!-\!\pgood}& 0 \\
\sqrt{\pgood}        & 0
& 0                          & \sqrt{1\!-\!\pgood}
\end{smallmatrix}\right],\ 
U^{\bra{1}} \defeq \left[ \begin{smallmatrix}
\sqrt{1\!-\!\pbad}  & 0
& 0                          & -\sqrt{\pbad}  \\
0                            & \sqrt{1\!-\!\pbad}
&\sqrt{\pbad}      & 0 \\
0                            & -\sqrt{\pbad}
& \sqrt{1\!-\!\pbad}& 0 \\
\sqrt{\pbad}        & 0
& 0                          & \sqrt{1\!-\!\pbad}
\end{smallmatrix}\right].
\]
\par
%*******************************************************************************
% Numerical Results
In Figure~\ref{fig:QGEC:plot:1}--\ref{fig:QGEC:plot:4}, we present some numerical information rate lower bounds estimated for a QGEC and a variant of a QGEC (as depicted in Figure~\ref{fig:QGEC}), equipped with ``trivial'' orthonormal ensemble and projective measurements.
Namely, the original channel in Figure~\ref{fig:QGEC:plot:1} and~\ref{fig:QGEC:plot:3} can be described by the CC-QSC
\begin{equation} \label{eq:qgec:qsc:1}
\operator{N}^{y|x}(\rho_\system{S}) = \tr_\system{B} \left( (V_\system{S}^\Herm V_\system{S} \tensor \braket{y}) \cdot \Phi^{\mathrm{CBF}}(\rho_\system{S} \tensor \braket{x}) \right),
\end{equation}
whereas Figure~\ref{fig:QGEC:plot:2} and~\ref{fig:QGEC:plot:4} is described by
\begin{equation} \label{eq:qgec:qsc:2}
\operator{N}^{y|x}(\rho_\system{S}) = \tr_\system{B} \left( (\tilde{V}_\system{S}^\Herm \tilde{V}_\system{S} \tensor
    \braket{y}) \cdot (\id \tensor \Phi^{\mathrm{CBF}})(\rho_\system{S} \tensor \braket{x}) \right),
\end{equation}
where $\{\bra{x}\}_{x\in\set{X}}$ and $\{\bra{y}\}_{y\in\set{Y}}$ are some orthonormal basis of $\hilbert_{\system{A}}$ and $\hilbert_{\system{B}}$, respectively.
In the latter case, the memory system $\system{S}$ is extended as $\hilbert_{\system{S}} = \hilbert_{\system{S}_{1}} \tensor \hilbert_{\system{S}_{0}}$.
More specifically, in~\eqref{eq:qgec:qsc:2}, $\rho_\system{S}$ and $\tilde{V}_\system{S}$ are operators on $\hilbert_{\system{S}}$, and $\Phi^{\mathrm{CBF}}$ acts on $\DensOp(\hilbert_{\system{S}_{0}}\tensor \hilbert_{\system{A}})$, and $\id$ is the identity map on $\system{S}_{1}$.
For both scenarios, the input processes are binary symmetric {i.i.d.}~processes, \ie, $Q^{(n)}(\vx_1^n)\defeq 2^{-n}$ for all $\vx_1^n\in\{0,1\}^n$.
The lower bounds in those figures were obtained by minimizing the difference function $\Delta_{W}^{(n)}$ defined in~\eqref{eq:delta} \wrt different classes of auxiliary channels (subject to certain time and threshold constraints).
For the case where the auxiliary channels are CC-QSCs, Algorithm~\ref{alg:grad:1} was applied.
For FSMC auxiliary channels, we implemented the expectation-maximization type algorithm in~\cite{sadeghi2009optimization} for comparison.
As already emphasized beforehand, these lower bounds represent rates that are achievable with the help of a mismatched decoder~\cite{ganti2000mismatched}.
Figure~\ref{fig:QGEC:plot:5} is an example illustrating the typical convergence time of different methods (including our own) for minimizing the difference function.
In all of the above figures, $n=10,000$, and we have used Algorithm~\ref{alg:SPA:2} to \emph{estimate} the information rate.
According to our experience, the error of the estimation in this case lies within the line-width in the figures.
\begin{figure}
\begin{tikzpicture}[
	every axis/.append style={font=\footnotesize},
	every mark/.append style={scale=1.5}]
\begin{axis}[
	xlabel={$\pbad$},
	ylabel={bits per channel use},
	ylabel shift=-4 pt,
	xmin=0, xmax=1, ymin=0,  ymax=0.9,
	legend style={font=\scriptsize}, 
	mark repeat={4}, 
	width=\columnwidth, height=.618\columnwidth,
	ytick={0.1,0.2,0.3,0.4,0.5,0.6,0.7,0.8},
	extra y ticks={0,0.9},
	extra y tick style={grid=none},
	y tick label style={
		/pgf/number format/fixed,
		/pgf/number format/fixed zerofill,
		/pgf/number format/precision=1},
	xtick={0.2,0.4,0.6,0.8},
	extra x ticks={0.0,1.0},
	extra x tick style={grid=none},
	x tick label style={
		/pgf/number format/fixed,
		/pgf/number format/fixed zerofill,
		/pgf/number format/precision=1},
	grid=both, grid style={line width=0.1pt, dashed},
	legend cell align={left}]
	\addplot[dashed, thick, black] coordinates {
		(0.00, 0.839601745174) (0.05, 0.714120015688)
		(0.10, 0.618261781187) (0.15, 0.538893269110)
		(0.20, 0.470418506076) (0.25, 0.411550042780)
		(0.30, 0.361325752184) (0.35, 0.315858686126)
		(0.40, 0.276693640181) (0.45, 0.242634877108)
		(0.50, 0.213364467275) (0.55, 0.187646703615)
		(0.60, 0.166111113269) (0.65, 0.147285664160)
		(0.70, 0.130630670279) (0.75, 0.115773844628)
		(0.80, 0.101787384142) (0.85, 0.087948744743)
		(0.90, 0.073629956869) (0.95, 0.055286167568)
		(1.00, 0.030194299618)};
	\addlegendentry{Estimated IR}
	\addplot[solid, red, dashed] coordinates {
		(0.00, 0.833686348679) (0.05, 0.711496512010)
		(0.10, 0.622781653924) (0.15, 0.531787273503)
		(0.20, 0.456691524228) (0.25, 0.394254900201)
		(0.30, 0.353083453813) (0.35, 0.306942832207)
		(0.40, 0.268742612068) (0.45, 0.234172015381)
		(0.50, 0.197625745231) (0.55, 0.179086497776)
		(0.60, 0.154340415942) (0.65, 0.135508136380)
		(0.70, 0.118932943912) (0.75, 0.108928054432)
		(0.80, 0.097114178557) (0.85, 0.082541871698)
		(0.90, 0.067974760037) (0.95, 0.055364118464)
		(1.00, 0.030712333955)};
	\addlegendentry{EM-type Algorithm~[SVS09] with 4-state FSMC}
	\addplot[solid, blue, dashed] coordinates {
		(0.00, 0.831741905571) (0.05, 0.712427595410)
		(0.10, 0.613164794327) (0.15, 0.533958257063)
		(0.20, 0.457099591967) (0.25, 0.396516510418)
		(0.30, 0.333931200180) (0.35, 0.276208896384)
		(0.40, 0.231405341707) (0.45, 0.193564225477)
		(0.50, 0.156714715230) (0.55, 0.126192277424)
		(0.60, 0.094889516489) (0.65, 0.075748267428)
		(0.70, 0.057208456686) (0.75, 0.040925933725)
		(0.80, 0.030531953773) (0.85, 0.020645659910)
		(0.90, 0.015760096465) (0.95, 0.009658573661)
		(1.00, 0.008579667796)};
	\addlegendentry{EM-type Algorithm~[SVS09] with 2-state FSMC}
	\addplot[solid, mark = diamond, blue] coordinates {
		(0.00, 0.828760809658) (0.05, 0.713985929442)
		(0.10, 0.615994378803) (0.15, 0.528854930308)
		(0.20, 0.459058469386) (0.25, 0.389432982447)
		(0.30, 0.341934887183) (0.35, 0.307475093096)
		(0.40, 0.282804615465) (0.45, 0.244481977767)
		(0.50, 0.215103229847) (0.55, 0.190116309030)
		(0.60, 0.168708522517) (0.65, 0.146737015048)
		(0.70, 0.132491368761) (0.75, 0.116971911450)
		(0.80, 0.099518348532) (0.85, 0.086967374828)
		(0.90, 0.071632787144) (0.95, 0.051482497374)
		(1.00, 0.030962259381) };
	\addlegendentry{Algorithm~4.3 with 1-qubit QSC}
\end{axis}
\end{tikzpicture}
\caption{Quantum Gilbert–Elliott channel: $\pgood = 0.05$ is fixed; $\pbad$ varies from $0$ to $1$; $V_\system{S} = \exp(-j \alpha H)$, where $H$ is some fixed $2$-by-$2$ Hermitian matrix and where $\alpha = 1$ is fixed; $n=10^5$.}
\label{fig:QGEC:plot:1}
\end{figure}
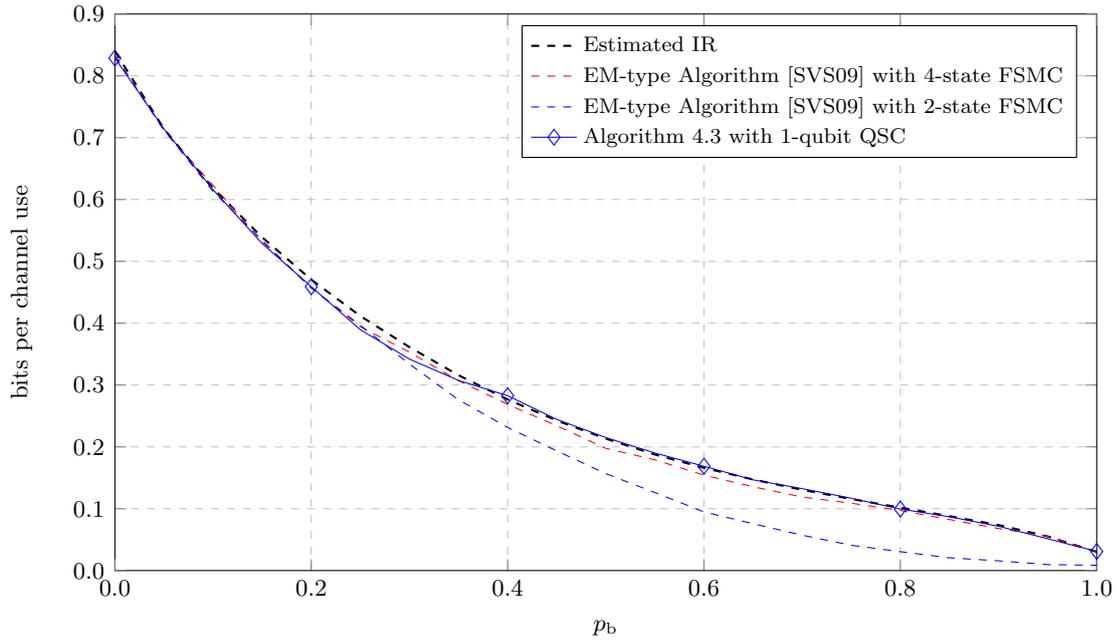
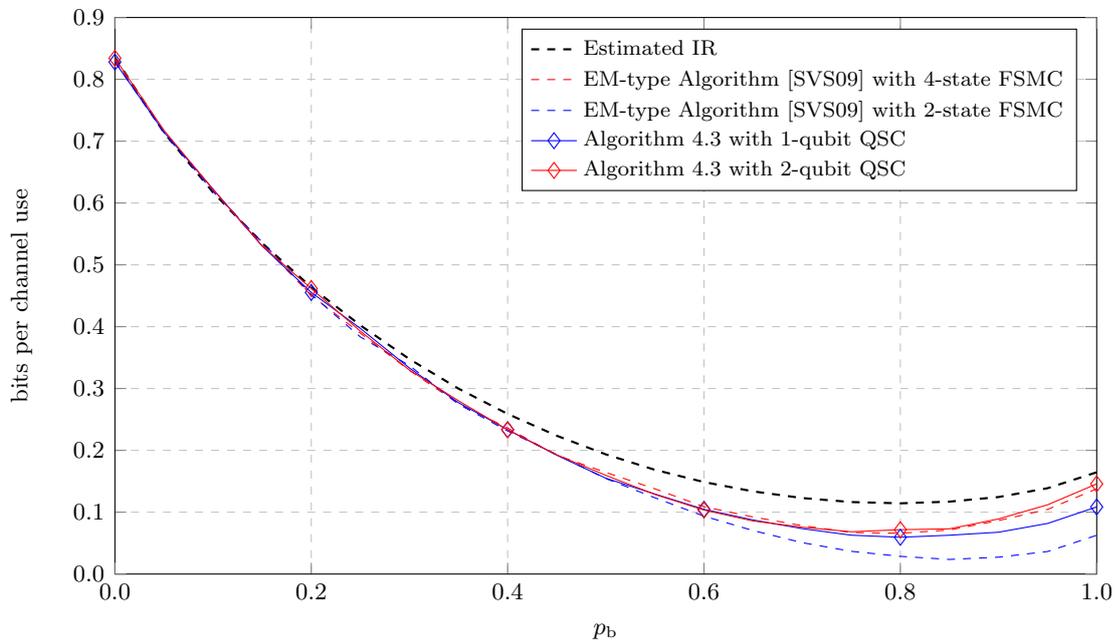
\begin{figure}
\begin{tikzpicture}[
	every axis/.append style={font=\footnotesize},
	every mark/.append style={scale=1.5}]
\begin{axis}[
	xlabel={$\pbad$}, 
	ylabel={bits per channel use},
	ylabel shift=-4 pt,
	xmin=0, xmax=1, ymin=0,  ymax=0.9,
	legend style={font=\scriptsize}, 
	mark repeat={4}, 
	width=\columnwidth, height=.618\columnwidth,
	y tick label style={
		/pgf/number format/fixed,
		/pgf/number format/fixed zerofill,
		/pgf/number format/precision=1},
	ytick={0.1,0.2,0.3,0.4,0.5,0.6,0.7,0.8},
	extra y ticks={0,0.9},
	extra y tick style={grid=none},
	xtick={0.2,0.4,0.6,0.8},
	extra x ticks={0.0,1.0},
	extra x tick style={grid=none},
	x tick label style={
		/pgf/number format/fixed,
		/pgf/number format/fixed zerofill,
		/pgf/number format/precision=1},
	grid=both, grid style={line width=0.1pt, dashed},
	legend cell align={left}]
	\addplot[dashed, black, thick] coordinates {
		(0.00, 0.835334547072) (0.05, 0.713817138457)
		(0.10, 0.616995721117) (0.15, 0.535170417981)
		(0.20, 0.463871511452) (0.25, 0.402218381517)
		(0.30, 0.347600395229) (0.35, 0.299817809855)
		(0.40, 0.258942093992) (0.45, 0.223665276283)
		(0.50, 0.193521758116) (0.55, 0.168561414250)
		(0.60, 0.148733205093) (0.65, 0.133686946797)
		(0.70, 0.122912164107) (0.75, 0.116431195033)
		(0.80, 0.114371410381) (0.85, 0.117102579296)
		(0.90, 0.124549747387) (0.95, 0.138826478090)
		(1.00, 0.164582211924)};
	\addlegendentry{Estimated IR}
	\addplot[solid, red, dashed] coordinates {
		(0.00, 0.831517212717) (0.05, 0.714764648111)
		(0.10, 0.620863015847) (0.15, 0.532861848575)
		(0.20, 0.450692752931) (0.25, 0.389546432757)
		(0.30, 0.329290008573) (0.35, 0.275553175585)
		(0.40, 0.236190228422) (0.45, 0.192090727923)
		(0.50, 0.164389649731) (0.55, 0.137826355554)
		(0.60, 0.108796192389) (0.65, 0.092290544995)
		(0.70, 0.078342868664) (0.75, 0.067132930525)
		(0.80, 0.065980300478) (0.85, 0.071230799275)
		(0.90, 0.086533592418) (0.95, 0.104559960603)
		(1.00, 0.137889776752)};
	\addlegendentry{EM-type Algorithm~[SVS09] with 4-state FSMC}
	\addplot[solid, blue, dashed] coordinates {
		(0.00, 0.828833977289) (0.05, 0.713401902069)
		(0.10, 0.618383090362) (0.15, 0.536796068619)
		(0.20, 0.452103146643) (0.25, 0.383414537377)
		(0.30, 0.337204267691) (0.35, 0.274681469889)
		(0.40, 0.231386008521) (0.45, 0.193295725224)
		(0.50, 0.154478319411) (0.55, 0.123336899762)
		(0.60, 0.093653236031) (0.65, 0.070335085534)
		(0.70, 0.050775825078) (0.75, 0.036778413585)
		(0.80, 0.028658913432) (0.85, 0.023617608576)
		(0.90, 0.027313636512) (0.95, 0.036456654565)
		(1.00, 0.062739939978)};
	\addlegendentry{EM-type Algorithm~[SVS09] with 2-state FSMC}
	\addplot[solid, mark = diamond, blue] coordinates {
		(0.00, 0.827985156616) (0.05, 0.714668236552)
		(0.10, 0.622321223551) (0.15, 0.530413033978)
		(0.20, 0.455413779164) (0.25, 0.397241693779)
		(0.30, 0.333206279519) (0.35, 0.277076762749)
		(0.40, 0.233492782705) (0.45, 0.192237937528)
		(0.50, 0.154956559171) (0.55, 0.129387917312)
		(0.60, 0.104722394072) (0.65, 0.087434544657)
		(0.70, 0.073503966393) (0.75, 0.062853525376)
		(0.80, 0.059483982086) (0.85, 0.062788063958)
		(0.90, 0.067646229180) (0.95, 0.081847657751)
		(1.00, 0.108539715332)};
	\addlegendentry{Algorithm~4.3 with 1-qubit QSC}
	\addplot[solid, mark = diamond, red] coordinates {
		(0.00, 0.834042100608) (0.05, 0.717419247603)
		(0.10, 0.621893971509) (0.15, 0.530616289996)
		(0.20, 0.461705781004) (0.25, 0.392711338005)
		(0.30, 0.329821195930) (0.35, 0.280192615343)
		(0.40, 0.233190489326) (0.45, 0.192814793929)
		(0.50, 0.159635194920) (0.55, 0.128742999094)
		(0.60, 0.103480124610) (0.65, 0.085811168978)
		(0.70, 0.075800104426) (0.75, 0.068367278168)
		(0.80, 0.071753040441) (0.85, 0.073104108032)
		(0.90, 0.089161880328) (0.95, 0.112025283761)
		(1.00, 0.145801283377)};
	\addlegendentry{Algorithm~4.3 with 2-qubit QSC}
\end{axis}
\end{tikzpicture}
\caption{Variant of the quantum Gilbert–Elliott channel described in the RHS of Figure~\ref{fig:QGEC}.
Parameters: $\pgood = 0.05$; $\pbad\in[0,1]$; $\tilde{V}_\system{S} = \exp(-j\alpha H)$, where $H$ is some fixed $4$-by-$4$ Hermitian matrix and where $\alpha = 1$ is fixed; $n = 10^5$.}
\label{fig:QGEC:plot:2}
\end{figure}
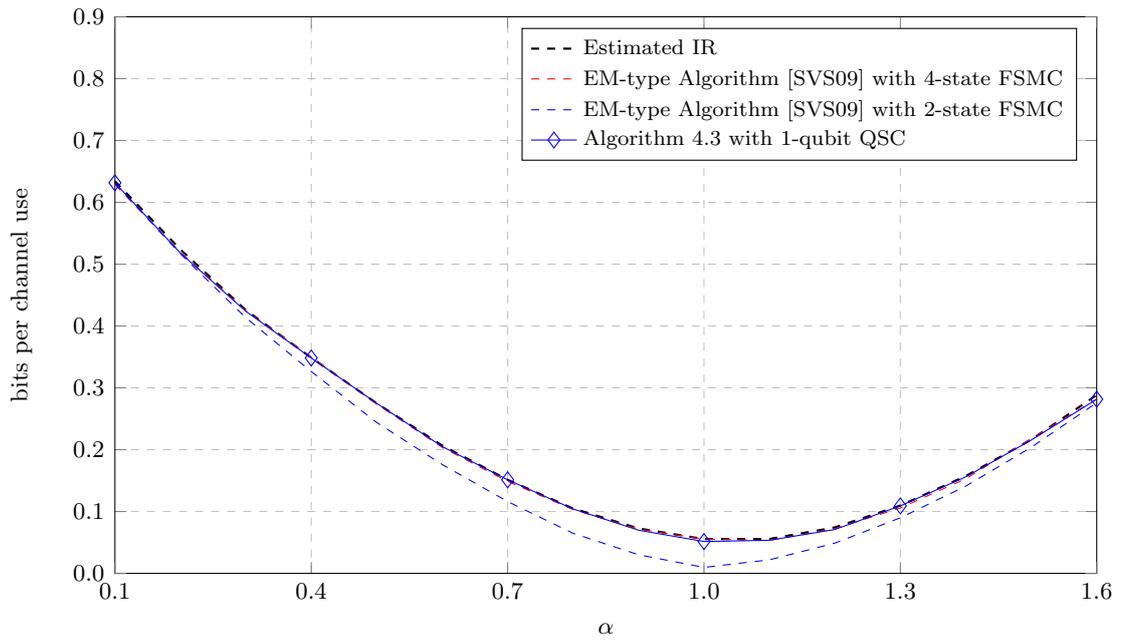
\begin{figure}
\begin{tikzpicture}[
	every axis/.append style={font=\footnotesize},
	every mark/.append style={scale=1.5}]
\begin{axis}[xlabel={$\alpha$}, 
	ylabel={bits per channel use},
	ylabel shift=-4 pt,
	xmin=0.1, xmax=1.6, ymin=0,  ymax=0.9,
	legend style={font=\scriptsize}, 
	mark repeat={3}, 
	width=\columnwidth, height=.618\columnwidth,
	y tick label style={
		/pgf/number format/fixed,
		/pgf/number format/fixed zerofill,
		/pgf/number format/precision=1},
	ytick={0.1,0.2,0.3,0.4,0.5,0.6,0.7,0.8},
	extra y ticks={0,0.9},
	extra y tick style={grid=none},
	xtick={0.4,0.7,1.0,1.3},
	extra x ticks={0.1,1.6},
	extra x tick style={grid=none},
	x tick label style={
		/pgf/number format/fixed,
		/pgf/number format/fixed zerofill,
		/pgf/number format/precision=1},
	grid=both, grid style={line width=0.1pt, dashed},
	legend cell align={left}]
	\addplot[dashed, black, thick] coordinates {
		(0.1, 0.634246436353) (0.2, 0.523724519704)
		(0.3, 0.425843507895) (0.4, 0.348767614102)
		(0.5, 0.274949201240) (0.6, 0.206807177630)
		(0.7, 0.151334883495) (0.8, 0.104972113295)
		(0.9, 0.072702348555) (1.0, 0.055474107518)
		(1.1, 0.055705182488) (1.2, 0.074364078064)
		(1.3, 0.109701367088) (1.4, 0.157408739672)
		(1.5, 0.216248463821) (1.6, 0.288383375526)};
	\addlegendentry{Estimated IR}
	\addplot[solid, red, dashed] coordinates {
		(0.1, 0.628757790182) (0.2, 0.519154991495)
		(0.3, 0.424800128488) (0.4, 0.349505951177)
		(0.5, 0.273691145329) (0.6, 0.203689311260)
		(0.7, 0.148728055607) (0.8, 0.103202914342)
		(0.9, 0.071644182670) (1.0, 0.055364118464)
		(1.1, 0.053273311931) (1.2, 0.072222309436)
		(1.3, 0.105810164407) (1.4, 0.153058635133)
		(1.5, 0.218322703723) (1.6, 0.281969975885)};
	\addlegendentry{EM-type Algorithm~[SVS09] with 4-state FSMC}
	\addplot[solid, blue, dashed] coordinates {
		(0.1, 0.630537416883) (0.2, 0.516839131244)
		(0.3, 0.413357758219) (0.4, 0.326340217404)
		(0.5, 0.244019998069) (0.6, 0.176033806634)
		(0.7, 0.116161198638) (0.8, 0.064796428188)
		(0.9, 0.030125167285) (1.0, 0.009658573661)
		(1.1, 0.021889961949) (1.2, 0.048929762310)
		(1.3, 0.089895723650) (1.4, 0.141175364024)
		(1.5, 0.203941640222) (1.6, 0.275966557336)};
	\addlegendentry{EM-type Algorithm~[SVS09] with 2-state FSMC}
	\addplot[solid, mark = diamond, blue] coordinates {
		(0.100, 0.631573943140) (0.200, 0.518053760695)
		(0.300, 0.423761299113) (0.400, 0.348255291277)
		(0.500, 0.274996844752) (0.600, 0.205089306142)
		(0.700, 0.151321209016) (0.800, 0.104187042637)
		(0.900, 0.069819436449) (1.000, 0.051482497374)
		(1.100, 0.053512560728) (1.200, 0.070980855572)
		(1.300, 0.109099922099) (1.400, 0.156097904004)
		(1.500, 0.215592931815) (1.600, 0.281598938534) };
	\addlegendentry{Algorithm~4.3 with 1-qubit QSC}
\end{axis}
\end{tikzpicture}
\caption{Quantum Gilbert–Elliott channel: $\pgood = 0.05$ is fixed; $\pbad =0.95$ is fixed; $V_\system{S} = \exp(-j \alpha H)$, where $H$ is the same $2$-by-$2$ Hermitian matrix as in Figure~\ref{fig:QGEC:plot:1} and where $\alpha$ varies from $0.1$ to $+1.5$; $n = 10^5$.}
\label{fig:QGEC:plot:3}
\end{figure}
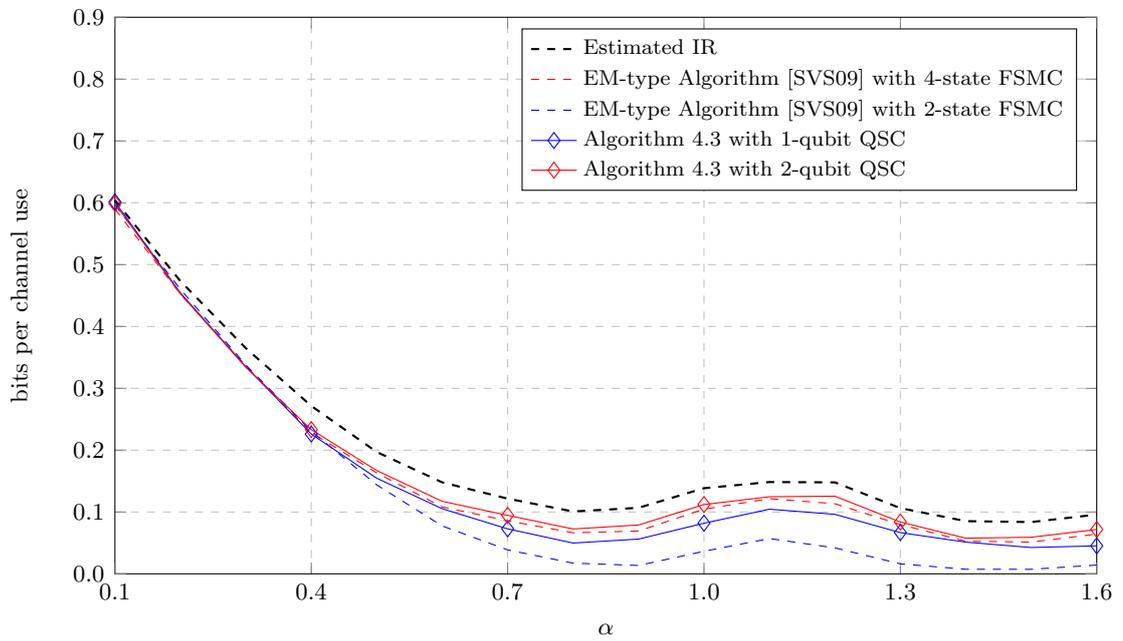
\begin{figure}
\begin{tikzpicture}[
	every axis/.append style={font=\footnotesize},
	every mark/.append style={scale=1.5}]
\begin{axis}[
	xlabel={$\alpha$}, 
	ylabel={bits per channel use},
	ylabel shift=-4 pt,
	xmin=0.1, xmax=1.6, ymin=0,  ymax=0.9,
	legend style={font=\scriptsize}, 
	mark repeat={3}, 
	width=\columnwidth, height=.618\columnwidth,
	y tick label style={
		/pgf/number format/fixed,
		/pgf/number format/fixed zerofill,
		/pgf/number format/precision=1},
	ytick={0.1,0.2,0.3,0.4,0.5,0.6,0.7,0.8},
	extra y ticks={0,0.9},
	extra y tick style={grid=none},
	xtick={0.4,0.7,1.0,1.3},
	extra x ticks={0.1,1.6},
	extra x tick style={grid=none},
	x tick label style={
		/pgf/number format/fixed,
		/pgf/number format/fixed zerofill,
		/pgf/number format/precision=1},
	grid=both, grid style={line width=0.1pt, dashed},
	legend cell align={left}]
	\addplot[dashed, black, thick] coordinates {
		(0.1, 0.604100063356) (0.2, 0.473401110905)
		(0.3, 0.364883790848) (0.4, 0.271535196296)
		(0.5, 0.196980525949) (0.6, 0.147892471850)
		(0.7, 0.121579631965) (0.8, 0.100800736650)
		(0.9, 0.106741942766) (1.0, 0.138425721984)
		(1.1, 0.148411159003) (1.2, 0.147873920267)
		(1.3, 0.106213495216) (1.4, 0.085179900573)
		(1.5, 0.083684790791) (1.6, 0.095801652946)};
	\addlegendentry{Estimated IR}
	\addplot[solid, red, dashed] coordinates {
		(0.1, 0.590030613376) (0.2, 0.452373967863)
		(0.3, 0.336164446674) (0.4, 0.228386828122)
		(0.5, 0.163171954985) (0.6, 0.108116007087)
		(0.7, 0.085836992514) (0.8, 0.066406308630)
		(0.9, 0.069213426994) (1.0, 0.104559960603)
		(1.1, 0.121506799903) (1.2, 0.113347442020)
		(1.3, 0.079694559896) (1.4, 0.052635475909)
		(1.5, 0.051345016335) (1.6, 0.064265969362)};
	\addlegendentry{EM-type Algorithm~[SVS09] with 4-state FSMC}
	\addplot[solid, blue, dashed] coordinates {
		(0.1, 0.596576347658) (0.2, 0.460288337401)
		(0.3, 0.337629300284) (0.4, 0.230703408080)
		(0.5, 0.143804485696) (0.6, 0.077428614807)
		(0.7, 0.038734579522) (0.8, 0.017294618485)
		(0.9, 0.013494816112) (1.0, 0.036456654565)
		(1.1, 0.057035136631) (1.2, 0.041939227789)
		(1.3, 0.016095530905) (1.4, 0.007397882922)
		(1.5, 0.007319105337) (1.6, 0.014187593031)};
	\addlegendentry{EM-type Algorithm~[SVS09] with 2-state FSMC}
	\addplot[solid, mark = diamond, blue] coordinates {
		(0.1, 0.601528340475) (0.2, 0.454463594618)
		(0.3, 0.335020097620) (0.4, 0.226052428000)
		(0.5, 0.154639172891) (0.6, 0.105263825537)
		(0.7, 0.072710389970) (0.8, 0.049759820290)
		(0.9, 0.056187338246) (1.0, 0.081847657751)
		(1.1, 0.104479497919) (1.2, 0.096208892501)
		(1.3, 0.066664735539) (1.4, 0.051219739998)
		(1.5, 0.042557259501) (1.6, 0.045233504158)};
	\addlegendentry{Algorithm~4.3 with 1-qubit QSC}
	\addplot[solid, mark = diamond, red] coordinates {
		(0.1, 0.598825525120) (0.2, 0.452793632504)
		(0.3, 0.333983755168) (0.4, 0.233252042999)
		(0.5, 0.167063892604) (0.6, 0.117258612189)
		(0.7, 0.094584955365) (0.8, 0.072521235449)
		(0.9, 0.079008564509) (1.0, 0.112025283761)
		(1.1, 0.124629280264) (1.2, 0.125535074107)
		(1.3, 0.083770542149) (1.4, 0.057645455501)
		(1.5, 0.059133475287) (1.6, 0.071769388212) };
	\addlegendentry{Algorithm~4.3 with 2-qubit QSC}
\end{axis}
\end{tikzpicture}
\caption{Same variant of the quantum Gilbert–Elliott channel as in Figure~\ref{fig:QGEC:plot:2} with different parameters: $\pgood = 0.05$; $\pbad = 0.95$; $V_\system{S} = \exp(-j \alpha H)$, where $H$ is the same $4$-by-$4$ Hermitian matrix as in Figure~\ref{fig:QGEC:plot:2} and where $\alpha$ varies from $0.1$ to $+1.5$; $n = 10^5$.}
\label{fig:QGEC:plot:4}
\end{figure}
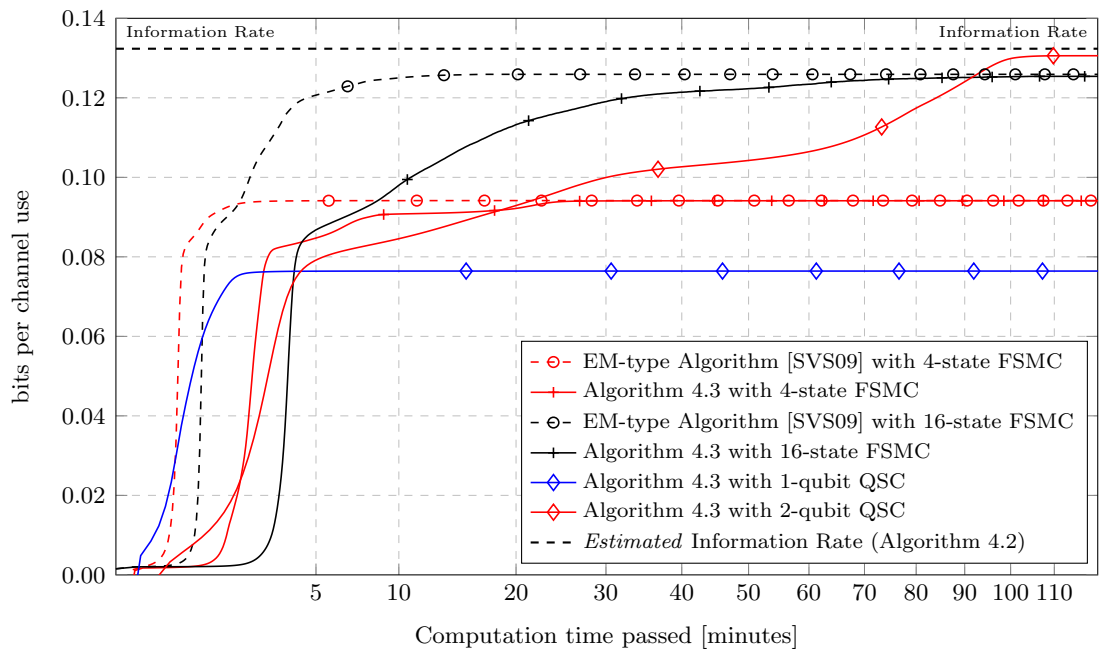
\begin{figure}
\begin{tikzpicture}[
	every axis/.append style={font=\footnotesize},
	every mark/.append style={scale=1.5}]
\begin{axis}[
	xlabel={Computation time passed [minutes]}, 
	ylabel={bits per channel use},
	ylabel shift=-4 pt,
	xmin=0, xmax=85, ymin=0,  ymax=0.14,
	legend style={font=\scriptsize},
	mark repeat=200, 
	mark phase = 200,
	width=\columnwidth, height=.618\columnwidth,
	y tick label style={
		/pgf/number format/fixed,
		/pgf/number format/fixed zerofill,
		/pgf/number format/precision=2},
	ytick={0.02,0.04,0.06,0.08,0.1,0.12},,
	extra y ticks={0,0.14},
	extra y tick style={grid=none},
	xticklabels={5,10,20,30,40,50,60,70,80,90,100,110},
	xtick={17.3205,24.4949,34.6410,42.4264,48.9898,54.7723,60.0000,64.8074,69.2820,73.4847,77.4597,81.2404},
	x tick label style={
		/pgf/number format/fixed,
		/pgf/number format/fixed zerofill,
		/pgf/number format/precision=0},
	grid=both, grid style={line width=0.1pt, dashed},
	legend cell align={left},
	legend style={at={(0.99,0.02)},anchor=south east},
	every axis plot/.append style={semithick}]
	\addplot[red, dashed, mark = o, mark options={solid}]
		table[x expr=((\thisrow{step}-1)*0.848424200)^(0.5),y=IR_L]
		{1155001913-cao_michael_xuan-202104-phd-Data_Performance_Gap_cem_4.txt};
	\addlegendentry{EM-type Algorithm~[SVS09] with 4-state FSMC}
	\addplot[red, solid, mark = +, mark options={solid}]
		table[x expr=((\thisrow{step}-1)*1.343300420)^(0.5),y=IR_L]
		{1155001913-cao_michael_xuan-202104-phd-Data_Performance_Gap_cgd_4.txt};
	\addlegendentry{Algorithm~4.3 with 4-state FSMC}
	\addplot[black, dashed, mark = o, mark options={solid}]
		table[x expr=((\thisrow{step}-1)*1.010961860)^(0.5),y=IR_L]
		{1155001913-cao_michael_xuan-202104-phd-Data_Performance_Gap_cem_16.txt};
	\addlegendentry{EM-type Algorithm~[SVS09] with 16-state FSMC}
	\addplot[black, solid, mark = +, mark options={solid}]
		table[x expr=((\thisrow{step}-1)*1.599697330)^(0.5),y=IR_L]
		{1155001913-cao_michael_xuan-202104-phd-Data_Performance_Gap_cgd_16.txt};
	\addlegendentry{Algorithm~4.3 with 16-state FSMC}
	\addplot[blue, mark=diamond]
		table[x expr=((\thisrow{step}-1)*2.298483867)^(0.5),y=IR_L]
		{1155001913-cao_michael_xuan-202104-phd-Data_Performance_Gap_qgd_1.txt};
	\addlegendentry{Algorithm~4.3 with 1-qubit QSC}
	\addplot[red, mark=diamond]
		table[x expr=((\thisrow{step}-1)*5.480697869)^(0.5),y=IR_L]
		{1155001913-cao_michael_xuan-202104-phd-Data_Performance_Gap_qgd_2.txt};
	\addlegendentry{Algorithm~4.3 with 2-qubit QSC}
	\addplot[dashed,black,thick] coordinates{(0000,0.132352)(90,0.132352)};
	\addlegendentry{\emph{Estimated} Information Rate (Algorithm~4.2)};
	\node[anchor=north west, font=\tiny, yshift=1pt] at (rel axis cs:0,1) {Information Rate};
	\node[anchor=north east, font=\tiny, yshift=1pt] at (rel axis cs:1,1) {Information Rate};
\end{axis}
\end{tikzpicture}
\caption{Minimizing the difference function $\Delta_{W}^{(n)}$ using different methods.
	The markers appear after every 400 updates.}
\label{fig:QGEC:plot:5}
\end{figure}
%*******************************************************************************
%*******************************************************************************
\chapter*{Summary and Outlook\markboth{SUMMARY AND OUTLOOK}{}}
\addcontentsline{toc}{chapter}{Summary and Outlook}
Chapter~\ref{chapter:DeFGs} proposed a generalized factor graph model called double-edge factor graphs (DeFGs) for representing quantum systems.
Compared with the factor graph for quantum probabilities~\cite{loeliger2012factor, loeliger2017factor}, DeFGs provide a neat generalization of the ``closing-the-box'' operations and is also more generic.
We also proposed generalized versions of the belief-propagation (BP) algorithms for DeFGs and derived the loop calculus expansion for DeFGs (see~\eqref{eq:DeFG:lc}) at BP fixed points.
However, it is still unclear how the induced partition sum $Z_\mathsf{induced}(\{m_{j \to a}, m_{a \to j}\}_{(j,a)\in\set{E}})$ at BP fixed points can be related to the actual partition sum $Z(\set{G})$.
In comparison, for factor graphs, it was shown that the induced partition sum $Z_\mathsf{induced}$ coincides with the Bethe partition sum $Z_{\mathsf{B}}$ at BP fixed points (see Corollary~\ref{cor:Zb:decompose}), and how well the latter estimates $Z(\set{G})$ depends on how $\bethe$ estimates $\gibbs$.
Nevertheless, we studied several numerical examples.
We observed a promising performance of BP algorithms for DeFGs.
\par
Chapter~\ref{chapter:QFGs} studied a graphical model called quantum factor graphs.
In particular, we investigated how ``closing-the-box'' operations behave when all the local operators are proportionally close to identity operators.
Under such a setup, the ``closing-the-box'' operations provide an \emph{approximation} to the partition sum of a QFG, as we found that the distributivity of the $\star$-product over the (partial) trace operations holds approximately.
We also introduced BP algorithms for QFGs as a natural generalization of the ``closing-the-box'' operations.
The algorithm coincided with the quantum belief-propagation algorithm for bifactor networks~\cite{Leifer08}.
We were able to generalize Bethe's approximation to QFGs and show that the positive BP fixed points \emph{approximately} correspond to the quantum Bethe free energy's interior stationary points.
\par
Chapter~\ref{chapter:QCwM} considered the scenario of transmitting classical information over a quantum channel with finite memory using separable-state ensembles and local measurements.
We defined the classical-input classical-output quantum-state channel (CC-QSC) as an equivalent way to describe such communication setups and demonstrated how normal factor graphs could be used to visualize such channels.
We showed that a CC-QSC's information rate is independent of the initial density operator under suitable conditions and proposed algorithms for estimating and bounding such information rate.
The computations in such algorithms can be carried out using the corresponding factor graphs of the CC-QSC.
We emphasize that our approach for optimizing the lower bound is data-driven and does not require the knowledge of the actual channel model.
\par
Despite the obstacles we have encountered in analyzing the corresponding generalized versions of BP algorithms in Chapter~\ref{chapter:DeFGs} and  Chapter~\ref{chapter:QFGs}, we have observed promising numerical results in a number of examples.
This suggests potential connections between the generalized BP-fixed points and the corresponding generalized partition sums of DeFGs and QFGs, respectively, yet to be discovered.
One of the directions is to generalize the method of graph covers to these new factor graphs.
(For some initial results, see~\cite{huang2020characterizing}.)
Another is to consider the regional method as in~\cite{yedidia2005constructing}.
In some earlier studies of (classical) factor graphs~\cite{mori2015loop}, some methods from information geometry have been proven useful in analyzing the loop calculus expansion at the fixed points.
Though it is currently unclear whether (or how much) quantum information geometry can be helpful in analyzing the loop calculus expansion for DeFGs, it is still an interesting direction to look into.
%*******************************************************************************
% Appendix**********************************************************************
%\cleardoublepage
\begin{appendices}
%*******************************************************************************
\chapter[Alternative Distributivity]{Alternative Approximated Distributivity of $\star$ over (Partial) Trace} \label{app:exp:approx}
In this appendix, we discuss how $\tr\left( \rho_\system{A}\star \sigma_{AB} \right)$ and $\tr_\system{B}\left( \rho_\system{A}\star \sigma_{AB} \right)$ can be approximated by $\tr_\system{A} \left(\rho_\system{A}\star\tr_\system{B}(\sigma_{AB}) \right)$ and $\rho_\system{A}\star\tr_\system{B}(\sigma_{AB})$, respectively, when the operators $\rho_\system{A}$ and $\sigma_\system{AB}$ are proportionally close to identity operator in a nonlinear manner, \ie, 
\begin{align}
	\rho_\system{A} &\propto \exp\left(t\cdot X\right)\\
	\sigma_\system{AB} &\propto \exp\left(t\cdot Y\right)
\end{align}
for some Hermitian matrices $X$ and $Y$ (each of proper size) and some small $t>0$.
We present some results similar to Theorem~\ref{thm:approx:distri} and Proposition~\ref{prop:approx:distri}.
\begin{theorem}\label{thm:approx:distri:exp}
Consider finite-dimensional Hilbert spaces $\hilbert_\system{A}$ and $\hilbert_\system{B}$.
Given $X \in \HermitianOp(\hilbert_\system{A})$, and $Y \in \HermitianOp(\hilbert_\system{A} \tensor \hilbert_\system{B})$,
it holds that
\begin{equation}\label{eq:approx:distri:exp}
\tr\left( e^{tX} \star e^{tY} \right) = \tr_\system{A} \left( e^{tX} \star \tr_\system{B}(e^{tY}) \right) + O(t^4),
\end{equation}
where the real number $t>0$ is in some neighborhood of $0$.
\end{theorem}
\begin{proof}
We use similar techniques in the proof of Theorem~\ref{thm:approx:distri} in proving this theorem.
Using the notation of normalized trace functions, we rewrite~\eqref{eq:approx:distri:exp} as 
\begin{equation}\label{eq:approx:distri:exp:normalized}
\overline{\tr}\left( e^{tX} \star e^{tY} \right) = \overline{\tr}_\system{A} \left( e^{tX} \star \overline{\tr}_\system{B}(e^{tY}) \right) + O(t^4).
\end{equation}
Let $\tilde{X}\defeq X\tensor I_\system{B}$.
The Taylor series expansion of each side of~\eqref{eq:approx:distri:exp:normalized} (without $O(t^4)$) can be expressed as
\begin{align}
\label{eq:taylor:LHS:ADT:exp}
\text{LHS}
&=\begin{aligned}[t]
1 &+ t \cdot \overline{\tr}(\tilde{X}+Y) + \frac{1}{2}t^2 \cdot \overline{\tr}(\tilde{X}+Y)^2 \\
&+ \frac{1}{3!} t^3 \cdot \overline{\tr}(\tilde{X}+Y)^3 + \frac{1}{4!} t^4 \cdot \overline{\tr}(\tilde{X}+Y)^4 + O(t^5),
\end{aligned}\\
\label{eq:taylor:RHS:ADT:exp}
\text{RHS}
&=\begin{aligned}[t]
1 &+ t \cdot \overline{\tr}_\system{A} \left( X + \overline{\tr}_\system{B}(Y) \right) + \frac{t^2}{2} \cdot \overline{\tr}_\system{A} \left( X^2 + X\overline{\tr}_2(Y) + \overline{\tr}_2(Y)X + \overline{\tr}_\system{B}(Y^2) \right) \\
&+ \frac{t^3}{6} \cdot \overline{\tr}_\system{A}\left( X^3 + 3\cdot X^2 \cdot \overline{\tr}_2(Y) + 3 \cdot X \cdot \overline{\tr}_\system{B}(Y^2) + \overline{\tr}_2(Y^3) \right) + O(t^4).
\end{aligned}\!\!\!
\end{align}
Note that~\eqref{eq:taylor:LHS:ADT:exp} and~\eqref{eq:taylor:RHS:ADT:exp} agree up to $t^3$, which concludes the proof.
\end{proof}
We can also show the following proposition in a similar way.
\begin{proposition}\label{prop:approx:distri:exp}
Under the same setup as in Theorem~\ref{thm:approx:distri:exp}, it holds that
\begin{equation}\label{eq:approx:distri:prop:exp}
\tr_\system{B}\left( e^{tX}\star e^{tY} \right)  = e^{tX} \star \tr_\system{B}(e^{tY}) + O(t^3).
\end{equation}
\end{proposition}
\begin{proof}
Rewrite~\eqref{eq:approx:distri:prop:exp} as
\begin{equation}
\overline{\tr}_\system{B}\left( e^{tX}\star e^{tY} \right)  = e^{tX} \star \overline{\tr}_\system{B}(e^{tY}) + O(t^3),
\end{equation}\label{eq:approx:distri:prop:exp:normalized}
and note that each side of the above equation (without $O(t^3)$) can be expressed as
\begin{align}
\label{eq:taylor:LHS:ADP:exp}
\text{LHS}
&= I + t \cdot \overline{\tr}_\system{B}(\tilde{X}+Y)+\frac{1}{2} t^2 \cdot \overline{\tr}_\system{B}(\tilde{X}+Y)^2 + \frac{1}{3!}t^3 \cdot \overline{\tr}_\system{B}(\tilde{X}+Y)^3 + O(t^4),\\
\label{eq:taylor:RHS:ADP:exp}
\text{RHS}
&= I + t \cdot \left(X+\overline{\tr}_\system{B}(Y)\right) + \frac{1}{2} t^2 \cdot \left(X^2 + X\overline{\tr}_\system{B}(Y) + \overline{\tr}_\system{B}(Y)X + \overline{\tr}_\system{B}(Y^2)\right) + O(t^3),
\end{align} respectively.
The proposition can be justified since the above two expressions agree up to $t^2$.
\end{proof}
%*******************************************************************************
\chapter{Quantum Exponential Family}\label{app:QEF}
\begin{definition}[Quantum Exponential Family~\cite{mori2015loop}]
Similar to classical exponential families, a \emph{quantum exponential family} (of degree $d$) is a parametric family of quantum operators in the form of 
\begin{equation}\label{eq:defQEF}
\rho_\boldtheta \defeq \exp\left(\sum_{k=1}^d \theta_k\cdot\boldtau_k - \Psi(\boldtheta)\right)
\end{equation}
for \emph{natural parameter} $\boldtheta$ in some open subset $\Theta \in \mathbb{R}^d$, where $\boldtau_k\in\HermitianOp(\hilbert_{\nb{k}}) =\HermitianOp(\Tensor_{i\in\nb{k}}\hilbert_{i})$ are some given Hermitian operators ($\partial k\subset \left\{1,\ldots,N\right\}$), and conventions as in~\eqref{eq:blowupConvention} are applied in the summation in~\eqref{eq:defQEF}.
Moreover, the scaler $\Psi(\boldtheta)$ is defined as
\begin{equation}
\Psi(\boldtheta) \defeq \log\left( \tr\left( \exp\left( \sum_{k=1}^d \theta_k \cdot \boldtau_k \right) \right) \right).
\end{equation}
Thus, $\rho_\boldtheta$ is a density operator.
\end{definition}
Note that if $\{\boldtau_k\}_{k=1}^d$ are linearly independent, then the mapping $\boldtheta\mapsto\rho_\boldtheta$ is injective.
In this appendix, we assume $\{\boldtau_k\}_{k=1}^d$ to be linearly independent.
In this case, the (strict) convexity of the function $\Psi$ follows naturally from the (strict) convexity of the exponential function (see, \eg,~\cite{carlen2010trace}).
\par
\begin{example}\label{example:QEFwrtG}
Consider the quantum exponential family
\begin{equation}\label{eq:QEFwrtG}
\sigma_\boldtheta = \exp\left(\sum_{a\in\set{F}} \sum_{k}\theta^{(a)}_k\cdot\boldtau^{(a)}_k - \Psi(\boldtheta)\right),
\end{equation}
where $\boldtheta\in\mathbb{R}^d$ and where $\{\boldtau^{(a)}_k\}_k$ form a basis of $\HermitianOp(\hilbert_\nb{a})$.
The parametrization of $\sigma_\boldtheta$ corresponds to all the density operators $\sigma$ that can be decomposed as 
\begin{equation}
\sigma \propto \bigstar_{a \in \set{F}} \sigma_a
= \exp\left(\sum_{a\in\set{F}}\log{\sigma_a}\right),
\end{equation}
where $\sigma_a\in\DensOp(\hilbert_\nb{a})$ for each $a\in\set{F}$.
\end{example}
\begin{definition}[Dual Parameters]
Given a quantum exponential family as in~\eqref{eq:QEFwrtG}.
We define the \emph{dual parameter} (\wrt $\boldtheta$) to be $\boldeta = (\eta_l)_{l=1}^d$ where
\begin{equation}
	\eta_l \defeq \frac{\partial}{\partial \theta_l} \Psi(\boldtheta)
\end{equation}
for each $l=1,\ldots,n.$
\end{definition}
We can re-express $\boldeta$ as a function of $\rho_\boldtheta$, \ie,
\begin{align}
\nonumber
\eta_l &= \frac{\partial}{\partial \theta_l} \Psi(\boldtheta) 
= \frac{\partial}{\partial \theta_l} \log\left( \tr\left(
    \exp\left( \sum_{k=1}^d \theta_k \cdot \boldtau_k \right) \right) \right), \\
\nonumber
&= {\tr\left( \exp\left( \sum_{k=1}^d \theta_k \boldtau_k \right) \right) }^{-1} \frac{\partial}{\partial \theta_l} \tr\left( \exp\left( \sum_{k=1}^d \theta_k \boldtau_l \right) \right), \\
\nonumber
&\overset{\text{(a)}}{=} {\tr\left( \exp\left( \sum_{k=1}^d \theta_k \boldtau_k \right) \right) }^{-1} \tr\left( \exp\left( \sum_{k=1}^d \theta_k \boldtau_k \right) \cdot \boldtau_l \right), \\
&= \tr\left( \rho_\boldtheta \cdot \boldtau_l \right),
\end{align}
where step (a) was obtained by applying the first-order perturbation theory~\cite{Kato1966}.
\par
Due to the strict convexity of $\Psi$, the mapping $\boldeta(\boldtheta): \boldtheta \mapsto (\tr\left( \rho_\boldtheta \cdot \boldtau_l \right) )_{l=1}^n$ is injective.
On the other hand, by considering the conjugate function of $\Psi$ defined as 
(which is also strictly convex)
\begin{equation}
\Phi(\boldeta) \defeq \sup_\boldtheta \left\{ \sum_{k=1}^d \theta_k \cdot \eta_k - \Psi(\boldtheta) \right\},
\end{equation}
the inverse mapping can be expressed as 
\begin{equation}
\boldtheta(\boldeta): \boldeta \mapsto \left( \frac{\partial}{\partial \eta_k} \Phi(\boldeta) \right)_k.
\end{equation}
Thus, the correspondence between the natural parameters and the dual parameters is bijective.
\par
\begin{example}
We continue Example~\ref{example:QEFwrtG}.
The dual parameters can be expressed as
\begin{equation}\label{eq:sigma2eta}
\eta_k^{(a)} = \tr\left( \sigma_\boldtheta \cdot \boldtau_k^{(a)} \right)
= \tr_{\partial a} \left( \tr_{\set{V}\xk\nb{a}}\left( \sigma_\boldtheta \cdot \boldtau_k^{(a)} \right) \right)
= \tr\left( \sigma_a \cdot \boldtau_k^{(a)}\right),
\end{equation}
where $\sigma_a \defeq \tr_{\set{V}\xk\nb{a}} \left( \sigma_\boldtheta \right)$.
Since $\left\{ \boldtau^{(a)}_k \right\}_k$ is a basis of $\HermitianOp(\hilbert_{\partial a})$, \eqref{eq:sigma2eta} has established an injection from $\sigma_a$ to $\boldeta^{(a)}$.
In other words, given some local densities operators, there exists at most one  global density operator in the quantum exponential family such that its partial traces match these local densities operators.
\end{example}
%*******************************************************************************
\chapter{Differentiability of $-\tr\left( \sigma \cdot \log{\rho(\boldeta)} \right)$} \label{app:differentiability}
First, we verify the differentiability of the bijective mapping 
$\boldeta: \boldtheta \mapsto \left(\tr\left(\rho_\theta \cdot \boldtau_i \right) \right)_{i}$.
Note that,
\begin{align}
\frac{\partial \eta_i}{\partial \theta_j} = 
\frac{\partial\tr\left(\rho_\theta\cdot\boldtau_i\right)}{\partial \theta_j}
&= \left.\frac{\D}{\D t}\right|_{t=0} \tr\left( \exp\left( \sum_{k=1}^d \theta_k \cdot \boldtau_k + t\cdot \boldtau_i - \Psi(\boldtheta + t \cdot \vect{e}_i) \right) \cdot \boldtau_i \right)\\
&= \left.\frac{\D}{\D t}\right|_{t=0} \frac{\tr\left( \exp\left( \sum_{k=1}^d \theta_k \cdot \boldtau_k + t \cdot \boldtau_i \right) \cdot \boldtau_i \right)}{\exp\left( \Psi(\boldtheta + t \cdot \vect{e}_i) \right)}.
\end{align}
Since the denominator $\exp(\Psi(\boldtheta+t\cdot\vect{e}_i))$ is clearly differentiable, it suffice to show the differentiability of $t\mapsto \tr\left( \exp\left( \sum_{k=1}^d \theta_k \cdot \boldtau_k + t \cdot \boldtau_i \right) \cdot \boldtau_i \right)$ at $t=0$.
By the Taylor series expansion, we can write 
\begin{align}
&\tr\left( \exp\left( \sum_{k=1}^d \theta_k \cdot \boldtau_k + t \cdot \boldtau_i \right) \cdot \boldtau_i \right) - \tr\left( \exp\left( \sum_{k=1}^d \theta_k \cdot \boldtau_k \right) \cdot \boldtau_i \right)\\
= &\tr\left( \sum_{n=0}^\infty \frac{\left( \sum_{k=1}^d \theta_k \cdot \boldtau_k + t \cdot \boldtau_i \right)^n}{n!} \cdot \boldtau_i \right) - \tr\left( \sum_{n=0}^\infty \frac{\left( \sum_{k=1}^d \theta_k \cdot \boldtau_k \right)^n}{n!} \cdot \boldtau_i \right)\\
= & \tr\left( \sum_{n=0}^\infty \frac{\left( \sum_{k=1}^d \theta_k \cdot \boldtau_k + t \cdot \boldtau_i\right)^n -\left( \sum_{k=1}^d \theta_k \cdot \boldtau_k \right)^n}{n!} \cdot \boldtau_i \right)\\
= & \tr\left(\sum_{n=1}^\infty \frac{ t\cdot \left( \sum_{\ell=0}^{n-1} \left(\sum_{k=1}^d \theta_k \cdot\boldtau_k\right)^{n-1-\ell} \!\cdot \boldtau_i \cdot \left(\sum_{k=1}^d \theta_k \cdot\boldtau_k\right)^{\ell}\right) + O(t^2)}{n!} \cdot \boldtau_i \right)\\
= &\ t\cdot \sum_{n=1}^\infty \frac{\tr\left( \left(\sum_{k=1}^d \theta_k \cdot\boldtau_k\right)^{n-1} \cdot \boldtau_i \right)}{(n-1)!}+O(t^2).
\end{align}
Notice that
\begin{align}
\label{eq:seriesBounded1}
\sum_{n=1}^\infty \abs{\frac{\tr\left( \left(\sum_{k=1}^d \theta_k \cdot\boldtau_k\right)^{n-1} \cdot \boldtau_i \right)}{(n-1)!} }
&\leqslant \sum_{n=1}^\infty
    \frac{ \norm{\sum_{k=1}^d \theta_k \cdot\boldtau_k}^{n-1}
    \cdot \abs{\tr(\boldtau_i)}}{(n-1)!}\\
\label{eq:seriesBounded2}
&= \exp\left( \norm{\sum_{k=1}^d \theta_k \cdot\boldtau_k}\right) \cdot \abs{\tr(\boldtau_i)},
\end{align}
where, for a matrix $A$, $\norm{A}$ stands for the absolute value of the largest-in-absolute-value eigenvalue of $A$ and where we have applied the inequalities
\begin{align}
\abs{\tr(A\cdot B)} &\leqslant \norm{A}\cdot\abs{\tr(B)}\\
\norm{A\cdot B} &\leqslant \norm{A} \cdot \norm{B}
\end{align}
in deriving~\eqref{eq:seriesBounded1}.
Equation~\eqref{eq:seriesBounded2} implies that the series
\[
\sum_{n=1}^\infty \frac{\tr\left( \left(\sum_{k=1}^d \theta_k \cdot\boldtau_k\right)^{n-1} \cdot \boldtau_i \right)}{(n-1)!}
\]
is absolute convergent and thus is convergent.
Therefore, the limit
\begin{equation}
\lim_{t\to 0} \frac{1}{t} \cdot \left( \tr\left( \exp\left( \sum_{k=1}^d \theta_k \cdot \boldtau_k + t \cdot \boldtau_i \right) \cdot \boldtau_i \right) - \tr\left( \exp\left( \sum_{k=1}^d \theta_k \cdot \boldtau_k \right) \cdot \boldtau_i \right) \right)
\end{equation}
exists, which justifies the differentiability of $t\mapsto \tr\left( \exp\left( \sum_{k=1}^d \theta_k \cdot \boldtau_k + t \cdot \boldtau_i \right) \cdot \boldtau_i \right)$ at $t=0$, and of $\boldeta: \boldtheta \mapsto \left(\tr\left(\rho_\theta \cdot \boldtau_i \right) \right)_{i}$ as well.
\par
Second, we consider the function $\hat{f}(\boldtheta)\defeq-\tr\left(\sigma\cdot\log{\rho_\boldtheta}\right)$.
Consider a small change in $\boldtheta$ along its $i$-th component.
The change in $\hat{f}$ can be expressed as 
\begin{equation}
\hat{f}(\boldtheta+h\cdot \mathbf{e}_j) - \hat{f}(\boldtheta) = 
h\cdot\tr\left(\sigma \cdot \boldtau_i \right) - ( \Psi(\boldtheta+h\cdot\mathbf{e}_i) - \Psi(\boldtheta) ).
\end{equation}
Clearly, $\hat{f}$ is differentiable.
\par
Finally, note that $f(\boldeta)=\hat{f}(\boldtheta(\boldeta))$.
Since the bijective mapping $\boldtheta\mapsto\boldeta$ is differentiable, then so is its inverse mapping $\boldeta\mapsto\boldtheta$.
Therefore, the differentiability of $f$ follows directly from the differentiability of $\hat{f}$.
%*******************************************************************************
\chapter{Additional Figures for Chapter~\ref{chapter:QCwM}}\label{app:figures}
%*******************************************************************************
\noindent\begin{minipage}[c][.341\textheight]{\textwidth}
\centering
\begin{tikzpicture}[scale=.85,every node/.style={transform shape},
    factor/.style={rectangle, minimum width=1cm, minimum height=.7cm, draw},
    sfactor/.style={rectangle, minimum size=.4cm, draw},
    darksolid/.style={rectangle, minimum size=.15cm, draw,fill = black,
    inner sep=0pt, outer sep = 0pt},
    label/.style={red,anchor=north east,xshift = .1cm}]
\node[darksolid] (S) {}; \node [left=0pt of S] {$\cs_0$};
\node[factor] (E1) [right=.7cm of S] {$W$};
\draw (S) -- (E1);
\node[darksolid] (X1) [above=.8cm of E1] {}; \node[above = 0pt of X1] {$\cx_1$};
\draw (X1) -- (E1);
\draw (E1.south) -- ([yshift=-.8cm]E1.south) node (Y1) [right] {$y_1$};
    
\node[factor] (E2) [right=1cm of E1] {$W$};
\draw (E1) -- (E2) node[above,midway] {$s_1$};
\node[darksolid] (X2) [above=.8cm of E2] {}; \node[above = 0pt of X2] {$\cx_2$};
\draw (X2) -- (E2);
\draw (E2.south) -- ([yshift=-.8cm]E2.south) node[right] {$y_2$};
    
\node[factor, draw=none] (Edummy1) [right=1cm of E2] {$\cdots$};
\draw (E2) -- (Edummy1) node[above,midway] {$s_2$};
\node at (X1-|Edummy1) {$\cdots$};
\node at (Y1-|Edummy1) {$\cdots$};
    
\node[factor] (El) [right=1cm of Edummy1] {$W$};
\draw (Edummy1) -- (El) node[above,midway] {$\cs_{\ell-1}$};
\node[darksolid] (Xl) [above=.8cm of El] {};
\node[above = 0pt of Xl] {$\cx_\ell$};
\draw (Xl) -- (El);
\draw (El.south) -- ([yshift=-.8cm]El.south) node[right] {$y_\ell$};
    
\node[factor] (El2) [right=1cm of El] {$W$};
\draw (El) -- (El2) node[above,midway] {$s_\ell$};
\node[darksolid] (Xl2) [above=.8cm of El2] {};
    \node[above = 0pt of Xl2] {$\cx_{\ell\!+\!1}$};
\draw (Xl2) -- (El2);
\draw (El2.south) -- ([yshift=-.8cm]El2.south) node[right] {$y_{\ell\!+\!1}$};
    
\node[factor, draw=none] (Edummy2) [right=1cm of El2] {$\cdots$};
\draw (El2) -- (Edummy2) node[above,midway] {$s_{\ell\!+\!1}$};
\node at (X1-|Edummy2) {$\cdots$};
\node at (Y1-|Edummy2) {$\cdots$};
    
\node[factor] (En) [right=1cm of Edummy2] {$W$};
\draw (Edummy2) -- (En) node[above,pos=0.4] {$s_{n\!-\!1}$};
\node[darksolid] (Xn) [above=.8cm of En] {}; \node[above = 0pt of Xn] {$\cx_n$};
\draw (Xn) -- (En);
\draw (En.south) -- ([yshift=-.8cm]En.south) node[right] {$y_{n}$};
	
\node[sfactor, inner sep=0pt] (ee) [right=.5cm of En] {$\mathbf{1}$};
\draw (En) -- (ee) node[above,midway] {$s_n$};
    
\begin{pgfonlayer}{bg}
    \draw[dashed, blue, line width=1.5pt, fill=blue!10]
        ([xshift=-2.1cm,yshift=2.8cm]E1) rectangle
        ([xshift=1.5cm,yshift=-2.4cm]ee);
    \draw[dashed, blue, line width=1.5pt, fill=blue!20]
        ([xshift=-.6cm,yshift=2.6cm]E1) rectangle
        ([xshift=1.3cm,yshift=-2.2cm]ee);
    \draw[dashed, blue, line width=1.5pt, fill=blue!30]
    ([xshift=-.7cm,yshift=2.4cm]E2) rectangle
    ([xshift=1.1cm,yshift=-2cm]ee);
    \draw[dashed, blue, line width=1.5pt, fill=blue!40]
        ([xshift=-.6cm,yshift=2.2cm]El) rectangle
        ([xshift=.9cm,yshift=-1.8cm]ee);
    \draw[dashed, blue, line width=1.5pt, fill=blue!50]
        ([xshift=-.7cm,yshift=2cm]El2) rectangle
        ([xshift=.7cm,yshift=-1.6cm]ee);
    \draw[dashed, blue, line width=1.5pt, fill=blue!60]
        ([xshift=-.7cm,yshift=1.8cm]En) rectangle
        ([xshift=.5cm,yshift=-1.4cm]ee);
\end{pgfonlayer}
\end{tikzpicture}

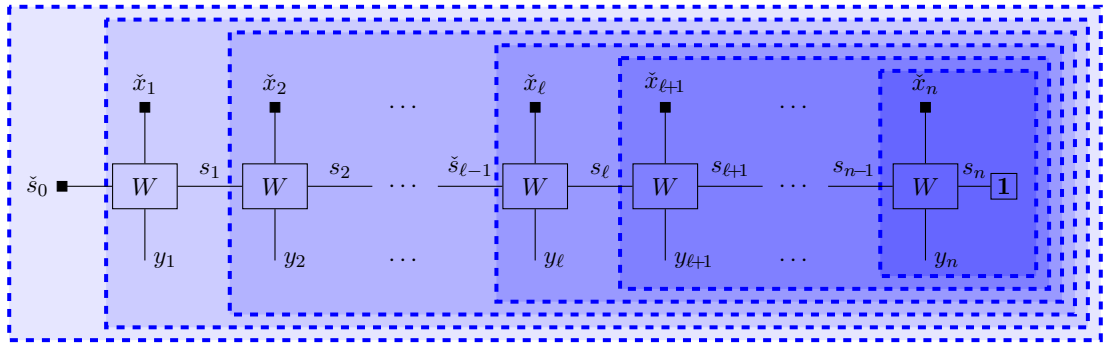
\captionof{figure}{Verification of~\eqref{eq:verfy:FSMC:conditional:distribution}.
Note that every ``closing-the-box'' operation yields a function node
        representing the constant function~$1$.}
\label{fig:FMSC:high:level:2}
\end{minipage}
%*******************************************************************************
\noindent\begin{minipage}[c][.341\textheight]{\textwidth}
\centering
\begin{tikzpicture}[scale=.85,every node/.style={transform shape},
    factor/.style={rectangle, minimum width=1cm, minimum height=1.7cm, draw},
    sfactor/.style={rectangle, minimum size=.4cm, draw},
    darksolid/.style={rectangle, minimum size=.15cm, draw,fill = black,
    inner sep=0pt, outer sep = 0pt},
    label/.style={red,anchor=north east,xshift = .1cm}]
\node[sfactor,inner sep=0pt] (S) {$\rho_{\system{S}_0}$};
\node[factor] (E1) [right=.7cm of S] {$W$};
\draw (S.north) |- ([yshift=.6cm]E1.west) node[above=-0.05cm,pos=.75] {$s_0$};
\draw (S.south) |- ([yshift=-.6cm]E1.west) node[below,pos=.75] {$\tilde{s}_0$};
\node[darksolid] (X1) [above=.7cm of E1] {}; \node[above = 0pt of X1] {$\cx_1$};
\draw (X1) -- (E1);
\draw (E1.south) -- ([yshift=-.7cm]E1.south) node (Y1) [right] {$y_1$};
    
\node[factor] (E2) [right=1cm of E1] {$W$};
\draw ([yshift=.6cm]E1.east) |- ([yshift=.6cm]E2.west) 
    node[above=-0.05cm,pos=.75] {$s_1$};
\draw ([yshift=-.6cm]E1.east) |- ([yshift=-.6cm]E2.west) 
    node[below,pos=.75] {$\tilde{s}_1$};
\node[darksolid] (X2) [above=.7cm of E2] {}; \node[above = 0pt of X2] {$\cx_2$};
\draw (X2) -- (E2);
\draw (E2.south) -- ([yshift=-.7cm]E2.south) node[right] {$y_2$};
    
\node[factor, draw=none] (Edummy1) [right=1cm of E2] {};
\node at ([yshift=.6cm]E2.east-|Edummy1) {$\cdots$};
\node at ([yshift=-.6cm]E2.east-|Edummy1) {$\cdots$};
\draw ([yshift=.6cm]E2.east) |- ([yshift=.6cm]Edummy1.west)
    node[above=-0.05cm,pos=.75] {$s_2$};
\draw ([yshift=-.6cm]E2.east) |- ([yshift=-.6cm]Edummy1.west) 
    node[below,pos=.75] {$\tilde{s}_2$};
\node at (X1-|Edummy1) {$\cdots$};
\node at (Y1-|Edummy1) {$\cdots$};
    
\node[factor] (El) [right=1cm of Edummy1] {$W$};
\draw ([yshift=.6cm]Edummy1.east) |- ([yshift=.6cm]El.west)
    node[above=-0.05cm,pos=.75] {$s_{\ell\!-\!1}$};
\draw ([yshift=-.6cm]Edummy1.east) |- ([yshift=-.6cm]El.west)
    node[below,pos=.75] {$\tilde{s}_{\ell\!-\!1}$};
\node[darksolid] (Xl)[above=.7cm of El] {};\node[above = 0pt of Xl]{$\cx_\ell$};
\draw (Xl) -- (El);
\draw (El.south) -- ([yshift=-.7cm]El.south) node[right] {$y_\ell$};
    
\node[factor] (El2) [right=1cm of El] {$W$};
\draw ([yshift=.6cm]El.east) |- ([yshift=.6cm]El2.west)
    node[above=-0.05cm,pos=.75] {$s_\ell$};
\draw ([yshift=-.6cm]El.east) |- ([yshift=-.6cm]El2.west)
    node[below,pos=.75] {$\tilde{s}_\ell$};
\node[darksolid] (Xl2) [above=.7cm of El2] {};
    \node[above = 0pt of Xl2] {$\cx_{\ell\!+\!1}$};
\draw (Xl2) -- (El2);
\draw (El2.south) -- ([yshift=-.7cm]El2.south) node[right] {$y_{\ell\!+\!1}$};
    
\node[factor, draw=none] (Edummy2) [right=1cm of El2] {};
\node at ([yshift=.6cm]El2.east-|Edummy2) {$\cdots$};
\node at ([yshift=-.6cm]El2.east-|Edummy2) {$\cdots$};
\draw ([yshift=.6cm]El2.east) |- ([yshift=.6cm]Edummy2.west)
    node[above=-0.05cm,pos=.75] {$s_{\ell\!+\!1}$};
\draw ([yshift=-.6cm]El2.east) |- ([yshift=-.6cm]Edummy2.west)
    node[below,pos=.75] {$\tilde{s}_{\ell\!+\!1}$};
\node at (X1-|Edummy2) {$\cdots$};
\node at (Y1-|Edummy2) {$\cdots$};
    
\node[factor] (En) [right=1cm of Edummy2] {$W$};
\draw ([yshift=.6cm]Edummy2.east) |- ([yshift=.6cm]En.west)
    node[above=-0.05cm,pos=.75] {$s_{n\!-\!1}$};
\draw ([yshift=-.6cm]Edummy2.east) |- ([yshift=-.6cm]En.west)
    node[below,pos=.75] {$\tilde{s}_{n\!-\!1}$};
\node[darksolid] (Xn) [above=.7cm of En] {}; \node[above = 0pt of Xn] {$\cx_n$};
\draw (Xn) -- (En);
\draw (En.south) -- ([yshift=-.7cm]En.south) node[right] {$y_n$};
    
\node[sfactor, inner sep=0pt] (ee) [right=.5cm of En] {$=$};
\draw ([yshift=.6cm]En.east) -| (ee.north) node[above=-0.05cm,pos=.25] {$s_n$};
\draw ([yshift=-.6cm]En.east) -| (ee.south) node[below,pos=.25] {$\tilde{s}_n$};
    
\begin{pgfonlayer}{bg}
    \draw[dashed, blue, line width=1.5pt, fill=blue!10]
        ([xshift=-2.2cm,yshift=3.2cm]E1) rectangle
        ([xshift=1.3cm,yshift=-2.8cm]ee);
    \draw[dashed, blue, line width=1.5pt, fill=blue!20]
        ([xshift=-.7cm,yshift=3cm]E1) rectangle
        ([xshift=1.1cm,yshift=-2.6cm]ee);
    \draw[dashed, blue, line width=1.5pt, fill=blue!30]
        ([xshift=-.7cm,yshift=2.8cm]E2) rectangle
        ([xshift=.9cm,yshift=-2.4cm]ee);
    \draw[dashed, blue, line width=1.5pt, fill=blue!40]
    ([xshift=-.7cm,yshift=2.6cm]El) rectangle 
    ([xshift=.7cm,yshift=-2.2cm]ee);
    \draw[dashed, blue, line width=1.5pt, fill=blue!50]
    ([xshift=-.7cm,yshift=2.4cm]El2) rectangle
    ([xshift=.5cm,yshift=-2cm]ee);
    \draw[dashed, blue, line width=1.5pt, fill=blue!60]
    ([xshift=-.7cm,yshift=2.2cm]En) rectangle
    ([xshift=.3cm,yshift=-1.8cm]ee);
\end{pgfonlayer}
\end{tikzpicture}

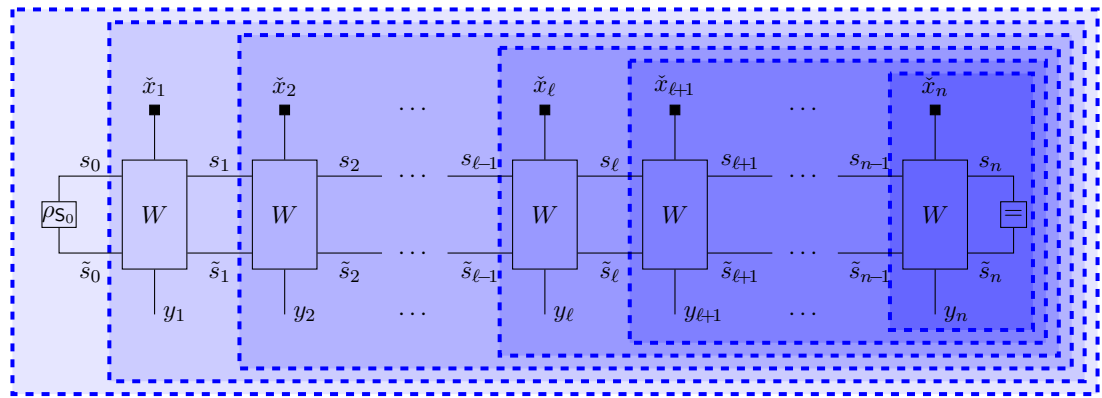
\captionof{figure}{Counterpart of Figure~\ref{fig:FMSC:high:level:2} for QSCs.
Note that every ``closing-the-box'' operation yields a function node representing
    a Kronecker-delta function node, \ie, a degree-two equality function node.}
\label{fig:QFSM:closing:the:box}
\end{minipage}
%*******************************************************************************
\noindent\begin{minipage}[c][.5\textheight]{\textwidth}
\centering
\begin{tikzpicture}[scale=.85,every node/.style={transform shape},
    node/.style={draw=none},
    factor/.style={rectangle, minimum width=1cm, minimum height=.7cm, draw},
    sfactor/.style={rectangle, minimum size=.4cm, draw},
    darksolid/.style={rectangle, minimum size=.15cm, draw, fill = black,
        inner sep=0pt, outer sep = 0pt}]
\node[sfactor,inner sep=0pt] (S) {$\prob_{\rv{S}_0}$};
\node[factor] (E1) [right=.7cm of S] {$W$};
\draw (S) -- (E1) node[above,midway] {$s_0$};
\draw (E1.north) -- ([yshift=.8cm]E1.north) node (X1) [darksolid]{} 
    node[right] {$\cx_1$};
\draw (E1.south) -- ([yshift=-.8cm]E1.south) node (Y1) [darksolid]{} 
    node[right] {$\cy_1$};

\node[factor, draw=none] (Edummy1) [right=1cm of E1] {$\cdots$};
\draw (E1) -- (Edummy1) node[above,midway] {$s_1$};
\node at (X1-|Edummy1) {$\cdots$};
\node at (Y1-|Edummy1) {$\cdots$};

\node[factor] (El0) [right=1cm of Edummy1] {$W$};
\draw (Edummy1) -- (El0) node[above,midway] {$s_{\ell-2}$};
\draw (El0.north) -- ([yshift=.8cm]El0.north) node (Xl0) [darksolid]{} 
    node[right] {$\cx_{\ell\!-\!1}$};
\draw (El0.south) -- ([yshift=-.8cm]El0.south)
    node[darksolid]{} node[right] {$\cy_{\ell\!-\!1}$};
    
\node[factor] (El) [right=1cm of El0] {$W$};
\draw (El0) -- (El) node[above,pos=.65] {$s_{\ell\!-\!1}$};
\draw (El.north) -- ([yshift=.8cm]El.north) node (Xl) [darksolid]{} 
    node[right] {$\cx_\ell$};
\draw (El.south) -- ([yshift=-1.8cm]El.south) node[right] {$y_\ell$};
    
\node[factor] (El2) [right=1cm of El] {$W$};
\draw (El) -- (El2) node[above,midway] {$s_\ell$};
\draw (El2.north) -- ([yshift=.8cm]El2.north) node (Xl2) [darksolid]{} 
    node[right] {$\cx_{\ell\!+\!1}$};
\draw (El2.south) -- ([yshift=-.8cm]El2.south) node[right] {$y_{\ell\!+\!1}$};
	
\node[factor, draw=none] (Edummy2) [right=1cm of El2] {$\cdots$};
\draw (El2) -- (Edummy2) node[above,midway] {$s_{\ell+1}$};
\node at (X1-|Edummy2) {$\cdots$};
\node at (Y1-|Edummy2) {$\cdots$};
    
\node[factor] (En) [right=1cm of Edummy2] {$W$};
\draw (Edummy2) -- (En) node[above,midway] {$s_{n\!-\!1}$};
\draw (En.north) -- ([yshift=.8cm]En.north) node (Xn) [darksolid]{} 
    node[right] {$\cx_n$};
\draw (En.south) -- ([yshift=-.8cm]En.south) node[right] {$y_n$};
    
\node[sfactor, inner sep=0pt] (ee) [right=.5cm of En] {$\mathbf{1}$};
\draw (En) -- (ee) node[above,midway] {$s_n$};
    
\begin{pgfonlayer}{bg}
    \draw[dashed, black, line width=1.5pt, fill=yellow!20]
        ([xshift=-.8cm,yshift=1.6cm]S) rectangle 
        ([xshift=.6cm,yshift=-1.6cm]ee);
    \node[anchor=north west] at ([xshift=-.6cm,yshift=-1.6cm]S|-ee)
        {$\prob_{\rv{Y}_\ell,\rv{Y}_1^{\ell\!-\!1}|\rv{X}_1^\ell}
        (y_\ell,\cvy_1^{\ell\!-\!1}|\cvx_1^\ell)$};
    \draw[dashed, blue, line width=1.5pt, fill=blue!20]
        ([xshift=-.7cm,yshift=1.4cm]El2) rectangle 
        ([xshift=.4cm,yshift=-1.4cm]ee);
    \node[anchor=south east, fill=blue!50] at ([xshift=.4cm,yshift=-1.4cm]ee)
        {$\mathbf{1}$};
    \draw[dashed, red, line width=1.5pt, fill=red!20]
    ([xshift=-.6cm,yshift=1.4cm]S) rectangle 
    ([xshift=.8cm,yshift=-1.4cm]El0);
    %\node[anchor=south west,fill=red!50] at 
    %([xshift=-.6cm,yshift=-1.4cm]S|-El0) 
    %{$\mu_{\rv{S}_{\ell\!-\!1}}^{(\cvx_1^{\ell\!-\!1}\!,\cvy_1^{\ell\!-\!1})}$};
\end{pgfonlayer}
\end{tikzpicture}
\captionof{figure}{Efficient simulation of the channel output at step $\ell$ given the channel input $\cvx_1^n$ and the channel output $\cvy_1^{\ell-1}$ for an FSMC.}
\label{fig:CFSM:channel:simulation:Y}
\end{minipage}
%*******************************************************************************
\noindent\begin{minipage}[c][.5\textheight]{\textwidth}
\centering
\begin{tikzpicture}[scale=.85,every node/.style={transform shape},
    factor/.style={rectangle, minimum width=1cm, minimum height=1.7cm, draw},
    sfactor/.style={rectangle, minimum size=.4cm, draw},
    darksolid/.style={rectangle, minimum size=.15cm, draw,fill = black,
    inner sep=0pt, outer sep = 0pt}]
\node[sfactor,inner sep=0pt] (S) {$\rho_{\system{S}_0}$};
\node[factor] (E1) [right=.7cm of S] {$W$};
\draw (S.north) |- ([yshift=.6cm]E1.west) node[above=-0.05cm,pos=.75] {$s_0$};
\draw (S.south) |- ([yshift=-.6cm]E1.west) node[below,pos=.75] {$\tilde{s}_0$};
\draw (E1.north) -- ([yshift=.8cm]E1.north) node (X1) [darksolid]{} 
    node[right] {$\cx_1$};
\draw (E1.south) -- ([yshift=-.8cm]E1.south) node[darksolid]{}
    node (Y1) [right] {$\cy_1$};
    
\node[factor, draw=none] (Edummy1) [right=1cm of E1] {};
\node at ([yshift=.6cm]E1.east-|Edummy1) {$\cdots$};
\node at ([yshift=-.6cm]E1.east-|Edummy1) {$\cdots$};
\draw ([yshift=.6cm]E1.east) |- ([yshift=.6cm]Edummy1.west)
    node[above=-0.05cm,pos=.75] {$s_1$};
\draw ([yshift=-.6cm]E1.east) |- ([yshift=-.6cm]Edummy1.west)
    node[below,pos=.75] {$\tilde{s}_1$};
\node at (X1-|Edummy1) {$\cdots$};
\node at (Y1-|Edummy1) {$\cdots$};
    
\node[factor] (El0) [right=1cm of Edummy1] {$W$};
\draw ([yshift=.6cm]Edummy1.east) |- ([yshift=.6cm]El0.west)
    node[above=-0.05cm,pos=.75] {$s_{\ell\!-\!2}$};
\draw ([yshift=-.6cm]Edummy1.east) |- ([yshift=-.6cm]El0.west)
    node[below,pos=.75] {$\tilde{s}_{\ell\!-\!2}$};
\draw (El0.north) -- ([yshift=.8cm]El0.north) node (Xl0) [darksolid]{} 
    node[right] {$\cx_{\ell\!-\!1}$};
\draw (El0.south) -- ([yshift=-.8cm]El0.south) node[darksolid]{}
    node[right] {$\cy_{\ell\!-\!1}$};
    
\node[factor] (El) [right=1cm of El0] {$W$};
\draw ([yshift=.6cm]El0.east) |- ([yshift=.6cm]El.west)
    node[above=-0.05cm,pos=.85] {$s_{\ell\!-\!1}$};
\draw ([yshift=-.6cm]El0.east) |- ([yshift=-.6cm]El.west)
    node[below,pos=.85] {$\tilde{s}_{\ell\!-\!1}$};
\draw (El.north) -- ([yshift=.8cm]El.north) node (Xl) [darksolid]{} 
    node[right] {$\cx_\ell$};
\draw (El.south) -- ([yshift=-2.2cm]El.south) node[right] {$y_\ell$};
    
\node[factor] (El2) [right=1cm of El] {$W$};
\draw ([yshift=.6cm]El.east) |- ([yshift=.6cm]El2.west)
    node[above=-0.05cm,pos=.75] {$s_\ell$};
\draw ([yshift=-.6cm]El.east) |- ([yshift=-.6cm]El2.west)
    node[below,pos=.75] {$\tilde{s}_\ell$};
\draw (El2.north) -- ([yshift=.8cm]El2.north) node (Xl2) [darksolid]{} 
    node[right] {$\cx_{\ell\!+\!1}$};
\draw (El2.south) -- ([yshift=-.8cm]El2.south) node[right] {$y_{\ell\!+\!1}$};
    
\node[factor, draw=none] (Edummy2) [right=1cm of El2] {};
\node at ([yshift=.6cm]El2.east-|Edummy2) {$\cdots$};
\node at ([yshift=-.6cm]El2.east-|Edummy2) {$\cdots$};
\draw ([yshift=.6cm]El2.east) |- ([yshift=.6cm]Edummy2.west)
    node[above=-0.05cm,pos=.75] {$s_{\ell\!+\!1}$};
\draw ([yshift=-.6cm]El2.east) |- ([yshift=-.6cm]Edummy2.west)
    node[below,pos=.75] {$\tilde{s}_{\ell\!+\!1}$};
\node at (X1-|Edummy2) {$\cdots$};
\node at (Y1-|Edummy2) {$\cdots$};
    
\node[factor] (En) [right=1cm of Edummy2] {$W$};
\draw ([yshift=.6cm]Edummy2.east) |- ([yshift=.6cm]En.west)
    node[above=-0.05cm,pos=.75] {$s_{n\!-\!1}$};
\draw ([yshift=-.6cm]Edummy2.east) |- ([yshift=-.6cm]En.west)
    node[below,pos=.75] {$\tilde{s}_{n\!-\!1}$};
\draw (En.north) -- ([yshift=.8cm]En.north) node (Xn) [darksolid]{} 
    node[right] {$\cx_n$};
\draw (En.south) -- ([yshift=-.8cm]En.south) node[right] {$y_n$};
    
\node[sfactor, inner sep=0pt] (ee) [right=.5cm of En] {$=$};
\draw ([yshift=.6cm]En.east) -| (ee.north) node[above=-0.05cm,pos=.25] {$s_n$};
\draw ([yshift=-.6cm]En.east) -| (ee.south) node[below,pos=.25] {$\tilde{s}_n$};
    
\begin{pgfonlayer}{bg}
    \draw[dashed, black, line width=1.5pt, fill=yellow!20]
        ([xshift=-.8cm,yshift=2.4cm]S) rectangle 
        ([xshift=.8cm,yshift=-2.3cm]ee);
    \node[anchor=north west] at ([xshift=-.8cm,yshift=-2.3cm]S|-ee)
        {$\prob_{\rv{Y}_\ell,\rv{Y}_1^{\ell\!-\!1}|\rv{X}_1^\ell,\system{S}}
        (y_\ell,\cvy_1^{\ell\!-\!1}|\cvx_1^\ell;\rho_{\system{S}_0})$};
    \draw[dashed, blue, line width=1.5pt, fill=blue!20]
        ([xshift=-.7cm,yshift=2.2cm]El2) rectangle
        ([xshift=.6cm,yshift=-2.1cm]ee);
    \node[anchor=south east, fill=blue!50] at ([xshift=.6cm,yshift=-2.1cm]ee)
        {$\delta$};
    \draw[dashed, red, line width=1.5pt, fill=red!20]
        ([xshift=-.6cm,yshift=2.2cm]S) rectangle 
        ([xshift=.8cm,yshift=-2.1cm]El0);
\end{pgfonlayer}
\end{tikzpicture}
\captionof{figure}{Efficient simulation of the channel output at step $\ell$ given the channel input $\cvx_1^n$ and the channel output $\cvy_1^{\ell-1}$ for a QSC.}
\label{fig:QFSM:channel:simulation:Y}
\end{minipage}
%*******************************************************************************
\noindent\begin{minipage}[c][.5\textheight]{\textwidth}
\centering
\begin{tikzpicture}[scale=.85,every node/.style={transform shape},
    factor/.style={rectangle, minimum width=1cm, minimum height=.7cm, draw},
    sfactor/.style={rectangle, minimum size=.5cm, draw},
    darksolid/.style={rectangle, minimum size=.15cm, draw,fill = black,
    inner sep=0pt, outer sep = 0pt},
    label/.style={red,anchor=north east,xshift = .1cm}]
\node[sfactor,inner sep=0pt] (S) {$\prob_{\rv{S}_0}$};
\node[factor] (E1) [right=.7cm of S] {$W$};
\draw (S) -- (E1) node[above,midway] {$s_0$};
\node[sfactor] (X1) [above=.8cm of E1] {$Q$};
\draw (X1) -- (E1) node[right, midway] {$x_1$};
\draw (E1.south) -- ([yshift=-.8cm]E1.south) node (Y1) [darksolid]{}
    node[right] {$\cy_1$};
    
\node[factor] (E2) [right=1cm of E1] {$W$};
\draw (E1) -- (E2) node[above,midway] {$s_1$};
\node[sfactor] (X2) [above=.8cm of E2] {$Q$};
\draw (X2) -- (E2) node[right, midway] {$x_2$};
\draw (E2.south) -- ([yshift=-.8cm]E2.south) node[darksolid]{}
    node[right] {$\cy_2$};
    
\node[factor, draw=none] (Edummy1) [right=1cm of E2] {$\cdots$};
\draw (E2) -- (Edummy1) node[above,midway] {$s_2$};
\node at (X1-|Edummy1) {$\cdots$};
\node at (Y1-|Edummy1) {$\cdots$};
    
\node[factor] (El) [right=1cm of Edummy1] {$W$};
\draw (Edummy1) -- (El) node[above,midway] {$s_{\ell-1}$};
\node[sfactor] (Xl) [above=.8cm of El] {$Q$};
\draw (Xl) -- (El) node[right, midway] {$x_\ell$};
\draw (El.south) -- ([yshift=-.8cm]El.south) node[darksolid]{}
    node[right] {$\cy_\ell$};
    
%\node[factor] (El2) [right=1cm of El] {$W$};
%\draw (El) -- (El2) node[above,midway] {$s_\ell$};
%\node[sfactor] (Xl2) [above=.8cm of El2] {$Q$};
%\draw (Xl2) -- (El2) node[right, midway] {$x_{\ell\!+\!1}$};
%\draw (El2.south) -- ([yshift=-.8cm]El2.south) node[darksolid]{}
%    node[right] {$\cy_{\ell\!+\!1}$};
    
\node[factor, draw=none] (Edummy2) [right=1cm of El] {$\cdots$};
\draw (El) -- (Edummy2) node[above,pos=.7] {$s_{\ell\!+\!1}$};
\node at (X1-|Edummy2) {$\cdots$};
\node at (Y1-|Edummy2) {$\cdots$};
    
\node[factor] (En) [right=1cm of Edummy2] {$W$};
\draw (Edummy2) -- (En) node[above,midway] {$s_{n\!-\!1}$};
\node[sfactor] (Xn) [above=.8cm of En] {$Q$};
\draw (Xn) -- (En) node[right, midway] {$x_n$};
\draw (En.south) -- ([yshift=-.8cm]En.south) node[darksolid]{}
    node[right] {$\cy_{n}$};
    
\node[sfactor, inner sep=0pt] (ee) [right=1cm of En] {$\mathbf{1}$};
\draw (En) -- (ee) node[above,midway] {$s_n$};
    
\begin{pgfonlayer}{bg}
    \draw[dashed, red, line width=1.5pt, fill=red!15]
        ([xshift=-1.2cm,yshift=2.5cm]S) rectangle
        ([xshift=.7cm,yshift=-3.2cm]En);
    \node[label] at ([xshift=.7cm,yshift=-3.2cm]En) {$\muY_{n}$};
%    \draw[dashed, red, line width=1.5pt, fill=red!20]
%        ([xshift=-1.2cm,yshift=2.5cm]S) rectangle
%        ([xshift=.8cm,yshift=-3.2cm]El2);
%    \node[label] at ([xshift=.8cm,yshift=-3.2cm]El2) {$\muY_{\ell+1}$};
    \draw[dashed, red, line width=1.5pt, fill=red!30]
        ([xshift=-1cm,yshift=2.3cm]S) rectangle
        ([xshift=.7cm,yshift=-2.6cm]El);
    \node[label] at ([xshift=.7cm,yshift=-2.6cm]El) {$\muY_{\ell}$};
    \draw[dashed, red, line width=1.5pt, fill=red!45]
        ([xshift=-.8cm,yshift=2.1cm]S) rectangle
        ([xshift=.7cm,yshift=-2cm]E2);
    \node[label] at ([xshift=.7cm,yshift=-2cm]E2) {$\muY_2$};
    \draw[dashed, red, line width=1.5pt, fill=red!60]
        ([xshift=-.6cm,yshift=1.9cm]S) rectangle
        ([xshift=.7cm,yshift=-1.4cm]E1);
    \node[label] at ([xshift=.7cm,yshift=-1.4cm]E1) {$\muY_1$};
\end{pgfonlayer}
\end{tikzpicture}
\captionof{figure}{The iterative computation of $\muY_{\ell}$ as in~\eqref{eq:recursive:state:metric:Y:1} can be understood as a sequence of ``closing-the-box'' operations as shown above.}
\label{fig:CFSM:estimate:hY}
\end{minipage}
%*******************************************************************************
\noindent\begin{minipage}[c][.5\textheight]{\textwidth}
\centering
\begin{tikzpicture}[scale=.85,every node/.style={transform shape},
    factor/.style={rectangle, minimum width=1cm, minimum height=1.7cm, draw},
    sfactor/.style={rectangle, minimum size=.4cm, draw},
    darksolid/.style={rectangle, minimum size=.15cm, draw, fill = black,
    inner sep=0pt, outer sep = 0pt},
    label/.style={red, anchor=north east, xshift = .1cm}]
\node[sfactor,inner sep=0pt] (S) {$\rho_{\system{S}_0}$};
\node[factor] (E1) [right=.7cm of S] {$W$};
\draw (S.north) |- ([yshift=.6cm]E1.west) node[above=-0.05cm,pos=.75] {$s_0$};
\draw (S.south) |- ([yshift=-.6cm]E1.west) node[below,pos=.75] {$\tilde{s}_0$};
\node[sfactor] (X1) [above=.7cm of E1] {$Q$};
\draw (X1) -- (E1) node[right, midway] {$x_1$};
\draw (E1.south) -- ([yshift=-.7cm]E1.south) node[darksolid]{}
    node (Y1) [right] {$\cy_1$};
    
\node[factor] (E2) [right=1cm of E1] {$W$};
\draw ([yshift=.6cm]E1.east) |- ([yshift=.6cm]E2.west)
    node[above=-0.05cm,pos=.75] {$s_1$};
\draw ([yshift=-.6cm]E1.east) |- ([yshift=-.6cm]E2.west)
    node[below,pos=.75] {$\tilde{s}_1$};
\node[sfactor] (X2) [above=.7cm of E2] {$Q$};
\draw (X2) -- (E2) node[right, midway] {$x_2$};
\draw (E2.south) -- ([yshift=-.7cm]E2.south) node[darksolid]{}
    node[right] {$\cy_2$};
    
\node[factor, draw=none] (Edummy1) [right=1cm of E2] {};
\node at ([yshift=.6cm]E2.east-|Edummy1) {$\cdots$};
\node at ([yshift=-.6cm]E2.east-|Edummy1) {$\cdots$};
\draw ([yshift=.6cm]E2.east) |- ([yshift=.6cm]Edummy1.west) node[above=-0.05cm,pos=.75] {$s_2$};
\draw ([yshift=-.6cm]E2.east) |- ([yshift=-.6cm]Edummy1.west)
    node[below,pos=.75] {$\tilde{s}_2$};
\node at (X1-|Edummy1) {$\cdots$};
\node at (Y1-|Edummy1) {$\cdots$};
    
\node[factor] (El) [right=1cm of Edummy1] {$W$};
\draw ([yshift=.6cm]Edummy1.east) |- ([yshift=.6cm]El.west)
    node[above=-0.05cm,pos=.75] {$s_{\ell-1}$};
\draw ([yshift=-.6cm]Edummy1.east) |- ([yshift=-.6cm]El.west)
    node[below,pos=.75] {$\tilde{s}_{\ell-1}$};
\node[sfactor] (Xl) [above=.7cm of El] {$Q$};
\draw (Xl) -- (El) node[right, midway] {$x_\ell$};
\draw (El.south) -- ([yshift=-.7cm]El.south) node[darksolid]{}
    node[right] {$\cy_\ell$};
    
%\node[factor] (El2) [right=1cm of El] {$W$};
%\draw ([yshift=.6cm]El.east) |- ([yshift=.6cm]El2.west) node[above=-0.05cm,pos=.75] {$s_\ell$};
%\draw ([yshift=-.6cm]El.east) |- ([yshift=-.6cm]El2.west)
%    node[below,pos=.75] {$\tilde{s}_\ell$};
%\node[sfactor] (Xl2) [above=.7cm of El2] {$Q$};
%\draw (Xl2) -- (El2) node[right, midway] {$x_{\ell\!+\!1}$};
%\draw (El2.south) -- ([yshift=-.7cm]El2.south) node[darksolid]{} node[right] {$\cy_{\ell\!+\!1}$};
    
\node[factor, draw=none] (Edummy2) [right=1cm of El] {};
\node at ([yshift=.6cm]El.east-|Edummy2) {$\cdots$};
\node at ([yshift=-.6cm]El.east-|Edummy2) {$\cdots$};
\draw ([yshift=.6cm]El.east) |- ([yshift=.6cm]Edummy2.west)
    node[above=-0.05cm,pos=.85] {$s_{\ell\!+\!1}$};
\draw ([yshift=-.6cm]El.east) |- ([yshift=-.6cm]Edummy2.west)
    node[below,pos=.85] {$\tilde{s}_{\ell\!+\!1}$};
\node at (X1-|Edummy2) {$\cdots$};
\node at (Y1-|Edummy2) {$\cdots$};
    
\node[factor] (En) [right=1cm of Edummy2] {$W$};
\draw ([yshift=.6cm]Edummy2.east) |- ([yshift=.6cm]En.west)
    node[above=-0.05cm,pos=.75] {$s_{n\!-\!1}$};
\draw ([yshift=-.6cm]Edummy2.east) |- ([yshift=-.6cm]En.west)
    node[below,pos=.75] {$\tilde{s}_{n\!-\!1}$};
\node[sfactor] (Xn) [above=.7cm of En] {$Q$};
\draw (Xn) -- (En) node[right, midway] {$x_n$};
\draw (En.south) -- ([yshift=-.7cm]En.south) node[darksolid]{}
    node[right] {$\cy_n$};
    
\node[sfactor, inner sep=0pt] (ee) [right=1cm of En] {$=$};
\draw ([yshift=.6cm]En.east) -| (ee.north) node[above=-0.05cm,pos=.25] {$s_n$};
\draw ([yshift=-.6cm]En.east) -| (ee.south) node[below,pos=.25] {$\tilde{s}_n$};
    
\begin{pgfonlayer}{bg}
    \draw[dashed, red, line width=1.5pt, fill=red!15]
        ([xshift=-1.2cm,yshift=2.9cm]S) rectangle
        ([xshift=.7cm,yshift=-3.6cm]En);
    \node[label] at ([xshift=.7cm,yshift=-3.6cm]En) {$\sigmaY_n$};
%    \draw[dashed, red, line width=1.5pt, fill=red!20]
%        ([xshift=-1.2cm,yshift=2.9cm]S) rectangle
%        ([xshift=.8cm,yshift=-3.6cm]El2);
%    \node[label] at ([xshift=.7cm,yshift=-3.6cm]El2) {$\sigmaY_{\ell+1}$};
    \draw[dashed, red, line width=1.5pt, fill=red!30]
        ([xshift=-1cm,yshift=2.7cm]S) rectangle
        ([xshift=.7cm,yshift=-3cm]El);
    \node[label] at ([xshift=.7cm,yshift=-3cm]El) {$\sigmaY_{\ell}$};
    \draw[dashed, red, line width=1.5pt, fill=red!45]
        ([xshift=-.8cm,yshift=2.5cm]S) rectangle
        ([xshift=.7cm,yshift=-2.4cm]E2);
    \node[label] at ([xshift=.7cm,yshift=-2.4cm]E2) {$\sigmaY_2$};
    \draw[dashed, red, line width=1.5pt, fill=red!60]
        ([xshift=-.6cm,yshift=2.3cm]S) rectangle
        ([xshift=.7cm,yshift=-1.8cm]E1);
    \node[label] at ([xshift=.7cm,yshift=-1.8cm]E1) {$\sigmaY_1$};
\end{pgfonlayer}
\end{tikzpicture}
\captionof{figure}{The iterative computation of $\sigmaY_{\ell}$ as in~\eqref{eq:recursive:quantum:state:metric:Y:1} can be understood as a sequence of ``closing-the-box'' operations as shown above.}
\label{fig:QFSM:estimate:hY}
\end{minipage}
%*******************************************************************************
\noindent\begin{minipage}[c][.5\textheight]{\textwidth}
\centering
\begin{tikzpicture}[scale=.85,every node/.style={transform shape},
    factor/.style={rectangle, minimum width=1cm, minimum height=.7cm, draw},
    sfactor/.style={rectangle, minimum size=.4cm, draw},
    darksolid/.style={rectangle, minimum size=.15cm, draw, fill = black,
    inner sep=0pt, outer sep = 0pt},
    label/.style={red, anchor=north east, xshift = .1cm}]
\node[sfactor,inner sep=0pt] (S) {$\prob_{\rv{S}_0}$};
\node[factor] (E1) [right=.7cm of S] {$W$};
\draw (S) -- (E1) node[above,midway] {$s_0$};
\node[darksolid] (X1) [above=.3cm of E1.north,anchor = south] {};
\draw (X1) -- (E1) node[right, pos = 0] {$\cx_1$};
\draw (X1) -- ([yshift=.3cm]X1.north) 
    node (pX1) [sfactor, anchor = south, pos=1] {$Q$};
\draw (E1.south) -- ([yshift=-.8cm]E1.south) node (Y1) [darksolid]{} 
    node[right] {$\cy_1$};
    
\node[factor] (E2) [right=1cm of E1] {$W$};
\draw (E1) -- (E2) node[above,midway] {$s_1$};
\node[darksolid] (X2) [above=.3cm of E2.north,anchor = south] {};
\draw (X2) -- (E2) node[right, pos = 0] {$\cx_2$};
\draw (X2) -- ([yshift=.3cm]X2.north)
    node[sfactor, anchor = south, pos=1] {$Q$};
\draw (E2.south) -- ([yshift=-.8cm]E2.south) 
    node[darksolid]{} node[right] {$\cy_2$};
    
\node[factor, draw=none] (Edummy1) [right=1cm of E2] {$\cdots$};
\draw (E2) -- (Edummy1) node[above,midway] {$s_2$};
\node at (pX1-|Edummy1) {$\cdots$};
\node at (Y1-|Edummy1) {$\cdots$};
    
\node[factor] (El) [right=1cm of Edummy1] {$W$};
\draw (Edummy1) -- (El) node[above,midway] {$s_{\ell-1}$};
\node[darksolid] (Xl) [above=.3cm of El.north,anchor = south] {};
\draw (Xl) -- (El) node[right, pos = 0] {$\cx_\ell$};
\draw (Xl) -- ([yshift=.3cm]Xl.north)
    node[sfactor, anchor = south, pos=1] {$Q$};
\draw (El.south) -- ([yshift=-.8cm]El.south) node[darksolid]{}
    node[right] {$\cy_\ell$};
    
%\node[factor] (El2) [right=1cm of El] {$W$};
%\draw (El) -- (El2) node[above,midway] {$s_{\ell}$};
%\node[darksolid] (Xl2) [above=.3cm of El2.north,anchor = south] {};
%\draw (Xl2) -- (El2) node[right, pos = 0] {$\cx_{\ell\!+\!1}$};
%\draw (Xl2) -- ([yshift=.3cm]Xl2.north)
%    node[sfactor, anchor = south, pos=1] {$Q$};
%\draw (El2.south) -- ([yshift=-.8cm]El2.south) node[darksolid]{}
%    node[right] {$\cy_{\ell\!+\!1}$};
    
\node[factor, draw=none] (Edummy2) [right=1cm of El] {$\cdots$};
\draw (El) -- (Edummy2) node[above,pos=.7] {$s_{\ell\!+\!1}$};
\node at (pX1-|Edummy2) {$\cdots$};
\node at (Y1-|Edummy2) {$\cdots$};
    
\node[factor] (En) [right=1cm of Edummy2] {$W$};
\draw (Edummy2) -- (En) node[above,midway] {$s_{n\!-\!1}$};
\node[darksolid] (Xn) [above=.3cm of En.north,anchor = south] {};
\draw (Xn) -- (En) node[right, pos = 0] {$\cx_n$};
\draw (Xn) -- ([yshift=.3cm]Xn.north)
    node[sfactor, anchor = south, pos=1] {$Q$};
\draw (En.south) -- ([yshift=-.8cm]En.south) node[darksolid]{}
    node[right] {$\cy_{n}$};
    
\node[sfactor, inner sep=0pt] (ee) [right=1cm of En] {$\mathbf{1}$};
\draw (En) -- (ee) node[above,midway] {$s_n$};
    
\begin{pgfonlayer}{bg}
    \draw[dashed, red, line width=1.5pt, fill=red!15]
        ([xshift=-1.2cm,yshift=2.5cm]S) rectangle
        ([xshift=.7cm,yshift=-3.2cm]En);
    \node[label] at ([xshift=.7cm,yshift=-3.2cm]En) {$\muXY_{n}$};
    \draw[dashed, red, line width=1.5pt, fill=red!30]
        ([xshift=-1cm,yshift=2.3cm]S) rectangle
        ([xshift=.7cm,yshift=-2.6cm]El);
    \node[label] at ([xshift=.7cm,yshift=-2.6cm]El) {$\muXY_{\ell}$};
    \draw[dashed, red, line width=1.5pt, fill=red!45]
        ([xshift=-.8cm,yshift=2.1cm]S) rectangle
        ([xshift=.7cm,yshift=-2cm]E2);
    \node[label] at ([xshift=.7cm,yshift=-2cm]E2) {$\muXY_2$};
    \draw[dashed, red, line width=1.5pt, fill=red!60]
        ([xshift=-.6cm,yshift=1.9cm]S) rectangle
        ([xshift=.7cm,yshift=-1.4cm]E1);
    \node[label] at ([xshift=.7cm,yshift=-1.4cm]E1) {$\muXY_1$};
\end{pgfonlayer}
\end{tikzpicture}
\captionof{figure}{The iterative computation of $\muXY_{\ell}$ can be understood as a sequence of ``closing-the-box'' operations as shown above.}
\label{fig:CFSM:estimate:hXY}
\end{minipage}
%*******************************************************************************
\noindent\begin{minipage}[c][.49\textheight]{\textwidth}
\centering
\begin{tikzpicture}[scale=.85,every node/.style={transform shape},
    factor/.style={rectangle, minimum width=1cm, minimum height=1.7cm, draw},
    sfactor/.style={rectangle, minimum size=.4cm, draw},
    darksolid/.style={rectangle, minimum size=.15cm, draw, fill = black,
    inner sep=0pt, outer sep = 0pt},
    label/.style={red, anchor=north east, xshift = .1cm}]
\node[sfactor,inner sep=0pt] (S) {$\rho_{\system{S}_0}$};
\node[factor] (E1) [right=.7cm of S] {$W$};
\draw (S.north) |- ([yshift=.6cm]E1.west) node[above=-0.05cm,pos=.75] {$s_0$};
\draw (S.south) |- ([yshift=-.6cm]E1.west) node[below,pos=.75] {$\tilde{s}_0$};
\node[darksolid] (X1) [above=.3cm of E1.north,anchor = south] {};
\draw (X1) -- (E1) node[right, pos = 0] {$\cx_1$};
\draw (X1) -- ([yshift=.3cm]X1.north)
    node (pX1) [sfactor, anchor = south, pos=1] {$Q$};
\draw (E1.south) -- ([yshift=-.7cm]E1.south) node[darksolid]{}
    node (Y1) [right] {$\cy_1$};
    
\node[factor] (E2) [right=1cm of E1] {$W$};
\draw ([yshift=.6cm]E1.east) |- ([yshift=.6cm]E2.west)
    node[above=-0.05cm,pos=.75] {$s_1$};
\draw ([yshift=-.6cm]E1.east) |- ([yshift=-.6cm]E2.west)
    node[below,pos=.75] {$\tilde{s}_1$};
\node[darksolid] (X2) [above=.3cm of E2.north,anchor = south] {};
\draw (X2) -- (E2) node[right, pos = 0] {$\cx_2$};
\draw (X2) -- ([yshift=.3cm]X2.north) node[sfactor, anchor = south, pos=1] {$Q$};
\draw (E2.south) -- ([yshift=-.7cm]E2.south) node[darksolid]{} 
    node[right] {$\cy_2$};
    
\node[factor, draw=none] (Edummy1) [right=1cm of E2] {};
\node at ([yshift=.6cm]E2.east-|Edummy1) {$\cdots$};
\node at ([yshift=-.6cm]E2.east-|Edummy1) {$\cdots$};
\draw ([yshift=.6cm]E2.east) |- ([yshift=.6cm]Edummy1.west)
    node[above=-0.05cm,pos=.75] {$s_2$};
\draw ([yshift=-.6cm]E2.east) |- ([yshift=-.6cm]Edummy1.west)
    node[below,pos=.75] {$\tilde{s}_2$};
\node at (pX1-|Edummy1) {$\cdots$};
\node at (Y1-|Edummy1) {$\cdots$};
    
\node[factor] (El) [right=1cm of Edummy1] {$W$};
\draw ([yshift=.6cm]Edummy1.east) |- ([yshift=.6cm]El.west)
    node[above=-0.05cm,pos=.75] {$s_{\ell-1}$};
\draw ([yshift=-.6cm]Edummy1.east) |- ([yshift=-.6cm]El.west)
    node[below,pos=.75] {$\tilde{s}_{\ell-1}$};
\node[darksolid] (Xl) [above=.3cm of El.north,anchor = south] {};
\draw (Xl) -- (El) node[right, pos = 0] {$\cx_\ell$};
\draw (Xl) -- ([yshift=.3cm]Xl.north) node[sfactor, anchor = south, pos=1] {$Q$};
\draw (El.south) -- ([yshift=-.7cm]El.south) node[darksolid]{}
    node[right] {$\cy_\ell$};

\node[factor, draw=none] (Edummy2) [right=1cm of El] {};
\node at ([yshift=.6cm]El.east-|Edummy2) {$\cdots$};
\node at ([yshift=-.6cm]El.east-|Edummy2) {$\cdots$};
\draw ([yshift=.6cm]El.east) |- ([yshift=.6cm]Edummy2.west)
    node[above=-0.05cm,pos=.85] {$s_{\ell\!+\!1}$};
\draw ([yshift=-.6cm]El.east) |- ([yshift=-.6cm]Edummy2.west)
    node[below,pos=.85] {$\tilde{s}_{\!\ell+\!1}$};
\node at (pX1-|Edummy2) {$\cdots$};
\node at (Y1-|Edummy2) {$\cdots$};

\node[factor] (En) [right=1cm of Edummy2] {$W$};
\draw ([yshift=.6cm]Edummy2.east) |- ([yshift=.6cm]En.west)
    node[above=-0.05cm,pos=.75] {$\tilde{s}_{n\!-\!1}$};
\draw ([yshift=-.6cm]Edummy2.east) |- ([yshift=-.6cm]En.west)
    node[below,pos=.75] {$\tilde{s}_{n\!-\!1}$};
\node[darksolid] (Xn) [above=.3cm of En.north,anchor = south] {};
\draw (Xn) -- (En) node[right, pos = 0] {$\cx_n$};
\draw (Xn) -- ([yshift=.3cm]Xn.north)
    node[sfactor, anchor = south, pos=1] {$Q$};
\draw (En.south) -- ([yshift=-.7cm]En.south) node[darksolid]{}
    node[right] {$\cy_n$};
    
\node[sfactor, inner sep=0pt] (ee) [right=1cm of En] {$=$};
\draw ([yshift=.6cm]En.east) -| (ee.north) node[above=-0.05cm,pos=.25] {$s_n$};
\draw ([yshift=-.6cm]En.east) -| (ee.south) node[below,pos=.25] {$\tilde{s}_n$};
    
\begin{pgfonlayer}{bg}
    \draw[dashed, red, line width=1.5pt, fill=red!15]
        ([xshift=-1.2cm,yshift=3cm]S) rectangle
        ([xshift=.7cm,yshift=-3.6cm]En);
    \node[label] at ([xshift=.7cm,yshift=-3.6cm]En) {$\sigmaXY_n$};
    \draw[dashed, red, line width=1.5pt, fill=red!30]
        ([xshift=-1cm,yshift=2.8cm]S) rectangle
        ([xshift=.7cm,yshift=-3cm]El);
    \node[label] at ([xshift=.7cm,yshift=-3cm]El) {$\sigmaXY_{\ell}$};
    \draw[dashed, red, line width=1.5pt, fill=red!45]
        ([xshift=-.8cm,yshift=2.6cm]S) rectangle
        ([xshift=.7cm,yshift=-2.4cm]E2);
    \node[label] at ([xshift=.7cm,yshift=-2.4cm]E2) {$\sigmaXY_2$};
    \draw[dashed, red, line width=1.5pt, fill=red!60]
        ([xshift=-.6cm,yshift=2.4cm]S) rectangle
        ([xshift=.7cm,yshift=-1.8cm]E1);
    \node[label] at ([xshift=.7cm,yshift=-1.8cm]E1) {$\sigmaXY_1$};
\end{pgfonlayer}
\end{tikzpicture}
\captionof{figure}{The iterative computation of $\sigmaXY_{\ell}$ as in~\eqref{eq:def:quantum:channel:state:metric:XY:1} can be understood as a sequence of ``closing-the-box'' operations as shown above.}
\label{fig:QFSM:estimate:hXY}
\end{minipage}
\end{appendices}
% bibliography******************************************************************
\addcontentsline{toc}{chapter}{Bibliography}
\printbibliography
% Index*************************************************************************
\printindex{}
\addcontentsline{toc}{chapter}{Index}
\end{document}